\documentclass[leqno,a4paper,12pt]{article}

\usepackage{latexsym}
\usepackage{amsmath}
\usepackage{amsthm}
\usepackage{amssymb}
\usepackage{enumerate}
\usepackage{multicol}
\usepackage{graphicx}

\usepackage{pict2e}

\allowdisplaybreaks[4]

\theoremstyle{plain}
\newtheorem{theorem}{Theorem}[section]
\newtheorem{proposition}[theorem]{Proposition}
\newtheorem{lemma}[theorem]{Lemma}
\newtheorem{corollary}[theorem]{Corollary}

\theoremstyle{definition}

\newtheorem{example}[theorem]{Example}
\newtheorem{remark}[theorem]{Remark}
\numberwithin{equation}{section}

\def\ba{{\mathbf{a}}}
\def\sa{{\mathsf{a}}}
\def\sA{{\mathsf{A}}}

\def\bb{{\mathbf{b}}}

\def\cB{{\mathcal{B}}}

\def\C{{\mathbb{C}}}
\def\cC{{\mathcal{C}}}

\def\sc{{\mathsf{c}}}

\def\be{{\mathbf{e}}}

\def\sF{{\mathsf{F}}}

\def\cG{{\mathcal{G}}}

\def\cH{{\mathcal{H}}}
\def\sH{{\mathsf{H}}}

\def\bk{{\mathbf{k}}}

\def\cM{{\mathcal{M}}}

\def\N{{\mathbb{N}}}

\def\sn{{\mathsf{n}}}

\def\cO{{\mathcal{O}}}

\def\bp{{\mathbf{p}}}

\def\bq{{\mathbf{q}}}
\def\cQ{{\mathcal{Q}}}

\def\R{{\mathbb{R}}}
\def\cR{{\mathcal{R}}}

\def\sr{{\mathsf{r}}}

\def\S{{\mathbb{S}}}
\def\cS{{\mathcal{S}}}

\def\sS{{\mathsf{S}}}
\def\sfs{{\mathsf{s}}}

\def\T{{\mathbb{T}}}

\def\cU{{\mathcal{U}}}
\def\bu{{\mathbf{u}}}

\def\cV{{\mathcal{V}}}
\def\sV{{\mathsf{V}}}
\def\bv{{\mathbf{v}}}
\def\hbv{{\hat{\mathbf{v}}}}

\def\bW{{\mathbf{W}}}
\def\cW{{\mathcal{W}}}
\def\sW{{\mathsf{W}}}

\def\bx{{\mathbf{x}}}
\def\bX{{\mathbf{X}}}

\def\hbx{{\hat{\mathbf{x}}}}

\def\by{{\mathbf{y}}}
\def\bY{{\mathbf{Y}}}

\def\hby{{\hat{\mathbf{y}}}}

\def\Z{{\mathbb{Z}}}

\def\bZ{{\mathbf{Z}}}

\def\ophi{{\overline{\phi}}}

\def\opsi{{\overline{\psi}}}

\def\oeta{{\overline{\eta}}}

\def\oxi{{\overline{\xi}}}

\def\spin{{\{\uparrow,\downarrow\}}}
\def\ua{{\uparrow}}
\def\da{{\downarrow}}

\def\la{{\lambda}}
\def\bla{{\mbox{\boldmath$\lambda$}}}

\def\O{{\Omega}}
\def\o{{\omega}}

\def\eps{{\varepsilon}}

\def\g{{\gamma}}

\def\G{{\Gamma}}
\def\s{{\sigma}}

\def\hrho{\hat{\rho}}
\def\orho{\overline{\rho}}
\def\heta{\hat{\eta}}
\def\oeta{\overline{\eta}}
\def\D{{\Delta}}

\def\<{{\langle}}
\def\>{{\rangle}}

\def\Tr{\mathop{\mathrm{Tr}}\nolimits}

\def\Mat{\mathop{\mathrm{Mat}}}
\def\Map{\mathop{\mathrm{Map}}}

\def\sgn{\mathop{\mathrm{sgn}}\nolimits}

\def\Re{\mathop{\mathrm{Re}}}

\def\b0{{\mathbf{0}}}

\def\frah{{\left(\frac{1}{h}\right)}}

\def\0betah{{[0,\beta)_h}}

\begin{document}

\title{Superconducting phase in the BCS model with imaginary magnetic
field. II. \\
Multi-scale infrared analysis}

\author{Yohei Kashima
\medskip\\
Center for Mathematical Modeling and Data Science, \\
Osaka University, \\
Toyonaka, Osaka 560-8531, Japan}

\date{}

\maketitle

\begin{abstract} We analyze the reduced BCS model with an imaginary
 magnetic field in a large domain of the temperature and the imaginary
 magnetic field. The magnitude of the attractive reduced BCS interaction
 is fixed to be small but independent of the temperature and the
 imaginary magnetic field unless the temperature is high. 
We impose a series of conditions on the free dispersion
 relation. These conditions are typically satisfied by free electron 
models with degenerate Fermi surface. For example, our theory applies to the model with nearest-neighbor hopping on 3 or 4-dimensional (hyper-)cubic lattice
 having degenerate free Fermi surface or the model with nearest-neighbor
 hopping on the honeycomb lattice with zero chemical potential. We prove
 that a spontaneous $U(1)$-symmetry breaking (SSB) and an off-diagonal 
 long range order (ODLRO) occur in many areas of the parameter
 space. The SSB and the ODLRO are proved to occur in low temperatures
 arbitrarily close to zero in particular. However, it turns out that the
 SSB and the ODLRO are not present in the zero-temperature limit. The
 proof is based on Grassmann Gaussian integral formulations and a
 multi-scale infrared analysis of the formulations. We keep using
 notations and lemmas of our previous work [Y. Kashima,
 accepted for publication in
 J. Math. Sci. Univ. Tokyo, arXiv:1609.06121] implementing the double-scale integration
 scheme. 
So the multi-scale analysis this paper presents is a continuation of the
 previous work.
\footnote[0]{Present address: Division of Mathematical
Science, Graduate School of Engineering Science, Osaka University,
 Toyonaka, Osaka 560-8531, Japan. E-mail: kashima@sigmath.es.osaka-u.ac.jp
\textit{keywords and phrases.} The BCS model, spontaneous symmetry
 breaking, off-diagonal long range order, Grassmann integral formulation,
 multi-scale IR analysis}
\end{abstract}

\tableofcontents

\section{Introduction}

\subsection{Introduction}\label{subsec_introduction}

The Bardeen-Cooper-Schrieffer (BCS) theory of superconductivity
(\cite{BCS}) has been a paradigm of modern physics. 
The BCS model Hamiltonian of interacting electrons lies at the core of
the theory. A large amount of knowledge on how to analyze the BCS
model have been accumulated. 
A history of mathematical development around the BCS model is summarized
in e.g. \cite{BP}. 
However, it is still a fair remark 
that we have not yet achieved a consensus on the possibility of 
completely rigorous, explicit analysis of the full BCS
model. Here we mean a Fermionic Hamiltonian consisting of a quadratic
kinetic term and a quartic interacting term by the BCS model. It is
necessary to investigate in which parameter region the BCS model can be
rigorously analyzed in order to clarify and increase our understanding
of the model in its original definition as the Fermionic field operator.

To supplement overviews of the literature given in the
introduction of our previous work \cite{K_BCS}, here let us comment
on two of the most studied mathematical approaches to the theory of the
BCS model. Analysis of the BCS functional has been vigorously developed
by the authors of the review article \cite{HS} and their coauthors. The
BCS functional is derived from the Gibbs variational principle as a
functional of generalized one-body density matrices. Above all the
derivation is based on an assumption that to characterize equilibrium
states it suffices to minimize the pressure functional over a set of
quasi-free states. To my knowledge, the equivalence between a quasi-free
state minimizing the BCS functional and the Gibbs state of the BCS model 
has not been proved. This means that we cannot rigorously
relate the superconducting order in terms of the minimizer of the BCS
functional to that in the BCS model. At this point it is natural to
consider that the recent papers summarized in \cite{HS} feature a
well-recognized approach to the BCS theory, rather than analysis of the
BCS model Hamiltonian itself. As for the BCS model Hamiltonian, it is
known that its eigenstates can be constructed by using solutions to a
system of nonlinear equations called Richardson's equations (\cite{R},
\cite{RS}). Nowadays Richardson's method is formulated within the framework
of algebraic Bethe ansatz (see e.g. \cite{vDP}, \cite{Z}). Though there
are many applications of this approach, Richardson's equations in
principle need to be solved numerically. It seems that it has not been
applied to rigorously prove existence of superconducting order in the
form of finite-temperature correlation functions in the
BCS model. 

In our previous work \cite{K_BCS} we studied the reduced BCS model,
where the quartic interacting term is a product of the Cooper pair
operators, at positive temperature by extending the external magnetic
field to be purely imaginary. We reached the conclusion that under the
imaginary magnetic field the BCS model is mathematically analyzable at positive
temperatures and especially the superconducting phase characterized by
spontaneous $U(1)$-symmetry breaking (SSB) and off-diagonal long range 
order (ODLRO) can be proven. Let us remark that the BCS model with the
imaginary magnetic field is not Hermitian and thus it does not a priori
define the Gibbs state. At present it seems that this model is not
analyzable within the methods of \cite{HS}, \cite{BP1} based on the Gibbs
variational principle.
One serious constraint in the previous work 
  \cite{K_BCS} is that the possible magnitude of the reduced BCS
  interaction heavily depends on the imaginary magnetic field and the
  temperature. In our previous construction, the closer the imaginary
  magnetic field is to the critical values or the lower the temperature
  is, the smaller the magnitude of the interaction must be. We have
  already mentioned in the introduction of \cite{K_BCS} that the
  temperature-dependency of the allowed magnitude of the interaction
  should be improved by a multi-scale infrared integration. In line with
  this purpose, here we develop a theory where the magnitude of the
  interaction is allowed to be largely independent of the temperature
  and the imaginary magnetic field.

More precisely, in this paper we consider the reduced BCS model
interacting with the imaginary magnetic field at positive temperature
and prove the existence of SSB and ODLRO in the form of the
infinite-volume limit of the thermal expectations over the full
Fermionic Fock space under periodic boundary conditions. 
The magnitude of the attractive interaction must be small. However, the
imaginary magnetic field and the temperature can take almost every value
of a low temperature region of the parameter space without lowering the
magnitude of the interaction. In order to substantially enlarge the
possible parameter region, we need to impose restrictive assumptions on the
free dispersion relation. Here, unlike in our previous paper, we
construct the theory by assuming a series of conditions on the generalized
free dispersion relation. These conditions are typically satisfied by
a free dispersion relation with degenerate Fermi surface. Examples of the
free Hamiltonian covered by our theory are the free electron model of
nearest-neighbor hopping on 3 or 4-dimensional (hyper-)cubic lattice
with a critical chemical potential or the free electron model of nearest-neighbor
hopping on the honeycomb lattice with zero chemical potential. The free
Hamiltonians with non-degenerate Fermi surface treated in \cite{K_BCS}
do not belong to the model class of this paper. See Remark
\ref{rem_nondegenerate_surface} for a mathematical confirmation of this
fact. As a new observation, we
show that for a fixed small coupling constant and a non-zero imaginary
magnetic field the SSB and the ODLRO occur in arbitrarily low
temperatures. However, it turns out that the SSB and the ODLRO are not
present in the zero-temperature, infinite-volume limit of the thermal
expectations. Moreover, the zero-temperature limit of the free energy
density is proved to be equal to that of the free electron model, which
does not depend on either the coupling constant or the imaginary magnetic
field. In terms of the superconducting order, the zero-temperature
limits derived as a corollary of the main results at positive
temperature seem plain and negative. However, if we think of the fact
that the superconducting order exists in arbitrarily low temperatures,
the whole scenario of the phase transitions in this system is unusual and
counterintuitive. In Section \ref{sec_phase_transition} we study the
nature of the phase transitions by focusing on the free energy density
characterized in the main theorem and under a couple of reasonable
additional assumptions on the free dispersion relation we
prove that the phase transitions are of second order.

Though our free Hamiltonian is qualitatively different from that of the
previous work, the basis of our approach is same. We formulate the grand
canonical partition function into a time-continuum limit of
finite-dimensional Grassmann Gaussian integration and perform
mathematical analysis of the Grassmann integral formulation. Moreover
we apply the key proposition \cite[\mbox{Proposition 4.16}]{K_BCS}
concerning the uniform convergence of the Grassmann Gaussian integral
having the modified interacting term in its action in order to deduce
the convergence of the finite-volume thermal expectations to the
infinite-volume limits in the final stage of the paper. While the
previous analysis of the Grassmann Gaussian integral formulation was completed
only by the double-scale integration, here we implement a multi-scale
infrared integration with the aim of easing the temperature-dependency
of the possible magnitude of the interaction. As in \cite{K_BCS}, we
deal with the ultra-violet part with large Matsubara frequencies by
simply applying Pedra-Salmhofer's determinant bound
(\cite{PS}). Many general tools for the double-scale integration
developed in the previous paper are applicable to our multi-scale
integration. We need some more estimation tools to complete our
scheme. We prepare them in accordance with the previous
format of general lemmas. Therefore, from a technical view point of  
the constructive Fermionic field theory this work is seen as a 
continuation of the previous construction \cite{K_BCS}. 

We should explain exceptional subsets of the parameter space of the
temperature and the imaginary magnetic field where we are unable to
construct our theory. If the imaginary magnetic field divided by 2
belongs to the set of Matsubara frequencies, the free
covariance is not well-defined. This is because in this case the
denominator of the free covariance in momentum space can be zero. 
As the free covariance is
a central object in this approach, we have to exclude these points,
which only amount to a 1-dimensional submanifold of the 2-dimensional
parameter space. We claim the main results of this paper for the
temperature and the imaginary magnetic field belonging to the complement
of the union of these subsets.
Also, we have to assume a nontrivial dependency of the
possible magnitude of the coupling constant on the temperature and the
imaginary magnetic field if the temperature is high. This constraint
stems from a determinant bound of the full covariance and has no effect
if the temperature is low. See Remark \ref{rem_nontrivial_dependency}
for details of this constraint.

Taking the zero-temperature limit in interacting many-electron systems
is still a challenging problem of mathematical physics. In the preceding
examples of taking the zero-temperature limit in the systems with spatial
dimension larger than 1 (\cite{GM}, \cite{GJ}, \cite{K_RG},
\cite{K_zero}) 
not only the degeneracy of the Fermi surface but also
symmetries of the whole Hamiltonian are essential. In the infrared
analysis of the Grassmann Gaussian integral of the correction term of
the reduced BCS interaction, we have an advantage that quadratic
Grassmann polynomials are always bounded by the inverse volume factor,
which is incomparably smaller than any support size of infrared
cut-off. We do not need to use symmetries to keep track of the zero set
of the effective dispersion relation, the kernel of the quadratic
Grassmann output, during the iterative infrared integration process. We
only need a priori information of the infrared properties of the free
dispersion relation in order to ensure that Grassmann polynomials of degree $\ge 4$
remain bounded in the iterative scale-dependent norm estimations. For the
above reason the free Hamiltonian can be chosen much more
flexibly in this paper than in the preceding zero-temperature limit
constructions based on multi-scale infrared integrations. The relative
generality of the free Hamiltonian is one novelty of our low temperature analysis.

Here let us explain more about key ideas of our multi-scale analysis 
in order to help the readers proceed to the main technical sections 
and recognize technical novelties of this paper. Let us allow ourselves
to use formulas informally and simplified notations in the following 
for illustrative purposes.  As in \cite{K_BCS} we begin with the 
Grassmann Gaussian 
integral formulation which has the correction term in its
exponent. 
\begin{align}\label{eq_illustrative_formulation}
\int e^{V^0(\psi)}d\mu_{C}(\psi),
\end{align}
where the Grassmann polynomial $V^0(\psi)$ denotes the correction term
and $C$ denotes the full covariance. 
The full Grassmann integral formulation is officially presented in Lemma
\ref{lem_final_Grassmann_formulation}.
By using much simpler notations than those actually used in the main
body of this paper we can write the correction term $V^0(\psi)$ as
follows.
\begin{align*}
&V^0(\psi)=V^0_s(\psi)+V^0_v(\psi),\\
&V^0_s(\psi):=\frac{\gamma}{L}\sum_{x=0}^{L-1}\sum_{t=0}^{n-1}\opsi_{xt}\psi_{xt},\\
&V^0_v(\psi):=\frac{\gamma}{L}\sum_{x,y=0}^{L-1}\sum_{t,u=0}^{n-1}\left(\delta_{t,u}-\frac{1}{n}\right)\opsi_{xt}\psi_{xt}\opsi_{yu}\psi_{yu}.
\end{align*}
Here $\gamma$ is a real number and $L$, $n$ are positive integers. We
should think of $\gamma$, $\{0,1,\cdots,L-1\}$, $\{0,1,\cdots,n-1\}$ as
coupling constant, set of spatial lattice points, set of values of
discretized imaginary time variable, respectively. 
In the following we sketch the
analysis performed in Subsection
\ref{subsec_generalized_covariances}, Subsection
\ref{subsec_integration_without} and Subsection
\ref{subsec_final_integration}. A norm of $V_s^0(\psi)$ is
bounded by the magnitude of the coupling constant $|\gamma|$ and 
the inverse volume factor $L^{-1}$.
$$\|V_s^0\|\le |\gamma|L^{-1}.$$
The norm $\|\cdot\|$ of a Grassmann polynomial is defined by summing its
unique anti-symmetric kernel function over all but one variables. 
More explicitly, the above bound is derived as follows. 
Writing $\psi_{xt,1}$, $\psi_{xt,-1}$ in place of $\opsi_{xt}$,
$\psi_{xt}$, respectively, the unique anti-symmetric kernel function of
$V^0_s(\psi)$ is that
\begin{align*}
&((x,t,\xi),(y,u,\zeta))\mapsto 
\frac{\gamma}{2L}\delta_{x,y}\delta_{t,u}(\delta_{\xi,1}\delta_{\zeta,-1}
-\delta_{\xi,-1}\delta_{\zeta,1})\\
&:(\{0,\cdots,L-1\}\times\{0,\cdots,n-1\}\times\{1,-1\})^2\to\C,
\end{align*}
and thus
\begin{align*}
\|V^0_s\|&=\sup_{(x,t,\xi)\in
 \{0,\cdots,L-1\}\times\{0,\cdots,n-1\}\times\{1,-1\}}\sum_{(y,u,\zeta)\in
 \{0,\cdots,L-1\}\times\{0,\cdots,n-1\}\times\{1,-1\}}\\
&\quad\cdot \left|\frac{\gamma}{2L}\delta_{x,y}\delta_{t,u}(\delta_{\xi,1}\delta_{\zeta,-1}-\delta_{\xi,-1}\delta_{\zeta,1})\right|\\
&=\frac{1}{2}|\gamma|L^{-1}.
\end{align*}
Though its norm cannot be bounded by $L^{-1}$, the Grassmann
polynomial $V^0_v(\psi)$ has a particular vanishing property that 
\begin{align}\label{eq_illustrative_vanishing}
\int V_v^0(\psi)f(\psi)d\mu_{\hat{C}}(\psi)=0
\end{align}
for any Grassmann polynomial $f(\psi)$ and covariance
$\hat{C}:(\{0,\cdots,L-1\}\times\{0,\cdots,n-1\})^2\to\C$ satisfying
that 
\begin{align}
&\hat{C}(xt,yu)=\hat{C}(x0,y0),
\quad(\forall x,y\in\{0,\cdots,L-1\},\ u,t\in \{0,\cdots,n-1\}).\label{eq_formal_time_independence}
\end{align}
In fact the equality \eqref{eq_illustrative_vanishing} can be confirmed
as follows. 
\begin{align*}
\int V_v^0(\psi)f(\psi)d\mu_{\hat{C}}(\psi)=
\frac{\gamma}{L}\sum_{x,y=0}^{L-1}\sum_{t,u=0}^{n-1}\left(\delta_{t,u}-\frac{1}{n}\right)\int\opsi_{x0}\psi_{x0}\opsi_{y0}\psi_{y0}f(\psi)d\mu_{\hat{C}}(\psi)=0.
\end{align*}
By inserting cut-off
functions inside the integral over momentum we can write the full
covariance as a sum of partial covariances. $C=\sum_{l=0}^{l_{end}}C_l$, where $l_{end}\in
\Z_{\le 0}$ denotes the final scale of cut-off and $C_l$ is the
covariance containing the cut-off function of $l$-th scale. 
We remark that $l_{end}$ is independent of $L$ and proportional to $\log
\beta^{-1}$ with the inverse temperature $\beta$ if the temperature is
low, i.e. $\beta\ge 1$. 
The multi-scale integration iterates as follows.
\begin{align*}
\int e^{V^0(\psi)}d\mu_C(\psi)&= \int\int
 e^{V^0(\psi+\psi')}d\mu_{C_0}(\psi')d\mu_{\sum_{l=-1}^{l_{end}}C_l}(\psi)
=\int e^{V^{-1}(\psi)}d\mu_{\sum_{l=-1}^{l_{end}}C_l}(\psi)\\
&=\int e^{V^m(\psi)}d\mu_{\sum_{l=m}^{l_{end}}C_l}(\psi),
\end{align*}
where 
\begin{align*}
&V^m(\psi)=\log\left(\int
 e^{V^{m+1}(\psi+\psi')}d\mu_{C_{m+1}}(\psi')\right),\ (m=-1,-2,\cdots,l_{end}).\end{align*}
At each step of the integration we can decompose the Grassmann
polynomial $V^m(\psi)$ into 2 terms. $V^m(\psi)=V^m_s(\psi)+V^m_v(\psi)$,
where the norm of $V^m_s(\psi)$ is bounded by $L^{-1}$ and $V^m_v(\psi)$ satisfies the vanishing property
\eqref{eq_illustrative_vanishing}.  We can
manipulate the support of the cut-off functions and perform a gauge 
transform so that the final covariance $C_{l_{end}}$, which has the most
 intense infrared singularity, satisfies
 \eqref{eq_formal_time_independence}. Thus, by the
property \eqref{eq_illustrative_vanishing} we reach that 
\begin{align*}
\int e^{V^0(\psi)}d\mu_C(\psi)=\int
e^{V^{l_{end}}_s(\psi)}d\mu_{C_{l_{end}}}(\psi).
\end{align*}
The heavy contribution from $C_{l_{end}}$ 
can be effectively absorbed by the inverse volume factor $L^{-1}$ 
which bounds
the norm of $V^{l_{end}}_s(\psi)$. Also, the factor $L^{-1}$ can be
taken smaller than any power of the inverse temperature or $l_{end}$ and
thus any extra contribution from these parameters does not lower the possible
magnitude of the coupling constant.
This is where we take best advantage 
of the mean-field scaling property and the vanishing property \eqref{eq_illustrative_vanishing} 
that the initial correction term $V^0(\psi)$ has.
The integration with the covariances $C_l$ $(l=0,-1,\cdots,l_{end}+1)$
is performed in Subsection \ref{subsec_integration_without} and 
the integration with the final covariance $C_{l_{end}}$ is specifically 
performed in
Subsection \ref{subsec_final_integration}.
As the result the formulation \eqref{eq_illustrative_formulation} is proved to
be uniformly bounded with respect to the coupling
constant in a good neighborhood of the origin which is independent of
the temperature and the imaginary magnetic field. In fact this mechanism was already implemented at the 
level of double-scale integration in \cite{K_BCS}, which did not require 
mathematical induction with the discrete energy scale. In this paper we 
implement this idea based on the classification of Grassmann polynomials 
inductively with respect to the scale index of infrared cut-off as
described above. We also 
have to incorporate various scale-dependent bound properties into the 
classification of Grassmann polynomials. The mathematical justification 
of the whole inductive procedure is what this paper newly offers in
terms of technical aspects. 

Let us comment on key differences between this paper and \cite{M} one by 
one, as both concern analysis of Grassmann integral formulations of BCS 
type-models. The paper \cite{M} treats a quartic long range interaction 
which is derived from the reduced BCS interaction by inserting a Kac 
potential into the time integral. The essential goal of \cite{M} is to 
ensure the solvability of the BCS gap equation in parameter regions
where the correction part obtained after extracting the main reference 
model can be proved to vanish in the infinite-volume limit. The analysis 
of the correction part was done on the assumption that the parameter 
$\kappa$ determining the range of the inserted Kac potential is bounded 
from above by some negative power of the coupling constant and the
inverse temperature. This assumption does not affect the solvability 
of the BCS gap equation, since the gap equation is independent of the
parameter $\kappa$, and thus the goal was achieved. We should add that 
the solvability of the gap equation is also due to that the free Fermi
surface of the model in \cite{M} is non-degenerate. The 
assumption on $\kappa$ means that the modified BCS-type interaction 
depends on temperature and in particular it approaches to the doubly 
reduced BCS interaction which contains a double time integral, rather 
than to the original reduced BCS interaction in low temperatures. No 
multi-scale infrared integration was performed to improve the 
temperature-dependency of the interacting term. Conceptually this paper 
aims at completing the same story, though we have the reduced BCS 
interaction and the imaginary magnetic field from the beginning. We
prove the solvability of a gap equation together with the fact that 
the correction part becomes negligible in the infinite-volume limit. 
However, we prove the irrelevance of the correction part without
assuming that the interaction is temperature-dependent in low 
temperatures. In order to establish the temperature-independence of the 
interaction, we perform the multi-scale infrared integration which 
requires restrictive degeneracy of the free Fermi surface instead. 
Our gap equation explicitly depends on the imaginary magnetic field 
and thus admits a positive solution regardless of the degeneracy of 
the free Fermi surface. In summary, the properties of quartic
interaction, the degeneracy of free Fermi surface and the presence of
imaginary magnetic field are the key differences between \cite{M} and the
present paper. Among them, the temperature-dependency of 
interaction is considered as the main difference, since it largely
affects the design of constructive theory of interacting Fermions. 
  
If we face a question about whether SSB and ODLRO in the BCS model 
without imaginary magnetic field or in many-electron systems with 
realistic short range interaction can be proved by extending this 
paper's method, we realize that there are many essential problems 
to overcome. This paper's result implies that as long as the same 
free Hamiltonian is adopted, the BCS model without imaginary magnetic 
field can be analyzed down to zero temperature by keeping the magnitude
 of the coupling constant positive. However, we cannot prove that the allowed 
magnitude of the coupling constant is large enough to ensure the 
existence of a positive solution to the BCS gap equation and thus cannot 
prove SSB and ODLRO, either. See Remark \ref{rem_imaginary_zero} for 
a more detailed explanation of this issue. Because of the relatively 
simple form of the reduced BCS interaction, we can apply the 
Hubbard-Stratonovich transformation and reformulate the system into a 
hybrid of Grassmann Gaussian integral and Gaussian integral with a 
single classical field, where the quartic Grassmann field only appears
as a controllable correction term. It is well known that one can also 
apply the Hubbard-Stratonovich transformation to derive a classical
system with many degrees of freedom from the Grassmann integral 
formulation of a many-electron model with short range interaction. Since 
infinitely many classical fields come into play in the infinite-volume 
limit in the standard reformulation of a Hubbard-type short range 
interaction, it seems at present that its complete solution is beyond 
the reach of an immediate extension of this paper's methods. Let us 
remark that equivalence between the minimum configuration of an 
effective potential for many classical fields whose number can be
proportional to the number of finite spatial lattice points and that of an
approximate BCS-type potential, which is expressed as a truncated sum
over the Matsubara frequencies, for a single classical field was proved in
\cite{Leh}. However, such a partial equivalence has not led to complete 
characterization of the thermodynamic limit of the original
many-electron system with short range interaction, to the author's 
knowledge. For these reasons, possible new contributions of this paper 
may not be a construction of necessary steps toward complete solutions 
of the standard BCS model or realistic many-electron models with short range 
interaction, but should be a positive proposal for studying these models 
in a non-standard parameter region of complex plane by means of 
multi-scale analysis and a construction of its necessary tools. The 
proposal should make sense if a structurally rich phase transition can 
be proved as a result. 

The outline of this paper is as follows. In the rest of
this section we define the model Hamiltonian, state the main theorem concerning the superconducting phase at
positive temperature and its corollary about the zero-temperature
limit and present concrete examples of the model. 
In Section \ref{sec_phase_transition} we separately analyze the
free energy density obtained in the main theorem, draw a schematic phase
diagram on the plane of the inverse temperature and the imaginary
magnetic field and prove that the phase transitions are of second order. In
Section \ref{sec_formulation} we state the Grassmann Gaussian integral
formulations of the grand canonical partition function. In Section
\ref{sec_multiscale_integration} we perform the multi-scale infrared
integration by assuming scale-dependent bound properties of generalized
covariances. In Section \ref{sec_proof_theorem} first we confirm that
the actual covariance introduced as the free 2-point correlation function
can be decomposed into a family of scale-dependent covariances
satisfying the properties required in the general multi-scale
analysis of Section \ref{sec_multiscale_integration}. Then we prove the main theorem by applying the results of the
general multi-scale analysis and its corollary. In Appendix
\ref{app_infinite_volume_limit} we summarize basic lemmas which are used
to complete the proof of the main theorem in Section
\ref{sec_proof_theorem}. In addition, we present a supplementary list of
notations which are newly introduced in this paper or were introduced in
the previous paper \cite{K_BCS} with some different meaning. The list 
should be used together with that of
\cite{K_BCS}, since many notations used in this paper are intentionally same or
close to those in \cite{K_BCS}.

\subsection{Models and the main results}\label{subsec_model_theorem}

Let us start by defining our model Hamiltonian. Throughout the paper the spatial
dimension is represented by $d$. Let $\bv_1,\bv_2,\cdots, \bv_d$ be a
basis of $\R^d$. Let $\hbv_1,\hbv_2,\cdots,\hbv_d$ be vectors of $\R^d$
satisfying that $\<\bv_i,\hbv_j\>=\delta_{i,j}$
$(i,j\in\{1,2,\cdots,d\})$, where $\<\cdot,\cdot\>$ denotes the
canonical inner product of $\R^d$. With $L\in\N(=\{1,2,\cdots\})$ we
define the spatial lattice $\G$ and the momentum lattice $\G^*$ as
follows. 
\begin{align*}
&\G:=\left\{\sum_{j=1}^dm_j\bv_j\ \Big|\ m_j\in \{0,1,\cdots,L-1\}\
 (j=1,2,\cdots,d)\right\},\\
&\G^*:=\left\{\sum_{j=1}^d\hat{m}_j\hbv_j\ \Big|\ \hat{m}_j\in
 \left\{0,\frac{2\pi}{L},\frac{4\pi}{L},\cdots,2\pi-\frac{2\pi}{L}\right\}\ (j=1,2,\cdots,d)\right\}.
\end{align*}
In the infinite-volume limit the finite sets $\G$, $\G^*$ are replaced
by the infinite sets $\G_{\infty}$, $\G_{\infty}^*$ defined by
\begin{align*}
&\G_{\infty}:=\left\{\sum_{j=1}^dm_j\bv_j\ \Big|\ m_j\in \Z\
 (j=1,2,\cdots,d)\right\},\\
&\G_{\infty}^*:=\left\{\sum_{j=1}^d\hat{k}_j\hbv_j\ \Big|\ \hat{k}_j\in
 [0,2\pi]\ (j=1,2,\cdots,d)\right\}.
\end{align*}
We plan to construct our theory by assuming a series of conditions on
the free dispersion relation of the model Hamiltonian. We consider
multi-band Hamiltonians since they can have a variety of free dispersion
relations. The number of sites in the unit cell is denoted by $b(\in\N)$. Set
$\cB:=\{1,2,\cdots,b\}$. A crystalline lattice having $b$ sites per the
unit cell is modeled by $\cB\times \G$.
 We define our $b$-band Hamiltonian on the
Fermionic Fock space $F_f(L^2(\cB\times \G\times \spin))$. As in
\cite{K_BCS}, we focus on the reduced BCS interaction defined by 
\begin{align*}
\sV:=\frac{U}{L^d}\sum_{(\rho,\bx),(\eta,\by)\in\cB\times \G}\psi_{\rho\bx\ua}^*\psi_{\rho\bx\da}^*\psi_{\eta\by\da}\psi_{\eta\by\ua},
\end{align*}
where $U(\in\R_{<0})$ is the negative coupling constant. Let us define the map
$r_L:\G_{\infty}\to \G$ by 
\begin{align*}
r_L\left(\sum_{j=1}^dm_j\bv_j\right):=\sum_{j=1}^dm_j'\bv_j,
\end{align*}
where $m_j\in\Z$, $m_j'\in \{0,1,\cdots,L-1\}$, $m_j=m_j'$ (mod $L$)
$(\forall j\in \{1,2,\cdots,d\})$. Throughout the paper we assume
periodic boundary conditions so that for any $\bx\in \G_{\infty}$,
$\psi_{\rho\bx\s}^{(*)}$ is identified with
$\psi_{\rho r_L(\bx)\s}^{(*)}$. We define the free Hamiltonian by giving
a generalized hopping matrix. For $n\in \N$ let $\Mat(n,\C)$ denote the set
of all $n\times n$ complex matrices. For $A\in \Mat(n,\C)$ let 
\begin{align*}
\|A\|_{n\times n}:=\sup_{\bv\in \C^n\atop \text{with
}\|\bv\|_{\C^n}=1}\|A\bv\|_{\C^n},
\end{align*}
where $\|\cdot\|_{\C^n}$ is the norm of $\C^n$ induced by the canonical
Hermitian inner product. $\Mat(n,\C)$ is a Banach space with the norm
$\|\cdot\|_{n\times n}$. For any proposition $P$ let $1_P:=1$ if $P$ is
true, 0 otherwise. We assume that the matrix-valued function $E:\R^d\to
\Mat(b,\C)$ satisfies the following conditions. 
\begin{align}
&E\in C^{\infty}(\R^d,\Mat(b,\C)),\notag\\
&E(\bk)=E(\bk)^*,\ (\forall \bk\in \R^d),\label{eq_self_adjointness}\\
&E(\bk+2\pi \hbv_j)=E(\bk),\ (\forall \bk\in \R^d,\ j\in
 \{1,2,\cdots,d\}),\notag\\
&E(\bk)=\overline{E(-\bk)},\ (\forall
 \bk\in\R^d).\label{eq_spatial_reflection_symmetry}
\end{align}
Moreover, there exist a function $e:\R^d\to \R_{\ge 0}$ and the
constants $\sc\in\R_{\ge 1}$, $\sn_j\in\N$ $(j=1,2,\cdots,d)$, $\sa\in
\R_{>1}$ such that 
\begin{itemize}
\item 
\begin{align}
&e(\bk)\le \inf_{\bv\in \C^b\atop\text{with
 }\|\bv\|_{\C^b}=1}\|E(\bk)\bv\|_{\C^b}\le \sc e(\bk),\ (\forall \bk\in
 \R^d),\label{eq_dispersion_upper_lower_bound}
\end{align}
\item
\begin{align}
&\sup_{\bk\in\R^d}e(\bk)\le \sc,\label{eq_dispersion_upper_bound}
\end{align}
\item
\begin{align}
&e\in C(\R^d,\R),\quad e^2\in C^{\infty}(\R^d,\R),\notag\\
&e(\bk+2\pi \hbv_j)=e(\bk),\ (\forall \bk\in\R^d,\ j\in
 \{1,2,\cdots,d\}),\notag
\end{align}
\item
\begin{align}
&\left| \left(\frac{\partial}{\partial \hat{k}_j}\right)^ne\left(\sum_{i=1}^d\hat{k}_i\hbv_i\right)^2\right|
\le \sc \left(1_{n\le
 2\sn_j}e\left(\sum_{i=1}^d\hat{k}_i\hbv_i\right)^{2-\frac{n}{\sn_j}}+1_{2\sn_j<n}\right),\label{eq_dispersion_derivative}\\
&(\forall (\hat{k}_1,\hat{k}_2,\cdots,\hat{k}_d)\in \R^d,\ n\in \{0,1,\cdots,d+2\},\ j\in
 \{1,2,\cdots,d\}),\notag
\end{align}
\item
\begin{align}
&\left\|\left(\frac{\partial}{\partial
 \hat{k}_j}\right)^nE\left(\sum_{i=1}^d\hat{k}_i\hbv_i\right)\right\|_{b\times b}
\le \sc \left(1_{n\le
 \sn_j}e\left(\sum_{i=1}^d\hat{k}_i\hbv_i\right)^{1-\frac{n}{\sn_j}}+1_{\sn_j<n}\right),\label{eq_hopping_matrix_derivative}\\
&(\forall (\hat{k}_1,\hat{k}_2,\cdots,\hat{k}_d)\in \R^d,\ n\in \{1,2,\cdots,d+2\},\ j\in
 \{1,2,\cdots,d\}),\notag
\end{align}
\item
\begin{align}
&\int_{\G^*_{\infty}}d\bk 1_{e(\bk)\le R}\le \sc \min\{R^{\sa},1\},\
 (\forall R\in\R_{>0}),\label{eq_dispersion_measure}
\end{align}
\item
\begin{align}
&\int_{\G^*_{\infty}}d\bk \frac{1_{e(\bk)\le R}}{e(\bk)}\le \sc \min\{R^{\sa-1},1\},\
 (\forall R\in\R_{>0}),\label{eq_dispersion_measure_divided}
\end{align}
\item
\begin{align}
&\lim_{\eps \searrow 0}\int_{\G^*_{\infty}}d\bk\frac{1}{e(\bk)^2+\eps}=\infty
.\label{eq_dispersion_measure_divergence}
\end{align}
\end{itemize}
Furthermore, we assume the following condition.
\begin{align}
2\sa-1-\sum_{j=1}^d\frac{1}{\sn_j}>0.\label{eq_dispersion_power}
\end{align}
We define the free part of the Hamiltonian by
\begin{align*}
\sH_0:=\frac{1}{L^d}\sum_{(\rho,\bx),(\eta,\by)\in\cB\times
 \G}\sum_{\s\in\spin}\sum_{\bk\in\G^*}e^{i\<\bx-\by,\bk\>}E(\bk)(\rho,\eta)\psi_{\rho\bx\s}^*\psi_{\eta\by\s}.
\end{align*}
By the condition \eqref{eq_self_adjointness} $\sH_0$ is self-adjoint.
The Hamiltonian $\sH$ is defined by $\sH:=\sH_0+\sV$, which is a
self-adjoint operator on $F_f(L^{2}(\cB\times \G\times \spin))$. Because
of the form of the interaction and the generality of the hopping matrix,
we can consider that $\sH$ represents a class of the reduced BCS model. As in
\cite{K_BCS}, we analyze the system under the influence of imaginary
magnetic field. Let $\sS_z$ be the $z$-component of the spin operator,
which is defined by 
\begin{align*}
\sS_z:=\frac{1}{2}\sum_{(\rho,\bx)\in\cB\times
\G}(\psi_{\rho\bx\ua}^*\psi_{\rho\bx\ua}-
 \psi_{\rho\bx\da}^*\psi_{\rho\bx\da}).
\end{align*}
With the parameter $\theta(\in\R)$ we add the operator $i\theta \sS_z$, which
we formally consider as the interacting term with the imaginary magnetic
field, to the Hamiltonian $\sH$ and study the existence or non-existence
of SSB and ODLRO in the infinite-volume limit of the thermal
averages. To study SSB, we introduce the symmetry breaking external field
$\sF$ by
\begin{align*}
\sF:=\g\sum_{(\rho,\bx)\in\cB\times
 \G}(\psi_{\rho\bx\ua}^*\psi_{\rho\bx\da}^*+\psi_{\rho\bx\da}\psi_{\rho\bx\ua}),\
 (\g\in\R).
\end{align*}
Since the operator $\sH+i\theta \sS_z+\sF$ is not Hermitian, it is
nontrivial that the partition function and the thermal expectations of
our interest are real-valued. We should confirm these basic properties
at this stage. 

\begin{lemma}\label{lem_thermal_expectation_real}
For any $\hrho,\heta\in\cB$, $\hbx,\hby\in \G_{\infty}$, 
\begin{align*}
&\Tr e^{-\beta (\sH+i\theta \sS_z+\sF)},\ \Tr (e^{-\beta (\sH+i\theta
 \sS_z+\sF)}\psi_{\hrho \hbx \ua}^* \psi_{\hrho \hbx \da}^*),\
\Tr (e^{-\beta (\sH+i\theta
 \sS_z+\sF)}\psi_{\hrho \hbx \ua}^* \psi_{\hrho \hbx \da}^*\psi_{\heta
 \hby \da}\psi_{\heta \hby \ua})\\
&\in\R
\end{align*}
and 
\begin{align*}
\Tr (e^{-\beta (\sH+i\theta
 \sS_z+\sF)}\psi_{\hrho \hbx \ua}^* \psi_{\hrho \hbx \da}^*)
=\Tr (e^{-\beta (\sH+i\theta
 \sS_z+\sF)}\psi_{\hrho \hbx \da} \psi_{\hrho \hbx \ua}).
\end{align*}
\end{lemma}
\begin{proof}
Observe that 
\begin{align}
&\Tr e^{-\beta (\sH-i\theta \sS_z+\sF)}=\overline{\Tr e^{-\beta
 (\sH+i\theta \sS_z+\sF)}},\label{eq_trace_conjugate}\\
&\Tr (e^{-\beta (\sH-i\theta \sS_z+\sF)}\cO^*)=\overline{\Tr (e^{-\beta
 (\sH+i\theta \sS_z+\sF)}\cO)},\notag\\
&\Tr (e^{-\beta (\sH-i\theta \sS_z+\sF)}\cO)=\overline{\Tr (e^{-\beta
 (\sH+i\theta \sS_z+\sF)}\cO)},\notag\\
&(\forall \cO\in \{\psi_{\hrho\hbx\ua}^*\psi_{\hrho\hbx\da}^*,\
 \psi_{\hrho \hbx \da} \psi_{\hrho \hbx \ua},\ \psi_{\hrho \hbx \ua}^* \psi_{\hrho \hbx \da}^*\psi_{\heta
 \hby \da}\psi_{\heta \hby \ua}\}).\notag
\end{align}
To derive the third equality, one can use the property
 \eqref{eq_spatial_reflection_symmetry} and the periodicity of $E(\cdot)$. On the other hand, by using the
 transforms 
\begin{align*}
&(\psi_{\rho\bx \s},\psi_{\rho\bx \s}^*)\to (\psi_{\rho\bx
 -\s},\psi_{\rho\bx -\s}^*),\\
&(\psi_{\rho\bx \s},\psi_{\rho\bx \s}^*)\to (-i\psi_{\rho\bx
 \s},i\psi_{\rho\bx \s}^*),\quad ((\rho,\bx,\s)\in\cB\times \G\times \spin) 
\end{align*}
in this order we can show that
\begin{align}
&\Tr e^{-\beta (\sH+i\theta \sS_z+\sF)}=\Tr e^{-\beta (\sH-i\theta
 \sS_z+\sF)},\label{eq_trace_theta_reflection}\\
&\Tr (e^{-\beta (\sH+i\theta \sS_z+\sF)}\cO)=\Tr (e^{-\beta (\sH-i\theta
 \sS_z+\sF)}\cO),\notag\\
&(\forall \cO\in \{\psi_{\hrho\hbx\ua}^*\psi_{\hrho\hbx\da}^*,\
 \psi_{\hrho \hbx \da} \psi_{\hrho \hbx \ua},\ \psi_{\hrho \hbx \ua}^* \psi_{\hrho \hbx \da}^*\psi_{\heta
 \hby \da}\psi_{\heta \hby \ua}\}).\notag
\end{align}
The claims follow from \eqref{eq_trace_conjugate} and \eqref{eq_trace_theta_reflection}.
\end{proof}

To state the main theorem, let us fix some notational conventions, which will
be used throughout the paper. For $\bk\in \R^d$ let $e_j(\bk)$
$(j=1,2,\cdots,b')$ be the eigenvalues of $E(\bk)$ satisfying
$e_1(\bk)>e_2(\bk)>\cdots >e_{b'}(\bk)$. With the projection matrix
$P_j(\bk)$ corresponding to the eigenvalue $e_j(\bk)$
$(j=1,2,\cdots,b')$ the spectral decomposition of $E(\bk)$ is that 
\begin{align}
E(\bk)=\sum_{j=1}^{b'}e_j(\bk)P_j(\bk).\label{eq_spectral_decomposition}
\end{align}
For any function $f:\R\to \C$ we define $f(E(\bk))\in \Mat(b,\C)$ by
\begin{align*}
f(E(\bk)):=\sum_{j=1}^{b'}f(e_j(\bk))P_j(\bk).
\end{align*}
It is important in our applications that for $f\in C(\R,\C)$ the
function $\bk\mapsto \Tr f(E(\bk))$ $:\R^d\to \C$ is continuous. This is
essentially because the  roots of the characteristic polynomial of
$E(\bk)$ continuously depend on $\bk$. Rouch\'e's theorem ensures this fact.
 
The statements of our main theorem involve a solution to our gap equation.
Let us confirm the unique solvability
of our gap equation, which is written by the above convention. We admit
that for any $x\in\R$, $y\in\R_{>0}$, $x+\infty=\infty>x$, $y\times
\infty=\infty$, $\infty^{-1}=0$ and
\begin{align*}
\int_{\G_{\infty}^*}d\bk\Tr \left(\frac{\sinh(\beta
|E(\bk)|)}{(-1+\cosh(\beta E(\bk)))|E(\bk)|}\right)=\infty,
\end{align*}
which is consistent with the conditions
\eqref{eq_dispersion_upper_lower_bound},
\eqref{eq_dispersion_measure_divergence}. In fact \eqref{eq_dispersion_upper_lower_bound},
\eqref{eq_dispersion_measure_divergence} imply that
\begin{align*}
\lim_{\eps\searrow 0}\int_{\G_{\infty}^*}d\bk\Tr \left(\frac{\sinh(\beta
|E(\bk)|)}{(\eps -1+\cosh(\beta E(\bk)))|E(\bk)|}\right)=\infty.
\end{align*}
Set 
\begin{align*}
D_d:=|\det(\hbv_1,\hbv_2,\cdots,\hbv_d)|^{-1}(2\pi)^{-d}.
\end{align*}

\begin{lemma}\label{lem_gap_equation_solvability}
Let $U\in \R_{<0}$, $\beta \in \R_{>0}$, $\theta \in\R$. Then the following
 statements hold true. The equation 
\begin{align}
&-\frac{2}{|U|}+D_d\int_{\G_{\infty}^*}d\bk\Tr \left(\frac{\sinh(\beta
\sqrt{E(\bk)^2+\D^2})}{(\cos(\beta\theta/2)+\cosh(\beta
\sqrt{E(\bk)^2+\D^2}))\sqrt{E(\bk)^2+\D^2}}\right)=0\label{eq_gap_equation}
\end{align}
has a solution $\D$ in $[0,\infty)$ if and only if
\begin{align*}
&-\frac{2}{|U|}+D_d\int_{\G_{\infty}^*}d\bk\Tr \left(\frac{\sinh(\beta
|E(\bk)|)}{(\cos(\beta\theta/2)+\cosh(\beta
E(\bk)))|E(\bk)|}\right)\ge 0.
\end{align*}  
Moreover, if a solution exists, it is unique.
\end{lemma}

\begin{proof}
Observe that the functions 
\begin{align*}
&x\mapsto \frac{\sinh x}{(\eps +\cosh x)x}:[0,\infty)\to \R,\ (\eps\in
 (-1,1]),\\
&x\mapsto \frac{\sinh x}{(-1 +\cosh x)x}:(0,\infty)\to \R
\end{align*}
are strictly monotone decreasing and converge to 0 as $x\to \infty$. See
 e.g. \cite[\mbox{Lemma 4.19}]{K_BCS} for hints of the proof. Thus the
 left-hand side of \eqref{eq_gap_equation} is strictly monotone
 decreasing with $\D$ as the map from $\R_{\ge 0}$ to $\R\cup
 \{\infty\}$ and converges to $-2/|U|$ as $\D\to \infty$. Moreover, it
 is continuous with $\D$ as a real-valued function in $\R_{\ge 0}$ if
 $\cos(\beta \theta/2 )\neq -1$, or in $\R_{>0}$ if
 $\cos(\beta\theta/2)=-1$. Furthermore, by
 \eqref{eq_dispersion_upper_lower_bound} and \eqref{eq_dispersion_measure_divergence}
\begin{align*}
\lim_{\D\searrow 0}\int_{\G_{\infty}^*}d\bk\Tr \left(\frac{\sinh(\beta
\sqrt{E(\bk)^2+\D^2})}{(-1+\cosh(\beta
\sqrt{E(\bk)^2+\D^2}))\sqrt{E(\bk)^2+\D^2}}\right)=\infty.
\end{align*}
By using these facts we can deduce the claim.
\end{proof}

For a function $f:\G_{\infty}\times \G_{\infty}\to \C$ and $a\in\C$ we
write $\lim_{\|\bx-\by\|_{\R^d}\to\infty}f(\bx,\by)=a$ if for any
$\eps\in\R_{>0}$ there exists $\delta \in \R_{>0}$ such that
$\|f(\bx,\by)-a\|_{\C}<\eps$ for any $\bx,\by\in\G_{\infty}$ satisfying
$\|\bx-\by\|_{\R^d}>\delta$. Here $\|\cdot\|_{\R^d}$ denotes the
Euclidean norm of $\R^d$. 

For a sequence $(s_n)_{n=n_0}^{\infty}$ and an element $s$ of a normed
space with the norm $|||\cdot|||$ we write $\lim_{n\to\infty,n\in\N}s_n=s$
if for any $\eps\in\R_{>0}$ there exists $m\in \N$ such that
$|||s_n-s|||<\eps$ for any $n\in \N$ satisfying $n\ge m$. The point of
this convention is that we write $\lim_{n\to\infty,n\in \N}s_n$ even if
$s_1,s_2,\cdots,s_{n_0-1}$ are undefined. We use this convention
especially when we consider the infinite-volume limit $L\to \infty$. 

Our main result is stated as follows.

\begin{theorem}\label{thm_main_theorem}
We let $\D(\in\R_{\ge 0})$ be the solution to \eqref{eq_gap_equation} if 
\begin{align*}
&-\frac{2}{|U|}+D_d\int_{\G_{\infty}^*}d\bk\Tr \left(\frac{\sinh(\beta
|E(\bk)|)}{(\cos(\beta\theta/2)+\cosh(\beta
E(\bk)))|E(\bk)|}\right)\ge 0.
\end{align*}
We let $\D:=0$ if   
\begin{align*}
&-\frac{2}{|U|}+D_d\int_{\G_{\infty}^*}d\bk\Tr \left(\frac{\sinh(\beta
|E(\bk)|)}{(\cos(\beta\theta/2)+\cosh(\beta
E(\bk)))|E(\bk)|}\right)< 0.
\end{align*}
Then there exists a positive constant $c_1$ depending only on
 $d,b,(\hbv_j)_{j=1}^d,\sa,(\sn_j)_{j=1}^d,\sc$ such that the following
 statements hold for any $\beta \in\R_{>0}$, $\theta\in \R$ satisfying
 $\beta\theta/2\notin \pi(2\Z+1)$ and 
\begin{align}
U\in \left(-c_1\left(1_{\beta\ge 1}+1_{\beta
 <1}\max\left\{\beta^2,\min_{m\in\Z}\left|\frac{\beta\theta}{2}-\pi(2m+1)\right|^2\right\}\right),0\right).\label{eq_coupling_constant_interval}
\end{align}
\begin{enumerate}[(i)]
\item\label{item_partition_positivity}
There exists $L_0\in \N$ such that 
\begin{align*}
\Tr e^{-\beta (\sH+i\theta \sS_z+\sF)}\in\R_{>0},\quad(\forall
 L\in\N\text{ with }L\ge L_0,\ \g\in[0,1]).
\end{align*}
\item\label{item_free_energy_density}
\begin{align}
&\lim_{L\to \infty\atop L\in \N}\left(-\frac{1}{\beta L^d}\log(\Tr
 e^{-\beta(\sH+i\theta \sS_z)})\right)\label{eq_free_energy_density}\\
&=\frac{\D^2}{|U|}-\frac{D_d}{\beta}\int_{\G_{\infty}^*}d\bk\Tr\log\Bigg(2\cos\left(\frac{\beta\theta}{2}\right)e^{-\beta
 E(\bk)}\notag\\
&\qquad\qquad\qquad\qquad\qquad\qquad +e^{\beta(\sqrt{E(\bk)^2+\D^2}-E(\bk))}
+e^{-\beta(\sqrt{E(\bk)^2+\D^2}+E(\bk))}\Bigg).\notag
\end{align}
\item\label{item_SSB}
\begin{align}
&\lim_{\g\searrow 0\atop \g\in (0,1]}\lim_{L\to \infty\atop L\in\N}\frac{\Tr(e^{-\beta
(\sH+i\theta \sS_z+\sF)}\psi_{\hrho\hbx\ua}^*\psi_{\hrho\hbx\da}^*)}{\Tr e^{-\beta (\sH+i\theta \sS_z+\sF)}}=\lim_{\g\searrow 0\atop \g\in (0,1]}\lim_{L\to \infty\atop L\in\N}\frac{\Tr(e^{-\beta
(\sH+i\theta \sS_z+\sF)}\psi_{\hrho\hbx\da}\psi_{\hrho\hbx\ua})}{\Tr
 e^{-\beta (\sH+i\theta \sS_z+\sF)}}
\label{eq_SSB}\\
&=-\frac{\D D_d}{2}\int_{\G_{\infty}^*}d\bk
\frac{\sinh(\beta
\sqrt{E(\bk)^2+\D^2})}{(\cos(\beta\theta/2)+\cosh(\beta
\sqrt{E(\bk)^2+\D^2}))\sqrt{E(\bk)^2+\D^2}}(\hrho,\hrho),\notag\\
&(\forall \hrho\in \cB,\ \hbx\in \G_{\infty}).\notag
\end{align}
\item\label{item_ODLRO}
If 
\begin{align*}
|U|\neq \left(
\frac{D_d}{2}\int_{\G_{\infty}^*}d\bk\Tr \left(\frac{\sinh(\beta
|E(\bk)|)}{(\cos(\beta\theta/2)+\cosh(\beta
E(\bk)))|E(\bk)|}\right)\right)^{-1},
\end{align*}
\begin{align}
&\lim_{\|\hbx-\hby\|_{\R^d}\to\infty}\lim_{L\to \infty\atop L\in\N}\frac{\Tr (e^{-\beta
(\sH+i\theta \sS_z)}\psi_{\hrho\hbx\ua}^*\psi_{\hrho\hbx\da}^*\psi_{\heta\hby\da}\psi_{\heta\hby\ua})}{\Tr
 e^{-\beta (\sH+i\theta \sS_z)}}\label{eq_ODLRO_explicit}\\
&=
\D^2\prod_{\rho\in \{\hrho,\heta\}}\left(
\frac{D_d}{2}\int_{\G_{\infty}^*}d\bk
\frac{\sinh(\beta
\sqrt{E(\bk)^2+\D^2})}{(\cos(\beta\theta/2)+\cosh(\beta
\sqrt{E(\bk)^2+\D^2}))\sqrt{E(\bk)^2+\D^2}}(\rho,\rho)\right),\notag\\
&(\forall\hrho,\heta \in\cB).\notag
\end{align}
If 
\begin{align*}
|U|= \left(
\frac{D_d}{2}\int_{\G_{\infty}^*}d\bk\Tr \left(\frac{\sinh(\beta
|E(\bk)|)}{(\cos(\beta\theta/2)+\cosh(\beta
E(\bk)))|E(\bk)|}\right)\right)^{-1},
\end{align*}
\begin{align*}
&\lim_{\|\hbx-\hby\|_{\R^d}\to\infty}\limsup_{L\to \infty\atop L\in\N}\left|\frac{\Tr (e^{-\beta
(\sH+i\theta \sS_z)}\psi_{\hrho\hbx\ua}^*\psi_{\hrho\hbx\da}^*\psi_{\heta\hby\da}\psi_{\heta\hby\ua})}{\Tr
 e^{-\beta (\sH+i\theta \sS_z)}}\right|=0,\quad(\forall\hrho,\heta
 \in\cB).
\end{align*}
\item\label{item_CPD}
\begin{align*}
&\lim_{L\to \infty\atop
 L\in\N}\frac{1}{L^{2d}}\sum_{(\hrho,\hbx),(\heta,\hby)\in\cB\times \G}\frac{\Tr (e^{-\beta
(\sH+i\theta \sS_z)}\psi_{\hrho\hbx\ua}^*\psi_{\hrho\hbx\da}^*\psi_{\heta\hby\da}\psi_{\heta\hby\ua})}{\Tr
 e^{-\beta (\sH+i\theta \sS_z)}}=\frac{\D^2}{U^2}.
\end{align*}
\item\label{item_gap_equation_solvable}
There exists $\delta\in \R_{>0}$ such that if
     $\min_{m\in\Z}|\beta\theta/2-\pi(2m+1)|<\delta$, then 
\begin{align*}
-\frac{2}{|U|}+D_d\int_{\G_{\infty}^*}d\bk\Tr \left(\frac{\sinh(\beta
|E(\bk)|)}{(\cos(\beta\theta/2)+\cosh(\beta
E(\bk)))|E(\bk)|}\right)>0
\end{align*}
and $\D>0$.
\end{enumerate}
\end{theorem}

In the rest of the paper except Section \ref{sec_phase_transition} we
always assume that 
$$
\frac{\beta\theta}{2}\notin \pi(2\Z+1).
$$
This is because the free partition function can vanish if
$\beta\theta/2\in \pi(2\Z+1)$ and thus we are unable to define the free
covariance, which is indispensable for our construction. Only in Section
\ref{sec_phase_transition} we lift this condition.

For $(x,y,z)\in \R_{>0}\times \R\times \R$ with $xy/2\notin \pi (2\Z+1)$
we define the matrix-valued function $G_{x,y,z}:\R^d\to \Mat(b,\C)$ by 
\begin{align}
G_{x,y,z}(\bk):= 
\frac{\sinh(x
\sqrt{E(\bk)^2+z^2})}{(\cos(xy/2)+\cosh(x
\sqrt{E(\bk)^2+z^2}))\sqrt{E(\bk)^2+z^2}}.\label{eq_matrix_valued_notation}
\end{align}
This notation helps to shorten formulas in subsequent arguments. We
can prove the claim \eqref{item_gap_equation_solvable} here. There
uniquely exists $m_0\in\Z$ such that
$\beta\theta/(2\pi)\in[2m_0,2m_0+2)$ and
$\min_{m\in\Z}|\beta\theta/2-\pi(2m+1)|=|\beta\theta/2-\pi(2m_0+1)|$. By
\eqref{eq_dispersion_upper_lower_bound} and \eqref{eq_dispersion_upper_bound},
\begin{align*}
\int_{\G_{\infty}^*}d\bk \Tr G_{\beta,\theta,0}(\bk)&\ge \beta
 \int_{\G_{\infty}^*}d\bk \frac{1}{\cos(\beta \theta/2)+\cosh(\beta \sc
 e(\bk))}\\
&\ge \beta \int_{\G_{\infty}^*}d\bk \frac{1}{e^{\beta\sc^2}\beta^2\sc^2
 e(\bk)^2+1-\cos(|\beta\theta/2-\pi(2m_0+1)|)}\\
&\ge \frac{\beta}{\max\{1,e^{\beta\sc^2}\beta^2\sc^2\}}
\int_{\G_{\infty}^*}d\bk
 \frac{1}{e(\bk)^2+|\beta\theta/2-\pi(2m_0+1)|^2}.
\end{align*}
By \eqref{eq_dispersion_measure_divergence} there exists $\delta
\in\R_{>0}$ such that if $|\beta\theta/2-\pi(2m_0+1)|<\delta$, the
right-hand side of the above inequality is larger than
$2D_d^{-1}|U|^{-1}$. Then Lemma \ref{lem_gap_equation_solvability}
implies that $\D>0$.

\begin{remark}
The claim \eqref{item_partition_positivity} ensures the well-definedness
 of the free energy density and the thermal expectations for $L\in \N$
 with $L\ge L_0$. By following the above-mentioned convention we write
 $\lim_{L\to\infty,L\in\N}$ in Theorem \ref{thm_main_theorem}, though
 these objects are not defined for $L\in \{1,2,\cdots,L_0-1\}$.
\end{remark}

\begin{remark}
For any $\D\in\R_{\ge 0}$, $\hrho\in\cB$,
\begin{align*}
&\int_{\G_{\infty}^*}d\bk G_{\beta,\theta,\D}(\bk)(\hrho,\hrho)\\
&\ge D_d^{-1}\sinh\left(\beta\sqrt{\sup_{\bk\in\R^d}\|E(\bk)\|_{b\times
 b}^2+\D^2}\right)\\
&\quad\cdot \left(\cos(\beta\theta/2)+\cosh\left(\beta\sqrt{\sup_{\bk\in\R^d}\|E(\bk)\|_{b\times
 b}^2+\D^2}\right)\right)^{-1}\left(\sqrt{\sup_{\bk\in\R^d}\|E(\bk)\|_{b\times
 b}^2+\D^2}\right)^{-1}\\
&>0.
\end{align*}
From this estimate we can see that the theorem implies the occurrence of
 SSB and ODLRO in the case
\begin{align*}
-\frac{2}{|U|}+D_d \int_{\G_{\infty}^*}d\bk \Tr
 G_{\beta,\theta,0}(\bk)>0.
\end{align*}
\end{remark}

\begin{remark}
If $\theta\neq 0$, for any $\beta_0$, $\delta\in\R_{>0}$ there exists
 $\beta\in [\beta_0,\infty)$ such that
 $0<\min_{m\in\Z}|\beta\theta/2-\pi(2m+1)|<\delta$. Thus we can read from
 the claim \eqref{item_gap_equation_solvable} that if $\theta\neq 0$,
 the SSB and the ODLRO occur in arbitrarily low temperatures.
\end{remark}

\begin{remark}\label{rem_nontrivial_dependency}
The $\beta$-dependency in the case $\beta<1$ in
 \eqref{eq_coupling_constant_interval} stems from a determinant bound on
 the full covariance, which is essentially governed by the integral
\begin{align}
\frac{1}{\beta}\int_{\G_{\infty}^*}d\bk \frac{1}{e(\bk)+\min_{m\in
 \Z}|\theta/2-\pi(2m+1)/\beta|}
\label{eq_essential_nontrivial}
\end{align}
if $\beta <1$.
See the proof of Lemma \ref{lem_covariance_construction}
 \eqref{item_covariance_construction_necessary}. In fact a lower bound
 of the term 
\begin{align*}
\left(\frac{1}{\beta}\int_{\G_{\infty}^*}d\bk \frac{1}{e(\bk)+\min_{m\in
 \Z}|\theta/2-\pi(2m+1)/\beta|}\right)^{-2}
\end{align*}
leads to the $\beta$-dependency in
 \eqref{eq_coupling_constant_interval}. If $\theta =0$ (mod
 $4\pi/\beta$), the term \eqref{eq_essential_nontrivial} is bounded by a
 $\beta$-independent constant and the determinant bound on the full
 covariance becomes independent of $\beta$ as usual. Thus we can explain
 that the nontrivial $\beta$-dependency in
 \eqref{eq_coupling_constant_interval} is caused by the insertion of the
 imaginary magnetic field.
\end{remark}

\begin{remark}
For $\phi\in \C$ let us define the operator $\sF(\phi)$ by
\begin{align*}
\sF(\phi):=\sum_{(\rho,\bx)\in\cB\times\G}(\phi\psi_{\rho\bx\ua}^*\psi_{\rho\bx\da}^*
+
\overline{\phi}\psi_{\rho\bx\da}\psi_{\rho\bx\ua}).
\end{align*}
Take any $\hbx\in\G_{\infty}$, $\hrho\in\cB$, $\xi\in\R$. It follows
 from the claim \eqref{item_SSB} and the gauge transform $\psi_X^*\to
 e^{-i\frac{\xi}{2}}\psi_X^*$, $\psi_X \to  e^{i\frac{\xi}{2}}\psi_X$
 $(X\in \cB\times \G\times \spin)$ that 
\begin{align*}
&\lim_{\g\searrow 0\atop \g\in (0,1]}\lim_{L\to \infty\atop L\in\N}\frac{\Tr(e^{-\beta
(\sH+i\theta \sS_z+\sF(\g
 e^{i\xi}))}\psi_{\hrho\hbx\ua}^*\psi_{\hrho\hbx\da}^*)}{\Tr e^{-\beta
 (\sH+i\theta \sS_z+\sF(\g e^{i\xi}))}}
=-\frac{e^{-i\xi}\D D_d}{2}\int_{\G_{\infty}^*}d\bk
 G_{\beta,\theta,\D}(\bk)(\hrho,\hrho),\\
&\lim_{\g\searrow 0\atop \g\in (0,1]}\lim_{L\to \infty\atop L\in\N}\frac{\Tr(e^{-\beta
(\sH+i\theta \sS_z+\sF(\g
 e^{i\xi}))}\psi_{\hrho\hbx\da}\psi_{\hrho\hbx\ua})}{\Tr e^{-\beta
 (\sH+i\theta \sS_z+\sF(\g e^{i\xi}))}}
=-\frac{e^{i\xi}\D D_d}{2}\int_{\G_{\infty}^*}d\bk
 G_{\beta,\theta,\D}(\bk)(\hrho,\hrho).
\end{align*}
These convergent properties imply that the limit
\begin{align*}
&\lim_{\phi\to 0\atop \phi\in \C\backslash\{0\}}\lim_{L\to \infty\atop L\in\N}\frac{\Tr(e^{-\beta
(\sH+i\theta \sS_z+\sF(\phi))}\psi_{\hrho\hbx\ua}^*\psi_{\hrho\hbx\da}^*)}{\Tr e^{-\beta
 (\sH+i\theta \sS_z+\sF(\phi))}},\ \lim_{\phi\to 0\atop \phi\in \C\backslash\{0\}}\lim_{L\to \infty\atop L\in\N}\frac{\Tr(e^{-\beta
(\sH+i\theta \sS_z+\sF(\phi))}\psi_{\hrho\hbx\da}\psi_{\hrho\hbx\ua})}{\Tr e^{-\beta
 (\sH+i\theta \sS_z+\sF(\phi))}}
\end{align*}
do not exist when $\D>0$. However, 
\begin{align*}
&\lim_{\phi\to 0\atop \phi\in \C\backslash\{0\}}\left|\lim_{L\to \infty\atop L\in\N}\frac{\Tr(e^{-\beta
(\sH+i\theta \sS_z+\sF(\phi))}\psi_{\hrho\hbx\ua}^*\psi_{\hrho\hbx\da}^*)}{\Tr e^{-\beta
 (\sH+i\theta \sS_z+\sF(\phi))}}\right|\\
&=\lim_{\phi\to 0\atop \phi\in \C\backslash\{0\}}\left|\lim_{L\to \infty\atop L\in\N}\frac{\Tr(e^{-\beta
(\sH+i\theta \sS_z+\sF(\phi))}\psi_{\hrho\hbx\da}\psi_{\hrho\hbx\ua})}{\Tr e^{-\beta
 (\sH+i\theta \sS_z+\sF(\phi))}}\right|\\
&=\frac{\D D_d}{2}\int_{\G_{\infty}^*}d\bk
 G_{\beta,\theta,\D}(\bk)(\hrho,\hrho).
\end{align*}
\end{remark}

\begin{remark}
The claim \eqref{item_ODLRO} does not imply the convergence of 
\begin{align*} 
\lim_{L\to \infty\atop L\in\N}
\frac{\Tr (e^{-\beta
(\sH+i\theta \sS_z)}\psi_{\hrho\hbx\ua}^*\psi_{\hrho\hbx\da}^*\psi_{\heta\hby\da}\psi_{\heta\hby\ua})}{\Tr
 e^{-\beta (\sH+i\theta \sS_z)}}
\end{align*}
in the case 
\begin{align}
|U|=\left(\frac{D_d}{2}\int_{\G_{\infty}^*}d\bk\Tr
 G_{\beta,\theta,0}(\bk)\right)^{-1}.\label{eq_on_phase_boundary_short}
\end{align}
 In fact in this case
 we cannot prove the convergence of the finite-volume 4-point
 correlation function as $L\to \infty$. We can prove that the global
 maximum point of the function 
\begin{align*}
x\mapsto &-\frac{x^2}{|U|}+\frac{1}{\beta L^d}\sum_{\bk\in\G^*}\Tr
 \log\left(\cos\left(\frac{\beta
 \theta}{2}\right)+\cosh(\beta\sqrt{E(\bk)^2+x^2})\right)\\
&-\frac{1}{\beta L^d}\sum_{\bk\in\G^*}\Tr
 \log\left(\cos\left(\frac{\beta
 \theta}{2}\right)+\cosh(\beta E(\bk))\right):\R_{\ge 0}\to \R
\end{align*}
converges to 0 as $L\to \infty$. According to the proof of the theorem
 in Subsection \ref{subsec_proof_theorem} and Lemma
 \ref{lem_infinite_volume_correlation} in Appendix
 \ref{app_infinite_volume_limit}, we must have more detailed information
 about how the maximum point and derivatives of the function at the
 maximum point converge as $L\to \infty$ to complete the proof. We are
 unable to extract the necessary information from our assumptions on $E(\cdot)$.
 On the contrary, the theorem guarantees the convergence of the thermal
 expectations 
\begin{align*}
\frac{\Tr (e^{-\beta
(\sH+i\theta \sS_z+\sF)}\psi_{\hrho\hbx\ua}^*\psi_{\hrho\hbx\da}^*)}{\Tr
 e^{-\beta (\sH+i\theta \sS_z+\sF)}},\quad
 \frac{1}{L^{2d}}\sum_{(\hat{\rho},\hat{\bx}),(\hat{\eta},\hat{\by})\in
 \cB\times \G}\frac{\Tr (e^{-\beta
(\sH+i\theta \sS_z)}\psi_{\hrho\hbx\ua}^*\psi_{\hrho\hbx\da}^*\psi_{\heta\hby\da}\psi_{\heta\hby\ua})}{\Tr
 e^{-\beta (\sH+i\theta \sS_z)}}
\end{align*}
and the free energy density 
$$
-\frac{1}{\beta L^d}\log(\Tr e^{-\beta(\sH+i\theta \sS_z)})
$$
as $L\to \infty$ as long as $\beta(\in \R_{>0})$, $\theta(\in \R)$
 satisfies $\beta\theta/2\notin \pi (2\Z+1)$ and $U$ satisfies \eqref{eq_coupling_constant_interval}, whether \eqref{eq_on_phase_boundary_short} holds or not.
\end{remark}

\begin{remark}\label{rem_imaginary_zero}
As we can see from the claim
 \eqref{item_gap_equation_solvable}, the non-zero imaginary magnetic field
  is crucial to ensure the existence of a positive solution to
 the gap equation \eqref{eq_gap_equation} for any small coupling
 constant and accordingly the existence of SSB and ODLRO in this regime. 
One natural question we face is whether we can prove SSB and ODLRO when
 the imaginary magnetic field is switched off. To find an answer to this
 question, let us examine the solvability of the gap equation when
 $\theta=0$. By \eqref{eq_dispersion_upper_lower_bound},
 \eqref{eq_dispersion_upper_bound},
 \eqref{eq_dispersion_measure_divided}
\begin{align*}
\int_{\G_{\infty}^*}d\bk \Tr G_{\beta,0,0}(\bk)\le
 \int_{\G_{\infty}^*}d\bk\Tr\left(\frac{1}{|E(\bk)|}\right)
\le b\int_{\G_{\infty}^*}d\bk\frac{1}{e(\bk)}
\le b\sc\min\{\sc^{\sa-1},1\}.
\end{align*}
Thus, a necessary condition for existence of a positive solution $\D$ to
 \eqref{eq_gap_equation} with $\theta=0$ 
is that 
\begin{align}\label{eq_necessary_lower_bound_theta}
|U|>\frac{2}{D_db\sc \min\{\sc^{\sa-1},1\}}.
\end{align}
We can compute the
 thermal expectation values for some $U$ independent of $\beta$ and
 $\theta$ as described in 
\eqref{eq_coupling_constant_interval},
 which is an advantageous result of the multi-scale
 integration. However, our multi-scale analysis has no advantage to make
 the allowed magnitude of $U$ quantitatively explicit, as we need to
 go through a pile of calculations. Whether the
 necessary condition \eqref{eq_necessary_lower_bound_theta} holds in
 this regime is highly nontrivial and we cannot give an affirmative
 answer to the question at present.
\end{remark}

Since the upper bound on $|U|$ does not depend on $\beta$ if $\beta\ge
1$, we can consider the zero-temperature limit $\beta\to \infty$ of the
free energy density and the thermal expectations. It turns out that in
the weak coupling region where our construction is valid the
zero-temperature limit does not exhibit the characteristics of
superconductivity.

\begin{corollary}\label{cor_zero_temperature_limit}
There exists a positive constant $c_2$ depending only on
 $d,b,(\hbv_j)_{j=1}^d,\sa$, $(\sn_j)_{j=1}^d,\sc$ such that the following
 statements hold for any $U\in (-c_2,0)$, $\theta\in \R$.
\begin{enumerate}[(i)]
\item\label{item_gap_bound}
\begin{align*}
\D\le \frac{1}{\beta},\quad (\forall \beta\in\R_{\ge 1}\text{ with
 }\beta\theta/2\notin \pi (2\Z+1)).
\end{align*}
\item\label{item_ground_state_energy}
\begin{align*}
&\lim_{\beta\to \infty,\beta\in\R_{>0}\atop\text{with }\frac{\beta\theta}{2}\notin \pi (2\Z+1)}
\lim_{L\to \infty\atop L\in \N}\left(-\frac{1}{\beta L^d}\log(\Tr
 e^{-\beta(\sH+i\theta \sS_z)})\right)=
D_d\int_{\G_{\infty}^*}d\bk \Tr (E(\bk)-|E(\bk)|)\\
&=\lim_{\beta\to \infty\atop\beta\in\R_{>0}}
\lim_{L\to \infty\atop L\in \N}\left(-\frac{1}{\beta L^d}\log(\Tr
 e^{-\beta \sH_0})\right).
\end{align*}
\item\label{item_SSB_zero}
\begin{align*}
&\lim_{\g\searrow 0\atop \g\in (0,1]}
\lim_{\beta\to \infty,\beta\in\R_{>0}\atop\text{with }\frac{\beta\theta}{2}\notin \pi (2\Z+1)}
\lim_{L\to \infty\atop L\in\N}\frac{\Tr(e^{-\beta
(\sH+i\theta \sS_z+\sF)}\psi_{\hrho\hbx\ua}^*\psi_{\hrho\hbx\da}^*)}{\Tr
 e^{-\beta (\sH+i\theta \sS_z+\sF)}}\\
&=\lim_{\g\searrow 0\atop \g\in (0,1]}\lim_{\beta\to \infty,\beta\in\R_{>0}\atop\text{with }\frac{\beta\theta}{2}\notin \pi (2\Z+1)}\lim_{L\to \infty\atop L\in\N}\frac{\Tr(e^{-\beta
(\sH+i\theta \sS_z+\sF)}\psi_{\hrho\hbx\da}\psi_{\hrho\hbx\ua})}{\Tr
 e^{-\beta (\sH+i\theta \sS_z+\sF)}}\\
&=\lim_{\beta\to \infty,\beta\in\R_{>0}\atop\text{with
 }\frac{\beta\theta}{2}\notin \pi (2\Z+1)}\lim_{\g\searrow 0\atop \g\in
 (0,1]}\lim_{L\to \infty\atop L\in\N}\frac{\Tr(e^{-\beta
(\sH+i\theta \sS_z+\sF)}\psi_{\hrho\hbx\ua}^*\psi_{\hrho\hbx\da}^*)}{\Tr
 e^{-\beta (\sH+i\theta \sS_z+\sF)}}\\
&=\lim_{\beta\to \infty,\beta\in\R_{>0}\atop\text{with }\frac{\beta\theta}{2}\notin \pi (2\Z+1)}\lim_{\g\searrow 0\atop \g\in (0,1]}\lim_{L\to \infty\atop L\in\N}\frac{\Tr(e^{-\beta
(\sH+i\theta \sS_z+\sF)}\psi_{\hrho\hbx\da}\psi_{\hrho\hbx\ua})}{\Tr
 e^{-\beta (\sH+i\theta \sS_z+\sF)}}=0,\quad (\forall \hrho\in\cB,\
 \hbx\in \G_{\infty}).
\end{align*}
\item\label{item_ODLRO_zero}
\begin{align*}
&\lim_{\|\hbx-\hby\|_{\R^d}\to\infty}
 \limsup_{\beta\to \infty,\beta\in\R_{>0}\atop\text{with }\frac{\beta\theta}{2}\notin \pi (2\Z+1)}\limsup_{L\to \infty\atop L\in\N}\left|\frac{\Tr (e^{-\beta
(\sH+i\theta \sS_z)}\psi_{\hrho\hbx\ua}^*\psi_{\hrho\hbx\da}^*\psi_{\heta\hby\da}\psi_{\heta\hby\ua})}{\Tr
 e^{-\beta (\sH+i\theta \sS_z)}}\right|=0,\\
&(\forall\hrho,\heta
 \in\cB).
\end{align*}
\item\label{item_CPD_zero}
\begin{align*}
& \lim_{\beta\to \infty,\beta\in\R_{>0}\atop\text{with }\frac{\beta\theta}{2}\notin \pi (2\Z+1)}
\lim_{L\to \infty\atop L\in\N}
\frac{1}{L^{2d}}\sum_{(\hrho,\hbx),(\heta,\hby)\in\cB\times \G}
\frac{\Tr (e^{-\beta
(\sH+i\theta \sS_z)}\psi_{\hrho\hbx\ua}^*\psi_{\hrho\hbx\da}^*\psi_{\heta\hby\da}\psi_{\heta\hby\ua})}{\Tr
 e^{-\beta (\sH+i\theta \sS_z)}}=0.
\end{align*}
\end{enumerate}
\end{corollary}

\begin{remark} 
Though it does not show the sign of superconductivity, it is
 interesting that the zero-temperature, infinite-volume limit of the
 free energy density claimed in \eqref{item_ground_state_energy} is
 independent of both the coupling constant and the imaginary magnetic field. 
\end{remark}

\begin{remark} 
Among the assumptions listed in the beginning, the smoothness of
 $E(\cdot)$, $e(\cdot)^2$ is assumed only for simplicity. In fact we
 only need to differentiate $E(\cdot)$, $e(\cdot)^2$ finite times
 depending only on the dimension $d$. Thus the smoothness condition can
 be relaxed to be continuous differentiability of certain degree.
\end{remark}

\begin{remark} See Remark \ref{rem_dispersion_power} for the specific
 reason why we need to assume \eqref{eq_dispersion_power}. Also, Remark
 \ref{rem_strictly_larger_one} explains how we use the condition
 $\sa>1$.
\end{remark}

\subsection{Examples}\label{subsec_examples}

In order to see the applicability of Theorem \ref{thm_main_theorem} and
Corollary \ref{cor_zero_temperature_limit}, we should examine which
model satisfies the required conditions. We let `$c$' denote a
generic positive constant independent of any parameter not only in this
section but in the rest of the paper. Also, $I_n$ denotes the $n\times
n$ unit matrix throughout the paper.

\begin{example}[Nearest-neighbor hopping on the 3 or 4-dimensional
 (hyper-)cubic lattice with a critical chemical
 potential]\label{ex_degenerate_surface}

Let $d=3$ or 4 and $\{\bv_j\}_{j=1}^d$, $\{\hbv_j\}_{j=1}^d$ be the
 canonical basis of $\R^d$. In this case $\G=\{0,1,\cdots,L-1\}^d$,
 $\G^*=\{0,\frac{2\pi}{L},\cdots,2\pi-\frac{2\pi}{L}\}^d$,
 $\G_{\infty}^*=[0,2\pi]^d$. Let $b=1$ and set for $\bk=(k_1,k_2,\cdots,k_d)\in\R^d$
$$
E(\bk):=(-1)^{hop}2\sum_{j=1}^d\cos k_j -2d
$$
 with $hop\in\{0,1\}$. In
 this case $\sH_0$ describes free electrons hopping to nearest-neighbor
 sites under the chemical potential $2d$. The role of the fixed
 parameter $hop$ is to implement the negative and positive hopping at
 the same time. The applicability of the previous framework to this
 model was briefly studied in \cite[\mbox{Remark 1.9}]{K_BCS}.
Define the function $e:\R^d\to \R$ by 
$$
e(\bk):=4\sum_{j=1}^d\sin^2\left(\frac{k_j}{2}+1_{hop=1}\frac{\pi}{2}\right).
$$
We can check that $e(\bk)=|E(\bk)|$ for any $\bk\in\R^d$. It is clear
 that \eqref{eq_dispersion_upper_lower_bound},
 \eqref{eq_dispersion_upper_bound} hold with some $\sc(\in \R_{\ge 1})$
 and $e(\cdot)$ satisfies the required regularity and the
 periodicity. Moreover, the conditions \eqref{eq_dispersion_derivative},
 \eqref{eq_hopping_matrix_derivative} hold with $\sn_j=2$
 $(j=1,2,\cdots,d)$ and the
 conditions \eqref{eq_dispersion_measure},
 \eqref{eq_dispersion_measure_divided} hold with $\sa=d/2$ and
 some $\sc(\in \R_{\ge 1})$. By considering that $d=3,4$ we can confirm that the
 conditions \eqref{eq_dispersion_measure_divergence},
 \eqref{eq_dispersion_power} hold as well.
\end{example}

\begin{example}[Nearest-neighbor hopping on the honeycomb
 lattice]\label{ex_honeycomb}
Many-Fermion systems on the honeycomb lattice with nearest neighbor
 hopping are well studied in a branch of mathematical physics based on
 Grassmann integral formulations. See e.g. \cite{GM}. Let us confirm that
 the free electron model on the honeycomb lattice with zero chemical
 potential can be dealt in this
 framework as the free Hamiltonian. Take the basis $\bv_1=(1,0)^T$,
 $\bv_2=(\frac{1}{2},\frac{\sqrt{3}}{2})^T$ of $\R^2$. Then, $\hbv_1$,
 $\hbv_2$ are uniquely determined as
 follows. $\hbv_1=(1,-\frac{1}{\sqrt{3}})^T$,
 $\hbv_2=(0,\frac{2}{\sqrt{3}})^T$.
The honeycomb lattice with a spatial cut-off is identified with the
 product set $\{1,2\}\times \G$. The hopping matrix is given with
 momentum variables by
$$
E(\bk):=\left(\begin{array}{cc}
0 & 1 + e^{-i\<\bv_1,\bk\>} + e^{-i\<\bv_2,\bk\>} \\
1 + e^{i\<\bv_1,\bk\>} + e^{i\<\bv_2,\bk\>} & 0 \end{array}\right),\quad
 \bk\in\R^2.
$$
See Figure \ref{fig_honeycomb} for a portion of the honeycomb lattice
 linked by the nearest-neighbor hopping.

\begin{figure}
\begin{center}
\begin{picture}(245,130)(0,0)

\put(60,60){\circle*{5}}
\put(120,60){\circle*{5}}
\put(90,77.321){\circle*{5}}
\put(90,111.962){\circle*{5}}
\put(30,77.321){\circle*{5}}
\put(60,25.359){\circle*{5}}

\put(60,60){\line(1.732,1){30}}
\put(90,77.321){\line(0,1){34.641}}
\put(60,60){\line(-1.732,1){30}}
\put(60,60){\line(0,-1){34.641}}
\put(90,77.321){\line(1.732,-1){30}}

\put(60,25.359){\line(1.732,-1){15}}
\put(60,25.359){\line(-1.732,-1){15}}
\put(30,77.321){\line(-1.732,-1){15}}
\put(30,77.321){\line(0,1){17.321}}
\put(120,60){\line(1.732,1){15}}
\put(120,60){\line(0,-1){17.321}}
\put(90,111.962){\line(1.732,1){15}}
\put(90,111.962){\line(-1.732,1){15}}

\put(62,56){\tiny$(1,\bx)$}
\put(122,56){\tiny$(1,\bx+\bv_1)$}
\put(92,79.321){\tiny$(2,\bx)$}
\put(92,107.962){\tiny$(1,\bx+\bv_2)$}
\put(32,79.321){\tiny$(2,\bx-\bv_1)$}
\put(62,27.359){\tiny$(2,\bx-\bv_2)$}

\put(180,60){\vector(1,0){60}}
\put(180,60){\vector(1,1.732){30}}

\put(241,59){\small$\bv_1$}
\put(211,112.962){\small$\bv_2$}

\end{picture}
 \caption{A portion of the honeycomb lattice linked by the
 nearest-neighbor hopping.}\label{fig_honeycomb}
\end{center}
\end{figure}
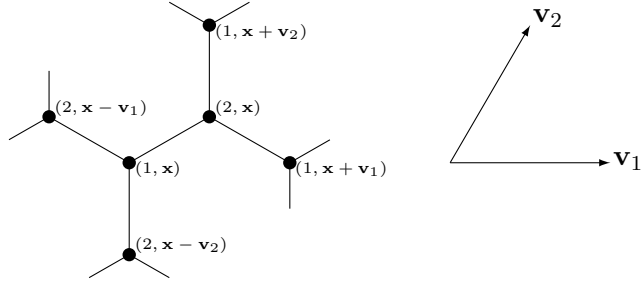

The eigenvalues of $E(\bk)$ are
 $(-1)^{\delta}|1+e^{i\<\bv_1,\bk\>}+e^{i\<\bv_2,\bk\>}|$, $(\delta
 \in\{0,1\})$. Let us set
 $e(\bk):=|1+e^{i\<\bv_1,\bk\>}+e^{i\<\bv_2,\bk\>}|$. 
The validity of the inequalities
 \eqref{eq_dispersion_upper_lower_bound},
 \eqref{eq_dispersion_upper_bound} and the regularity and the
 periodicity are clear. One can directly prove that
 \eqref{eq_dispersion_derivative}, \eqref{eq_hopping_matrix_derivative}
 hold with $\sn_1=\sn_2=1$. Observe that $e(\bk)=0$ and $\bk\in
 \G_{\infty}^*$ if and only if $\bk=\frac{2\pi}{3}\hbv_1+\frac{4\pi}{3}\hbv_2$ or
$\bk=\frac{4\pi}{3}\hbv_1+\frac{2\pi}{3}\hbv_2$. By making use of the expansions
\begin{align*}
&1+\cos x +\cos y\\
&=-\frac{\sqrt{3}}{2}\left(x-\frac{2}{3}\pi\right)+
 \frac{\sqrt{3}}{2}\left(y-\frac{4}{3}\pi\right)\\
&\quad+\sum_{n=2}^{\infty}\frac{1}{n!}\left(
\cos^{(n)}\left(\frac{2}{3}\pi\right)\left(x-\frac{2}{3}\pi\right)^n
+
\cos^{(n)}\left(\frac{4}{3}\pi\right)\left(y-\frac{4}{3}\pi\right)^n\right),\\ 
&\sin x +\sin y\\
& =-\frac{1}{2}\left(x-\frac{2}{3}\pi\right)-
 \frac{1}{2}\left(y-\frac{4}{3}\pi\right)\\
&\quad +\sum_{n=2}^{\infty}\frac{1}{n!}\left(
\sin^{(n)}\left(\frac{2}{3}\pi\right)\left(x-\frac{2}{3}\pi\right)^n
+
\sin^{(n)}\left(\frac{4}{3}\pi\right)\left(y-\frac{4}{3}\pi\right)^n\right),\quad
 (x,y\in \R),
\end{align*}
we can prove that there exist constants $\delta_1$, $\delta_2$,
 $\delta_3\in\R_{>0}$ such that for any $\hat{k}_1,\hat{k}_2\in [0,2\pi]$,
\begin{align}
&e(\hat{k}_1\hbv_1+\hat{k}_2\hbv_2)\label{eq_honeycomb_upper_property}\\
&\le \delta_1 \min\left\{
\left(\left(\hat{k}_1-\frac{2}{3}\pi\right)^2+ \left(\hat{k}_2-\frac{4}{3}\pi\right)^2
\right)^{1/2},
\left(\left(\hat{k}_1-\frac{4}{3}\pi\right)^2+ \left(\hat{k}_2-\frac{2}{3}\pi\right)^2
\right)^{1/2}\right\}\notag
\end{align}
and if $e(\hat{k}_1\hbv_1+\hat{k}_2\hbv_2)\le \delta_2$,
\begin{align}
&e(\hat{k}_1\hbv_1+\hat{k}_2\hbv_2)\label{eq_honeycomb_property}\\
&\ge \delta_3 \min\left\{
\left(\left(\hat{k}_1-\frac{2}{3}\pi\right)^2+ \left(\hat{k}_2-\frac{4}{3}\pi\right)^2
\right)^{1/2},
\left(\left(\hat{k}_1-\frac{4}{3}\pi\right)^2+ \left(\hat{k}_2-\frac{2}{3}\pi\right)^2
\right)^{1/2}\right\}.\notag
\end{align}
We can apply these properties to prove that the conditions
 \eqref{eq_dispersion_measure}, \eqref{eq_dispersion_measure_divided},
 \eqref{eq_dispersion_measure_divergence}, \eqref{eq_dispersion_power}
 hold with $\sa=2$, $d=2$, $\sn_1=\sn_2=1$.
\end{example}

\begin{example}[Hopping on the square littice with additional
 sites]\label{ex_6_band}
 To demonstrate the applicability of the multi-band formulation, let us
 consider a model on the square lattice with additional lattice
 points. The basis $\bv_1,\bv_2$ are equal to the canonical basis
 $\be_1,\be_2$ of $\R^2$. The lattice of our interest is identified with
 $\{1,2,3,4,5,6\}\times \G$. So we are going to construct a 6-band
 model. We define the hopping matrix with momentum variables by
\begin{align*}
&E(\bk):=\left(\begin{array}{cc}0 & E_0(\bk) \\ E_0(\bk)^* & 0
	       \end{array}\right),\quad 
E_0(\bk):=\left(\begin{array}{ccc} 2 & 1+e^{ik_2} & 1+e^{-ik_1} \\
                                   0 & 1+e^{ik_1} & 1+e^{-ik_2} \\
                                   1 & 0 & 1+e^{-ik_1} 
	       \end{array}\right),\quad\bk\in\R^2.
\end{align*}
A portion of the lattice $\{1,2,3,4,5,6\}\times \G$ linked by the
 hopping is pictured in Figure \ref{fig_6_band}.

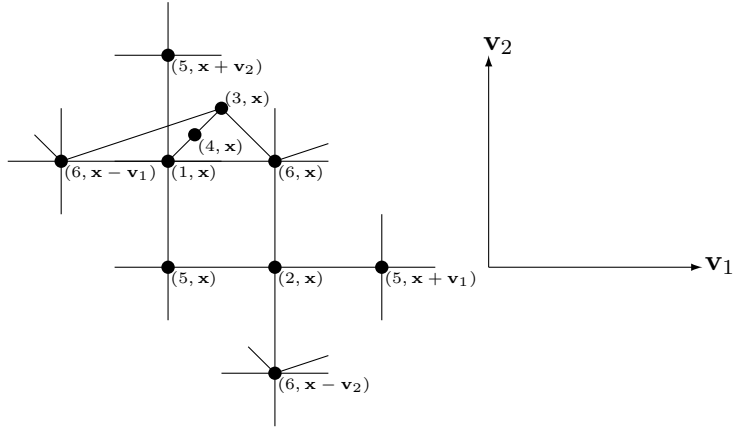
\begin{figure}
\begin{center}
\begin{picture}(270,160)(0,0)

\put(60,60){\circle*{5}}
\put(60,100){\circle*{5}}
\put(60,140){\circle*{5}}
\put(100,60){\circle*{5}}
\put(140,60){\circle*{5}}
\put(100,100){\circle*{5}}
\put(20,100){\circle*{5}}
\put(100,20){\circle*{5}}
\put(70,110){\circle*{5}}
\put(80,120){\circle*{5}}

\put(61,54){\tiny$(5,\bx)$}
\put(61,94){\tiny$(1,\bx)$}
\put(61,134){\tiny$(5,\bx+\bv_2)$}
\put(101,54){\tiny$(2,\bx)$}
\put(141,54){\tiny$(5,\bx+\bv_1)$}
\put(21,94){\tiny$(6,\bx-\bv_1)$}
\put(101,94){\tiny$(6,\bx)$}
\put(101,14){\tiny$(6,\bx-\bv_2)$}
\put(71,104){\tiny$(4,\bx)$}
\put(81,122){\tiny$(3,\bx)$}

\put(180,60){\vector(1,0){80}}
\put(180,60){\vector(0,1){80}}
\put(261,59){\small$\bv_1$}
\put(178,142){\small$\bv_2$}

\put(80,20){\line(1,0){40}}
\put(40,60){\line(1,0){120}}
\put(0,100){\line(1,0){120}}
\put(40,100){\line(1,0){40}}
\put(40,140){\line(1,0){40}}

\put(20,80){\line(0,1){40}}
\put(60,40){\line(0,1){120}}
\put(100,0){\line(0,1){120}}
\put(140,40){\line(0,1){40}}

\put(60,100){\line(1,1){20}}
\put(100,100){\line(-1,1){20}}
\put(20,100){\line(-1,1){10}}
\put(100,20){\line(-1,1){10}}

\put(20,100){\line(3,1){60}}
\put(100,100){\line(3,1){20}}
\put(100,20){\line(3,1){20}}

\end{picture}
 \caption{A portion of the lattice $\{1,2,3,4,5,6\}\times\G$ linked by the
 hopping.}\label{fig_6_band}
\end{center}
\end{figure}

To estimate the modulus of the eigenvalues of $E(\bk)$, it is efficient to estimate the
 eigenvalues of $E(\bk)^2$. Note that
\begin{align*}
&E(\bk)^2=\left(\begin{array}{cc} E_0(\bk)E_0(\bk)^* & 0 \\
                                0 & E_0(\bk)^*E_0(\bk) 
\end{array}\right),\\
&\det(xI_3-E_0(\bk)E_0(\bk)^*)=\det (xI_3- E_0(\bk)^*E_0(\bk))\\
&=x^3-(5+3|1+e^{ik_1}|^2+ 2|1+e^{ik_2}|^2)x^2\\
&\quad +\left(6\sum_{j=1}^2|1+e^{ik_j}|^2+2 |1+e^{ik_1}|^4 + |1+e^{ik_2}|^4
-|1+e^{ik_1}|^2|1+e^{ik_2}|^2\right)x\\
&\quad -\left(\sum_{j=1}^2|1+e^{ik_j}|^2\right)^2.
\end{align*}
Moreover, if $\sum_{j=1}^2|1+e^{ik_j}|^2\neq 0$, 
\begin{align*}
&\inf_{\bv\in \C^6\atop\text{with
 }\|\bv\|_{\C^6}=1}\|E(\bk)^2\bv\|_{\C^6}
\ge \frac{\det (E_0(\bk)E_0(\bk)^*)}{\sum_{l=1}^3\det
 (E_0(\bk)E_0(\bk)^*(i,j))_{1\le i,j\le 3\atop i,j\neq l}}\\
&=\frac{\left(\sum_{j=1}^2
 |1+e^{ik_j}|^2\right)^2}{6\sum_{j=1}^2|1+e^{ik_j}|^2+ 2|1+e^{ik_1}|^4
+ |1+e^{ik_2}|^4- |1+e^{ik_1}|^2|1+e^{ik_2}|^2}\\
&\ge \frac{\sum_{j=1}^2
 |1+e^{ik_j}|^2}{6+ 2\sum_{j=1}^2
 |1+e^{ik_j}|^2}\ge \frac{1}{22}\sum_{j=1}^2|1+e^{ik_j}|^2,\\
&\inf_{\bv\in \C^6\atop\text{with
 }\|\bv\|_{\C^6}=1}\|E(\bk)^2\bv\|_{\C^6}\le
\frac{\det (E_0(\bk)E_0(\bk)^*)}{\frac{1}{3}\sum_{l=1}^3\det
 (E_0(\bk)E_0(\bk)^*(i,j))_{1\le i,j\le 3\atop i,j\neq l}}\le\frac{1}{2}\sum_{j=1}^2
 |1+e^{ik_j}|^2.
\end{align*}
Therefore, if we define the function $e:\R^2\to \R$ by
$$
e(\bk):=\left(\frac{1}{22}\sum_{j=1}^2|1+e^{ik_j}|^2
\right)^{\frac{1}{2}},
$$
the condition \eqref{eq_dispersion_upper_lower_bound} holds with some
 positive constant $\sc$. It is apparent that $e(\cdot)$ satisfies
 \eqref{eq_dispersion_upper_bound} with some $\sc$ and the required regularity and
 periodicity. There is no difficulty to confirm that $e(\cdot)^2$,
 $E(\cdot)$ satisfy \eqref{eq_dispersion_derivative},
 \eqref{eq_hopping_matrix_derivative} with $\sn_1=\sn_2=1$. By using the
 inequalities 
\begin{align}
\frac{2}{\pi}\left(\sum_{j=1}^2|k_j-\pi|^2\right)^{\frac{1}{2}}\le 
\left(\sum_{j=1}^2|1+e^{i k_j}|^2\right)^{\frac{1}{2}}
\le \left(\sum_{j=1}^2|k_j-\pi|^2\right)^{\frac{1}{2}},\quad(\forall
 \bk\in\G^*_{\infty}),
\label{eq_6_band_useful_inequality}
\end{align}
we can check that \eqref{eq_dispersion_measure},
 \eqref{eq_dispersion_measure_divided},
 \eqref{eq_dispersion_measure_divergence}, 
\eqref{eq_dispersion_power} hold with $\sa=2$, $d=2$, $\sn_1=\sn_2=1$.
\end{example}

\begin{example}[3-dimensional model with nonuniform exponents]
\label{ex_different_exponent}
Let us give a 3-dimensional model where the exponents
 $\sn_1,\sn_2,\sn_3$ are not uniform. As in Example
 \ref{ex_degenerate_surface}, we let $\{\bv_j\}_{j=1}^3$,
 $\{\hbv_j\}_{j=1}^3$ be the canonical basis of $\R^3$. Set
\begin{align*}
E(\bk)=e(\bk):=\sum_{j=1}^2\cos k_j+2+(\cos k_3+1)^2.
\end{align*}
This is the dispersion relation of a one-band free electron model on the
 cubic lattice. The required regularity, periodicity,
 \eqref{eq_dispersion_upper_lower_bound} and \eqref{eq_dispersion_upper_bound} are
 clearly satisfied by $E(\cdot)$, $e(\cdot)$. By making use of the form
\begin{align}
e(\bk)=2\sum_{j=1}^2\sin^2\left(\frac{k_j-\pi}{2}\right)
+4\sin^4\left(\frac{k_3-\pi}{2}\right),
\label{eq_non_uniform_model_reform}
\end{align}
one can check that \eqref{eq_dispersion_derivative},
 \eqref{eq_hopping_matrix_derivative} hold with 
$\sn_1=\sn_2=2$, $\sn_3=4$. Moreover, for $R\in \R_{>0}$,
\begin{align*}
&\int_{\G_{\infty}^*}d\bk 1_{e(\bk)\le R}\le
 c\int_{[0,1]^3}d\bk1_{\sum_{j=1}^2k_j^2+k_3^4\le R}\le c
 \int_0^{R^{1/2}}dr r\int_0^1dk_3 1_{r^2+k_3^4\le R}\\
&\qquad\qquad\qquad\ \le c  \int_0^{R^{1/2}}dr r(R-r^2)^{\frac{1}{4}}\le c R^{\frac{5}{4}},\\
&\int_{\G_{\infty}^*}d\bk \frac{1_{e(\bk)\le R}}{e(\bk)}\le c\int_{[0,1]^3}d\bk
\frac{1_{\sum_{j=1}^2k_j^2+k_3^4\le R}}{\sum_{j=1}^2k_j^2+k_3^4}
\le c\int_0^{R^{1/4}}dk_3 \int_0^{\infty}dr \frac{r 1_{r^2+k_3^4\le
 R}}{r^2+k_3^4}\\
&\qquad\qquad\qquad\ \le c \int_0^{R^{1/4}}dk_3\log\left(\frac{R}{k_3^4}\right)\le
 cR^{\frac{1}{4}}.
\end{align*}
These calculations lead to the conclusion that
 \eqref{eq_dispersion_measure}, \eqref{eq_dispersion_measure_divided}
 hold with $\sa=\frac{5}{4}$. One can similarly confirm that
 \eqref{eq_dispersion_measure_divergence} holds. Since
$$
2\sa-1-\sum_{j=1}^3\frac{1}{\sn_j}=\frac{1}{4},
$$
the condition \eqref{eq_dispersion_power} holds as well.
\end{example}

\begin{example}[5-dimensional model whose Fermi surface does not
 degenerate into finite points]\label{ex_5_dimensional}
In the above examples the zero set of $e(\cdot)$ consists of finite
 points. Here let us give an example where the zero set of $e(\cdot)$
 does not degenerate into finite points. Let $d=5$ and let
 $\{\bv_j\}_{j=1}^5$, $\{\hbv_j\}_{j=1}^5$ be the canonical basis of $\R^5$. Define
 $E(\cdot):\R^5\to \R$ by
\begin{align*}
E(\bk):=(\cos k_1 + \cos k_2)^2+\sum_{j=3}^5\cos k_j+3
\end{align*}
and set $e(\bk):=E(\bk)$. It is possible to make an interpretation of
 this model in terms of hopping and chemical potential. We can see from
 the equality 
\begin{align}
e(\bk)=4\cos^2\left(\frac{k_1+k_2}{2}\right)\cos^2\left(\frac{k_1-k_2}{2}\right)+2\sum_{j=3}^5\sin^2\left(\frac{k_j-\pi}{2}\right)
\label{eq_factorized_dispersion}
\end{align}
that 
\begin{align*}
&\{\bk\in\G^*_{\infty}\ |\ e(\bk)=0\}\\
&=\left\{(k_1,k_2,\pi,\pi,\pi)\ \Big|\ \begin{array}{l} k_1,k_2\in
   [0,2\pi)\text{ satisfying
 }\\
 k_1+k_2\in \{\pi,3\pi\}\text{ or }k_1-k_2\in \{-\pi,\pi\}
\end{array}
\right\}.
\end{align*}
It is clear that $E(\cdot)$, $e(\cdot)$ satisfy
 \eqref{eq_dispersion_upper_lower_bound},
 \eqref{eq_dispersion_upper_bound} and the required regularity and
 periodicity. By using the equality
\begin{align*}
e(\bk)=(\cos k_1+\cos
 k_2)^2+2\sum_{j=3}^5\sin^2\left(\frac{k_j-\pi}{2}\right)
\end{align*}
we can check that \eqref{eq_dispersion_derivative},
 \eqref{eq_hopping_matrix_derivative} hold with $\sn_j=2$
 $(j=1,2,3,4,5)$. By using
 \eqref{eq_factorized_dispersion} and the inequality $\theta^2\ge
 \sin^2\theta \ge \frac{2^2}{\pi^2}\theta^2$ $(\theta\in
 [0,\frac{\pi}{2}])$ and changing the variables we have that for $R\in (0,1]$,
\begin{align*}
&\int_{\G_{\infty}^*}d\bk 1_{e(\bk)\le R}\le c \int_{[0,\frac{\pi}{2}]^5}d\bk
1_{\frac{2^6}{\pi^4}k_1^2k_2^2+\frac{2^3}{\pi^2}\sum_{j=3}^5k_j^2\le R}
\le c\int_{[0,1]^5}d\bk 1_{k_1^2k_2^2+\sum_{j=3}^5k_j^2\le R}\\
&\qquad\qquad\qquad\ \le c R^2 \int_{[0,R^{-1/4}]^2}dk_1dk_2\int_{[0,R^{-1/2}]^3}dk_3dk_4dk_5
1_{k_1^2k_2^2+\sum_{j=3}^5k_j^2\le 1}\\
&\qquad\qquad\qquad\ \le c R^2\int_0^1dr r^2  \int_{[0,R^{-1/4}]^2}dk_1dk_2
 1_{k_1k_2\le\sqrt{1-r^2}}\\
&\qquad\qquad\qquad\ \le c R^2 \int_0^1dr r^2  \left(\int_0^{R^{1/4}\sqrt{1-r^2}}dk_1
 R^{-\frac{1}{4}}+ \int_{R^{1/4}\sqrt{1-r^2}}^{R^{-1/4}}dk_1 \frac{\sqrt{1-r^2}}{k_1}\right)\\
&\qquad\qquad\qquad\ \le c R^2 \log(R^{-1}+1),\\
&\int_{\G_{\infty}^*}d\bk \frac{1_{e(\bk)\le R}}{e(\bk)}\le
 c\int_{[0,1]^5}d\bk \frac{1_{k_1^2k_2^2+\sum_{j=3}^5k_j^2\le
 R}}{k_1^2k_2^2+\sum_{j=3}^5k_j^2}\\
&\qquad\qquad\qquad\ \le c R \int_{[0,R^{-1/4}]^2}dk_1dk_2 \int_{[0,R^{-1/2}]^3}dk_3dk_4dk_5
\frac{1_{k_1^2k_2^2+\sum_{j=3}^5k_j^2\le
 1}}{k_1^2k_2^2+\sum_{j=3}^5k_j^2}\\
&\qquad\qquad\qquad\ \le c R \int_0^1dr \int_{[0,R^{-1/4}]^2}dk_1dk_21_{k_1k_2\le
 \sqrt{1-r^2}}\\
&\qquad\qquad\qquad\ \le c R\log(R^{-1}+1),\\
&\int_{\G_{\infty}^*}d\bk \frac{1}{e(\bk)}\le c\int_{[0,1]^5}d\bk
 \frac{1}{k_1^2k_2^2+\sum_{j=3}^5k_j^2}\le c.
\end{align*}
Since $\log(R^{-1}+1)\le c R^{-1/5}$ $(\forall R\in (0,1])$, 
the inequalities \eqref{eq_dispersion_measure},
 \eqref{eq_dispersion_measure_divided} hold with $\sa=\frac{9}{5}$. 
Let us check that the inequality
 \eqref{eq_dispersion_power} holds with
 $\sa=\frac{9}{5}$, $d=5$, $\sn_j=2$ ($j=1,2,3,4,5$). 
Moreover, 
\begin{align*}
\int_{\G_{\infty}^*}d\bk\frac{1}{e(\bk)^2+\eps}
&\ge c \int_{[0,1]^5}d\bk
 \left(\left(k_1^2k_2^2+\sum_{j=3}^5k_j^2\right)^2+\eps\right)^{-1}\\
&\ge
 c\int_0^1dr \int_{[0,1]^2}dk_1dk_2\frac{r^2}{(k_1^2k_2^2+r^2)^2+\eps}\\
&\ge c\int_0^1dr \int_{[0,r^{-1/2}]^2}dk_1dk_2
 \frac{r^3}{(k_1^2k_2^2+1)^2r^4+\eps}\to \infty,
\end{align*}
as $\eps\searrow 0$. Thus the condition
 \eqref{eq_dispersion_measure_divergence} holds as well. 
\end{example}

In summary, Theorem \ref{thm_main_theorem} and Corollary
\ref{cor_zero_temperature_limit} hold for the Hamiltonian
$\sH$ whose free part $\sH_0$ is defined with the hopping matrix $E(\cdot)$
given in Example
\ref{ex_degenerate_surface} -  Example \ref{ex_5_dimensional}. 

\begin{remark}\label{rem_nondegenerate_surface}
Let us see that the free dispersion relation of nearest-neighbor hopping
 electrons on the (hyper-)cubic lattice with non-degenerate Fermi
 surface does not satisfy the condition
 \eqref{eq_dispersion_measure_divided}, which will be essentially used
 to prove that $|U|$ can be taken independently of the temperature and
 the imaginary magnetic field in low temperature. Most of the necessary 
notations are defined in the same way as
 in Example \ref{ex_degenerate_surface}, apart from that now $d\in \N$
 and 
$$
E(\bk):=(-1)^{hop}2\sum_{j=1}^d\cos k_j -\mu
$$
with the chemical potential $\mu\in (-2d,2d)$. This free model was
 treated in our previous work \cite{K_BCS}. For any $R,\eps\in \R_{>0}$
 and a continuous function $e:\R^d\to \R_{\ge 0}$ satisfying
 \eqref{eq_dispersion_upper_lower_bound} we can derive by the coarea
 formula that 
\begin{align*}
\int_{\G_{\infty}^*}d\bk\frac{1_{e(\bk)\le R}}{e(\bk)+\eps}&\ge 
\int_{\G_{\infty}^*}d\bk\frac{1_{|E(\bk)|\le R}}{|E(\bk)|+\eps}\ge
 \frac{1}{2\sqrt{d}}\int_{\G_{\infty}^*}d\bk\frac{1_{|E(\bk)|\le
 R}|\nabla E(\bk)|}{|E(\bk)|+\eps}\\
&=\frac{1}{2\sqrt{d}}\int_{-R}^{R}d\eta \frac{\cH^{d-1}(\{\bk\in\G_{\infty}^*\ |\
 E(\bk)=\eta\})}{|\eta|+\eps},
\end{align*}
where $\cH^{d-1}$ is the $d-1$ dimensional Hausdorff measure. Set
 $K:=\frac{1}{2}(2d-|\mu|)$, $R':=\min\{K,R\}$. Then we can apply 
\cite[\mbox{Lemma 4.17}]{K_BCS} to derive that
\begin{align*}
\int_{\G_{\infty}^*}d\bk\frac{1_{e(\bk)\le R}}{e(\bk)+\eps}&\ge
 \frac{1}{2\sqrt{d}}\inf_{\eta\in [-K,K]}
\cH^{d-1}(\{\bk\in\G_{\infty}^*\ |\ E(\bk)=\eta\})
\int_{-R'}^{R'}dx
 \frac{1}{|x|+\eps}\\
&\ge
 \frac{1}{\sqrt{d}}\left(1_{d=1}+1_{d\ge
 2}\left(\frac{2d-|\mu|}{10(d-1)d}\right)^{d-1}\right)\log\left(\frac{R'+\eps}{\eps}\right).
\end{align*}
Since the right-hand side of the above inequality diverges to $\infty$
 as $\eps\searrow 0$, the condition
 \eqref{eq_dispersion_measure_divided} cannot be satisfied by this
 model.
\end{remark}

\section{Phase transitions}\label{sec_phase_transition}

In this section we analyze properties of the free energy density. We
focus on the right-hand side of \eqref{eq_free_energy_density} as a
function of $(\beta,\theta)\in \R_{>0}\times \R$ by fixing $U(\in
\R_{<0})$ with small magnitude. Mathematical arguments in this
section are essentially independent of the following sections, which aim
at proving Theorem \ref{thm_main_theorem} and Corollary
\ref{cor_zero_temperature_limit}. Our aim here is to describe the nature
of the phase transition happening in the system. 
The readers who want to prove Theorem \ref{thm_main_theorem} and
Corollary \ref{cor_zero_temperature_limit} first can skip to Section
\ref{sec_formulation} and come back to this section afterward. If we
think of the right-hand side of \eqref{eq_free_energy_density} alone, we
are free to substitute any large coupling constant and a hopping matrix
with different properties. However, we restrict our attention not
to deviate from the configuration where the derivation of
\eqref{eq_free_energy_density} is justified. We simply assume that $E:\R^d\to
\Mat(b,\C)$, $e:\R^d\to \R$ satisfy the same conditions as listed in
Subsection \ref{subsec_model_theorem} and $|U|$ is small as described subsequently.
We need to impose a couple more conditions on the function
$e(\cdot)$. Assume that there exist $\sr,\sfs\in\R_{>0}$, $\sc\in\R_{\ge 1}$
such that $0<\sr\le 1$, $1+2\sr\le \sfs$ and 
\begin{align}
&\int_{\G_{\infty}^*}d\bk\frac{1}{e(\bk)^2+A^2}\le \sc
 A^{-\sr},\label{eq_additional_upper_order}\\
&\int_{\G_{\infty}^*}d\bk\frac{1_{e(\bk)\le B}}{(e(\bk)^2+A^2)^2}\ge
 \sc^{-1}A^{-\sfs},\label{eq_additional_lower_order}\\
&(\forall A,B\in (0,1]\text{ with }0<A\le B\le 1).\notag
\end{align}
The conditions \eqref{eq_additional_upper_order},
\eqref{eq_additional_lower_order} are used only in this section and not
required to prove Theorem \ref{thm_main_theorem} and Corollary
\ref{cor_zero_temperature_limit}. Moreover, we assume that $U(\in
\R_{<0})$ satisfies 
\begin{align}
|U|< \frac{2}{bD_d \int_{\G_{\infty}^*}d\bk
 \frac{1}{e(\bk)}}.\label{eq_initial_smallness_coupling}
\end{align}
We will replace the upper bound on $|U|$ by a smaller constant in the
following. 

\begin{remark} In fact we do not use the conditions
 \eqref{eq_spatial_reflection_symmetry},
 \eqref{eq_dispersion_derivative}, \eqref{eq_hopping_matrix_derivative},
 \eqref{eq_dispersion_measure}, \eqref{eq_dispersion_power} in this
 section. Also, the regularity assumption of $E(\cdot)$, $e(\cdot)^2$
 can be relaxed.
\end{remark}

\subsection{Study of the models}\label{subsec_study_model}

In order to see that the conditions \eqref{eq_additional_upper_order},
\eqref{eq_additional_lower_order} are reasonable, let us check that 
the examples given in Subsection \ref{subsec_examples} satisfy these
additional conditions. 

For $x\in\R$ we let $\lfloor x\rfloor$ denote the largest integer which
does not exceed $x$. This notation will be used in the rest of the
paper. 

In the model given in Example \ref{ex_degenerate_surface} with $d=3$ the
conditions \eqref{eq_additional_upper_order},
\eqref{eq_additional_lower_order} hold with $\sr=\frac{1}{2}$,
$\sfs=\frac{5}{2}$ respectively. Since $1+2\sr\le \sfs$, the required
conditions are fulfilled in this case. In the case $d=4$ the condition
\eqref{eq_additional_lower_order} holds with $\sfs=2$. Note that 
\begin{align}
\int_{\G_{\infty}^*}d\bk\frac{1}{e(\bk)^2+A^2}\le c
\log(A^{-1}+1).\label{eq_additional_upper_log}
\end{align}
The condition \eqref{eq_additional_upper_order} holds with
e.g. $\sr=\frac{1}{2}$ and thus $1+2\sr\le \sfs$ in this case as well.  

By using \eqref{eq_honeycomb_upper_property} we can check that
$e(\cdot)$ introduced in Example \ref{ex_honeycomb} satisfies
\eqref{eq_additional_lower_order} with $\sfs=2$. It follows from
\eqref{eq_honeycomb_property} that the function $e(\cdot)$ satisfies
\eqref{eq_additional_upper_log} and thus
\eqref{eq_additional_upper_order} with $\sr=\frac{1}{2}$. Therefore, the
additional conditions are met in this example.

By using \eqref{eq_6_band_useful_inequality} we can confirm without
difficulty that $e(\cdot)$ introduced in Example \ref{ex_6_band}
satisfies \eqref{eq_additional_upper_order} with $\sr=\frac{1}{2}$ and
\eqref{eq_additional_lower_order} with $\sfs=2$ as well.

Let us study with the dispersion relation $e(\cdot)$ defined in Example
\ref{ex_different_exponent}. By using
\eqref{eq_non_uniform_model_reform} and changing variables
we have for $A,B\in (0,1]$ with $0<A\le B\le 1$ that 
\begin{align*}
\int_{\G_{\infty}^*}d\bk\frac{1}{e(\bk)^2+A^2}&\le c\int_{[0,1]^3}d\bk \left(
\left(\sum_{j=1}^2k_j^2+k_3^4\right)^2+A^2\right)^{-1}\\
&\le c
 A^{-\frac{3}{4}}\int_{\R^3}d\bk\left(
\left(\sum_{j=1}^2k_j^2+k_3^4\right)^2+1\right)^{-1}\le c A^{-\frac{3}{4}},\\
\int_{\G_{\infty}^*}d\bk\frac{1_{e(\bk)\le B}}{(e(\bk)^2+A^2)^2}&\ge
 c\int_{[0,1]^3}d\bk 1_{\sum_{j=1}^2k_j^2+k_3^4\le B}\left(
\left(\sum_{j=1}^2k_j^2+k_3^4\right)^2+A^2\right)^{-2}\\
&\ge c A^{-\frac{11}{4}}\int_{[0,1]^3}d\bk
1_{\sum_{j=1}^2k_j^2+k_3^4\le 1}
\left(
\left(\sum_{j=1}^2k_j^2+k_3^4\right)^2+1\right)^{-2}\\
&\ge  c A^{-\frac{11}{4}}.
\end{align*}
Thus the conditions \eqref{eq_additional_upper_order},
\eqref{eq_additional_lower_order} hold with $\sr=\frac{3}{4}$,
$\sfs=\frac{11}{4}$ respectively. Check that the condition $1+2\sr\le \sfs$
holds as well. 

Finally let us consider the function $e(\cdot)$ introduced in Example
\ref{ex_5_dimensional}. By \eqref{eq_factorized_dispersion} and change
of variables, 
\begin{align*}
&\int_{\G_{\infty}^*}d\bk\frac{1}{e(\bk)^2+A^2}\le c\int_{[0,1]^5}d\bk
\left(
\left(k_1^2k_2^2+\sum_{j=3}^5k_j^2\right)^2+A^2\right)^{-1}\\
&\le c \int_{[0,A^{-1/4}]^2}dk_1dk_2
 \int_{[0,A^{-1/2}]^3}dk_3dk_4dk_5\\
&\qquad \cdot \left(1_{k_1k_2\le 1}+\sum_{l=0}^{\lfloor \log (A^{-1/2})\rfloor}
 1_{e^l<k_1k_2\le e^{l+1}}\right) \left(
\left(k_1^2k_2^2+\sum_{j=3}^5k_j^2\right)^2+1\right)^{-1}\\
&\le c \int_{[0,A^{-1/4}]^2}dk_1dk_2 1_{k_1k_2\le 1} \int_{\R_{\ge
 0}^3}dk_3dk_4dk_5 \left(\left(\sum_{j=3}^5k_j^2\right)^2+1\right)^{-1}\\
&\quad+c\sum_{l=0}^{\lfloor \log (A^{-1/2})\rfloor}
\int_{[0,A^{-1/4}]^2}dk_1dk_2 1_{e^l<k_1k_2\le e^{l+1}}
\int_{\R_{\ge
 0}^3}dk_3dk_4dk_5 \left(e^{2l}+\sum_{j=3}^5k_j^2\right)^{-2}\\
&\le c \int_{[0,A^{-1/4}]^2}dk_1dk_2 1_{k_1k_2\le 1} 
+c\sum_{l=0}^{\lfloor \log (A^{-1/2})\rfloor} e^{-l}
\int_{[0,A^{-1/4}]^2}dk_1dk_2 1_{e^l<k_1k_2\le e^{l+1}}\\
&\le c\log(A^{-1}+1)
+c \sum_{l=0}^{\lfloor \log (A^{-1/2})\rfloor}
 e^{-l}\int_{e^lA^{1/4}}^{A^{-1/4}}dk_1\int_{e^lk_1^{-1}}^{e^{l+1}k_1^{-1}}dk_2\\
&\le c(\log(A^{-1}+1))^2,\\
&\int_{\G_{\infty}^*}d\bk\frac{1_{e(\bk)\le B}}{(e(\bk)^2+A^2)^2}\\
&\ge c
\int_{\G_{\infty}^*}d\bk
 1_{\frac{1}{4}(k_1+k_2-\pi)^2(k_1-k_2-\pi)^2+\frac{1}{2}\sum_{j=3}^5(k_j-\pi)^2\le B}\\
&\qquad\cdot \left(\left((k_1+k_2-\pi)^2(k_1-k_2-\pi)^2+\sum_{j=3}^5(k_j-\pi)^2\right)^2+A^2\right)^{-2}\\
&\ge  c
\int_{[0,\pi]^5}d\bk
 1_{(k_1+k_2)^2(k_1-k_2)^2+\sum_{j=3}^5k_j^2\le B}\\
&\qquad\cdot \left(\left((k_1+k_2)^2(k_1-k_2)^2+\sum_{j=3}^5k_j^2\right)^2+A^2\right)^{-2}\\
&\ge  c A^{-2}
\int_{[0,1]^5}d\bk
 1_{(k_1+k_2)^2(k_1-k_2)^2+\sum_{j=3}^5k_j^2\le 1}\\
&\qquad\cdot \left(\left((k_1+k_2)^2(k_1-k_2)^2+\sum_{j=3}^5k_j^2\right)^2+1\right)^{-2}\\
&\ge c A^{-2}.
\end{align*}
Thus the inequalities \eqref{eq_additional_upper_order},
\eqref{eq_additional_lower_order} hold with e.g. $\sr=\frac{1}{2}$, $\sfs=2$.

 We have seen that in each example of Subsection
 \ref{subsec_examples} the function $e(\cdot)$ satisfies the required conditions of this subsection.

\subsection{Phase boundaries}\label{subsec_phase_boundaries}

Let us define the map $g:\R_{>0}\times \R\times \R\to \R\cup\{\infty\}$
by 
\begin{align*}
&g(x,y,z)\\
&:=\left\{\begin{array}{l}\displaystyle
	   -\frac{2}{|U|}+D_d\int_{\G_{\infty}^*}d\bk \Tr \left(\frac{\sinh(x\sqrt{E(\bk)^2+z^2})}{(\cos(xy/2)+\cosh(x\sqrt{E(\bk)^2+z^2}))\sqrt{E(\bk)^2+z^2}}
\right)\\
\quad\text{ if }\frac{xy}{2}\notin \pi (2\Z+1)\text{ or }z\neq 0,\\
\infty\text{ if }\frac{xy}{2}\in \pi (2\Z+1)\text{ and }z=0.
\end{array}
\right.
\end{align*}
We can check that the function $(x,y,z)\mapsto g(x,y,z)$ is
$C^{\infty}$-class in the open set
\begin{align*}
\left\{(x,y,z)\in \R_{>0}\times \R\times \R\ \big|\ 
\frac{xy}{2}\notin \pi (2\Z+1)\text{ or }z\neq 0\right\}
\end{align*}
of $\R^3$.

For $(\beta,\theta)\in \R_{>0}\times \R$ let $\D(\beta,\theta)$ be such
that $\D(\beta,\theta)\ge 0$ and $g(\beta,\theta,\D(\beta,\theta))=0$ if
$g(\beta,\theta,0)\ge 0$, $\D(\beta,\theta)=0$ if $g(\beta,\theta,0)<
0$. This rule defines the function $\D:\R_{>0}\times \R\to \R_{\ge
0}$. The well-definedness of the function $\D(\cdot,\cdot)$ is
guaranteed by Lemma \ref{lem_gap_equation_solvability}.

The goal of this subsection is to characterize the set $
\{(\beta,\theta)\in\R_{>0}\times\R\ |\ \D(\beta,\theta)>0\}$. We will
see that this set consists of countable disjoint subsets. Let us call
the boundaries of the disjoint subsets phase boundaries. Our goal here
is equivalent to characterizing the phase boundaries. 

Define the subsets $O_{+}$, $O_-$ of $\R^2$ by
\begin{align*}
&O_+:=\{(x,y)\in\R_{>0}\times\R\ |\ g(x,y,0)>0\},\\
&O_-:=\{(x,y)\in\R_{>0}\times\R\ |\ g(x,y,0)<0\}.
\end{align*}
We can see that $O_+$, $O_-$ are open subsets of $\R^2$. 

\begin{lemma}\label{lem_regularity_Delta}
\begin{align*}
\D\in C(\R_{>0}\times \R),\quad \D|_{O_+\cup O_-}\in C^{\infty}(O_+\cup
 O_-).
\end{align*}
\end{lemma}

\begin{proof}
It is trivial that $\D|_{O_-}\in C^{\infty}(O_-)$. We have observed that
the functions 
\begin{align*}
x\mapsto \frac{\sinh x}{(\eps+\cosh x)x}:(0,\infty)\to \R,\quad (\eps\in
 [-1,1])
\end{align*}
are strictly monotone decreasing in the proof of Lemma
 \ref{lem_gap_equation_solvability}. It follows that 
\begin{align}
\frac{\partial g}{\partial z}(x,y,z)<0,\quad (\forall
 (x,y,z)\in\R_{>0}\times\R\times
 \R_{>0}).\label{eq_positivity_3rd_derivative}
\end{align}
Thus 
\begin{align*}
\frac{\partial g}{\partial z}(\beta,\theta,\D(\beta,\theta))<0,\quad (\forall
 (\beta,\theta)\in O_+).
\end{align*}
By the implicit function theorem we have that 
$\D|_{O_+}\in C^{\infty}(O_+)$, or $\D|_{O_+\cup
 O_-}\in C^{\infty}(O_+\cup O_-)$.

Let us prove that $\D\in C(\R_{>0}\times \R)$. It is sufficient to prove
 the continuity at each point belonging to $\R_{>0}\times \R\backslash O_+\cup O_-$. Let
 $(\beta_0,\theta_0)\in\R_{>0}\times \R \backslash O_+\cup O_-$. By
 definition $g(\beta_0,\theta_0,0)=0$, $\D(\beta_0,\theta_0)=0$ and
 $\beta_0\theta_0/2\notin \pi(2\Z+1)$. Suppose that there exists
 $\eps\in \R_{>0}$ such that for any $\delta\in \R_{>0}$ there exists
 $(\beta_{\delta},\theta_{\delta})\in\R_{>0}\times \R$ such that 
$\|(\beta_0,\theta_0)-(\beta_{\delta},\theta_{\delta})\|_{\R^2}<\delta$
 and $\D(\beta_{\delta},\theta_{\delta})\ge \eps$. Then,
\begin{align*}
0&=g(\beta_{\delta},\theta_{\delta},\D(\beta_{\delta},\theta_{\delta}))\le
 g(\beta_{\delta},\theta_{\delta},\eps)\le
 \sup_{(\beta,\theta)\in\R_{>0}\times \R\atop\text{with
 }\|(\beta_0,\theta_0)-(\beta,\theta)\|_{\R^2}<\delta}g(\beta,\theta,\eps).
\end{align*}
By sending $\delta$ to 0 we have that $0\le
 g(\beta_0,\theta_0,\eps)<g(\beta_0,\theta_0,0)=0$, which is a
 contradiction. Thus $\lim_{(\beta,\theta)\to
 (\beta_0,\theta_0)}\D(\beta,\theta)=0=\D(\beta_0,\theta_0)$, which
 implies that $\D\in C(\R_{>0}\times \R)$.
\end{proof}

The next lemma states the existence of critical values of the imaginary
magnetic field in $[0,4\pi/\beta]$. 

\begin{lemma}\label{lem_critical_theta}
For any $\beta\in \R_{>0}$ there uniquely exist $\theta_{c,1}\in
 (0,2\pi/\beta)$, $\theta_{c,2}\in (2\pi/\beta,4\pi/\beta)$ such that
\begin{align*}
&g(\beta,\theta_{c,1},0)=g(\beta,\theta_{c,2},0)=0,\\
&g(\beta,\theta,0)>0,\quad (\forall \theta\in
 (\theta_{c,1},\theta_{c,2})),\\
&g(\beta,\theta,0)<0,\quad (\forall \theta\in [0,\theta_{c,1})\cup
 (\theta_{c,2},4\pi/\beta]).
\end{align*}
\end{lemma}

\begin{proof}
One can see from the definition that $\theta\mapsto g(\beta,\theta,0)$
 is strictly monotone increasing in $(0,2\pi/\beta)$, strictly monotone
 decreasing in $(2\pi/\beta,4\pi/\beta)$ and continuous in
 $(0,2\pi/\beta)\cup (2\pi/\beta,4\pi/\beta)$. By the assumption
 \eqref{eq_initial_smallness_coupling}, 
\begin{align*}
g(\beta,0,0)=g\left(\beta,\frac{4\pi}{\beta},0\right)\le
 -\frac{2}{|U|}+bD_d\int_{\G_{\infty}^*}d\bk \frac{1}{e(\bk)}<0.
\end{align*}
Also, by \eqref{eq_dispersion_measure_divergence} $\lim_{\theta\to
 2\pi/\beta}g(\beta,\theta,0)=\infty$. We can deduce the claim from
 these properties.
\end{proof}

By Lemma \ref{lem_critical_theta} we can define the functions
$\theta_{c,1}:\R_{>0}\to (0,2\pi/\beta)$, $\theta_{c,2}:\R_{>0}\to
(2\pi/\beta,4\pi/\beta)$. 

For any parameters $\alpha_1,\alpha_2,\cdots,\alpha_n$ we let
$c(\alpha_1,\alpha_2,\cdots,\alpha_n)$ denote a positive constant
depending only on $\alpha_1,\alpha_2,\cdots,\alpha_n$. This notational
rule will be used not only in the proof of the next lemma but throughout
the rest of the paper.

\begin{lemma}\label{lem_critical_theta_property}
There exist positive constants $c_3,c_4$ depending only on
 $b,D_d,\sc,\sr,\sfs$ such that the following statements hold for any
 $U\in (-c_3,0)$. 
\begin{enumerate}[(i)]
\item\label{item_simple_temperature_bound}
\begin{align*}
\left|\frac{\theta_{c,j}(\beta)}{2}-\frac{\pi}{\beta}\right|\le
 \frac{\pi}{2\beta},\quad(\forall \beta\in\R_{>0},\ j\in\{1,2\}).
\end{align*}
\item\label{item_global_temperature_bound}
\begin{align*}
\left|\frac{\theta_{c,j}(\beta)}{2}-\frac{\pi}{\beta}\right|\le
 c_4 \left(\frac{|U|}{\beta}\right)^{\frac{1}{\sr}},\quad(\forall \beta\in\R_{>0},\ j\in\{1,2\}).
\end{align*}
\item\label{item_high_temperature_bound}
\begin{align*}
\left|\frac{\theta_{c,j}(\beta)}{2}-\frac{\pi}{\beta}\right|\le
 c_4 \left(\frac{|U|}{\beta}\right)^{\frac{1}{2}},\quad(\forall \beta\in(0,1],\ j\in\{1,2\}).
\end{align*}
\item\label{item_critical_theta_regularity}
$\theta_{c,j}\in C^{\infty}(\R_{>0})$ and 
\begin{align*}
\frac{d\theta_{c,j}}{d\beta}(\beta)<0,\quad (\forall \beta\in \R_{>0},\
 j\in\{1,2\}).
\end{align*}
\end{enumerate}
\end{lemma}

\begin{proof}
\eqref{item_simple_temperature_bound}: Take $\beta\in \R_{>0}$,
 $j\in\{1,2\}$. Assume that
 $|\theta_{c,j}(\beta)/2-\pi/\beta|>\pi/(2\beta)$. Then,
\begin{align*}
0&=g(\beta,\theta_{c,j}(\beta),0)\le
 -\frac{2}{|U|}+D_d\int_{\G_{\infty}^*}d\bk\Tr \left(\frac{\tanh (\beta
 |E(\bk)|)}{|E(\bk)|}\right)\\
&\le -\frac{2}{|U|}+b D_d\int_{\G_{\infty}^*}d\bk\frac{1}{e(\bk)},
\end{align*}
which contradicts the condition
 \eqref{eq_initial_smallness_coupling}. Thus the claim holds true.

\eqref{item_global_temperature_bound}: By
  \eqref{eq_dispersion_measure_divided} and \eqref{eq_additional_upper_order}, 
\begin{align*}
0&=g(\beta,\theta_{c,j}(\beta),0)\le
 -\frac{2}{|U|}+bD_d\int_{\G_{\infty}^*}d\bk\frac{\sinh(\beta
 e(\bk))}{(\cos(\beta \theta_{c,j}(\beta)/2)+\cosh(\beta e(\bk)))e(\bk)}\\
&\le
 -\frac{2}{|U|}+c(b,D_d)\beta^{-1}\int_{\G_{\infty}^*}d\bk\frac{1_{e(\bk)\le \beta^{-1}}}{e(\bk)^2+|\theta_{c,j}(\beta)/2-\pi/\beta|^2} 
+c(b,D_d)\int_{\G_{\infty}^*}d\bk\frac{1}{e(\bk)}\\
&\le 
-\frac{2}{|U|} + c(b,D_d,\sc)\\
&\quad +c(b,D_d,\sc)\beta^{-1}\Bigg(1_{|\theta_{c,j}(\beta)/2-\pi/\beta|\le
 1}\left|\frac{\theta_{c,j}(\beta)}{2}-\frac{\pi}{\beta}\right|^{-\sr}
\\
&\qquad\qquad\qquad\qquad +1_{|\theta_{c,j}(\beta)/2-\pi/\beta|>
 1}\left|\frac{\theta_{c,j}(\beta)}{2}-\frac{\pi}{\beta}\right|^{-2}\Bigg)\\
&\le
 -\frac{2}{|U|}+c(b,D_d,\sc)\left(\beta^{-1}\left|\frac{\theta_{c,j}(\beta)}{2}-\frac{\pi}{\beta}\right|^{-\sr}+1\right).
\end{align*}
To derive the last inequality, we also used that $0<\sr \le 1$.
If $|U|\le c(b,D_d,\sc)^{-1}$,
\begin{align*}
0\le
 -\frac{1}{|U|}+c(b,D_d,\sc)\beta^{-1}\left|\frac{\theta_{c,j}(\beta)}{2}-\frac{\pi}{\beta}\right|^{-\sr}.
\end{align*}
This leads to the result.

\eqref{item_high_temperature_bound}:
Since $\beta\le 1$, 
\begin{align*}
0&=g(\beta,\theta_{c,j}(\beta),0)\le -\frac{2}{|U|}+c(b,D_d,\sc)\beta^{-1}
\int_{\G_{\infty}^*}d\bk\frac{1}{e(\bk)^2+|\theta_{c,j}(\beta)/2-\pi/\beta|^2} \\
&\le  -\frac{2}{|U|}+c(b,D_d,\sc)\beta^{-1}
 \left|\frac{\theta_{c,j}(\beta)}{2}-\frac{\pi}{\beta}\right|^{-2},
\end{align*}
which implies the result.

\eqref{item_critical_theta_regularity}:
For $x\in\R_{>0}$, $y\in (0,2\pi/x)\cup (2\pi/x,4\pi/x)$,
\begin{align}
\frac{\partial g}{\partial
 y}(x,y,0)=\frac{D_d}{2}x\sin\left(\frac{xy}{2}\right)\int_{\G_{\infty}^*}d\bk\Tr
 \left(\frac{\sinh(x|E(\bk)|)}{(\cos(xy/2)+\cosh(x E(\bk)))^2|E(\bk)|}
\right).\label{eq_gap_equation_2nd_derivative}
\end{align}
Thus
\begin{align}
&\frac{\partial g}{\partial y}(x,y,0)>0,\quad (\forall y\in
 (0,2\pi/x)),\label{eq_gap_equation_2nd_derivative_sign}\\
&\frac{\partial g}{\partial y}(x,y,0)<0,\quad (\forall y\in
 (2\pi/x,4\pi/x)).\notag
\end{align}
Therefore the implicit function theorem ensures that $\theta_{c,j}\in
 C^{\infty}(\R_{>0})$.

 Let us determine the sign of $\frac{\partial g}{\partial
 x}(\beta,\theta_{c,1}(\beta),0)$, $\frac{\partial g}{\partial
 x}(\beta,\theta_{c,2}(\beta),0)$. For $x\in \R_{>0}$, $y\in \R$ with
 $xy/2\notin \pi(2\Z+1)$,
\begin{align}
\frac{\partial g}{\partial
 x}(x,y,0)=&\frac{D_d}{2}\sin\left(\frac{xy}{2}\right)y
 \int_{\G_{\infty}^*}d\bk\Tr\left(\frac{\sinh(x|E(\bk)|)}{(\cos(xy/2)+\cosh(x E(\bk)))^2|E(\bk)|}\right)\label{eq_gap_equation_1st_derivative}\\
&+D_d \int_{\G_{\infty}^*}d\bk\Tr\left(\frac{\cosh(x
 E(\bk))(1+\cos(xy/2))}{(\cos(xy/2)+\cosh(x E(\bk)))^2}\right)\notag\\
&-
D_d \int_{\G_{\infty}^*}d\bk\Tr\left(\frac{\cosh(x
 E(\bk))-1}{(\cos(xy/2)+\cosh(x E(\bk)))^2}\right).\notag
\end{align}
By the result of \eqref{item_simple_temperature_bound},
\begin{align}
&\theta_{c,1}(\beta)\sin\left(\frac{\beta\theta_{c,1}(\beta)}{2}\right)\ge
 c\left|\frac{\theta_{c,1}(\beta)}{2}-\frac{\pi}{\beta}\right|,\label{eq_critical_theta_1_lower}\\
&\theta_{c,2}(\beta)\sin\left(\frac{\beta\theta_{c,2}(\beta)}{2}\right)\le
 -c\left|\frac{\theta_{c,2}(\beta)}{2}-\frac{\pi}{\beta}\right|.\label{eq_critical_theta_2_upper}
\end{align}

Let us consider the case that $\beta\ge 1$. By the claim
 \eqref{item_global_temperature_bound} and $0<\sr\le 1$, if $|U|\le c_4^{-\sr}$,
\begin{align}
\left|\frac{\theta_{c,j}(\beta)}{2}-\frac{\pi}{\beta}\right|\le\frac{1}{\beta}\le
 1,\quad(\forall j\in\{1,2\}).\label{eq_assumption_for_additional}
\end{align}
Then by \eqref{eq_additional_upper_order},
 \eqref{eq_additional_lower_order}, 
\eqref{eq_gap_equation_1st_derivative},
\eqref{eq_critical_theta_1_lower},
\eqref{eq_assumption_for_additional}
 and the claim \eqref{item_global_temperature_bound} again
\begin{align*}
&\frac{\partial g}{\partial x}(\beta,\theta_{c,1}(\beta),0)\\
&\ge c D_d
\left|\frac{\theta_{c,1}(\beta)}{2}-\frac{\pi}{\beta}\right|\beta 
\int_{\G_{\infty}^*}d\bk\Tr
 \left(\frac{1}{(\cos(\beta\theta_{c,1}(\beta)/2)+\cosh(\beta
 E(\bk)))^2}\right)\\
&\quad-D_d \int_{\G_{\infty}^*}d\bk\Tr
 \left(\frac{1}{\cos(\beta\theta_{c,1}(\beta)/2)+\cosh(\beta
 E(\bk))}\right)\\
&\ge c(D_d,\sc)
\left|\frac{\theta_{c,1}(\beta)}{2}-\frac{\pi}{\beta}\right|\beta^{-3} 
\int_{\G_{\infty}^*}d\bk1_{e(\bk)\le
 \beta^{-1}}\left(e(\bk)^2+\left|\frac{\theta_{c,1}(\beta)}{2}-\frac{\pi}{\beta}\right|^2\right)^{-2}\\
&\quad -c(b,D_d)\beta^{-2}\int_{\G_{\infty}^*}d\bk\left(e(\bk)^2+\left|\frac{\theta_{c,1}(\beta)}{2}-\frac{\pi}{\beta}\right|^2\right)^{-1}\\
&\ge
 c(D_d,\sc)\beta^{-3}\left|\frac{\theta_{c,1}(\beta)}{2}-\frac{\pi}{\beta}\right|^{1-\sfs}-c(b,D_d,\sc)\beta^{-2}\left|\frac{\theta_{c,1}(\beta)}{2}-\frac{\pi}{\beta}\right|^{-\sr}\\
&\ge
 c(D_d,\sc)\beta^{-3}\left|\frac{\theta_{c,1}(\beta)}{2}-\frac{\pi}{\beta}\right|^{1-\sfs}\left(1-c(b,D_d,\sc)\beta \left|\frac{\theta_{c,1}(\beta)}{2}-\frac{\pi}{\beta}\right|^{-1+\sfs-\sr}\right)\\
&\ge
 c(D_d,\sc)\beta^{-3}\left|\frac{\theta_{c,1}(\beta)}{2}-\frac{\pi}{\beta}\right|^{1-\sfs}\left(1-c(b,D_d,\sc,\sr,\sfs)\beta^{-\frac{\sfs-1-2\sr}{\sr}}|U|^{\frac{\sfs-1-\sr}{\sr}}\right)\\
&\ge
 c(D_d,\sc)\beta^{-3}\left|\frac{\theta_{c,1}(\beta)}{2}-\frac{\pi}{\beta}\right|^{1-\sfs}\left(1-c(b,D_d,\sc,\sr,\sfs)|U|^{\frac{\sfs-1-\sr}{\sr}}\right).
\end{align*}
In the last inequality we used the conditions $\sfs-1-2\sr\ge 0$, $\beta
 \ge 1$. Thus if $c(b,D_d,\sc,\sr,\sfs)|U|^{(\sfs-1-\sr)/\sr}<1$,
$$
\frac{\partial g}{\partial x}(\beta,\theta_{c,1}(\beta),0)>0.
$$
Similarly by using \eqref{eq_additional_upper_order},
 \eqref{eq_additional_lower_order}, \eqref{eq_critical_theta_2_upper},
 \eqref{eq_assumption_for_additional}, the claim
 \eqref{item_global_temperature_bound} and the conditions
 $\sfs-1-2\sr\ge 0$, $\beta \ge 1$ we can derive from
 \eqref{eq_gap_equation_1st_derivative} that 
\begin{align*}
&\frac{\partial g}{\partial x}(\beta,\theta_{c,2}(\beta),0)\\
&\le -c D_d
\left|\frac{\theta_{c,2}(\beta)}{2}-\frac{\pi}{\beta}\right|\beta 
\int_{\G_{\infty}^*}d\bk\Tr
 \left(\frac{1}{(\cos(\beta\theta_{c,2}(\beta)/2)+\cosh(\beta
 E(\bk)))^2}\right)\\
&\quad+D_d \int_{\G_{\infty}^*}d\bk\Tr
 \left(\frac{\cosh(\beta E(\bk))(1+\cos(\beta \theta_{c,2}(\beta)/2))}{(\cos(\beta\theta_{c,2}(\beta)/2)+\cosh(\beta
 E(\bk)))^2}\right)\\
&\le -c(D_d,\sc)
\left|\frac{\theta_{c,2}(\beta)}{2}-\frac{\pi}{\beta}\right|\beta^{-3} 
\int_{\G_{\infty}^*}d\bk1_{e(\bk)\le
 \beta^{-1}}\left(e(\bk)^2+\left|\frac{\theta_{c,2}(\beta)}{2}-\frac{\pi}{\beta}\right|^2\right)^{-2}\\
&\quad + D_d \int_{\G_{\infty}^*}d\bk1_{e(\bk)\le \beta^{-1}}\Tr
 \left(\frac{\cosh(\beta E(\bk))}{\cos(\beta\theta_{c,2}(\beta)/2)+\cosh(\beta
 E(\bk))}\right)\\
&\quad +c(D_d) \int_{\G_{\infty}^*}d\bk 1_{e(\bk)> \beta^{-1}} \Tr
 \left(\frac{1+\cos(\beta
 \theta_{c,2}(\beta)/2)}{\cos(\beta\theta_{c,2}(\beta)/2)+\cosh(\beta
 E(\bk))}\right)\\
&\le -c(D_d,\sc)
\left|\frac{\theta_{c,2}(\beta)}{2}-\frac{\pi}{\beta}\right|\beta^{-3} 
\int_{\G_{\infty}^*}d\bk1_{e(\bk)\le
 \beta^{-1}}\left(e(\bk)^2+\left|\frac{\theta_{c,2}(\beta)}{2}-\frac{\pi}{\beta}\right|^2\right)^{-2}\\
&\quad+ c(b,D_d)\beta^{-2}\int_{\G_{\infty}^*}d\bk\left(e(\bk)^2+\left|\frac{\theta_{c,2}(\beta)}{2}-\frac{\pi}{\beta}\right|^2\right)^{-1}\\
&\le -
 c(D_d,\sc)\beta^{-3}\left|\frac{\theta_{c,2}(\beta)}{2}-\frac{\pi}{\beta}\right|^{1-\sfs}\left(1-c(b,D_d,\sc,\sr,\sfs)|U|^{\frac{\sfs-1-\sr}{\sr}}\right).
\end{align*}
To derive the third inequality, we also used that 
\begin{align*}
x\mapsto \frac{\cosh x}{\cos(\beta \theta_{c,2}(\beta)/2)+\cosh
 x}:[0,\infty)\to \R
\end{align*}
is non-increasing.
Thus on the assumption $c(b,D_d,\sc,\sr,\sfs)|U|^{(\sfs-1-\sr)/\sr}<1$,
$$
\frac{\partial g}{\partial x}(\beta,\theta_{c,2}(\beta),0)<0.
$$

Next let us assume that $\beta <1$. By
  \eqref{eq_additional_upper_order}, \eqref{eq_additional_lower_order},
 \eqref{eq_gap_equation_1st_derivative} and \eqref{eq_critical_theta_1_lower}
\begin{align*}
&\frac{\partial g}{\partial x}(\beta,\theta_{c,1}(\beta),0)\\
&\ge c D_d
\left|\frac{\theta_{c,1}(\beta)}{2}-\frac{\pi}{\beta}\right|\beta 
\int_{\G_{\infty}^*}d\bk\Tr
 \left(\frac{1}{(\cos(\beta\theta_{c,1}(\beta)/2)+\cosh(\beta
 E(\bk)))^2}\right)\\
&\quad-D_d \int_{\G_{\infty}^*}d\bk\Tr
 \left(\frac{1}{\cos(\beta\theta_{c,1}(\beta)/2)+\cosh(\beta
 E(\bk))}\right)\\
&\ge c(D_d,\sc)
\left|\frac{\theta_{c,1}(\beta)}{2}-\frac{\pi}{\beta}\right|\beta^{-3} 
\int_{\G_{\infty}^*}d\bk\left(e(\bk)^2+\left|\frac{\theta_{c,1}(\beta)}{2}-\frac{\pi}{\beta}\right|^2\right)^{-2}\\
&\quad -c(b,D_d)\beta^{-2}\int_{\G_{\infty}^*}d\bk\left(e(\bk)^2+\left|\frac{\theta_{c,1}(\beta)}{2}-\frac{\pi}{\beta}\right|^2\right)^{-1}\\
&\ge 1_{|\theta_{c,1}(\beta)/2-\pi/\beta|\le 1}\left(
c(D_d,\sc)\beta^{-3}\left|\frac{\theta_{c,1}(\beta)}{2}-\frac{\pi}{\beta}\right|^{1-\sfs}-c(b,D_d,\sc)\beta^{-2}\left|\frac{\theta_{c,1}(\beta)}{2}-\frac{\pi}{\beta}\right|^{-\sr}\right)\\
&\quad + 1_{|\theta_{c,1}(\beta)/2-\pi/\beta|> 1}\left(
c(D_d,\sc)\beta^{-3}\left|\frac{\theta_{c,1}(\beta)}{2}-\frac{\pi}{\beta}\right|^{-3}-c(b,D_d)\beta^{-2}\left|\frac{\theta_{c,1}(\beta)}{2}-\frac{\pi}{\beta}\right|^{-2}\right)\\
&\ge 1_{|\theta_{c,1}(\beta)/2-\pi/\beta|\le 1}
c(D_d,\sc)\beta^{-3}\left|\frac{\theta_{c,1}(\beta)}{2}-\frac{\pi}{\beta}\right|^{1-\sfs}\left(1-
 c(b,D_d,\sc)\beta \left|\frac{\theta_{c,1}(\beta)}{2}-\frac{\pi}{\beta}\right|^{\sfs-1-\sr}\right)\\
&\quad + 1_{|\theta_{c,1}(\beta)/2-\pi/\beta|> 1}
c(D_d,\sc)\beta^{-3}\left|\frac{\theta_{c,1}(\beta)}{2}-\frac{\pi}{\beta}\right|^{-3}
\left(1-c(b,D_d,\sc)\beta
 \left|\frac{\theta_{c,1}(\beta)}{2}-\frac{\pi}{\beta}\right|\right).
\end{align*}
By the claim \eqref{item_global_temperature_bound} and the condition
 $1+2\sr\le \sfs$, 
\begin{align*}
& 1_{|\theta_{c,1}(\beta)/2-\pi/\beta|\le 1}
\left(1-
 c(b,D_d,\sc)\beta
 \left|\frac{\theta_{c,1}(\beta)}{2}-\frac{\pi}{\beta}\right|^{\sfs-1-\sr}\right)\\
&\ge 
 1_{|\theta_{c,1}(\beta)/2-\pi/\beta|\le 1}
\left(1-
 c(b,D_d,\sc)\beta
 \left|\frac{\theta_{c,1}(\beta)}{2}-\frac{\pi}{\beta}\right|^{\sr}\right)\\
&\ge 1_{|\theta_{c,1}(\beta)/2-\pi/\beta|\le 1}
\left(1-
 c(b,D_d,\sc,\sr,\sfs)|U|\right).
\end{align*}
Also, by the claim \eqref{item_high_temperature_bound} and the
 assumption $\beta<1$, 
\begin{align*}
& 1_{|\theta_{c,1}(\beta)/2-\pi/\beta|> 1}
\left(1-
 c(b,D_d,\sc)\beta
 \left|\frac{\theta_{c,1}(\beta)}{2}-\frac{\pi}{\beta}\right|\right)\\
&\ge  1_{|\theta_{c,1}(\beta)/2-\pi/\beta|> 1}
\left(1-
 c(b,D_d,\sc,\sr,\sfs)|U|^{\frac{1}{2}}\right).
\end{align*}
Therefore if $|U|<c(b,D_d,\sc,\sr,\sfs)$, 
$$
\frac{\partial g}{\partial x}(\beta,\theta_{c,1}(\beta),0)>0.
$$
Similarly we can derive from \eqref{eq_additional_upper_order},
 \eqref{eq_additional_lower_order},
 \eqref{eq_gap_equation_1st_derivative},
 \eqref{eq_critical_theta_2_upper}, the claims
 \eqref{item_global_temperature_bound}, \eqref{item_high_temperature_bound} and the conditions $1+2\sr\le
 \sfs$, $\beta<1$ that
\begin{align*}
&\frac{\partial g}{\partial x}(\beta,\theta_{c,2}(\beta),0)\\
&\le -c D_d
\left|\frac{\theta_{c,2}(\beta)}{2}-\frac{\pi}{\beta}\right|\beta 
\int_{\G_{\infty}^*}d\bk\Tr
 \left(\frac{1}{(\cos(\beta\theta_{c,2}(\beta)/2)+\cosh(\beta
 E(\bk)))^2}\right)\\
&\quad+D_d \int_{\G_{\infty}^*}d\bk\Tr
 \left(\frac{\cosh(\beta E(\bk))}{\cos(\beta\theta_{c,2}(\beta)/2)+\cosh(\beta
 E(\bk))}\right)\\
&\le -c(D_d,\sc)
 \left|\frac{\theta_{c,2}(\beta)}{2}-\frac{\pi}{\beta}\right|\beta^{-3}
\int_{\G_{\infty}^*}d\bk
 \left(e(\bk)^2+\left|\frac{\theta_{c,2}(\beta)}{2}-\frac{\pi}{\beta}\right|^2\right)^{-2}\\
&\quad + c(b,D_d,\sc) \beta^{-2} \int_{\G_{\infty}^*}d\bk
 \left(e(\bk)^2+\left|\frac{\theta_{c,2}(\beta)}{2}-\frac{\pi}{\beta}\right|^2\right)^{-1}\\
&\le -1_{|\theta_{c,2}(\beta)/2-\pi/\beta|\le 1}
c(D_d,\sc) \beta^{-3}
 \left|\frac{\theta_{c,2}(\beta)}{2}-\frac{\pi}{\beta}\right|^{1-\sfs}
\left(1-
 c(b,D_d,\sc)\beta
 \left|\frac{\theta_{c,2}(\beta)}{2}-\frac{\pi}{\beta}\right|^{\sfs-1-\sr}\right)\\
&\quad -1_{|\theta_{c,2}(\beta)/2-\pi/\beta|> 1}
c(D_d,\sc) \beta^{-3}
 \left|\frac{\theta_{c,2}(\beta)}{2}-\frac{\pi}{\beta}\right|^{-3}
\left(1-
 c(b,D_d,\sc)\beta
 \left|\frac{\theta_{c,2}(\beta)}{2}-\frac{\pi}{\beta}\right|\right)\\
&\le -1_{|\theta_{c,2}(\beta)/2-\pi/\beta|\le 1}
c(D_d,\sc) \beta^{-3}
 \left|\frac{\theta_{c,2}(\beta)}{2}-\frac{\pi}{\beta}\right|^{1-\sfs}
(1-c(b,D_d,\sc,\sr,\sfs)|U|)\\
&\quad -1_{|\theta_{c,2}(\beta)/2-\pi/\beta|> 1}c(D_d,\sc)\beta^{-3}
 \left|\frac{\theta_{c,2}(\beta)}{2}-\frac{\pi}{\beta}\right|^{-3}
(1-c(b,D_d,\sc,\sr,\sfs)|U|^{\frac{1}{2}})\\
&<0.
\end{align*}
In the last inequality we assumed that $|U|<c(b,D_d,\sc,\sr,\sfs)$.

Thus we have proved that
\begin{align}
\frac{\partial g}{\partial x}(\beta,\theta_{c,1}(\beta),0)>0,\quad 
\frac{\partial g}{\partial x}(\beta,\theta_{c,2}(\beta),0)<0,\quad
 (\forall \beta\in \R_{>0}).\label{eq_gap_equation_1st_derivative_sign}
\end{align}
 
Now by combining \eqref{eq_gap_equation_2nd_derivative_sign} with
 \eqref{eq_gap_equation_1st_derivative_sign} we conclude that there
 exists a positive constant $c(b,D_d,\sc, \sr,\sfs)$ such that if
 $|U|<c(b,D_d,\sc,\sr,\sfs)$,
\begin{align*}
\frac{d\theta_{c,j}}{d\beta}(\beta)=-\frac{\frac{\partial g}{\partial
 x}(\beta,\theta_{c,j}(\beta),0)}{\frac{\partial g}{\partial
 y}(\beta,\theta_{c,j}(\beta),0)}<0,\quad (\forall j\in\{1,2\},\
 \beta\in \R_{>0}).
\end{align*}
\end{proof}

Let us assume that $U\in (-c_3,0)$ with the constant $c_3$ appearing in
Lemma \ref{lem_critical_theta_property}. For $m\in \N\cup\{0\}$, $j\in
\{1,2\}$ we define the function $\theta_{c,j,m}:\R_{>0}\to \R_{>0}$
by 
$$
\theta_{c,j,m}(x):=\theta_{c,j}(x)+\frac{4\pi}{x}m.
$$
By Lemma
\ref{lem_critical_theta_property} \eqref{item_critical_theta_regularity}
the continuous function $\theta_{c,j,m}:\R_{>0}\to \R_{>0}$ is monotone decreasing
and thus injective. By the fact that
$0<\theta_{c,1}(\beta)<\frac{2\pi}{\beta}<\theta_{c,2}(\beta)<\frac{4\pi}{\beta}$
and Lemma \ref{lem_critical_theta_property}
\eqref{item_simple_temperature_bound},
\begin{align}
&\frac{\pi}{\beta} \le \theta_{c,1,0}(\beta),\label{eq_critical_theta_interval}
\\
&\frac{4\pi}{\beta}m<\theta_{c,1,m}(\beta)<\frac{2\pi}{\beta}+\frac{4\pi}{\beta}m<\theta_{c,2,m}(\beta)<
 \frac{4\pi}{\beta}+\frac{4\pi}{\beta}m,\quad (\forall \beta \in \R_{>0}).\notag
\end{align}
This implies that the function
$\theta_{c,j,m}$ is surjective and thus bijective. Let $\beta_{c,j,m}$ denote the inverse
function of $\theta_{c,j,m}$.

The phase boundaries are characterized in the next proposition.

\begin{proposition}\label{prop_phase_boundaries}
Let $c_3$ be the positive constant appearing in Lemma
 \ref{lem_critical_theta_property}. Assume that $U\in (-c_3,0)$.
Then the following statements hold.
\begin{enumerate}[(i)]
\item\label{item_critical_theta_interval}
For any $m\in \N\cup \{0\}$, $x\in \R_{>0}$,
\begin{align*}
&\theta_{c,1,m}(x)\in \left(\frac{4\pi m}{x},\frac{2\pi + 4\pi
 m}{x}\right),\quad \theta_{c,2,m}(x)\in \left(\frac{2\pi+4\pi
 m}{x},\frac{4\pi(m+1)}{x}\right).
\end{align*}
\item\label{item_critical_temperature_interval}
For any $m\in \N\cup\{0\}$, $y\in \R_{>0}$, 
\begin{align*}
&\beta_{c,1,m}(y)\in \left(\frac{4\pi m}{y},\frac{2\pi + 4\pi
 m}{y}\right),\quad \beta_{c,2,m}(y)\in \left(\frac{2\pi+4\pi
 m}{y},\frac{4\pi(m+1)}{y}\right).
\end{align*}
\item\label{item_phase_boundaries}
Let $(\beta,\theta)\in \R_{>0}\times \R$. The following statements are
     equivalent to each other.
\begin{enumerate}[(a)]
\item\label{item_item_positivity}
$$
\D(\beta,\theta)>0.
$$
\item\label{item_item_theta}
$$
|\theta|\in \bigcup_{m\in\N\cup\{0\}}\left(\theta_{c,1,m}(\beta),\theta_{c,2,m}(\beta)\right).
$$
\item\label{item_item_beta}
$\theta\neq 0$ and 
$$
\beta\in \bigcup_{m\in \N\cup\{0\}}\left(\beta_{c,1,m}(|\theta|), \beta_{c,2,m}(|\theta|)\right).
$$
\end{enumerate}
\end{enumerate}
\end{proposition}

\begin{proof}
We have already seen the claim \eqref{item_critical_theta_interval}
 in \eqref{eq_critical_theta_interval}. The claim
 \eqref{item_critical_temperature_interval} follows from
the claim \eqref{item_critical_theta_interval} and 
the definition of
 $\beta_{c,j,m}(\cdot)$ $(j=1,2)$. Take any
 $(\beta,\theta)\in\R_{>0}\times \R$. There uniquely exist $m'\in
 \N\cup\{0\}$, $\theta'\in [0, 4\pi/\beta)$ such that
 $|\theta|=\theta'+\frac{4\pi}{\beta}m'$. Let us confirm the claim
 \eqref{item_phase_boundaries}. The statement (a) is equivalent to $g(\beta,\theta,0)>0$,
 which is equivalent to $\theta'\in (\theta_{c,1}(\beta),
 \theta_{c,2}(\beta))$ since $g(\beta,\theta,0)=g(\beta,\theta',0)$. The
 inclusion $\theta'\in (\theta_{c,1}(\beta),
 \theta_{c,2}(\beta))$ is equivalent to the statement
 (b). Thus the equivalence between (a) and (b) is proved.
The equivalence between (b) and (c)
 can be deduced from the definition of $\beta_{c,j,m}(\cdot)$ $(j=1,2)$.
\end{proof}

Based on Lemma \ref{lem_critical_theta_property}
\eqref{item_critical_theta_regularity} and Proposition
\ref{prop_phase_boundaries}, we can sketch the $\beta-|\theta|$ phase
diagram as in Figure \ref{fig_theta_beta}.

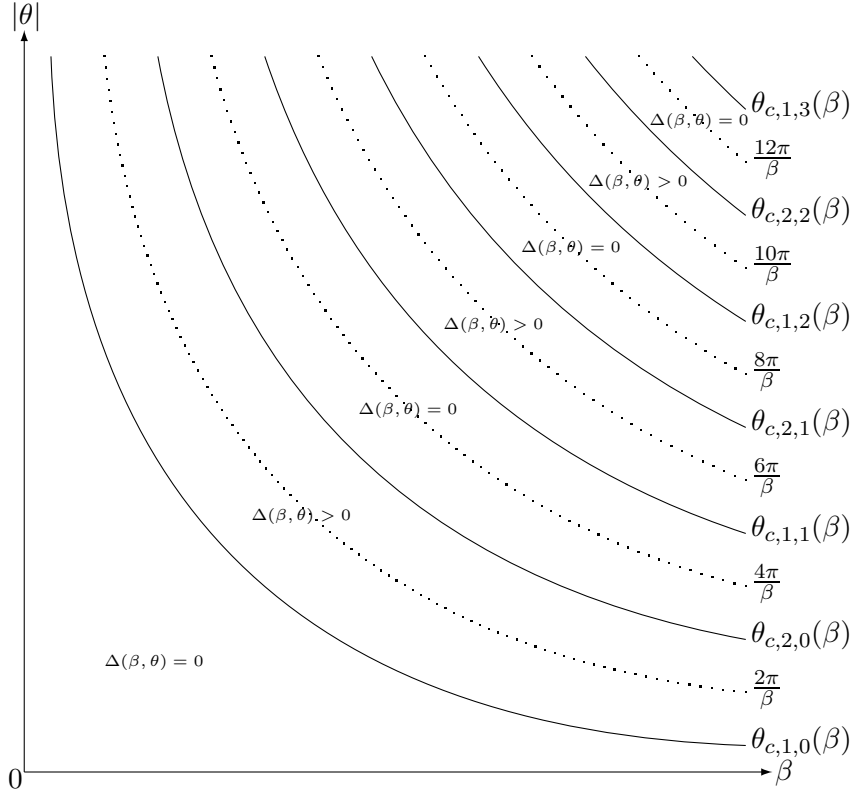
\begin{figure}
\begin{center}
\begin{picture}(300,300)(0,0)

\put(10,10){\vector(1,0){280}}
\put(10,10){\vector(0,1){280}}
\put(4,4){\small 0}
\put(291,7){\small$\beta$}
\put(5,292){\small$|\theta|$}

\put(282,19){\small$\theta_{c,1,0}(\beta)$}
\put(282,39){\small$\frac{2\pi}{\beta}$}
\put(282,59){\small$\theta_{c,2,0}(\beta)$}
\put(282,79){\small$\frac{4\pi}{\beta}$}
\put(282,99){\small$\theta_{c,1,1}(\beta)$}
\put(282,119){\small$\frac{6\pi}{\beta}$}
\put(282,139){\small$\theta_{c,2,1}(\beta)$}
\put(282,159){\small$\frac{8\pi}{\beta}$}
\put(282,179){\small$\theta_{c,1,2}(\beta)$}
\put(282,199){\small$\frac{10\pi}{\beta}$}
\put(282,219){\small$\theta_{c,2,2}(\beta)$}
\put(282,239){\small$\frac{12\pi}{\beta}$}
\put(282,259){\small$\theta_{c,1,3}(\beta)$}

\qbezier(20,280)(29.286,29.286)(280,20)
\qbezier[100](40,280)(64,64)(280,40)
\qbezier(60,280)(94.375,94.375)(280,60)
\qbezier[75](80,280)(121.176,121.176)(280,80)
\qbezier(100,280)(145,145)(280,100)
\qbezier[55](120,280)(166.316,166.316)(280,120)
\qbezier(140,280)(185.5,185.5)(280,140)
\qbezier[40](160,280)(202.857,202.857)(280,160)
\qbezier(180,280)(218.636,218.636)(280,180)
\qbezier[24](200,280)(233.043,233.043)(280,200)
\qbezier(220,280)(246.25,246.25)(280,220)
\qbezier[13](240,280)(258.4,258.4)(280,240)
\qbezier(260,280)(269.615,269.615)(280,260)

\put(40,50){\tiny$\D(\beta,\theta)=0$}

\put(95,105){\tiny$\D(\beta,\theta)>0$}

\put(135,145){\tiny$\D(\beta,\theta)=0$}

\put(167,177){\tiny$\D(\beta,\theta)>0$}

\put(196,206){\tiny$\D(\beta,\theta)=0$}

\put(221,231){\tiny$\D(\beta,\theta)>0$}

\put(244,254){\tiny$\D(\beta,\theta)=0$}

\end{picture}
 \caption{The schematic $\beta-|\theta|$ phase diagram.}\label{fig_theta_beta}
\end{center}
\end{figure}

We can understand from Proposition \ref{prop_phase_boundaries} that for
any fixed $\theta\in \R\backslash\{0\}$ the system repeatedly enters and
exits a superconducting phase where $\D(\beta,\theta)>0$ as $\beta$ varies from 0 to
$\infty$. It is notable that there are infinitely many critical temperatures.

\subsection{The second order phase
  transitions}\label{subsec_2nd_order_transition}

Using the function $\D:\R_{>0}\times\R\to\R_{\ge 0}$, we define the
function $F:\R_{>0}\times \R\to\R$ by
\begin{align*}
&F(x,y)\\
&:=\frac{\D(x,y)^2}{|U|}-\frac{D_d}{x}\int_{\G_{\infty}^*}d\bk \Tr
 \log\Bigg(2\cos\left(\frac{xy}{2}\right)e^{-xE(\bk)}+e^{x(\sqrt{E(\bk)^2+\D(x,y)^2}-E(\bk))}\\
&\qquad\qquad\qquad\qquad\qquad\qquad\qquad\quad+e^{-x(\sqrt{E(\bk)^2+\D(x,y)^2}+E(\bk))}\Bigg).\end{align*}
Equally, we can write as follows.
\begin{align*}
F(x,y)=&\frac{\D(x,y)^2}{|U|}-\frac{D_d}{x}\int_{\G_{\infty}^*}d\bk \Tr
 \log\left(\cos\left(\frac{xy}{2}\right)+\cosh(x\sqrt{E(\bk)^2+\D(x,y)^2})\right)\\
&-\frac{b\log 2}{x}+D_d\int_{\G_{\infty}^*}d\bk\Tr E(\bk).
\end{align*}
Since $\D(x,y)>0$ if $xy/2\in \pi(2\Z+1)$, $F$ is
well-defined. Recalling Theorem \ref{thm_main_theorem} \eqref{item_free_energy_density}, we see that
$F(\beta,\theta)$ is equal to the free energy density for
$(\beta,\theta,U)\in\R_{>0}\times\R\times\R_{<0}$ satisfying
$\beta\theta/2\notin \pi(2\Z+1)$ and \eqref{eq_coupling_constant_interval}.

We end this section by proving that the first order derivatives of $F$
are globally continuous and the second order derivatives of $F$ have jump
discontinuities across the phase boundaries. Since these properties hold
in the parameter region where $F$ is proved to be equal to the
free energy density by Theorem \ref{thm_main_theorem}, we can consider
that our many-electron system shows the second order phase transitions
driven by the temperature and the imaginary magnetic field. 

\begin{proposition}\label{prop_2nd_order_transition} 
Let $c_3$ be the positive constant appearing in Lemma
 \ref{lem_critical_theta_property} and $U\in (-c_3,0)$. Then the following
 statements hold true. 
\begin{enumerate}[(i)]
\item\label{item_phase_boundaries_confirmation}
\begin{align*}
&\R_{>0}\times \R\backslash O_+\cup O_-\\
&=\left\{(\beta,\delta \theta_{c,j,m}(\beta))\ |\ \beta \in\R_{>0},\
 j\in \{1,2\},\ m\in \N\cup\{0\},\ \delta\in \{1,-1\}\right\}\\
&=\left\{(\beta_{c,j,m}(\theta),\delta \theta)\ |\ \theta \in\R_{>0},\
 j\in \{1,2\},\ m\in \N\cup\{0\},\ \delta\in \{1,-1\}\right\}.
\end{align*}
\item\label{item_energy_regularity}
$$
F|_{O_+\cup O_-}\in C^{\infty}(O_+\cup O_-),\quad F\in C^1(\R_{>0}\times\R).
$$
\item\label{item_energy_beta_jump}
For any $\theta\in \R_{>0}$, $j\in\{1,2\}$, $m\in \N\cup \{0\}$,
     $\delta\in \{1,-1\}$, $\lim_{\beta\nearrow
     \beta_{c,j,m}(\theta)}\frac{\partial^2F}{\partial
     \beta^2}(\beta,\delta \theta)$, 
$\lim_{\beta\searrow
     \beta_{c,j,m}(\theta)}\frac{\partial^2F}{\partial
     \beta^2}(\beta,\delta \theta)$ converge and 
\begin{align*}
&\lim_{\beta\nearrow
     \beta_{c,1,m}(\theta)}\frac{\partial^2F}{\partial
     \beta^2}(\beta,\delta \theta)>\lim_{\beta\searrow
     \beta_{c,1,m}(\theta)}\frac{\partial^2F}{\partial
     \beta^2}(\beta,\delta \theta),\\
&\lim_{\beta\nearrow
     \beta_{c,2,m}(\theta)}\frac{\partial^2F}{\partial
     \beta^2}(\beta,\delta \theta)<\lim_{\beta\searrow
     \beta_{c,2,m}(\theta)}\frac{\partial^2F}{\partial
     \beta^2}(\beta,\delta \theta).
\end{align*}
\item\label{item_energy_theta_jump}
For any $\beta\in \R_{>0}$, $j\in\{1,2\}$, $m\in \N\cup \{0\}$,
     $\delta\in \{1,-1\}$, $\lim_{\theta\nearrow \delta
     \theta_{c,j,m}(\beta)}\frac{\partial^2 F}{\partial
     \theta^2}(\beta,\theta)$, $\lim_{\theta\searrow \delta
     \theta_{c,j,m}(\beta)}\frac{\partial^2 F}{\partial
     \theta^2}(\beta,\theta)$ converge and 
\begin{align*}
&\lim_{\theta\nearrow 
     \theta_{c,1,m}(\beta)}\frac{\partial^2 F}{\partial
     \theta^2}(\beta,\theta)> \lim_{\theta\searrow 
     \theta_{c,1,m}(\beta)}\frac{\partial^2 F}{\partial
     \theta^2}(\beta,\theta),\\
&\lim_{\theta\nearrow 
     \theta_{c,2,m}(\beta)}\frac{\partial^2 F}{\partial
     \theta^2}(\beta,\theta)< \lim_{\theta\searrow 
     \theta_{c,2,m}(\beta)}\frac{\partial^2 F}{\partial
     \theta^2}(\beta,\theta),\\
&\lim_{\theta\nearrow -\theta_{c,1,m}(\beta)}\frac{\partial^2 F}{\partial
     \theta^2}(\beta,\theta)< \lim_{\theta\searrow - \theta_{c,1,m}(\beta)}\frac{\partial^2 F}{\partial
     \theta^2}(\beta,\theta),\\
&\lim_{\theta\nearrow -\theta_{c,2,m}(\beta)}\frac{\partial^2 F}{\partial
     \theta^2}(\beta,\theta)> \lim_{\theta\searrow - \theta_{c,2,m}(\beta)}\frac{\partial^2 F}{\partial
     \theta^2}(\beta,\theta).
\end{align*}
\end{enumerate}
\end{proposition}

\begin{proof}
\eqref{item_phase_boundaries_confirmation}: We can deduce the claim from
 Lemma \ref{lem_critical_theta}, the definitions of
 $\theta_{c,j,m}(\cdot)$, $\beta_{c,j,m}(\cdot)$ and the fact that 
$g(\beta,\theta,0)=g(\beta,|\theta|+\frac{4\pi}{\beta}n,0)$ $(\forall n\in\Z)$.

\eqref{item_energy_regularity}: Set 
\begin{align*}
D:=\left\{(x,y,z)\in\R_{>0}\times \R\times \R\ \big|\ \frac{xy}{2}\notin
 \pi(2\Z+1)\text{ or }z>0\right\},
\end{align*}
which is an open set of $\R^3$. We define the function
 $\hat{F}:D\to\R$ by
\begin{align*}
\hat{F}(x,y,z):=\frac{z^2}{|U|}-\frac{D_d}{x}\int_{\G_{\infty}^*}d\bk
 \Tr
 \log\left(\cos\left(\frac{xy}{2}\right)+\cosh(x\sqrt{E(\bk)^2+z^2})\right).
\end{align*}
We can see that 
\begin{align}
&\hat{F}\in C^{\infty}(D),\label{eq_quasi_energy_regularity}\\
&(\beta,\theta,\D(\beta,\theta))\in D,\label{eq_gap_function_graph}\\
&F(\beta,\theta)=\hat{F}(\beta,\theta,\D(\beta,\theta))-\frac{b\log
 2}{\beta}+D_d\int_{\G_{\infty}^*}d\bk \Tr E(\bk),\label{eq_energy_quasi_energy}\\
&(\forall
 (\beta,\theta)\in \R_{>0}\times \R).\notag
\end{align}
Combined with Lemma \ref{lem_regularity_Delta}, the functions
 $(\beta,\theta)\mapsto \hat{F}(\beta,\theta,\D(\beta,\theta))$,
 $(\beta,\theta)\mapsto F(\beta,\theta)$ are seen to be continuous in
 $\R_{>0}\times\R$ and $C^{\infty}$-class in $O_+\cup O_-$. 

Let us prove that $F\in C^1(\R_{>0}\times\R)$.
For $(\beta,\theta)\in O_+\cup O_-$
\begin{align*}
\frac{\partial \hat{F}}{\partial
 z}(\beta,\theta,\D(\beta,\theta))=-g(\beta,\theta, \D(\beta,\theta))\D(\beta,\theta)=0,
\end{align*}
and thus
\begin{align}
\frac{\partial}{\partial
 \beta}\hat{F}(\beta,\theta,\D(\beta,\theta))&=\frac{\partial
 \hat{F}}{\partial x}(\beta,\theta,\D(\beta,\theta))+
\frac{\partial
 \hat{F}}{\partial z}(\beta,\theta,\D(\beta,\theta))\frac{\partial
 \D}{\partial \beta}(\beta,\theta)\label{eq_quasi_energy_beta_1st}\\
&=\frac{\partial \hat{F}}{\partial
 x}(\beta,\theta,\D(\beta,\theta)),\notag\\
\frac{\partial}{\partial
 \theta}\hat{F}(\beta,\theta,\D(\beta,\theta))&=\frac{\partial
 \hat{F}}{\partial y}(\beta,\theta,\D(\beta,\theta))+
\frac{\partial
 \hat{F}}{\partial z}(\beta,\theta,\D(\beta,\theta))\frac{\partial
 \D}{\partial \theta}(\beta,\theta)\label{eq_quasi_energy_theta_2nd}\\
&=\frac{\partial \hat{F}}{\partial
 y}(\beta,\theta,\D(\beta,\theta)).\notag
\end{align}
Note that
\begin{align}
\D(\beta,\theta)=0,\quad (\forall (\beta,\theta)\in\R_{>0}\times
 \R\backslash O_+\cup O_-).\label{eq_gap_function_vanish_condition}
\end{align}
It follows from the global continuity of $\D(\cdot,\cdot)$,
 \eqref{eq_quasi_energy_regularity}, \eqref{eq_gap_function_graph},
 \eqref{eq_gap_function_vanish_condition} and the characterization of 
$\R_{>0}\times \R\backslash O_{+}\cup O_-$ given in
 \eqref{item_phase_boundaries_confirmation} that
\begin{align}
&\lim_{(\beta,\theta)\to (\beta',\theta')\atop (\beta,\theta)\in O_+\cup
 O_-}\frac{\partial^{m+n}\hat{F}}{\partial x^m\partial
 y^n}(\beta,\theta,\D(\beta,\theta))=\frac{\partial^{m+n}\hat{F}}{\partial
 x^m\partial
 y^n}(\beta',\theta',0),\label{eq_quasi_energy_derivative_convergence}\\
&(\forall (\beta',\theta')\in \R_{>0}\times \R\backslash O_+\cup O_-,\
 m,n\in\N\cup \{0\}).\notag
\end{align}
By \eqref{eq_quasi_energy_beta_1st}, \eqref{eq_quasi_energy_theta_2nd},
 \eqref{eq_quasi_energy_derivative_convergence} we can observe that
for any $(\beta',\theta')\in\R_{>0}\times \R\backslash  O_+\cup O_-$
\begin{align*}
&\lim_{(\beta,\theta)\to (\beta',\theta')\atop (\beta,\theta)\in O_+\cup
 O_-}\frac{\partial}{\partial\beta}\hat{F}(\beta,\theta,\D(\beta,\theta))
=\lim_{(\beta,\theta)\to (\beta',\theta')\atop (\beta,\theta)\in O_+\cup
 O_-}\frac{\partial\hat{F}}{\partial x}(\beta,\theta,\D(\beta,\theta))
=\frac{\partial\hat{F}}{\partial x}(\beta',\theta',0),\\
&\lim_{(\beta,\theta)\to (\beta',\theta')\atop (\beta,\theta)\in O_+\cup
 O_-}\frac{\partial}{\partial\theta}\hat{F}(\beta,\theta,\D(\beta,\theta))
=\lim_{(\beta,\theta)\to (\beta',\theta')\atop (\beta,\theta)\in O_+\cup
 O_-}\frac{\partial\hat{F}}{\partial y}(\beta,\theta,\D(\beta,\theta))
=\frac{\partial\hat{F}}{\partial y}(\beta',\theta',0),
\end{align*}
which together with the characterization of $\R_{>0}\times \R\backslash
 O_+\cup O_-$ given in \eqref{item_phase_boundaries_confirmation}
 implies that $(\beta,\theta)\mapsto
 \hat{F}(\beta,\theta,\D(\beta,\theta))$ is partially differentiable in
 $\R_{>0}\times\R$ and 
\begin{align*}
&\frac{\partial}{\partial\beta}\hat{F}(\beta,\theta,\D(\beta,\theta))=\frac{\partial
 \hat{F}}{\partial x}(\beta,\theta,\D(\beta,\theta)),\\
&\frac{\partial}{\partial\theta}\hat{F}(\beta,\theta,\D(\beta,\theta))=\frac{\partial
 \hat{F}}{\partial y}(\beta,\theta,\D(\beta,\theta)),\quad (\forall
 (\beta,\theta)\in \R_{>0}\times \R).
\end{align*}
Since $(\beta,\theta)\mapsto \frac{\partial
 \hat{F}}{\partial x}(\beta,\theta,\D(\beta,\theta))$, $(\beta,\theta)\mapsto \frac{\partial
 \hat{F}}{\partial y}(\beta,\theta,\D(\beta,\theta))$ are continuous in
 $\R_{>0}\times \R$, we can conclude that the function $(\beta,\theta)\mapsto
 \hat{F}(\beta,\theta,\D(\beta,\theta))$ belongs to
 $C^1(\R_{>0}\times\R)$ and so does the function $F$.

\eqref{item_energy_beta_jump}: By \eqref{eq_quasi_energy_beta_1st}, for
 $(\beta,\theta)\in O_+\cup O_-$
\begin{align}
\frac{\partial^2}{\partial
 \beta^2}\hat{F}(\beta,\theta,\D(\beta,\theta))&=\frac{\partial}{\partial
 \beta}\frac{\partial \hat{F}}{\partial x}(\beta,\theta,\D(\beta,\theta))
\label{eq_quasi_energy_beta_beta}\\
&=\frac{\partial^2 \hat{F}}{\partial x^2}(\beta,\theta,\D(\beta,\theta))
+\frac{\partial^2 \hat{F}}{\partial z\partial
 x}(\beta,\theta,\D(\beta,\theta))\frac{\partial \D}{\partial
 \beta}(\beta,\theta).\notag
\end{align}
In particular for $(\beta,\theta)\in O_-$
\begin{align}
\frac{\partial^2}{\partial
 \beta^2}\hat{F}(\beta,\theta,\D(\beta,\theta))=\frac{\partial^2
 \hat{F}}{\partial
 x^2}(\beta,\theta,\D(\beta,\theta)).\label{eq_quasi_energy_vanish_beta}
\end{align}
It follows from the claim \eqref{item_phase_boundaries_confirmation} and 
 \eqref{eq_quasi_energy_derivative_convergence} that
for any $\theta\in\R_{>0}$, $j\in \{1,2\}$, $m\in \N\cup \{0\}$, $\delta
 \in \{1,-1\}$
\begin{align}
\lim_{\beta\searrow
 \beta_{c,j,m}(\theta)}\frac{\partial^2\hat{F}}{\partial
 x^2}(\beta,\delta \theta,\D(\beta,\delta \theta))
&=\lim_{\beta\nearrow
 \beta_{c,j,m}(\theta)}\frac{\partial^2\hat{F}}{\partial
 x^2}(\beta,\delta \theta,\D(\beta,\delta \theta))\label{eq_quasi_energy_pre_convergence_beta}\\
&=\frac{\partial^2\hat{F}}{\partial
 x^2}(\beta_{c,j,m}(\theta),\delta \theta,0).\notag
\end{align}
For $(\beta,\theta)\in O_+$ one can derive that
\begin{align}
\frac{\partial^2\hat{F}}{\partial x\partial
 z}(\beta,\theta,\D(\beta,\theta))=-\D(\beta,\theta)\frac{\partial
 g}{\partial x}(\beta,\theta,\D(\beta,\theta)).
\label{eq_quasi_energy_1st_3rd}
\end{align}
Also by taking into account \eqref{eq_positivity_3rd_derivative},
\begin{align}
\frac{\partial \D}{\partial \beta}(\beta,\theta)=-\frac{\frac{\partial
 g}{\partial x}(\beta,\theta,\D(\beta,\theta))}{\frac{\partial
 g}{\partial
 z}(\beta,\theta,\D(\beta,\theta))}.\label{eq_gap_function_beta}
\end{align}
By combining \eqref{eq_quasi_energy_1st_3rd} with
 \eqref{eq_gap_function_beta} we obtain that
\begin{align}
\frac{\partial^2\hat{F}}{\partial x\partial
 z}(\beta,\theta,\D(\beta,\theta))\frac{\partial \D}{\partial
 \beta}(\beta,\theta)=\D(\beta,\theta)\frac{(\frac{\partial g}{\partial
 x}(\beta,\theta,\D(\beta,\theta)))^2}{\frac{\partial g}{\partial
 z}(\beta,\theta,\D(\beta,\theta))}.
\label{eq_quasi_energy_term_characterization_beta}
\end{align}
Observe that for any $(x,y,z)\in\R_{>0}\times \R\times
 \R_{>0}$
\begin{align}
&z^{-1}\frac{\partial g}{\partial z}(x,y,z)\label{eq_quasi_energy_denominator}\\
&=-D_dx \int_{\G_{\infty}^*}d\bk \Tr\Bigg(
\frac{1}{(\cos(xy/2)+\cosh(x\sqrt{E(\bk)^2+z^2}))^2(E(\bk)^2+z^2)}\notag\\
&\qquad\qquad\qquad\qquad\quad \cdot \Bigg(\sinh^2(x\sqrt{E(\bk)^2+z^2})\notag\\
&\qquad\qquad\qquad\qquad\qquad\quad
 -\left(\cosh(x\sqrt{E(\bk)^2+z^2})-\frac{\sinh(x\sqrt{E(\bk)^2+z^2})}{x\sqrt{E(\bk)^2+z^2}}\right)\notag\\
&\qquad\qquad\qquad\qquad\qquad\qquad\quad\cdot
\left(\cos\left(\frac{xy}{2}\right)+\cosh(x\sqrt{E(\bk)^2+z^2})\right)\Bigg)\Bigg)\notag\\
&\le -D_d x \int_{\G_{\infty}^*}d\bk \Tr\Bigg(
\frac{1}{(\cos(xy/2)+\cosh(x\sqrt{E(\bk)^2+z^2}))^2(E(\bk)^2+z^2)}\notag\\
&\qquad\qquad\qquad\qquad\quad \cdot (1+\cosh(x\sqrt{E(\bk)^2+z^2}))\left(
\frac{\sinh(x\sqrt{E(\bk)^2+z^2})}{x\sqrt{E(\bk)^2+z^2}}-1
\right)\Bigg)\notag\\
&\le -\frac{D_d}{3!}x^3 \int_{\G_{\infty}^*}d\bk\Tr 
\left(
\frac{1}{\cos(xy/2)+\cosh(x\sqrt{E(\bk)^2+z^2})}\right)
\notag\\
&<0.\notag
\end{align}
By Proposition \ref{prop_phase_boundaries}
 \eqref{item_critical_temperature_interval} for any $\theta\in \R_{>0}$,
 $m\in \N\cup\{0\}$, $\delta \in\{1,-1\}$, $j\in \{1,2\}$, $\delta
 \theta\beta_{c,j,m}(\theta)/2\notin \pi (2\Z+1)$. Thus we can see from
 Proposition \ref{prop_phase_boundaries} \eqref{item_phase_boundaries},
 the global continuity of $\D(\cdot,\cdot)$ and
 \eqref{eq_quasi_energy_denominator} that 
\begin{align*}
\lim_{\beta\searrow
 \beta_{c,1,m}(\theta)}\D(\beta,\delta\theta)^{-1}\frac{\partial
 g}{\partial z}(\beta,\delta \theta,\D(\beta,\delta \theta)),\quad 
\lim_{\beta\nearrow
 \beta_{c,2,m}(\theta)}\D(\beta,\delta\theta)^{-1}\frac{\partial
 g}{\partial z}(\beta,\delta \theta,\D(\beta,\delta \theta))
\end{align*}
converge to negative values. On the other hand, we can see from
 \eqref{eq_gap_equation_1st_derivative},
 \eqref{eq_gap_equation_1st_derivative_sign} and Proposition
 \ref{prop_phase_boundaries} \eqref{item_critical_theta_interval} 
that for any $\theta\in\R_{>0}$,
 $m\in \N\cup \{0\}$, $\delta \in \{1,-1\}$,
 \begin{align*}
&\lim_{\beta\searrow
 \beta_{c,1,m}(\theta)} \frac{\partial
 g}{\partial x}(\beta,\delta \theta,\D(\beta,\delta \theta))=\frac{\partial g}{\partial x}(\beta_{c,1,m}(\theta),\theta,0)\\
&=\frac{\partial g}{\partial
 x}(\beta_{c,1,m}(\theta),\theta_{c,1,m}(\beta_{c,1,m}(\theta)),0)
\ge \frac{\partial g}{\partial
 x}(\beta_{c,1,m}(\theta),\theta_{c,1}(\beta_{c,1,m}(\theta)),0)>0,\\
&\lim_{\beta\nearrow
 \beta_{c,2,m}(\theta)} \frac{\partial
 g}{\partial x}(\beta,\delta \theta,\D(\beta,\delta \theta))
=\frac{\partial g}{\partial
 x}(\beta_{c,2,m}(\theta),\theta,0)\\
&=\frac{\partial g}{\partial
 x}(\beta_{c,2,m}(\theta),\theta_{c,2,m}(\beta_{c,2,m}(\theta)),0)
\le \frac{\partial g}{\partial
 x}(\beta_{c,2,m}(\theta),\theta_{c,2}(\beta_{c,2,m}(\theta)),0)<0.
\end{align*}

It follows from Proposition \ref{prop_phase_boundaries}
 \eqref{item_phase_boundaries}, \eqref{eq_quasi_energy_beta_beta},
 \eqref{eq_quasi_energy_vanish_beta},
 \eqref{eq_quasi_energy_pre_convergence_beta}, 
\eqref{eq_quasi_energy_term_characterization_beta} and the above
 convergence results that for any $\theta\in \R_{>0}$, $m\in \N\cup
 \{0\}$, $\delta \in \{1,-1\}$
\begin{align*}
&\lim_{\beta\searrow
 \beta_{c,1,m}(\theta)}\frac{\partial^2}{\partial
 \beta^2}\hat{F}(\beta,\delta\theta,\D(\beta,\delta\theta))\\
&=\frac{\partial^2\hat{F}}{\partial
 x^2}(\beta_{c,1,m}(\theta),\delta\theta,0)+
\frac{(\frac{\partial g}{\partial x}(\beta_{c,1,m}(\theta),\theta,0))^2}{\lim_{\beta\searrow
 \beta_{c,1,m}(\theta)}\D(\beta,\delta\theta)^{-1}\frac{\partial
 g}{\partial z}(\beta,\delta\theta,\D(\beta,\delta\theta))}\\
&<\frac{\partial^2 \hat{F}}{\partial x^2}(\beta_{c,1,m}(\theta),\delta
 \theta,0)\\
&=\lim_{\beta\nearrow
 \beta_{c,1,m}(\theta)}\frac{\partial^2}{\partial
 \beta^2}\hat{F}(\beta,\delta\theta,\D(\beta,\delta\theta)),\\
&\lim_{\beta\nearrow
 \beta_{c,2,m}(\theta)}\frac{\partial^2}{\partial
 \beta^2}\hat{F}(\beta,\delta\theta,\D(\beta,\delta\theta))\\
&=\frac{\partial^2\hat{F}}{\partial
 x^2}(\beta_{c,2,m}(\theta),\delta\theta,0)+
\frac{(\frac{\partial g}{\partial x}(\beta_{c,2,m}(\theta),\theta,0))^2}{\lim_{\beta\nearrow
 \beta_{c,2,m}(\theta)}\D(\beta,\delta\theta)^{-1}\frac{\partial
 g}{\partial z}(\beta,\delta\theta,\D(\beta,\delta\theta))}\\
&<\frac{\partial^2 \hat{F}}{\partial x^2}(\beta_{c,2,m}(\theta),\delta
 \theta,0)\\
&=\lim_{\beta\searrow
 \beta_{c,2,m}(\theta)}\frac{\partial^2}{\partial
 \beta^2}\hat{F}(\beta,\delta\theta,\D(\beta,\delta\theta)),
\end{align*}
which together with the equality \eqref{eq_energy_quasi_energy} implies
 the claim. 

\eqref{item_energy_theta_jump}: By \eqref{eq_quasi_energy_theta_2nd}, for
 $(\beta,\theta)\in O_+\cup O_-$
\begin{align}
\frac{\partial^2}{\partial
 \theta^2}\hat{F}(\beta,\theta,\D(\beta,\theta))&=\frac{\partial}{\partial
 \theta}\frac{\partial \hat{F}}{\partial
 y}(\beta,\theta,\D(\beta,\theta))\label{eq_quasi_energy_theta_theta}\\
&= \frac{\partial^2\hat{F}}{\partial
 y^2}(\beta,\theta,\D(\beta,\theta))+\frac{\partial^2\hat{F}}{\partial
 z\partial y}(\beta,\theta,\D(\beta,\theta))\frac{\partial
 \D}{\partial\theta}(\beta,\theta).\notag
\end{align}
For $(\beta,\theta)\in O_-$
\begin{align}
\frac{\partial^2}{\partial\theta^2}\hat{F}(\beta,\theta,\D(\beta,\theta))=\frac{\partial^2\hat{F}}{\partial
 y^2}(\beta,\theta,\D(\beta,\theta)).\label{eq_quasi_energy_vanish_theta}
\end{align}
By the claim \eqref{item_phase_boundaries_confirmation} and 
 \eqref{eq_quasi_energy_derivative_convergence}, for any $\beta\in
 \R_{>0}$, $j\in \{1,2\}$, $m\in \N\cup\{0\}$
\begin{align}
\lim_{\theta\searrow \theta_{c,j,m}(\beta)}\frac{\partial^2
 \hat{F}}{\partial y^2}(\beta,\theta,\D(\beta,\theta))=
\lim_{\theta\nearrow \theta_{c,j,m}(\beta)}\frac{\partial^2
 \hat{F}}{\partial y^2}(\beta,\theta,\D(\beta,\theta))
=\frac{\partial^2
 \hat{F}}{\partial
 y^2}(\beta,\theta_{c,j,m}(\beta),0).\label{eq_quasi_energy_pre_convergence_theta}
\end{align}
For $(\beta,\theta)\in O_+$ we can derive in the same way as the
 derivation of \eqref{eq_quasi_energy_term_characterization_beta} that
\begin{align}
&\frac{\partial^2\hat{F}}{\partial z\partial
 y}(\beta,\theta,\D(\beta,\theta))\frac{\partial \D}{\partial
 \theta}(\beta,\theta)
=\D(\beta,\theta)\frac{(\frac{\partial g}{\partial
 y}(\beta,\theta,\D(\beta,\theta)))^2}{\frac{\partial g}{\partial
 z}(\beta,\theta,\D(\beta,\theta))}.\label{eq_quasi_energy_term_characterization_theta}
\end{align}
By Proposition \ref{prop_phase_boundaries}
 \eqref{item_critical_theta_interval} for any $\beta\in \R_{>0}$, $m\in
 \N\cup \{0\}$, $j\in \{1,2\}$, $\beta\theta_{c,j,m}(\beta)/2\notin \pi
 (2\Z+1)$. Thus we can deduce from Proposition
 \ref{prop_phase_boundaries} \eqref{item_phase_boundaries}, the global
 continuity of $\D(\cdot,\cdot)$ and \eqref{eq_quasi_energy_denominator}
 that
\begin{align*}
\lim_{\theta\searrow
 \theta_{c,1,m}(\beta)}\D(\beta,\theta)^{-1}\frac{\partial
 g}{\partial z}(\beta,\theta,\D(\beta,\theta)),\quad 
\lim_{\theta\nearrow
 \theta_{c,2,m}(\theta)}\D(\beta,\theta)^{-1}\frac{\partial
 g}{\partial z}(\beta,\theta,\D(\beta,\theta))
\end{align*}
converge to negative values. On the other hand, it follows from
 \eqref{eq_gap_equation_2nd_derivative},
 \eqref{eq_gap_equation_2nd_derivative_sign} that
\begin{align*}
&\lim_{\theta\searrow
 \theta_{c,1,m}(\beta)}\frac{\partial
 g}{\partial y}(\beta,\theta,\D(\beta,\theta))
=\frac{\partial g}{\partial y}(\beta,\theta_{c,1,m}(\beta),0)
=\frac{\partial g}{\partial y}(\beta,\theta_{c,1}(\beta),0)>0,\\
&\lim_{\theta\nearrow
 \theta_{c,2,m}(\beta)}\frac{\partial
 g}{\partial y}(\beta,\theta,\D(\beta,\theta))
=\frac{\partial g}{\partial y}(\beta,\theta_{c,2}(\beta),0)<0.
\end{align*}

By Proposition \ref{prop_phase_boundaries}
 \eqref{item_phase_boundaries}, \eqref{eq_quasi_energy_theta_theta},
 \eqref{eq_quasi_energy_vanish_theta},
 \eqref{eq_quasi_energy_pre_convergence_theta}, 
\eqref{eq_quasi_energy_term_characterization_theta} and the above
 convergent properties, for any $\beta\in \R_{>0}$, $m\in \N\cup \{0\}$
\begin{align*}
&\lim_{\theta\searrow
 \theta_{c,1,m}(\beta)}\frac{\partial^2}{\partial
 \theta^2}\hat{F}(\beta,\theta,\D(\beta,\theta))\\
&=\frac{\partial^2\hat{F}}{\partial
 y^2}(\beta, \theta_{c,1,m}(\beta),0)+
\frac{(\frac{\partial g}{\partial y}(\beta,\theta_{c,1}(\beta),0))^2}{\lim_{\theta\searrow
 \theta_{c,1,m}(\beta)}\D(\beta,\theta)^{-1}\frac{\partial
 g}{\partial z}(\beta,\theta,\D(\beta,\theta))}\\
&<\frac{\partial^2 \hat{F}}{\partial y^2}(\beta,\theta_{c,1,m}(\beta),0)\\
&=\lim_{\theta\nearrow
 \theta_{c,1,m}(\beta)}\frac{\partial^2}{\partial
 \theta^2}\hat{F}(\beta,\theta,\D(\beta,\theta)),\\
&\lim_{\theta\nearrow
 \theta_{c,2,m}(\beta)}\frac{\partial^2}{\partial
 \theta^2}\hat{F}(\beta,\theta,\D(\beta,\theta))\\
&=\frac{\partial^2\hat{F}}{\partial
 y^2}(\beta,\theta_{c,2,m}(\beta),0)+
\frac{(\frac{\partial g}{\partial y}(\beta,\theta_{c,2}(\beta),0))^2}{\lim_{\theta\nearrow
 \theta_{c,2,m}(\beta)}\D(\beta,\theta)^{-1}\frac{\partial
 g}{\partial z}(\beta,\theta,\D(\beta,\theta))}\\
&<\frac{\partial^2 \hat{F}}{\partial y^2}(\beta,\theta_{c,2,m}(\beta),0)\\
&=\lim_{\theta\searrow
 \theta_{c,2,m}(\beta)}\frac{\partial^2}{\partial
 \theta^2}\hat{F}(\beta,\theta,\D(\beta,\theta)).
\end{align*}
Now recalling \eqref{eq_energy_quasi_energy}, we reach the conclusion
 that 
\begin{align*}
&\lim_{\theta\searrow \theta_{c,1,m}(\beta)}\frac{\partial^2 F}{\partial
 \theta^2}(\beta,\theta) < \lim_{\theta\nearrow \theta_{c,1,m}(\beta)}\frac{\partial^2 F}{\partial
 \theta^2}(\beta,\theta),\\
&\lim_{\theta\nearrow \theta_{c,2,m}(\beta)}\frac{\partial^2 F}{\partial
 \theta^2}(\beta,\theta) < \lim_{\theta\searrow \theta_{c,2,m}(\beta)}\frac{\partial^2 F}{\partial
 \theta^2}(\beta,\theta),\quad (\forall \beta\in\R_{>0}).
\end{align*}
Note that for any $(\beta,\theta)\in\R_{>0}\times \R$,
 $\D(\beta,\theta)=\D(\beta,-\theta)$ and thus
\begin{align*}
&F(\beta,\theta)=F(\beta,-\theta),\\
&(\beta,\theta)\in O_+\cup O_-\text{ if and only if }(\beta,-\theta)\in
 O_+\cup O_-.
\end{align*}
Thus for $(\beta,\theta)\in O_+\cup O_-$, $\frac{\partial^2 F}{\partial
 \theta^2}(\beta,-\theta)=\frac{\partial^2 F}{\partial
 \theta^2}(\beta,\theta)$. Therefore
\begin{align*}
\lim_{\theta\nearrow -\theta_{c,1,m}(\beta)}\frac{\partial^2
 F}{\partial \theta^2}(\beta,\theta)&=  \lim_{\theta\searrow
 \theta_{c,1,m}(\beta)}\frac{\partial^2
 F}{\partial \theta^2}(\beta,\theta)<
\lim_{\theta\nearrow
 \theta_{c,1,m}(\beta)}\frac{\partial^2
 F}{\partial \theta^2}(\beta,\theta)\\
&=\lim_{\theta\searrow
 -\theta_{c,1,m}(\beta)}\frac{\partial^2
 F}{\partial \theta^2}(\beta,\theta),\\
\lim_{\theta\searrow -\theta_{c,2,m}(\beta)}\frac{\partial^2
 F}{\partial \theta^2}(\beta,\theta)&=  \lim_{\theta\nearrow
 \theta_{c,2,m}(\beta)}\frac{\partial^2
 F}{\partial \theta^2}(\beta,\theta)<
\lim_{\theta\searrow
 \theta_{c,2,m}(\beta)}\frac{\partial^2
 F}{\partial \theta^2}(\beta,\theta)\\
&=\lim_{\theta\nearrow
 -\theta_{c,2,m}(\beta)}\frac{\partial^2
 F}{\partial \theta^2}(\beta,\theta).
\end{align*}
The claims have been proved.
\end{proof}

\section{Formulation}\label{sec_formulation}

In this section we derive Grassmann integral formulations of the grand
canonical partition function of the model Hamiltonian. In essence the derivation
can be completed by following \cite[\mbox{Section 2}]{K_BCS}. In order to support the readers, we state several lemmas
leading to Lemma \ref{lem_final_Grassmann_formulation} step by step 
along the same lines as \cite[\mbox{Section 2}]{K_BCS}. One should be able
to prove Lemma \ref{lem_final_Grassmann_formulation} by following the
outline given in this section and the proofs presented in
\cite[\mbox{Section 2}]{K_BCS}. We intend to adopt the notations used to formulate the 1-band
problem in \cite[\mbox{Section 2}]{K_BCS} as much as possible so that
the formulation procedure can be seen parallel.

Thanks to the next lemma, we can restrict the value of $\theta$ to prove
the main results of this paper.

\begin{lemma}\label{lem_restriction_theta}
Assume that $\theta'\in (-2\pi/\beta,2\pi/\beta]$ and $\theta=\theta'$
 $(\text{mod }4\pi/\beta)$. Then
\begin{align*}
&\Tr e^{-\beta (\sH+i\theta \sS_z+\sF)}=\Tr e^{-\beta(\sH+i|\theta'|\sS_z+\sF)},\\
&\Tr (e^{-\beta (\sH+i\theta \sS_z+\sF)}\cO)=\Tr (e^{-\beta
 (\sH+i|\theta'| \sS_z+\sF)}\cO),\\
&(\forall\cO\in \{\psi_{\hrho\hbx\ua}^*\psi_{\hrho\hbx\da}^*,\ 
\psi_{\heta\hby\da}\psi_{\heta\hby\ua},\ 
\psi_{\hrho\hbx\ua}^*\psi_{\hrho\hbx\da}^*\psi_{\heta\hby\da}\psi_{\heta\hby\ua}\}).
\end{align*}
\end{lemma}

\begin{proof} This is essentially same as \cite[\mbox{Lemma
 1.2}]{K_BCS}. By using the fact that $\sS_z$ commutes with $\sH$,
 $\sF$, $\cO$ and identifying the Fock space with the direct sum of
 the eigenspaces of $\sS_z$ we can replace $\theta$ by $\theta'$ inside the
 trace operations. Then by \eqref{eq_trace_theta_reflection} we can
 replace $\theta'$ by $|\theta'|$. 
\end{proof}

In the rest of the paper for $\beta\in\R_{>0}$, $\theta\in \R$ we let
$\theta(\beta)$ denote $|\theta'|$, where 
$\theta'\in (-2\pi/\beta,2\pi/\beta]$ and $\theta=\theta'$
$(\text{mod }4\pi/\beta)$. By the assumption $\beta\theta/2\notin
\pi(2\Z+1)$ we have that $\theta(\beta)\in [0,2\pi/\beta)$. 

We are going to formulate the normalized partition function
$$
\frac{\Tr e^{-\beta(\sH+i\theta(\beta)\sS_z+\sF+\sA)}}{\Tr e^{-\beta (\sH_0+i\theta(\beta)\sS_z)}}
$$
into a time-continuum limit of a finite-dimensional Grassmann Gaussian
integral, where we set 
\begin{align*}
&\sA:=\la_1\sA_1+\la_2\sA_2,\\
&\sA_1:=\psi_{\hrho\hbx\ua}^*\psi_{\hrho\hbx\da}^*,\\
&\sA_2:=\psi_{\hrho\hbx\ua}^*\psi_{\hrho\hbx\da}^*\psi_{\heta\hby\da}\psi_{\heta\hby\ua}
\end{align*}
with the artificial parameters $\la_1,\la_2\in\C$ and fixed sites
$(\hrho,\hbx),(\heta,\hby)\in \cB\times \G_{\infty}$. 
The reason why we insert the operator $\sA$ is that we can simply derive the
thermal expectations of our interest by differentiating the partition
function with the parameters $\la_1,\la_2$. We can compute the
denominator and check that it is non-zero because of the property $\theta(\beta)\in [0,2\pi/\beta)$.

\begin{lemma}\label{lem_free_partition_function}
\begin{align*}
\Tr e^{-\beta(\sH_0+i\theta(\beta)\sS_z)}&=\prod_{\bk\in \G^*}\det\left(
1+2\cos\left(\frac{\beta\theta(\beta)}{2}\right)e^{-\beta
 E(\bk)}+e^{-2\beta E(\bk)}\right)\\
&=e^{-\beta \sum_{\bk\in\G^*}\Tr E(\bk)}2^{bL^d}
\prod_{\bk\in \G^*}\det\left(
\cos\left(\frac{\beta\theta(\beta)}{2}\right)+\cosh(\beta
 E(\bk))\right).
\end{align*}
\end{lemma}

\begin{proof} This is a $b$-band version of \cite[\mbox{Lemma
 2.1}]{K_BCS}. One can diagonalize $\sH_0+i\theta(\beta)\sS_z$ with respect to
 the band index and derive the result.
\end{proof}

To state the first Grassmann integral formulation, let us introduce the
Grassmann algebra and the covariance. With the parameter $h\in
\frac{2}{\beta}\N$, set $[0,\beta)_h:=\{0,1/h,2/h,\cdots,\beta-1/h\}$,
which is a discretization of $[0,\beta)$. Define the index sets $J_0$,
$J$ by $J_0:=\cB\times \G\times \spin \times [0,\beta)_h$, $J:=J_0\times
\{1,-1\}$.  
Let $\cW$ be the complex vector space spanned by the abstract basis
$\{\psi_X\}_{X\in J}$. We let $\bigwedge \cW$ denote the Grassmann
algebra generated by $\{\psi_X\}_{X\in J}$. For $X\in J_0$ we also use the
notation $\opsi_X$, $\psi_X$ in place of $\psi_{(X,1)}$, $\psi_{(X,-1)}$
respectively. We do not restate the definitions and basic properties of
finite-dimensional Grassmann algebra and Grassmann integrations in detail. The
readers should refer to \cite[\mbox{Subsection 2.1}]{K_BCS} for the
summary of them in line with our purposes or to \cite{FKT} for more
general statements. Let us introduce the Grassmann polynomials
$\sV(\psi)$, $\sF(\psi)$, $\sA^1(\psi)$, $\sA^2(\psi)$, $\sA(\psi)$ $(\in \bigwedge
\cW)$ formulating the operators $\sV$, $\sF$, $\sA^1$, $\sA^2$, $\sA$
respectively as follows. 
\begin{align*}
&\sV(\psi):=\frac{U}{L^dh}\sum_{(\rho,\bx),(\eta,\by)\in \cB\times
 \G}\sum_{s\in [0,\beta)_h}\opsi_{\rho\bx\ua s}\opsi_{\rho\bx\da
 s}\psi_{\eta\by\da s}\psi_{\eta\by\ua s},\\
&\sF(\psi):=\frac{\g}{h}\sum_{(\rho,\bx)\in \cB\times
 \G}\sum_{s\in [0,\beta)_h}(\opsi_{\rho\bx\ua s}\opsi_{\rho\bx\da
 s}+\psi_{\rho\bx\da s}\psi_{\rho\bx\ua s}),\\
&\sA^1(\psi):=\frac{1}{h}\sum_{s\in[0,\beta)_h}\opsi_{\hrho r_L(\hbx)\ua
 s}\opsi_{\hrho r_L(\hbx)\da s},\\
&\sA^2(\psi):=\frac{1}{h}\sum_{s\in[0,\beta)_h}\opsi_{\hrho r_L(\hbx)\ua
 s}\opsi_{\hrho r_L(\hbx)\da s}\psi_{\heta r_L(\hby)\da s}\psi_{\heta r_L(\hby)\ua
 s},\\
&\sA(\psi):=\sum_{j=1}^2\la_j\sA^j(\psi).
\end{align*}

The covariance $G$ for the Grassmann Gaussian integral is defined as the
free 2-point correlation function. For
$(\rho,\bx,\s,s),(\eta,\by,\tau,t)\in \cB\times \G_{\infty}\times
\spin\times [0,\beta)$
\begin{align*}
&G(\rho\bx\s s,\eta\by \tau t):=\frac{\Tr (e^{-\beta
 (\sH_0+i\theta(\beta)\sS_z)}(1_{s\ge
 t}\psi_{\rho\bx\s}^*(s)\psi_{\eta\by\tau}(t)
-1_{s< t}\psi_{\eta\by\tau}(t)\psi_{\rho\bx\s}^*(s)))}{\Tr e^{-\beta
 (\sH_0+i\theta(\beta)\sS_z)}},
\end{align*}
where
$\psi_{\rho\bx\s}^{(*)}(s):=e^{s(\sH_0+i\theta(\beta)\sS_z)}\psi_{\rho\bx\s}^{(*)}e^{-s(\sH_0+i\theta(\beta)\sS_z)}$.
According to the conventional definition, any covariance for Grassmann
Gaussian integral on $\bigwedge \cW$ is a map from $J_0^2$ to $\C$. If
we follow the convention, we should introduce our covariance as the
restriction $G|_{J_0^2}$. However, we call $G$ covariance and omit the
sign $|_{J_0^2}$ even when the argument is restricted to $J_0^2$ for
simplicity. 

For $r\in \R_{>0}$ let $D(r)$ denote the open disk $\{z\in \C\ |\
|z|<r\}$.

\begin{lemma}\label{lem_first_Grassmann_formulation}
For any $r\in \R_{>0}$
\begin{align*}
\lim_{h\to\infty\atop h\in
 \frac{2}{\beta}\N}\sup_{\bla\in\overline{D(r)}^2}
\left|\int e^{-\sV(\psi)-\sF(\psi)-\sA(\psi)}d\mu_{G}(\psi)-\frac{\Tr
 e^{-\beta(\sH+i\theta(\beta)\sS_z+\sF+\sA)}}{\Tr e^{-\beta
 (\sH_0+i\theta(\beta)\sS_z)}}\right|=0.
\end{align*}
Here $\bla$ denotes $(\la_1,\la_2)$.
\end{lemma}

\begin{proof}
The proof is parallel to that of \cite[\mbox{Lemma 2.2}]{K_BCS}.
\end{proof}

The next step is to reformulate the Grassmann Gaussian integral given in
Lemma \ref{lem_first_Grassmann_formulation} into a hybrid of a
Gaussian integral with Grassmann variables and a Gaussian integral
with real variables. Define $\sV_+(\psi)$, $\sV_-(\psi)$, $\sW(\psi)\in
\bigwedge \cW$ by
\begin{align*}
&\sV_+(\psi):=\frac{|U|^{\frac{1}{2}}}{\beta^{\frac{1}{2}}L^{\frac{d}{2}}h}\sum_{(\rho,\bx)\in\cB\times\G}\sum_{s\in
 [0,\beta)_h}\opsi_{\rho\bx\ua s} \opsi_{\rho\bx\da s},\\
&\sV_-(\psi):=\frac{|U|^{\frac{1}{2}}}{\beta^{\frac{1}{2}}L^{\frac{d}{2}}h}\sum_{(\rho,\bx)\in\cB\times\G}\sum_{s\in
 [0,\beta)_h}\psi_{\rho\bx\da s} \psi_{\rho\bx\ua s},\\
&\sW(\psi):=\frac{U}{\beta L^dh^2}\sum_{(\rho,\bx),(\eta,\by)\in\cB\times\G}\sum_{s,t\in
 [0,\beta)_h}\opsi_{\rho\bx\ua s} \opsi_{\rho\bx\da s}\psi_{\eta\by\da
 t} \psi_{\eta\by\ua t}.
\end{align*}

\begin{lemma}\label{lem_hybrid_Grassmann_formulation}
\begin{align*}
&\int e^{-\sV(\psi)-\sF(\psi)-\sA(\psi)}d\mu_G(\psi)\\
&=\frac{1}{\pi}\int_{\R^2}d\phi_1d\phi_2 e^{-|\phi|^2}\int
 e^{-\sV(\psi)+\sW(\psi)-\sF(\psi)-\sA(\psi)+\phi \sV_+(\psi)+\ophi \sV_-(\psi)}
d\mu_G(\psi),
\end{align*}
where $\phi=\phi_1+i\phi_2$, $|\phi|=\|\phi\|_{\C}$. 
\end{lemma}

\begin{proof}
The proof is same as that of \cite[\mbox{Lemma 2.3}]{K_BCS} based on the
 Hubbard-\\
Stratonovich transformation.
\end{proof}

As the final step of the formulation, we introduce the index $\{1,2\}$
and derive the integral formulation on Grassmann algebra indexed by
$\{1,2\}$ rather than by the spin $\spin$. The new index sets are
defined by
\begin{align*}
I_0:=\{1,2\}\times\cB\times \G\times [0,\beta)_h,\quad I:=I_0\times \{1,-1\}.
\end{align*}
Let $\cV$ be the complex vector space spanned by the basis
$\{\psi_{X}\}_{X\in I}$. We define the Grassmann polynomials $V(\psi)$,
$W(\psi)$, $A^1(\psi)$, $A^2(\psi)$, $A(\psi)\in \bigwedge\cV$ by 
\begin{align}
&V(\psi):=\frac{U}{L^dh}\sum_{(\rho,\bx)\in \cB\times
 \G}\sum_{s\in [0,\beta)_h}\opsi_{1\rho\bx s}\psi_{1\rho\bx
 s}\notag\\
&\qquad\qquad +\frac{U}{L^dh}\sum_{(\rho,\bx),(\eta,\by)\in \cB\times
 \G}\sum_{s\in [0,\beta)_h}\opsi_{1\rho\bx s}\psi_{2\rho\bx
 s}\opsi_{2\eta\by s}\psi_{1\eta\by s},\notag\\
&W(\psi):=\frac{U}{\beta L^dh^2}\sum_{(\rho,\bx),(\eta,\by)\in \cB\times
 \G}\sum_{s,t\in [0,\beta)_h}\opsi_{1\rho\bx s}\psi_{2\rho\bx
 s}\opsi_{2\eta\by t}\psi_{1\eta\by t},\notag\\
&A^1(\psi):=\frac{1}{h}\sum_{s\in[0,\beta)_h}\opsi_{1\hrho r_L(\hbx)
 s}\psi_{2\hrho r_L(\hbx) s},\notag\\
&A^2(\psi):=1_{(\hrho,r_L(\hbx))=(\heta,r_L(\hby))}\frac{1}{h}\sum_{s\in[0,\beta)_h}\opsi_{1\hrho r_L(\hbx)
 s}\psi_{1\hrho r_L(\hbx) s}\notag\\
&\qquad\qquad +\frac{1}{h}\sum_{s\in[0,\beta)_h}\opsi_{1\hrho r_L(\hbx)
 s}\psi_{2\hrho r_L(\hbx) s}\opsi_{2\heta r_L(\hby) s}\psi_{1\heta r_L(\hby)
 s},\notag\\
&A(\psi):=\sum_{j=1}^2\la_j
 A^j(\psi).\label{eq_Grassmann_artificial_term}
\end{align}
Though the final formulation Lemma \ref{lem_final_Grassmann_formulation}
does not explicitly involve any partition
function of a Hamiltonian on the Fock space $F_f(L^2(\{1,2\}\times
\cB\times \G))$, the final formulation can be systematically derived by 
relating such a partition function to the Grassmann
Gaussian integral over $\bigwedge \cV$. To this end, let us define a
free Hamiltonian on $F_f(L^2(\{1,2\}\times \cB\times\G))$. For any $n\in
\N$ let $I_n$ denote the $n\times n$ unit matrix. For $\phi\in \C$, set
\begin{align*}
H_0(\phi)
:=\frac{1}{L^d}\sum_{\bx,\by\in\G}&\sum_{\bk\in
 \G^*}e^{i\<\bk,\bx-\by\>}\\
&\cdot \<\left(\begin{array}{c}\Psi_{1\bx}^* \\ \Psi_{2\bx}^*\end{array}
 \right),
\left(\begin{array}{cc} i\frac{\theta(\beta)}{2}I_b+E(\bk) & \phi I_b \\
                       \ophi I_b   & i\frac{\theta(\beta)}{2}I_b-E(\bk)
\end{array}\right)
\left(\begin{array}{c}\Psi_{1\by} \\ \Psi_{2\by}\end{array}
 \right)\>,
\end{align*}
where $\Psi_{\orho\bx}^{(*)}:=(\psi_{\orho 1 \bx}^{(*)},\psi_{\orho 2
\bx}^{(*)},\cdots,\psi_{\orho b \bx}^{(*)})^T$ and 
$\psi_{\orho \rho \bx}^*$ $(\psi_{\orho \rho \bx})$ is the creation
(annihilation) operator on $F_f(L^2(\{1,2\}\times \cB\times \G))$ for
$\orho\in \{1,2\}$, $\rho\in\cB$, $\bx\in \G$. The covariance in the
final formulation is equal to the free 2-point correlation function
$C(\phi):(\{1,2\}\times \cB\times \G_{\infty}\times [0,\beta))^2\to\C$
defined by 
\begin{align*}
C(\phi)(\orho\rho \bx s,\oeta\eta \by t):=\frac{\Tr (e^{-\beta H_0(\phi)}(1_{s\ge t}\psi_{\orho\rho
 \bx}^*(s)\psi_{\oeta\eta \by}(t)-1_{s<t}\psi_{\oeta\eta \by}(t)\psi_{\orho\rho
 \bx}^*(s)))}{\Tr e^{-\beta H_0(\phi)}},
\end{align*}
where $\psi_{\orho \rho \bx}^{(*)}(s):=e^{s H_0(\phi)}\psi_{\orho \rho
\bx}^{(*)}e^{-s H_0(\phi)}$. For $\bx\in \G_{\infty}$ we identify
$\psi_{\orho \rho \bx}^{(*)}$ with $\psi_{\orho \rho
r_L(\bx)}^{(*)}$. The next lemma ensures the well-definedness of
$C(\phi)$ and gives its characterization and determinant bound. For any
$\bk\in\R^d$, $\phi\in\C$ we define $E(\phi)(\bk)\in \Mat(2b,\C)$ by 
$$
E(\phi)(\bk):=\left(\begin{array}{cc} E(\bk) & \ophi I_b \\
                                      \phi I_b & -E(\bk) \end{array}\right).
$$ 

\begin{lemma}\label{lem_characterization_covariance}
\begin{enumerate}[(i)]
\item\label{item_free_partition_function_field}
\begin{align*}
\Tr e^{-\beta H_0(\phi)}&=\prod_{\bk\in\G^*}\prod_{\delta \in
 \{1,-1\}}\det (1+ e^{-\beta (i\frac{\theta(\beta)}{2}+\delta
 \sqrt{E(\bk)^2+|\phi|^2})})\\
&=e^{-\frac{i}{2}\beta\theta(\beta)bL^d}2^{bL^d}\prod_{\bk\in
 \G^*}\det\left(\cos\left(\frac{\beta\theta(\beta)}{2}\right)
+\cosh(\beta\sqrt{E(\bk)^2+|\phi|^2})\right).
\end{align*}
\item\label{item_characterization_covariance}
For any $(\orho,\rho,\bx,s),(\oeta,\eta,\by,t)\in
     \{1,2\}\times\cB\times\G_{\infty}\times[0,\beta)$,
\begin{align}
&C(\phi)(\orho\rho\bx s,\oeta\eta\by
 t)\label{eq_characterization_covariance}\\
&=\frac{1}{L^d}\sum_{\bk\in
 \G^*}e^{i\<\bk,\bx-\by\>}e^{(s-t)(i\frac{\theta(\beta)}{2}I_{2b}+E(\phi)(\bk))}\notag\\
&\quad\cdot \left(1_{s\ge t} (I_{2b}+e^{\beta
 (i\frac{\theta(\beta)}{2}I_{2b}+E(\phi)(\bk))})^{-1}-
1_{s< t} (I_{2b}+e^{-\beta
 (i\frac{\theta(\beta)}{2}I_{2b}+E(\phi)(\bk))})^{-1}\right)\notag\\
&\quad\cdot ((\orho-1)b+\rho,(\oeta-1)b+\eta).\notag
\end{align}
\item\label{item_P_S_determinant_bound}
\begin{align*}
&|\det(\<\bu_i,\bv_j\>_{\C^m}C(\phi)(X_i,Y_j))_{1\le i,j\le n}|\\
&\le \left(\frac{2^4b}{L^d}\sum_{\bk\in\G^*}\Tr
 \left(1+2\cos\left(\frac{\beta\theta(\beta)}{2}\right)
 e^{-\beta\sqrt{E(\bk)^2+|\phi|^2}}+ e^{-2\beta\sqrt{E(\bk)^2+|\phi|^2}}
\right)^{-\frac{1}{2}}\right)^n\\
&\le
 \left(2^4b^2\left(1+\frac{\pi}{\beta}\left|\frac{\theta(\beta)}{2}-\frac{\pi}{\beta}\right|^{-1}\right)\right)^n,\\
&(\forall m,n\in\N,\ \bu_i,\bv_i\in\C^m\text{ with
 }\|\bu_i\|_{\C^m},\|\bv_i\|_{\C^m}\le 1,\\
&\quad  X_i,Y_i\in
 \{1,2\}\times\cB\times\G\times[0,\beta)\ (i=1,2,\cdots,n),\ \phi\in \C).
\end{align*}
Here $\<\cdot,\cdot\>_{\C^m}$ denotes the canonical Hermitian inner
     product of $\C^m$. 
\end{enumerate}
\end{lemma}

\begin{proof}
Let $e_{\rho}(\bk)$ $(\rho\in \cB)$ be the eigenvalues of $E(\bk)$. Then
 the eigenvalues of $E(\phi)(\bk)$ are
 $\sqrt{e_{\rho}(\bk)^2+|\phi|^2}$, $-\sqrt{e_{\rho}(\bk)^2+|\phi|^2}$
 ($\rho\in\cB$). There exists $U(\phi)\in \Map(\R^d,\Mat(2b,\C))$ such
 that $U(\phi)(\bk)$ is unitary and 
\begin{align}
&U(\phi)(\bk)^*E(\phi)(\bk)U(\phi)(\bk)((\orho-1)b+\rho,
 (\oeta-1)b+\eta)\label{eq_field_hopping_matrix_diagonalization} \\
&= 1_{(\orho,\rho)=(\oeta,\eta)}
 (-1)^{1_{\orho=2}}\sqrt{e_{\rho}(\bk)^2+|\phi|^2},\notag\\
&(\forall
(\orho,\rho),(\oeta,\eta)\in \{1,2\}\times\cB,\ \bk\in \R^d).\notag
\end{align}
Set
 $\alpha_{(\orho-1)b+\rho}(\phi)(\bk):=i\frac{\theta(\beta)}{2}+(-1)^{1_{\orho=2}}\sqrt{e_{\rho}(\bk)^2+|\phi|^2}$.
 Remark that
\begin{align*}
H_0(\phi)
=\frac{1}{L^d}\sum_{\bx,\by\in\G}\sum_{\bk\in
 \G^*}e^{i\<\bk,\bx-\by\>}\<\left(\begin{array}{c}\Psi_{1\bx}^* \\ \Psi_{2\bx}^*\end{array}
 \right),\left(i\frac{\theta(\beta)}{2}I_{2b}+E(\ophi)(\bk)\right)
\left(\begin{array}{c}\Psi_{1\by} \\ \Psi_{2\by}\end{array} \right)
\>.
\end{align*}
Then we can see that there is a unitary transform $\cU(\phi)$
 on $F_f(L^2(\{1,2\}\times\cB\times \G))$ such that
\begin{align}
&\cU(\phi)\psi_{\orho\rho\bx}^*\cU(\phi)^*\label{eq_field_unitary_creation}\\
&=\frac{1}{L^d}\sum_{(\oeta,\eta)\in
 \{1,2\}\times\cB}\sum_{\by\in\G}
\sum_{\bk\in\G^*}e^{-i\<\bk,\bx-\by\>}\overline{U(\ophi)(\bk)((\orho-1)b+\rho,(\oeta-1)b+\eta)}\psi_{\oeta\eta\by}^*,\notag\\
&\cU(\phi)H_0(\phi)\cU(\phi)^*=\frac{1}{L^d}\sum_{\bx,\by\in\G}\sum_{(\orho,\rho)\in
 \{1,2\}\times \cB}\sum_{\bk\in\G^*}e^{i\<\bk,\bx-\by\>}\alpha_{(\orho-1)b+\rho}(\phi)(\bk)\psi_{\orho\rho\bx}^*\psi_{\orho\rho\by}.\label{eq_free_field_hamiltonian_diagonalization}
\end{align}

\eqref{item_free_partition_function_field}:
Since $H_0(\phi)$ is diagonalized with respect to the band index in
 \eqref{eq_free_field_hamiltonian_diagonalization}, a standard argument
 yields that 
\begin{align*}
\Tr e^{-\beta H_0(\phi)}=\prod_{\bk\in\G^*}\prod_{\orho\in
 \{1,2\}}\prod_{\rho\in\cB}(1+e^{-\beta
 \alpha_{(\orho-1)b+\rho}(\phi)(\bk)}),
\end{align*}
which implies the claim.

\eqref{item_characterization_covariance}: 
Insertion of \eqref{eq_field_unitary_creation} gives that
\begin{align}
&C(\phi)(\orho\rho\bx s,\oeta \eta \by t)\label{eq_covariance_characterization_pre}\\
&=\frac{1}{L^{2d}}\sum_{\bx',\by'\in \G}\sum_{(\orho',\rho'),
 (\oeta',\eta')\in\{1,2\}\times
 \cB}\sum_{\bk,\bp\in\G^*}e^{-i\<\bk,\bx-\bx'\>+i\<\bp,\by-\by'\>}\notag\\
&\quad\cdot\overline{U(\ophi)(\bk)((\orho-1)b+\rho,(\orho'-1)b+\rho')}
U(\ophi)(\bp)((\oeta-1)b+\eta,(\oeta'-1)b+\eta')\notag\\
&\quad\cdot 
\frac{\Tr (e^{-\beta \cU(\phi)H_0(\phi)\cU(\phi)^*}(1_{s\ge
 t}\tilde{\psi}_{\orho'\rho'\bx'}^*(s)\tilde{\psi}_{\oeta'\eta'\by'}(t)-1_{s< t}\tilde{\psi}_{\oeta'\eta'\by'}(t)\tilde{\psi}^*_{\orho'\rho'\bx'}(s)))}{\Tr e^{-\beta \cU(\phi)H_0(\phi)\cU(\phi)^*}},\notag
\end{align}
where
 $\tilde{\psi}_{\orho\rho\bx}^*(s):=e^{s\cU(\phi)H_0(\phi)\cU(\phi)^*}\psi_{\orho\rho\bx}^{(*)}e^{-s\cU(\phi)H_0(\phi)\cU(\phi)^*}$.
Since $\cU(\phi)H_0(\phi)\cU(\phi)^*$ is diagonalized with the band
 index, an argument parallel to the proof of \cite[\mbox{Lemma
 B.10}]{K_9} yields that 
\begin{align*}
&\frac{\Tr (e^{-\beta \cU(\phi)H_0(\phi)\cU(\phi)^*}(1_{s\ge
 t}\tilde{\psi}_{\orho'\rho'\bx'}^*(s)\tilde{\psi}_{\oeta'\eta'\by'}(t)-1_{s<
 t}\tilde{\psi}_{\oeta'\eta'\by'}(t)\tilde{\psi}^*_{\orho'\rho'\bx'}(s)))}{\Tr
 e^{-\beta \cU(\phi)H_0(\phi)\cU(\phi)^*}}\\
&=\frac{1_{(\orho',\rho')=(\oeta',\eta')}}{L^d}\sum_{\bq\in
 \G^*}e^{-i\<\bq,\bx'-\by'\>}e^{(s-t)\alpha_{(\orho'-1)b+\rho'}(\phi)(\bq)}\\
&\quad\cdot \left(\frac{1_{s\ge
 t}}{1+e^{\beta \alpha_{(\orho'-1)b+\rho'}(\phi)(\bq)}}
-\frac{1_{s< t}}{1+e^{-\beta \alpha_{(\orho'-1)b+\rho'}(\phi)(\bq)}}
\right).
\end{align*}
We should remark that here we have the exponent $-i\<\bq,\bx'-\by'\>$
 not $i\<\bq,\bx'-\by'\>$. By substituting this equality into
 \eqref{eq_covariance_characterization_pre} and using
 \eqref{eq_field_hopping_matrix_diagonalization} we observe that
\begin{align*}
&C(\phi)(\orho\rho\bx s,\oeta\eta \by t)\\
&=\frac{1}{L^d}\sum_{\bk\in\G^*}e^{-i\<\bk,\bx-\by\>}e^{(s-t)(i\frac{\theta(\beta)}{2}I_{2b}+E(\ophi)(\bk))}\\
&\quad\cdot \left(1_{s\ge
 t} (I_{2b}+e^{\beta (i\frac{\theta(\beta)}{2}I_{2b}+E(\ophi)(\bk))})^{-1}-1_{s<t}
(I_{2b}+e^{-\beta
 (i\frac{\theta(\beta)}{2}I_{2b}+E(\ophi)(\bk))})^{-1}\right)\\
&\quad\cdot ((\oeta-1)b+\eta,(\orho-1)b+\rho).
\end{align*}
It follows from \eqref{eq_self_adjointness} and
 \eqref{eq_spatial_reflection_symmetry} that
 $E(\ophi)(\bk)^T=E(\phi)(-\bk)$, $(\forall \bk\in\R^d)$. By combining
 this equality with the above characterization of $C(\phi)$ and using periodicity
we obtain
 \eqref{eq_characterization_covariance}. 

\eqref{item_P_S_determinant_bound}: In \cite[\mbox{Proposition
 4.1}]{K_BCS} we stated a version of Pedra-Salmhofer's determinant bound 
\cite[\mbox{Theorem 1.3}]{PS}. In \cite[\mbox{Appendix A}]{K_BCS} we
 gave a short proof of \cite[\mbox{Proposition
 4.1}]{K_BCS}. By applying \cite[\mbox{Proposition
 4.1}]{K_BCS} we derived the determinant bound \cite[\mbox{Proposition
 4.2}]{K_BCS} which gives the claimed determinant bound in the case $b=1$.
Here let us use the proof of \cite[\mbox{Proposition 4.2}]{K_BCS} and
 \cite[\mbox{Lemma A.1}]{K_BCS}, which is a simple application of the
 Cauchy-Binet formula, to derive the claimed determinant bound in the
 general case. It
 follows from \eqref{eq_characterization_covariance} and
 \eqref{eq_field_hopping_matrix_diagonalization} that for any
 $(\orho,\rho,\bx,s)$, $(\oeta,\eta,\by,t)\in \{1,2\}\times \cB\times
 \G_{\infty}\times [0,\beta)$
\begin{align}
C(\phi)(\orho\rho\bx s,\oeta\eta \by t)=
\sum_{\rho'\in\cB}C_{\rho'}(\phi)(\orho\rho\bx s,\oeta \eta \by
 t),\label{eq_covariance_band_index_sum}
\end{align}
where
\begin{align*}
&C_{\rho'}(\phi)(\orho\rho\bx s,\oeta \eta \by t)\\
&:=\frac{1}{L^d}\sum_{\bk\in\G^*}\sum_{\orho'\in
 \{1,2\}}e^{i\<\bk,\bx-\by\>}\\
&\quad\cdot U(\phi)(\bk)((\orho-1)b+\rho,(\orho'-1)b+\rho')
U(\phi)(\bk)^*((\orho'-1)b+\rho',(\oeta-1)b+\eta)\\
&\quad\cdot e^{(s-t)\alpha_{(\orho'-1)b+\rho'}(\phi)(\bk)}\left(
\frac{1_{s\ge t}}{1+e^{\beta \alpha_{(\orho'-1)b+\rho
'}(\phi)(\bk)}}-
\frac{1_{s< t}}{1+e^{-\beta \alpha_{(\orho'-1)b+\rho
'}(\phi)(\bk)}}
\right).
\end{align*}
In the proof of \cite[\mbox{Proposition 4.2}]{K_BCS} we estimated the
 determinant bound of a covariance whose form is close to 
 $C_{\rho'}(\phi)$. By following the proof of \cite[\mbox{Proposition
 4.2}]{K_BCS} straightforwardly we can deduce that 
\begin{align}
&|\det(\<\bu_i,\bv_j\>_{\C^m}C_{\rho'}(\phi)(X_i,Y_j))_{1\le i,j\le n}|\label{eq_field_covariance_pre_determinant}\\
&\le \left(\frac{2^4}{L^d}\sum_{\bk\in\G^*}\left(1+2\cos\left(\frac{\beta\theta(\beta)}{2}
\right)e^{-\beta \sqrt{e_{\rho'}(\bk)^2+|\phi|^2}}+ e^{-2\beta \sqrt{e_{\rho'}(\bk)^2+|\phi|^2}}\right)^{-\frac{1}{2}}\right)^n,\notag\\
&(\forall m,n\in\N,\ \bu_i,\bv_i\in\C^m\text{ with
 }\|\bu_i\|_{\C^m},\|\bv_i\|_{\C^m}\le 1,\notag\\
&\quad  X_i,Y_i\in
 \{1,2\}\times\cB\times \G\times[0,\beta)\ (i=1,2,\cdots,n),\ \phi\in \C,\
 \rho'\in \cB).\notag
\end{align}
To support the readers, let us provide a guidance to derive
 \eqref{eq_field_covariance_pre_determinant}. By using vectors of
 $\Map(\{1,2\}\times\cB\times\G\times \R,L^2(\G^*\times\R))$ and the
 inner product of $L^2(\G^*\times\R)$ we can rewrite the regularized
 version of $C_{\rho'}(\phi)$ in a form close to 
\cite[\mbox{(4.4)}]{K_BCS}. The vectors satisfy a uniform bound similar
 to \cite[\mbox{(4.5)}]{K_BCS}. In this situation we can apply a close
 variant of \cite[\mbox{Proposition 4.1}]{K_BCS} to the regularized version of
 $C_{\rho'}(\phi)$. Then by sending the parameter used to regularize   
 $C_{\rho'}(\phi)$ to zero we obtain \eqref{eq_field_covariance_pre_determinant}.
Since we have \eqref{eq_covariance_band_index_sum} and \eqref{eq_field_covariance_pre_determinant}, we can repeatedly apply
 \cite[\mbox{Lemma A.1}]{K_BCS} to derive that 
\begin{align}
&|\det(\<\bu_i,\bv_j\>_{\C^m}C(\phi)(X_i,Y_j))_{1\le i,j\le n}|\label{eq_derivation_via_previous_bound}\\
&\le \left(\sum_{\rho'\in\cB}\left(\frac{2^4}{L^d}\sum_{\bk\in\G^*}\left(1+2\cos\left(\frac{\beta\theta(\beta)}{2}
\right)e^{-\beta \sqrt{e_{\rho'}(\bk)^2+|\phi|^2}}+ e^{-2\beta \sqrt{e_{\rho'}(\bk)^2+|\phi|^2}}\right)^{-\frac{1}{2}}\right)^{\frac{1}{2}}\right)^{2n},\notag\\
&(\forall m,n\in\N,\ \bu_i,\bv_i\in\C^m\text{ with
 }\|\bu_i\|_{\C^m},\|\bv_i\|_{\C^m}\le 1,\notag\\
&\quad  X_i,Y_i\in
 \{1,2\}\times\cB\times \G\times[0,\beta)\ (i=1,2,\cdots,n),\ \phi\in \C),\notag
\end{align}
which together with Schwarz' inequality yields the first inequality of
 the claim \eqref{item_P_S_determinant_bound}. 
It is also possible to derive \eqref{eq_derivation_via_previous_bound}
 by directly applying \cite[\mbox{Theorem 1.3}]{PS}. In this case one
 should decompose $C(\phi)$ into a sum of $2b$ time-ordered covariances.
Note that
\begin{align*}
&1+2\cos\left(\frac{\beta\theta(\beta)}{2}
\right)e^{-\beta \sqrt{e_{\rho}(\bk)^2+|\phi|^2}}+ e^{-2\beta \sqrt{e_{\rho}(\bk)^2+|\phi|^2}}\\
&\ge 1_{\theta(\beta)\in [0,\pi/\beta]}+ 1_{\theta(\beta)\in
 (\pi/\beta,2\pi/\beta)}\sin^2\left(\frac{\beta\theta(\beta)}{2}\right)\\
&\ge  1_{\theta(\beta)\in [0,\pi/\beta]}+ 1_{\theta(\beta)\in
 (\pi/\beta,2\pi/\beta)}\frac{\beta^2}{\pi^2}\left(\frac{\theta(\beta)}{2}-\frac{\pi}{\beta}\right)^2,
\end{align*}
which implies the second inequality.
\end{proof}

We finalize our formulation in the next lemma. We should remark
that Lemma \ref{lem_first_Grassmann_formulation} and Lemma
\ref{lem_hybrid_Grassmann_formulation} will not see any application in
the rest of the paper. We stated these lemmas in the hope that the
readers can prove the next lemma by putting these lemmas together in the
same manner as in the proof of \cite[\mbox{Lemma 2.5}]{K_BCS}. While it
was quartic in \cite{K_BCS}, here the Grassmann polynomial $A^2(\psi)$ may contain
quadratic terms. This is because we assumed $r_L(\hbx)\neq r_L(\hby)$ in
the previous construction and here we need to drop this assumption in
order to study the Cooper pair density as claimed in Theorem
\ref{thm_main_theorem} \eqref{item_CPD}, Corollary
\ref{cor_zero_temperature_limit} \eqref{item_CPD_zero}. 

\begin{lemma}\label{lem_final_Grassmann_formulation}
The following statements hold true for any $r\in \R_{>0}$. 
\begin{enumerate}[(i)]
\item\label{item_formulation_uniform_convergence}
\begin{align*}
\lim_{h\to\infty\atop h\in \frac{2}{\beta}\N}\int e^{-V(\psi)+W(\psi)-A(\psi)}d\mu_{C(\phi)}(\psi)
\end{align*}
converges in $C(\overline{D(r)}^3)$ as a sequence of function with the
     variables $(\bla,\phi)\in \overline{D(r)}^3$.
\item\label{item_formulation_integrability}
The $C(\overline{D(r)}^2)$-valued function
\begin{align*}
(\phi_1,\phi_2)\mapsto &e^{-\frac{\beta
 L^d}{|U|}|\phi-\g|^2}\frac{\prod_{\bk\in\G^*}\det(\cos(\beta\theta(\beta)/2)+\cosh(\beta\sqrt{E(\bk)^2+|\phi|^2}))}
{\prod_{\bk\in\G^*}\det(\cos(\beta\theta(\beta)/2)+\cosh(\beta E(\bk)))}\\
&\cdot \lim_{h\to\infty\atop h\in \frac{2}{\beta}\N}\int
 e^{-V(\psi)+W(\psi)-A(\psi)}d\mu_{C(\phi)}(\psi)
\end{align*}
belongs to $L^1(\R^2,C(\overline{D(r)}^2))$.
\item\label{item_final_Grassmann_formulation}
\begin{align*}
&\frac{\Tr e^{-\beta (\sH+i\theta(\beta)\sS_z+\sF+\sA)}}{\Tr e^{-\beta
 (\sH_0+i\theta(\beta)\sS_z)}}\\
&=\frac{\beta L^d}{\pi |U|}\int_{\R^2}d\phi_1d\phi_2
e^{-\frac{\beta
 L^d}{|U|}|\phi-\g|^2}\\
&\quad\cdot \frac{\prod_{\bk\in\G^*}\det(\cos(\beta\theta(\beta)/2)+\cosh(\beta\sqrt{E(\bk)^2+|\phi|^2}))}
{\prod_{\bk\in\G^*}\det(\cos(\beta\theta(\beta)/2)+\cosh(\beta E(\bk)))}\\
&\quad\cdot \lim_{h\to\infty\atop h\in \frac{2}{\beta}\N}\int
 e^{-V(\psi)+W(\psi)-A(\psi)}d\mu_{C(\phi)}(\psi),\\
&\frac{\Tr( e^{-\beta (\sH+i\theta(\beta)\sS_z+\sF)}\sA_j)}{\Tr e^{-\beta
 (\sH_0+i\theta(\beta)\sS_z)}}\\
&=\frac{L^d}{\pi |U|}\int_{\R^2}d\phi_1d\phi_2
e^{-\frac{\beta
 L^d}{|U|}|\phi-\g|^2}\\
&\quad\cdot \frac{\prod_{\bk\in\G^*}\det(\cos(\beta\theta(\beta)/2)+\cosh(\beta\sqrt{E(\bk)^2+|\phi|^2}))}
{\prod_{\bk\in\G^*}\det(\cos(\beta\theta(\beta)/2)+\cosh(\beta E(\bk)))}\\
&\quad\cdot \lim_{h\to\infty\atop h\in \frac{2}{\beta}\N}\int
 e^{-V(\psi)+W(\psi)}A^j(\psi)d\mu_{C(\phi)}(\psi),\\
&(j=1,2).
\end{align*}
\item\label{item_formulation_reality}
For any $\phi\in\C$,
$$
\lim_{h\to\infty\atop h\in \frac{2}{\beta}\N}\int
 e^{-V(\psi)+W(\psi)}d\mu_{C(\phi)}(\psi)\in\R.
$$
\end{enumerate}
\end{lemma}

\begin{remark} Since we have obtained the $\phi$-independent determinant bound in Lemma
 \ref{lem_characterization_covariance}
 \eqref{item_P_S_determinant_bound}, we can readily prove that the
 integral with $(\phi_1,\phi_2)$ and the limit operation $h\to \infty$
 are interchangeable in the claim
 \eqref{item_final_Grassmann_formulation}. However, since we need to
 take large $h$ depending on fixed $(\phi_1,\phi_2)$ in the analysis of
 $C(\phi)$ in Subsection \ref{subsec_real_covariance}, taking the limit
 $h\to \infty$ after the integration with $(\phi_1,\phi_2)$ has no
 application in this paper.
\end{remark}

\begin{proof}[Proof of Lemma \ref{lem_final_Grassmann_formulation}]
Based on Lemma \ref{lem_free_partition_function}, Lemma
 \ref{lem_first_Grassmann_formulation}, Lemma
 \ref{lem_hybrid_Grassmann_formulation} and Lemma \ref{lem_characterization_covariance},
the claims \eqref{item_formulation_uniform_convergence},
 \eqref{item_formulation_integrability},
 \eqref{item_final_Grassmann_formulation} can be proved in the same way
 as in the proof of \cite[\mbox{Lemma 2.5}]{K_BCS}. Note that the locally uniform convergence with
 $(\bla,\phi)$ is claimed in
 \eqref{item_formulation_uniform_convergence}, while 
the convergence was claimed pointwise with $\phi$ 
in \cite[\mbox{Lemma 2.5 (i)}]{K_BCS}.
The uniform convergence property with $\phi$ can be deduced by making
 use of the $\phi$-independent determinant bound Lemma
 \ref{lem_characterization_covariance}
 \eqref{item_P_S_determinant_bound} in arguments parallel to the proof
 of \cite[\mbox{Lemma 2.5 (i)}]{K_BCS}. Concerning the form of
 the Grassmann polynomials $V(\psi)$, $A(\psi)$, which affects details
 of the forthcoming analysis, we should explain that they stem from the
 use of the unitary map $\cU:F_f(L^2(\cB\times \G\times \spin))\to
 F_f(L^2(\{1,2\}\times\cB\times \G))$ satisfying that
\begin{align*}
\cU\psi_{\rho\bx\ua}^*\cU^*= \psi_{1\rho\bx}^*,\quad
 \cU\psi_{\rho\bx\da}^*\cU^*= \psi_{2\rho\bx},\quad (\forall (\rho,\bx)\in\cB\times\G)
\end{align*}
and thus
\begin{align}
\cU\sV\cU^*=&\frac{U}{L^d}\sum_{(\rho,\bx)\in\cB\times
 \G}\psi_{1\rho\bx}^*\psi_{1\rho\bx}-
\frac{U}{L^d}\sum_{(\rho,\bx),(\eta,\by)\in\cB\times
 \G}\psi_{1\rho\bx}^*\psi_{2\eta\by}^*\psi_{2\rho\bx}
\psi_{1\eta\by},\label{eq_unitary_particle_hole_interaction}\\
\cU\sA\cU^*=&\la_1\psi_{1\hrho r_L(\hbx)}^*\psi_{2\hrho r_L(\hbx)}
\label{eq_unitary_particle_hole_artificial}\\
&+\la_2(1_{(\hrho,r_L(\hbx))=(\heta,r_L(\hby))} \psi_{1\hrho
 r_L(\hbx)}^*\psi_{1\hrho r_L(\hbx)}-
\psi_{1\hrho r_L(\hbx)}^*\psi_{2\heta r_L(\hby)}^*\psi_{2\hrho
 r_L(\hbx)}\psi_{1\heta
 r_L(\hby)}).\notag
\end{align}
The right-hand side of \eqref{eq_unitary_particle_hole_interaction},
 \eqref{eq_unitary_particle_hole_artificial} is formulated into
 $V(\psi)$, $A(\psi)$ respectively.

Let us prove the claim \eqref{item_formulation_reality}. Define the
 Grassmann polynomials $W_+(\psi)$, $W_-(\psi)$ $\in \bigwedge\cV$ by
\begin{align*}
&W_+(\psi):=\frac{i|U|^{\frac{1}{2}}}{\beta^{\frac{1}{2}}L^{\frac{d}{2}}h}\sum_{(\rho,\bx,s)\in\cB\times
 \G \times [0,\beta)_h}\opsi_{1\rho\bx s}\psi_{2\rho\bx s},\\
&W_-(\psi):=\frac{i|U|^{\frac{1}{2}}}{\beta^{\frac{1}{2}}L^{\frac{d}{2}}h}\sum_{(\rho,\bx,s)\in\cB\times
 \G \times [0,\beta)_h}\opsi_{2\rho\bx s}\psi_{1\rho\bx s}.
\end{align*}
Since $W(\psi)=W_+(\psi)W_-(\psi)$, the Hubbard-Stratonovich
 transformation yields that 
\begin{align*}
\int e^{-V(\psi)+W(\psi)}d\mu_{C(\phi)}(\psi)=\frac{1}{\pi}\int_{\R^2}d\xi_1d\xi_2e^{-|\xi|^2}\int e^{-V(\psi)+\xi
 W_+(\psi)+\oxi W_-(\psi)}d\mu_{C(\phi)}(\psi),
\end{align*}
where $\xi=\xi_1+i\xi_2$. See e.g. \cite[\mbox{Lemma 2.3}]{K_BCS} for
 the proof. By setting 
\begin{align*}
D(b,\beta,\theta):=2^4b^2\left(1+\frac{\pi}{\beta}\left|\frac{\theta(\beta)}{2}-\frac{\pi}{\beta}\right|^{-1}\right)
\end{align*}
and applying Lemma \ref{lem_characterization_covariance}
 \eqref{item_P_S_determinant_bound} we obtain that
\begin{align}
&\left|\int e^{-V(\psi)+\xi W_+(\psi)+\oxi W_-(\psi)}d\mu_{C(\phi)}(\psi)\right|\label{eq_determinant_bound_uniform_application}\\
&\le e^{|U|b\beta D(b,\beta,\theta)+|U|b^2L^d\beta
 D(b,\beta,\theta)^2+2|\xi|U|^{1/2}bL^{d/2}\beta^{1/2}D(b,\beta,\theta)}.\notag
\end{align}
Let us define the operators $V$, $W_+$, $W_-$ on
 $F_f(L^2(\{1,2\}\times \cB\times\G))$ by 
\begin{align*}
&V:=\frac{U}{L^d}\sum_{(\rho,\bx)\in
 \cB\times\G}\psi_{1\rho\bx}^*\psi_{1\rho\bx}-
\frac{U}{L^d}\sum_{(\rho,\bx),(\eta,\by)\in
 \cB\times\G}\psi_{1\rho\bx}^*\psi_{2\eta\by}^*\psi_{2\rho\bx}\psi_{1\eta\by},\\
&W_+:=\frac{|U|^{\frac{1}{2}}}{\beta^{\frac{1}{2}}L^{\frac{d}{2}}}\sum_{(\rho,\bx)\in\cB\times
 \G}\psi_{1\rho\bx}^*\psi_{2\rho\bx},\quad W_-:=\frac{|U|^{\frac{1}{2}}}{\beta^{\frac{1}{2}}L^{\frac{d}{2}}}\sum_{(\rho,\bx)\in\cB\times
 \G}\psi_{2\rho\bx}^*\psi_{1\rho\bx}.
\end{align*}
In the same way as in the Grassmann integral formulation procedure
 \cite[\mbox{Lemma 2.2}]{K_BCS} or that of \cite{K_9}, \cite{K_14}, \cite{K_RG} we
 have that for any $r\in\R_{>0}$
\begin{align}
\lim_{h\to\infty\atop h\in
 \frac{2}{\beta}\N}\sup_{\xi\in\overline{D(r)}}
\left|\int e^{-V(\psi)+\xi W_+(\psi)+\oxi W_-(\psi)}d\mu_{C(\phi)}(\psi)
-\frac{\Tr e^{-\beta (H_0(\phi)+V-i\xi W_+ -i\oxi W_-)}}{\Tr e^{-\beta H_0(\phi)}}
\right|=0.\label{eq_present_integral_formulation}
\end{align}
By \eqref{eq_determinant_bound_uniform_application},
 \eqref{eq_present_integral_formulation} we can apply the dominated
 convergence theorem to conclude that 
\begin{align*}
&\lim_{h\to\infty\atop h\in\frac{2}{\beta}\N}\int
 e^{-V(\psi)+W(\psi)}d\mu_{C(\phi)}(\psi)= \frac{1}{\pi}\int_{\R^2}d\xi_1d\xi_2 e^{-|\xi|^2}\frac{\Tr e^{-\beta
 (H_0(\phi)+V-i\xi W_+-i\oxi W_-)}}{\Tr e^{-\beta H_0(\phi)}}.
\end{align*}
To make clear the dependency on the parameter $\theta(\beta)$, let us
 write $H_0(\theta(\beta),\phi)$ in place of $H_0(\phi)$. Observe that 
\begin{align*}
\lim_{h\to \infty\atop h\in \frac{2}{\beta}\N}\overline{\int
 e^{-V(\psi)+W(\psi)}d\mu_{C(\phi)}(\psi)}
&= \frac{1}{\pi}\int_{\R^2}d\xi_1d\xi_2 e^{-|\xi|^2}\frac{\Tr
 e^{-\beta (H_0(\theta(\beta),\phi)+V-i\xi W_+-i\oxi
 W_-)^*}}{\Tr e^{-\beta H_0(\theta(\beta),\phi)^*}}\\
&= \frac{1}{\pi}\int_{\R^2}d\xi_1d\xi_2 e^{-|\xi|^2}\frac{\Tr
 e^{-\beta (H_0(-\theta(\beta),\phi)+V+i\xi W_++i\oxi
 W_-)}}{\Tr e^{-\beta H_0(-\theta(\beta),\phi)}}\\
&= \frac{1}{\pi}\int_{\R^2}d\xi_1d\xi_2 e^{-|\xi|^2}\frac{\Tr
 e^{-\beta (H_0(-\theta(\beta),\phi)+V-i\xi W_+-i\oxi
 W_-)}}{\Tr e^{-\beta H_0(-\theta(\beta),\phi)}}.
\end{align*}
To derive the last equality, we performed the change of variables
 $\xi_j\to -\xi_j$ $(j=1,2)$.

There is a unitary transform $\cU_0$ on $F_f(L^2(\{1,2\}\times \cB\times
 \G))$ satisfying that
\begin{align*}
&\cU_0\psi_{1\rho \bx}\cU_0^*=-\psi_{2\rho \bx}^*,\quad \cU_0\psi_{2\rho
 \bx}\cU_0^*=\psi_{1\rho \bx}^*,\quad (\forall (\rho,\bx)\in\cB\times
 \G).
\end{align*}
We can check that 
\begin{align*}
&\cU_0H_0(-\theta(\beta),\phi)\cU_0^*=-i\theta(\beta)L^db+H_0(\theta(\beta),\phi),\\
&\cU_0V\cU_0^*=V,\quad \cU_0W_+\cU_0^*=W_+,\quad \cU_0W_-\cU_0^*=W_-.
\end{align*}
In the derivation of the first equality we used
 \eqref{eq_self_adjointness} and 
 \eqref{eq_spatial_reflection_symmetry}. By using the unitary transform
 $\cU_0$,
\begin{align*}
&\lim_{h\to \infty\atop h\in \frac{2}{\beta}\N}\overline{\int
 e^{-V(\psi)+W(\psi)}d\mu_{C(\phi)}(\psi)}\\
&=  \frac{1}{\pi}\int_{\R^2}d\xi_1d\xi_2 e^{-|\xi|^2}\frac{\Tr
 e^{-\beta \cU_0(H_0(-\theta(\beta),\phi)+V-i\xi W_+-i\oxi
 W_-)\cU_0^*}}{\Tr e^{-\beta \cU_0H_0(-\theta(\beta),\phi)\cU_0^*}}\\
&= \frac{1}{\pi}\int_{\R^2}d\xi_1d\xi_2 e^{-|\xi|^2}\frac{\Tr
 e^{-\beta (H_0(\theta(\beta),\phi)+V-i\xi W_+-i\oxi
 W_-)}}{\Tr e^{-\beta H_0(\theta(\beta),\phi)}}\\
&=\lim_{h\to \infty\atop h\in \frac{2}{\beta}\N}\int
 e^{-V(\psi)+W(\psi)}d\mu_{C(\phi)}(\psi),
\end{align*}
which implies the claim.
\end{proof}

\section{Multi-scale integration}\label{sec_multiscale_integration}

As one can expect from the formulation Lemma \ref{lem_final_Grassmann_formulation},
the proof of the main theorem is based on 
analytical control of the Grassmann Gaussian integral
\begin{align*}
\int e^{-V(\psi)+W(\psi)-A(\psi)}d\mu_{C(\phi)}(\psi).
\end{align*}
We will achieve our purpose by means of multi-scale integration. In
principle our analysis is an extension of the double-scale integration
performed in the previous work \cite{K_BCS}. We intend to keep using the
previous framework as much as possible so that the readers can
smoothly connect it to this extended version. As in the
previous construction, after brief introductions or restatements of
necessary notations concerning estimation of kernel functions we
establish general bounds on Grassmann polynomials. Then by assuming
scale-dependent bound properties of covariances we inductively construct a
 multi-scale integration process running from the largest
scale to the smallest scale. In the next section we will confirm that
our actual covariance satisfies the properties assumed in this section.

\subsection{Necessary notions}\label{subsec_necessary_notions}

Our multi-scale analysis needs a little more detailed notions of
estimating kernel functions than the
double-scale integration required in \cite{K_BCS}. In order to avoid
unnecessary repetitions, we use some terminology and notational
convention without presenting the definitions in the following. 
The readers should refer to \cite[\mbox{Subsection 3.1}]{K_BCS} for
their meaning. We will not use any terminology or notational rule which
is not defined either in \cite[\mbox{Subsection 3.1}]{K_BCS} or in this
section and the preceding sections of this paper. 
As in the previous paper, we define the norms
$\|f\|_{1,\infty}$, $\|f\|_{1}$ of a function $f:I^n\to\C$ by
\begin{align*}
&\|f\|_{1,\infty}:=\sup_{j\in \{1,2,\cdots,n\}}\sup_{X_0\in
 I}\left(\frac{1}{h}\right)^{n-1}\sum_{\bX\in I^{j-1}}\sum_{\bY\in
 I^{n-j}}|f(\bX,X_0,\bY)|,\\ 
&\|f\|_1:=\left(\frac{1}{h}\right)^n\sum_{\bX\in I^n}|f(\bX)|.
\end{align*}
For $f_0\in \C$ we let $\|f_0\|_{1,\infty}=\|f_0\|_1:=|f_0|$. This
convention helps to organize formulas. We define the index set
$I^0(\subset I)$ by
$$
I^0:=\{1,2\}\times\cB\times\G\times\{0\}\times\{1,-1\}.
$$
Since we will frequently make use of bound properties of covariances, 
we need to introduce various norms on functions on $I^2$. For an
anti-symmetric function $g:I^2\to \C$ we define the norms
$\|g\|_{1,\infty}'$, $\|g\|$ as follows. 
\begin{align*}
&\|g\|_{1,\infty}':=\sup_{X_0\in I\atop s\in [0,\beta)_h}\sum_{X\in
 I^0}|g(X_0,X+s)|,\\
&\|g\|:=\|g\|_{1,\infty}'+(1+\beta^{-1})\|g\|_{1,\infty}.
\end{align*}
We should remark that the definition of the norm $\|\cdot\|$ is slightly
different from that in \cite{K_BCS}. 
We will also need to evaluate a function on $I^m\times I^n$ multiplied
by another anti-symmetric function on $I^2$. More specifically, for a 
function $f_{m,n}:I^m\times I^n\to
\C$ $(m,n\in\N_{\ge 2})$ and an anti-symmetric function $g:I^2\to \C$ 
we set 
\begin{align*}
&[f_{m,n},g]_{1,\infty}\\
&:=\max\Bigg\{\\
&\sup_{X_0\in I\atop j\in \{1,2,\cdots,m\}}\frah^{m-1}\sum_{\bX\in
 I^{m}}1_{X_j=X_0}
\Bigg(\sup_{Y_0\in I\atop k\in \{1,2,\cdots,n\}}\frah^n\sum_{\bY\in
 I^n}|f_{m,n}(\bX,\bY)||g(Y_0,Y_k)|\Bigg),\\
&\sup_{Y_0\in I\atop k\in \{1,2,\cdots,n\}}\frah^{n-1}\sum_{\bY\in
 I^{n}}1_{Y_k=Y_0}
\Bigg(\sup_{X_0\in I\atop j\in \{1,2,\cdots,m\}}\frah^{m}\sum_{\bX\in
 I^m}|f_{m,n}(\bX,\bY)||g(X_0,X_j)|\Bigg)\Bigg\},\\
&[f_{m,n},g]_{1}:=\sup_{j\in\{1,2,\cdots,m\}\atop k\in \{1,2,\cdots,n\}}
\frah^{m+n}\sum_{\bX\in I^m\atop \bY\in
 I^n}|f_{m,n}(\bX,\bY)||g(X_j,Y_k)|.
\end{align*}
Since we do not assume that $f_{m,n}$ is bi-anti-symmetric, the forms of 
$[f_{m,n},g]_{1,\infty}$, $[f_{m,n},g]_{1}$ are more complex than those
introduced in \cite[\mbox{Subsection 3.1}]{K_BCS}. If $f_{m,n}$ is
bi-anti-symmetric, they become same as before. 
\begin{align*}
&[f_{m,n},g]_{1,\infty}\\
&=\max\Bigg\{\sup_{X_0\in I}\frah^{m-1}\sum_{\bX\in
 I^{m-1}}
\Bigg(\sup_{Y_0\in I}\frah^n\sum_{\bY\in
 I^n}|f_{m,n}((X_0,\bX),\bY)||g(Y_0,Y_1)|\Bigg),\\
&\qquad\qquad \sup_{Y_0\in I}\frah^{n-1}\sum_{\bY\in
 I^{n-1}}
\Bigg(\sup_{X_0\in I}\frah^{m}\sum_{\bX\in
 I^m}|f_{m,n}(\bX,(Y_0,\bY))||g(X_0,X_1)|\Bigg)\Bigg\},\\
&[f_{m,n},g]_{1}=\frah^{m+n}\sum_{\bX\in I^m\atop \bY\in
 I^n}|f_{m,n}(\bX,\bY)||g(X_1,Y_1)|.
\end{align*}

Let $\bigwedge_{even}\cV$ denote the subspace of $\bigwedge \cV$
consisting of even polynomials. Each order term of the expansion of logarithm of a Grassmann Gaussian
integral with respect to the effective interaction can be expressed as a
finite sum over trees. Concerning the tree expansion, we can use the
same notations as in \cite[\mbox{Subsection 3.1}]{K_BCS}. The only
difference between the present setting and the previous setting is the definition
of the index sets $I_0$, $I$. By keeping in mind that $I_0$, $I$ count
the band index $\cB$ in this setting we can refer to \cite[\mbox{Subsection
3.1}]{K_BCS} for the meaning of the notations we use in the
following. The tree formula is applied as follows. For any covariance
$\cC:I_0^2\to \C$ and $f^j(\psi)\in\bigwedge_{even}\cV$
$(j=1,2,\cdots,n)$,
\begin{align}
&\frac{1}{n!}\prod_{j=1}^n\left(\frac{\partial}{\partial z_j}\right)\log\left(\int
 e^{\sum_{j=1}^nz_jf^j(\psi+\psi^1)}d\mu_{\cC}(\psi^1)\right)\Bigg|_{z_j=0\atop(\forall
 j\in\{1,2,\cdots,n\})}\label{eq_tree_expansion}\\
&=\frac{1}{n!}Tree(\{1,2,\cdots,n\},\cC)\prod_{j=1}^nf^j(\psi+\psi^j)
\Bigg|_{\psi^j=0\atop(\forall
 j\in\{1,2,\cdots,n\})}.\notag
\end{align}
The major part of our analysis is devoted to estimating Grassmann
polynomials produced by the operator $Tree(\{1,2,\cdots,n\},\cC)$.

\subsection{General estimation}\label{subsec_general_estimation}

Here we summarize bound properties of Grassmann polynomials produced by
the tree formula. Most of the necessary properties have essentially been
prepared in \cite[\mbox{Subsection 3.2}]{K_BCS}. However, as we need to
apply them repeatedly in the next subsection, let us present all the
necessary inequalities so that the readers can follow the
arguments without disruption. 

Here we do not fix details of the covariance. We only assume that the
covariance $\cC:I_0^2\to \C$ satisfies with a constant $D(\in \R_{>0})$
that 
\begin{align}
&\cC(\cR_{\beta}(\bX+s))=\cC(\bX),\quad \left(\forall \bX\in I_0^2,\
 s\in
 \frac{1}{h}\Z\right),\label{eq_general_covariance_time_translation}\\
&|\det(\<\bu_i,\bv_j\>_{\C^m}\cC(X_i,Y_j))_{1\le i,j\le n}|\le
 D^n,\label{eq_general_determinant_bound}\\
&(\forall m,n\in\N,\ \bu_i,\bv_i\in\C^m\text{ with
 }\|\bu_i\|_{\C^m},\|\bv_i\|_{\C^m}\le 1,\ X_i,Y_i\in I_0\
 (i=1,2,\cdots,n)).\notag
\end{align}
Here the map
$\cR_{\beta}:(\{1,2\}\times\cB\times\G\times\frac{1}{h}\Z)^n\to
I_0^n$ is defined by
\begin{align*}
&\cR_{\beta}((\orho_1,\rho_1,\bx_1,s_1),\cdots,(\orho_n,\rho_n,\bx_n,s_n))\\
&:=((\orho_1,\rho_1,\bx_1,r_{\beta}(s_1)),\cdots,(\orho_n,\rho_n,\bx_n,r_{\beta}(s_n))),
\end{align*}
where for any $s\in \frac{1}{h}\Z$, $r_{\beta}(s)\in [0,\beta)_h$ and
$r_{\beta}(s)=s$ in $\frac{1}{h}\Z/\beta\Z$. By abusing the notation we will sometimes consider
$\cR_{\beta}$ as a map from
$(\{1,2\}\times\cB\times\G\times\frac{1}{h}\Z\times\{1,-1\})^n$ to $I^n$ satisfying
the same condition on the time variables. The precise meaning of the map
$\cR_{\beta}$ should be understood from the context.

As in \eqref{eq_general_covariance_time_translation} we will often
impose the condition 
\begin{align}
F(\cR_{\beta}(\bX+s))=F(\bX),\quad \left(\forall \bX\in I^m,\ s\in
 \frac{1}{h}\Z\right)\label{eq_temperature_translation_invariance}
\end{align}
on a function $F:I^m\to\C$. For $j\in \N$ let $F^j(\psi)\in \bigwedge_{even}\cV$ be such that its
anti-symmetric kernels $F_m^j:I^m\to\C$ $(m=2,4,\cdots,N)$ satisfy
\eqref{eq_temperature_translation_invariance}.
The first lemma summarizes bound properties of $A^{(n)}(\psi)$ $(\in
\bigwedge_{even}\cV)$ $(n\in \N)$ defined by 
\begin{align*}
A^{(n)}(\psi):=Tree(\{1,2,\cdots,n\},\cC)\prod_{j=1}^nF^j(\psi^j+\psi)\Bigg|_{\psi^j=0\atop(\forall
 j\in\{1,2,\cdots,n\})}.
\end{align*}
Recall that the anti-symmetric extension $\tilde{\cC}:I^2\to \C$ of the
covariance $\cC$ is defined by 
\begin{align}
&\tilde{\cC}(X\xi,Y\zeta):=\frac{1}{2}(1_{(\xi,\zeta)=(1,-1)}\cC(X,Y)
-1_{(\xi,\zeta)=(-1,1)}\cC(Y,X)),\label{eq_covariance_anti_symmetric_extension}\\
&(X,Y\in I_0,\ \xi,\zeta\in \{1,-1\}).\notag
\end{align}
Let $N$ denote $4b\beta h L^d$, the cardinality of the index set $I$.  

\begin{lemma}\label{lem_simple_tree_expansion}
For any $m\in \{2,4,\cdots,N\}$, $n\in\N$ the anti-symmetric kernel
 $A_m^{(n)}(\cdot)$ satisfies
 \eqref{eq_temperature_translation_invariance}. Moreover, the following
 inequalities hold for any $m\in \{0,2,\cdots,N\}$, $n\in \N_{\ge 2}$.
\begin{align}
&\|A_m^{(1)}\|_{1,\infty}\le 
\sum_{p=m}^{N}\left(\frac{N}{h}\right)^{1_{m=0\land p\neq 0}}\left(\begin{array}{c}p\\ m\end{array}\right)
D^{\frac{p-m}{2}}\|F_p^1\|_{1,\infty}.\label{eq_simple_1_infinity_1}\\
&\|A_m^{(1)}\|_{1}\le  \sum_{p=m}^{N}\left(\begin{array}{c}p\\ m\end{array}\right)
D^{\frac{p-m}{2}}\|F_p^1\|_{1}.\label{eq_simple_1_1}\\
&\|A_m^{(n)}\|_{1,\infty}\le \left(\frac{N}{h}\right)^{1_{m=0}}
(n-2)!D^{-n+1-\frac{m}{2}}2^{-2m}\|\tilde{\cC}\|_{1,\infty}^{n-1}\label{eq_simple_1_infinity}\\
&\qquad\qquad\quad\cdot\prod_{j=1}^n\left(\sum_{p_j=2}^N2^{3p_j}D^{\frac{p_j}{2}}\|F_{p_j}^j\|_{1,\infty}
\right)1_{\sum_{j=1}^np_j-2(n-1)\ge m}.\notag\\
&\|A_m^{(n)}\|_{1}\le (n-2)!
 D^{-n+1-\frac{m}{2}}2^{-2m}\|\tilde{\cC}\|_{1,\infty}^{n-1}
\sum_{p_1=2}^N2^{3p_1}D^{\frac{p_1}{2}}\|F_{p_1}^1\|_{1}\label{eq_simple_1}\\
&\qquad\qquad\cdot\prod_{j=2}^n\left(\sum_{p_j=2}^N2^{3p_j}D^{\frac{p_j}{2}}\|F_{p_j}^j\|_{1,\infty}
\right)1_{\sum_{j=1}^np_j-2(n-1)\ge m}.\notag
\end{align}
\end{lemma}

\begin{proof} 
These are essentially same as \cite[\mbox{Lemma 3.1}]{K_BCS}.
\end{proof}

Next we deal with a Grassmann input with bi-anti-symmetric kernels. Let functions $F_{p,q}:I^p\times I^q\to \C$ $(p,q\in
\{2,4,\cdots,N\})$ be bi-anti-symmetric, satisfy
\eqref{eq_temperature_translation_invariance} and the following
property. For any functions $f:[0,\beta)_h^p\to\C$, $g:[0,\beta)_h^q\to\C$ satisfying \begin{align*}
&f(r_{\beta}(s_1+s),r_{\beta}(s_2+s),\cdots,r_{\beta}(s_p+s))=f(s_1,s_2,\cdots,s_p)\\
&\left(\forall (s_1,s_2,\cdots,s_p)\in [0,\beta)_h^p,\ s\in
 \frac{1}{h}\Z\right),\\
&g(r_{\beta}(s_1+s),r_{\beta}(s_2+s),\cdots,r_{\beta}(s_q+s))=g(s_1,s_2,\cdots,s_q)\notag\\&\left(\forall (s_1,s_2,\cdots,s_q)\in [0,\beta)_h^q,\ s\in
 \frac{1}{h}\Z\right),
\end{align*}
\begin{align}
&\sum_{(s_1,\cdots,s_p)\in
 [0,\beta)_h^p}F_{p,q}((\orho_1\rho_1\bx_1s_1\xi_1,\cdots,\orho_p\rho_p\bx_ps_p\xi_p),\bY)f(s_1,\cdots,s_p)=0,\label{eq_temperature_integral_vanish}\\
&(\forall \bY\in I^q,\ (\orho_j,\rho_j,\bx_j,\xi_j)\in
 \{1,2\}\times\cB\times \G\times\{1,-1\}\ (j=1,2,\cdots,p)),\notag\\
&\sum_{(t_1,\cdots,t_q)\in
 [0,\beta)_h^q}F_{p,q}(\bX,
(\oeta_1\eta_1\by_1t_1\zeta_1,\cdots,\oeta_q\eta_q\by_qt_q\zeta_q))g(t_1,\cdots,t_q)=0,\notag\\
&(\forall \bX\in I^p,\ (\oeta_j,\eta_j,\by_j,\zeta_j)\in
 \{1,2\}\times\cB\times \G\times\{1,-1\}\ (j=1,2,\cdots,q)).\notag
\end{align}
Then let us define the Grassmann polynomials $B^{(n)}(\psi)$ $(\in
\bigwedge_{even}\cV)$ $(n\in \N)$ by 
\begin{align*}
B^{(n)}(\psi)
:=&\sum_{p,q=2}^{N}1_{p,q\in 2\N}\frah^{p+q}\sum_{\bX\in I^p\atop\bY\in
 I^q}F_{p,q}(\bX,\bY)Tree(\{1,2,\cdots,n+1\},\cC)\\
&\cdot(\psi^1+\psi)_{\bX}(\psi^2+\psi)_{\bY}\prod_{j=3}^{n+1}F^j(\psi^j+\psi)\Bigg|_{\psi^j=0\atop(\forall
 j\in\{1,2,\cdots,n+1\})}.
\end{align*}
The kernels of the Grassmann polynomials $B^{(n)}(\psi)$ $(n\in \N)$ are
estimated as follows.

\begin{lemma}\label{lem_double_tree_expansion}
For any $m\in \{2,4,\cdots,N\}$, $n\in\N$ the anti-symmetric
 kernel $B_m^{(n)}(\cdot)$ satisfies
 \eqref{eq_temperature_translation_invariance}. Moreover, for any
 $m\in\{0,2,\cdots,N\}$, $n\in \N_{\ge 2}$,
\begin{align}
&\|B_m^{(1)}\|_{1,\infty}
\le\left(\frac{N}{h}\right)^{1_{m=0}}
 D^{-1-\frac{m}{2}}\sum_{p_1,p_2=2}^{N}1_{p_1,p_2\in
 2\N}2^{2p_1+2p_2}D^{\frac{p_1+p_2}{2}}[F_{p_1,p_2},\tilde{\cC}]_{1,\infty}1_{p_1+p_2-2\ge
 m}.\label{eq_double_1_infinity_1}\\
&\|B_m^{(1)}\|_1\le
 D^{-1-\frac{m}{2}}\sum_{p_1,p_2=2}^{N}1_{p_1,p_2\in
 2\N}2^{2p_1+2p_2}D^{\frac{p_1+p_2}{2}}[F_{p_1,p_2},\tilde{\cC}]_11_{p_1+p_2-2\ge
 m}.\label{eq_double_1_1}\\
&\|B_m^{(n)}\|_{1,\infty}\le \left(\frac{N}{h}\right)^{1_{m=0}}
(n-1)!D^{-n-\frac{m}{2}}2^{-2m}\|\tilde{\cC}\|_{1,\infty}^{n-1}\label{eq_double_1_infinity}\\
&\qquad\qquad\quad\cdot\sum_{p_1,p_2=2}^{N}1_{p_1,p_2\in2\N}2^{3p_1+3p_2}D^{\frac{p_1+p_2}{2}}[F_{p_1,p_2},\tilde{\cC}]_{1,\infty}
\prod_{j=3}^{n+1}\left(\sum_{p_j=2}^N2^{3p_j}D^{\frac{p_j}{2}}\|F_{p_j}^j\|_{1,\infty}
\right)\notag\\
&\qquad\qquad\quad\cdot 1_{\sum_{j=1}^{n+1}p_j-2n\ge m}.\notag\\
&\|B_m^{(n)}\|_{1}\le 
(n-1)!D^{-n-\frac{m}{2}}2^{-2m}\|\tilde{\cC}\|_{1,\infty}^{n-1}\label{eq_double_1}\\
&\qquad\qquad\quad\cdot\sum_{p_1,p_2=2}^{N}1_{p_1,p_2\in2\N}2^{3p_1+3p_2}D^{\frac{p_1+p_2}{2}}[F_{p_1,p_2},\tilde{\cC}]_{1,\infty}
\prod_{j=3}^{n}\left(\sum_{p_j=2}^N2^{3p_j}D^{\frac{p_j}{2}}\|F_{p_j}^j\|_{1,\infty}
\right)\notag\\
&\qquad\qquad\quad\cdot
\sum_{p_{n+1}=2}^N2^{3p_{n+1}}D^{\frac{p_{n+1}}{2}}\|F_{p_{n+1}}^{n+1}\|_1
1_{\sum_{j=1}^{n+1}p_j-2n\ge m}.\notag
\end{align}
\end{lemma}

\begin{proof} 
These are essentially proved in \cite[\mbox{Lemma 3.2}]{K_BCS}.
\end{proof}
  
\begin{remark}\label{rem_additional_remark_cancellation}
In fact the property \eqref{eq_temperature_integral_vanish} of
 $F_{p,q}(\cdot,\cdot)$ is not used to prove Lemma
 \ref{lem_double_tree_expansion}. It is only necessary to characterize
 kernels of the Grassmann polynomial denoted by $E^{(n)}(\psi)$ in Lemma
 \ref{lem_divided_tree_expansion}. It is not directly used to derive
 estimates of the kernels in Lemma \ref{lem_divided_tree_expansion},
 either. However, we assume the property
 \eqref{eq_temperature_integral_vanish} of $F_{p,q}(\cdot,\cdot)$
 throughout this subsection in order not to complicate the assumptions
 by unnecessary generalization of the lemmas. It is important to guarantee that
 some output polynomials inherit the property
 \eqref{eq_temperature_integral_vanish} from the input polynomial 
 as claimed in Lemma \ref{lem_divided_tree_expansion}, 
 since it enables us to classify Grassmann polynomials with or without
 the property \eqref{eq_temperature_integral_vanish} during the
 multi-scale integration process. It turns out that those with the property
 \eqref{eq_temperature_integral_vanish} vanish at the final
 integration, which is one essential reason why we can construct
this many-electron system by keeping the small
 coupling constant largely independent of the temperature and the
 imaginary magnetic field. 
\end{remark}
 
In the next lemma we claim several inequalities which were not proved in
\cite[\mbox{Subsection 3.2}]{K_BCS}. Using the same input polynomials as
in the above lemmas, we define $E^{(n)}(\psi)$ $(\in
\bigwedge_{even}\cV)$ $(n\in\N)$ as follows.
\begin{align*}
E^{(n)}(\psi)
:=&\sum_{p,q=2}^{N}1_{p,q\in 2\N}\frah^{p+q}\sum_{\bX\in I^p\atop
 \bY\in I^q}F_{p,q}(\bX,\bY)\\
&\cdot
 Tree(\{s_j\}_{j=1}^{m+1},\cC)(\psi^1+\psi)_{\bX}\prod_{j=2}^{m+1}F^{s_j}(\psi^{s_j}+\psi)\Bigg|_{\psi^{s_j}=0\atop(\forall
 j\in\{1,2,\cdots,m+1\})}\\
&\cdot
 Tree(\{t_k\}_{k=1}^{n-m},\cC)(\psi^1+\psi)_{\bY}\prod_{k=2}^{n-m}F^{t_k}(\psi^{t_k}+\psi)\Bigg|_{\psi^{t_k}=0\atop(\forall
 k\in\{1,2,\cdots,n-m\})},
\end{align*}
where 
\begin{align*}
&m\in \{0,1,\cdots,n-1\},\\
&1=s_1<s_2<\cdots <s_{m+1}\le n,\quad 1=t_1<t_2<\cdots <t_{n-m}\le n,\\
&\{s_j\}_{j=2}^{m+1}\cup \{t_k\}_{k=2}^{n-m}=\{2,3,\cdots,n\},\quad
 \{s_j\}_{j=2}^{m+1}\cap \{t_k\}_{k=2}^{n-m}=\emptyset.
\end{align*}
Here, unlike in \cite[\mbox{Lemma 3.3}]{K_BCS}, we present the
bi-anti-symmetric kernels of $E^{(n)}(\psi)$ beforehand. Let us define
functions $E_{a,b}^{(n)}:I^a\times I^b\to\C$ $(n\in\N,\ a,b\in
\{2,4,\cdots,N\})$ by
\begin{align}
E_{a,b}^{(n)}(\bX,\bY)
&:=\sum_{p_1,q_1=2}^{N}1_{p_1,q_1\in 2\N}\sum_{u_1=0}^{p_1}
(1_{m=0}+1_{m\neq 0}1_{u_1\le
 p_1-1})\left(\begin{array}{c} p_1 \\ u_1\end{array}\right)
\label{eq_divided_kernel}\\
&\quad\cdot \sum_{v_1=0}^{q_1}
(1_{m=n-1}+1_{m\neq n-1}1_{v_1\le
 q_1-1})\left(\begin{array}{c} q_1 \\ v_1\end{array}\right)\notag\\
&\quad\cdot \prod_{j=2}^{m+1}\left(\sum_{p_j=2}^N\sum_{u_j=0}^{p_j-1}
\left(\begin{array}{c} p_j \\ u_j\end{array}\right)\right)
\prod_{k=2}^{n-m}\left(\sum_{q_k=2}^N\sum_{v_k=0}^{q_k-1}
\left(\begin{array}{c} q_k \\ v_k\end{array}\right)\right)\notag\\
&\quad\cdot f_m^n((p_j)_{1\le j\le m+1},(u_j)_{1\le j\le
 m+1},(q_j)_{1\le j\le n-m},(v_j)_{1\le j\le
 n-m})\notag\\
&\quad\qquad ((\bX_1',\bX_2',\cdots,\bX_{m+1}'),
(\bY_1',\bY_2',\cdots,\bY_{n-m}'))\notag\\
&\quad\cdot 1_{\sum_{j=1}^{m+1}u_j=a}1_{\sum_{k=1}^{n-m}v_k=b}
1_{\sum_{j=1}^{m+1}p_j-2m\ge a}1_{\sum_{k=1}^{n-m}q_k-2(n-m-1)\ge b}\notag\\
&\quad\cdot \frac{1}{a!b!}\sum_{\s\in\S_a\atop \tau\in
 \S_b}\sgn(\s)\sgn(\tau)
1_{(\bX_1',\bX_2',\cdots,\bX_{m+1}')=\bX_{\s}}
1_{(\bY_1',\bY_2',\cdots,\bY_{n-m}')=\bY_{\tau}},\notag
\end{align}
where the function 
\begin{align*}
f_m^n((p_j)_{1\le j\le m+1},(u_j)_{1\le j\le
 m+1},(q_j)_{1\le j\le n-m},(v_j)_{1\le j\le
 n-m}):\prod_{j=1}^{m+1}I^{u_j}\times \prod_{k=1}^{n-m}I^{v_k}\to\C
\end{align*}
is defined by
\begin{align*}
&f_m^n((p_j)_{1\le j\le m+1},(u_j)_{1\le j\le
 m+1},(q_j)_{1\le j\le n-m},(v_j)_{1\le j\le
 n-m})\\
&\quad ((\bX_1,\bX_2,\cdots,\bX_{m+1}),
(\bY_1,\bY_2,\cdots,\bY_{n-m}))\notag\\
&:=\frah^{p_1+q_1-u_1-v_1}\sum_{\bW_1\in I^{p_1-u_1}}\sum_{\bZ_1\in
 I^{q_1-v_1}}F_{p_1,q_1}((\bW_1,\bX_1),(\bZ_1,\bY_1))\notag\\
&\quad\cdot \prod_{j=2}^{m+1}\left(
\frah^{p_j-u_j}\sum_{\bW_j\in
 I^{p_j-u_j}}F_{p_j}^{s_j}(\bW_j,\bX_j)\right)\notag\\
&\quad\cdot \prod_{k=2}^{n-m}
\left(
\frah^{q_k-v_k}\sum_{\bZ_k\in
 I^{q_k-v_k}}F_{q_k}^{t_k}(\bZ_k,\bY_k)\right)\notag\\
&\quad\cdot
 Tree(\{s_j\}_{j=1}^{m+1},\cC)\prod_{j=1}^{m+1}\psi_{\bW_j}^{s_j}\Bigg|_{\psi^{s_j}=0\atop(\forall
 j\in\{1,2,\cdots,m+1\})}\notag\\
&\quad\cdot Tree(\{t_k\}_{k=1}^{n-m},\cC)\prod_{k=1}^{n-m}\psi_{\bZ_k}^{t_k}\Bigg|_{\psi^{t_k}=0\atop(\forall
 k\in\{1,2,\cdots,n-m\})}\notag\\
&\quad\cdot
 (-1)^{\sum_{j=1}^mu_j\sum_{i=j+1}^{m+1}(p_i-u_i)+\sum_{k=1}^{n-m-1}v_k\sum_{i=k+1}^{n-m}(q_i-v_i)}.\notag
\end{align*}

\begin{lemma}\label{lem_divided_tree_expansion}
For any $n\in \N$, $a,b\in \{2,4,\cdots,N\}$, $E_{a,b}^{(n)}$ is
 bi-anti-symmetric, satisfies
 \eqref{eq_temperature_translation_invariance},
 \eqref{eq_temperature_integral_vanish} and 
\begin{align*}
E^{(n)}(\psi)=\sum_{a,b=2}^{N}1_{a,b\in 2\N}\frah^{a+b}\sum_{\bX\in
 I^a\atop \bY\in I^b}E_{a,b}^{(n)}(\bX,\bY)\psi_{\bX}\psi_{\bY}.
\end{align*}
Moreover, the following inequalities hold for any $a,b\in
 \{2,4,\cdots,N\}$, $n\in\N_{\ge 2}$ and anti-symmetric function
 $g:I^2\to \C$.
\begin{align}
&\|E_{a,b}^{(1)}\|_{1,\infty}\le
 \sum_{p=a}^{N}\sum_{q=b}^{N}1_{p,q\in 2\N}
\left(\begin{array}{c} p \\ a \end{array}\right)
\left(\begin{array}{c} q \\ b \end{array}\right)
D^{\frac{1}{2}(p+q-a-b)}\|F_{p,q}\|_{1,\infty}.\label{eq_divided_1_infinity_1}\\&\|E_{a,b}^{(1)}\|_{1}\le
 \sum_{p=a}^{N}\sum_{q=b}^{N}1_{p,q\in 2\N}
\left(\begin{array}{c} p \\ a \end{array}\right)
\left(\begin{array}{c} q \\ b \end{array}\right)
D^{\frac{1}{2}(p+q-a-b)}\|F_{p,q}\|_{1}.\label{eq_divided_1_1}\\
&[E_{a,b}^{(1)},g]_{1,\infty}\le
 \sum_{p=a}^{N}\sum_{q=b}^{N}1_{p,q\in 2\N}
\left(\begin{array}{c} p \\ a \end{array}\right)
\left(\begin{array}{c} q \\ b \end{array}\right)
D^{\frac{1}{2}(p+q-a-b)}[F_{p,q},g]_{1,\infty}.\label{eq_divided_multiplied_1_infinity_1}\\
&[E_{a,b}^{(1)},g]_{1}\le
 \sum_{p=a}^{N}\sum_{q=b}^{N}1_{p,q\in 2\N}
\left(\begin{array}{c} p \\ a \end{array}\right)
\left(\begin{array}{c} q \\ b \end{array}\right)
D^{\frac{1}{2}(p+q-a-b)}[F_{p,q},g]_{1}.\label{eq_divided_multiplied_1_1}\\
&\|E_{a,b}^{(n)}\|_{1,\infty}\label{eq_divided_1_infinity}\\
&\le (1_{m\neq 0}(m-1)!+1_{m=0})(1_{m\neq
 n-1}(n-m-2)!+1_{m=n-1})\notag\\
&\quad\cdot 2^{-2a-2b}D^{-n+1-\frac{1}{2}(a+b)}\|\tilde{\cC}\|_{1,\infty}^{n-1}\sum_{p_1,q_1=2}^{N}1_{p_1,q_1\in
 2\N}2^{3p_1+3q_1}D^{\frac{p_1+q_1}{2}}\|F_{p_1,q_1}\|_{1,\infty}\notag\\
&\quad\cdot\prod_{j=2}^{m+1}\Bigg(
\sum_{p_j=2}^N2^{3p_j}D^{\frac{p_j}{2}}\|F_{p_j}^{s_j}\|_{1,\infty}\Bigg)
\prod_{k=2}^{n-m}\Bigg(
\sum_{q_k=2}^N2^{3q_k}D^{\frac{q_k}{2}}
\|F_{q_k}^{t_k}\|_{1,\infty}\Bigg)\notag\\
&\quad\cdot 1_{\sum_{j=1}^{m+1}p_j-2m\ge a}
1_{\sum_{k=1}^{n-m}q_k-2(n-m-1)\ge b}.\notag\\
&\|E_{a,b}^{(n)}\|_{1}\label{eq_divided_1}\\
&\le (1_{m\neq 0}(m-1)!+1_{m=0})(1_{m\neq
 n-1}(n-m-2)!+1_{m=n-1})\notag\\
&\quad\cdot2^{-2a-2b}D^{-n+1-\frac{1}{2}(a+b)}\|\tilde{\cC}\|_{1,\infty}^{n-1} \sum_{p_1,q_1=2}^{N}1_{p_1,q_1\in
 2\N}2^{3p_1+3q_1}D^{\frac{p_1+q_1}{2}}\|F_{p_1,q_1}\|_{1,\infty}\notag\\
&\quad\cdot\prod_{j=2}^{m+1}\Bigg(
\sum_{p_j=2}^N2^{3p_j}D^{\frac{p_j}{2}}
(1_{s_j\neq
 n}\|F_{p_j}^{s_j}\|_{1,\infty}+1_{s_j=n}\|F_{p_j}^{s_j}\|_1)\Bigg)\notag\\
&\quad\cdot\prod_{k=2}^{n-m}\Bigg(
\sum_{q_k=2}^N2^{3q_k}D^{\frac{q_k}{2}}
(1_{t_k\neq
 n}\|F_{q_k}^{t_k}\|_{1,\infty}+1_{t_k=n}\|F_{q_k}^{t_k}\|_1)\Bigg)\notag\\
&\quad\cdot 1_{\sum_{j=1}^{m+1}p_j-2m\ge a}
1_{\sum_{k=1}^{n-m}q_k-2(n-m-1)\ge b}.\notag\\
&[E_{a,b}^{(n)},g]_{1,\infty}\label{eq_divided_multiplied_1_infinity}\\
&\le (1_{m\neq 0}(m-1)!+1_{m=0})(1_{m\neq
 n-1}(n-m-2)!+1_{m=n-1})\notag\\
&\quad\cdot
 2^{-2a-2b}D^{-n+1-\frac{1}{2}(a+b)}\|\tilde{\cC}\|_{1,\infty}^{n-2}\notag\\
&\quad\cdot \sum_{p_1,q_1=2}^{N}1_{p_1,q_1\in
 2\N}2^{3p_1+3q_1}D^{\frac{p_1+q_1}{2}}([F_{p_1,q_1},g]_{1,\infty}\|\tilde{\cC}\|_{1,\infty}+[F_{p_1,q_1},\tilde{\cC}]_{1,\infty}\|g\|_{1,\infty})\notag\\
&\quad\cdot\prod_{j=2}^{m+1}\Bigg(
\sum_{p_j=2}^N2^{3p_j}D^{\frac{p_j}{2}}\|F_{p_j}^{s_j}\|_{1,\infty}\Bigg)
\prod_{k=2}^{n-m}\Bigg(
\sum_{q_k=2}^N2^{3q_k}D^{\frac{q_k}{2}}\|F_{q_k}^{t_k}\|_{1,\infty}\Bigg)\notag\\
&\quad\cdot 1_{\sum_{j=1}^{m+1}p_j-2m\ge a}
1_{\sum_{k=1}^{n-m}q_k-2(n-m-1)\ge b}.\notag\\
&[E_{a,b}^{(n)},g]_{1}\label{eq_divided_multiplied_1}\\
&\le (1_{m\neq 0}(m-1)!+1_{m=0})(1_{m\neq
 n-1}(n-m-2)!+1_{m=n-1})\notag\\
&\quad\cdot
 2^{-2a-2b}D^{-n+1-\frac{1}{2}(a+b)}\|\tilde{\cC}\|_{1,\infty}^{n-2}\notag\\
&\quad\cdot \sum_{p_1,q_1=2}^{N}1_{p_1,q_1\in
 2\N}2^{3p_1+3q_1}D^{\frac{p_1+q_1}{2}}([F_{p_1,q_1},g]_{1,\infty}\|\tilde{\cC}\|_{1,\infty}+[F_{p_1,q_1},\tilde{\cC}]_{1,\infty}\|g\|_{1,\infty})\notag\\
&\quad\cdot\prod_{j=2}^{m+1}\Bigg(
\sum_{p_j=2}^N2^{3p_j}D^{\frac{p_j}{2}}
(1_{s_j\neq
 n}\|F_{p_j}^{s_j}\|_{1,\infty}+1_{s_j=n}\|F_{p_j}^{s_j}\|_1)\Bigg)\notag\\
&\quad\cdot\prod_{k=2}^{n-m}\Bigg(
\sum_{q_k=2}^N2^{3q_k}D^{\frac{q_k}{2}}
(1_{t_k\neq
 n}\|F_{q_k}^{t_k}\|_{1,\infty}+1_{t_k=n}\|F_{q_k}^{t_k}\|_1)\Bigg)\notag\\
&\quad\cdot 1_{\sum_{j=1}^{m+1}p_j-2m\ge a}
1_{\sum_{k=1}^{n-m}q_k-2(n-m-1)\ge b}.\notag
\end{align}
\end{lemma}

\begin{proof} 
The kernels $E_{a,b}^{(n)}$ $(a,b\in\{2,4,\cdots,N\})$ were
 essentially given in \\
\cite[\mbox{(3.41),(3.39)}]{K_BCS}. The claimed
 properties of $E_{a,b}^{(n)}$ and the inequalities
 \eqref{eq_divided_1_infinity_1},
 \eqref{eq_divided_1_1},
 \eqref{eq_divided_1_infinity},
 \eqref{eq_divided_1}
were essentially proved in \cite[\mbox{Lemma 3.3}]{K_BCS}.
We need to show \eqref{eq_divided_multiplied_1_infinity_1},
 \eqref{eq_divided_multiplied_1_1},
 \eqref{eq_divided_multiplied_1_infinity} and
 \eqref{eq_divided_multiplied_1}.
In fact these inequalities can be proved in similar ways to the
 derivations of \eqref{eq_divided_1_infinity_1},
 \eqref{eq_divided_1_1},
 \eqref{eq_divided_1_infinity},
 \eqref{eq_divided_1}. However, we provide the major part of the proof
 for completeness. Let $\bp:=(p_j)_{1\le j\le m+1}$,  $\bu:=(u_j)_{1\le
 j\le m+1}$, $\bq:=(q_j)_{1\le j\le n-m}$, 
  $\bv:=(v_j)_{1\le j\le n-m}$ for simplicity in the following. Since
 the norm bounds on $E_{a,b}^{(n)}$ follow from norm bounds on the
 function $f_m^n(\bp,\bu,\bq,\bv)$, let us focus on estimating
 $f_m^n(\bp,\bu,\bq,\bv)$. 

First let us consider the case $n=1$. By using the determinant bound
 \eqref{eq_general_determinant_bound} we have for any $\bX_1\in
 I^{u_1}$, $\bY_1\in I^{v_1}$ that
\begin{align*}
&|f_m^n(\bp,\bu,\bq,\bv)(\bX_1,\bY_1)|\\
&\le 
\frah^{p_1+q_1-u_1-v_1}\sum_{\bW_1\in I^{p_1-u_1}\atop \bZ_1\in
 I^{q_1-v_1}}|F_{p_1,q_1}((\bW_1,\bX_1),(\bZ_1,\bY_1))|D^{\frac{1}{2}(p_1+q_1-u_1-v_1)},\end{align*}
which implies that 
\begin{align}
[f_m^n(\bp,\bu,\bq,\bv),g]_{index}\le
 D^{\frac{1}{2}(p_1+q_1-u_1-v_1)}[F_{p_1,q_1},g]_{index},\label{eq_divided_multiplied_1_1_pre}
\end{align}
for $index =`1,\infty'$ or $index=1$ and any anti-symmetric function
 $g:I^2\to \C$. 

Let us consider the case $n\ge 2$. Let us take an anti-symmetric
 function $g:I^2\to \C$ and estimate $[f_m^n(\bp,\bu,\bq,\bv),g]_1$.
 If $n\in \{s_j\}_{j=1}^{m+1}$, we estimate the right-hand side of the
 following inequality.
\begin{align}
&[f_m^n(\bp,\bu,\bq,\bv),g]_1\label{eq_tactic_estimation_left}\\
&\le\sup_{k_1\in \{1,2,\cdots,\sum_{k=1}^{n-m}v_k\}}
\Bigg(\frah^{\sum_{j=1}^{m+1}u_j+\sum_{k=1}^{n-m}v_k}\sum_{\bX\in
 \prod_{j=1}^{m+1}I^{u_j}}\notag\\
&\qquad\qquad\qquad\qquad\qquad\cdot \sup_{Y_0\in I}\sum_{\bY\in \prod_{k=1}^{n-m}I^{v_k}}|f_m^n(\bp,\bu,\bq,\bv)(\bX,\bY)||g(Y_0,Y_{k_1})|\Bigg).\notag
\end{align}
For $k_1\in \{1,2,\cdots,\sum_{k=1}^{n-m}v_k\}$ there uniquely exists
 $k_0\in \{1,2,\cdots,n-m\}$ such that $Y_{k_1}$ is a component of the
 variable of the function $F_{p_1,q_1}$ if $k_0=1$ or
 $F_{q_{k_0}}^{t_{k_0}}$ if $k_0\neq 1$. We consider the vertex
 $t_{k_0}$ as the root of the tree $T\in \T(\{t_k\}_{k=1}^{n-m})$ and
 recursively estimate from the younger branches to the root $t_{k_0}$
 along the lines of $T$. In this procedure we obtain especially 
\begin{align*}
 \sup_{Y_0\in
 I}\left(\frac{1}{h}\right)^{q_1}\sum_{\bY\in I^{q_1}}
|F_{p_1,q_1}(\bX,\bY)||g(Y_0,Y_1)|
\end{align*}
if $k_0=1$,
\begin{align*}
 \sup_{Y_0\in
 I}\left(\frac{1}{h}\right)^{q_1}\sum_{\bY\in I^{q_1}}|F_{p_1,q_1}(\bX,\bY)||\tilde{\cC}(Y_0,Y_1)|
\end{align*}
if $k_0\neq 1$. Then we consider the vertex $n$ as the root of $S\in
 \T(\{s_j\}_{j=1}^{m+1})$ and recursively estimate from the younger
 branches to the root $n$ along the lines of $S$. In the end we
 obtain especially $[F_{p_1,q_1},g]_{1,\infty}$ if $k_0=1$, 
$[F_{p_1,q_1},\tilde{\cC}]_{1,\infty}$ if $k_0\neq 1$. In the case $n\in
 \{t_k\}_{k=1}^{n-m}$ we follow the other way round. We estimate the
 right-hand side of the following inequality.
\begin{align}
&[f_m^n(\bp,\bu,\bq,\bv),g]_1\label{eq_tactic_estimation_right}\\
&\le\sup_{j_1\in \{1,2,\cdots,\sum_{j=1}^{m+1}u_j\}}
\Bigg(\frah^{\sum_{j=1}^{m+1}u_j+\sum_{k=1}^{n-m}v_k}\sum_{\bY\in
 \prod_{k=1}^{n-m}I^{v_k}}\notag\\
&\qquad\qquad\qquad\qquad\qquad\cdot \sup_{X_0\in I}\sum_{\bX\in \prod_{j=1}^{m+1}I^{u_j}}|f_m^n(\bp,\bu,\bq,\bv)(\bX,\bY)||g(X_0,X_{j_1})|\Bigg).\notag
\end{align}
For $j_1\in \{1,2,\cdots,\sum_{j=1}^{m+1}u_j\}$ there uniquely exists
 $j_0\in \{1,2,\cdots,m+1\}$ such that $X_{j_1}$ is a component of the
 variable of the function $F_{p_1,q_1}$ if $j_0=1$ or
 $F_{p_{j_0}}^{s_{j_0}}$ if $j_0\neq 1$. We consider the vertex
 $s_{j_0}$ as the root of the tree $S\in \T(\{s_j\}_{j=1}^{m+1})$ and
 recursively estimate along the lines of $S$. Then we consider the
 vertex $n$ as the root of the tree $T\in \T(\{t_k\}_{k=1}^{n-m})$ and
 recursively estimate along the lines of $T$.
By applying the determinant bound \eqref{eq_general_determinant_bound}
 we have that for any $\bX_j\in I^{u_j}$ $(j=1,2,\cdots,m+1)$, $\bY_k\in
 I^{v_k}$ $(k=1,2,\cdots,n-m)$, 
\begin{align}
&|f_m^n(\bp,\bu,\bq,\bv)((\bX_1,\bX_2,\cdots,\bX_{m+1}),(\bY_1,\bY_2,\cdots,\bY_{n-m}))|
\label{eq_divided_core_kernel_inequality}\\
&\le
 2^{n-1}\sum_{S\in\T(\{s_j\}_{j=1}^{m+1})}\sum_{T\in\T(\{t_k\}_{k=1}^{n-m})}\frah^{p_1+q_1-u_1-v_1}\sum_{\bW_1\in I^{p_1-u_1-d_1(S)}\atop \bW_1'\in I^{d_1(S)}}
\sum_{\bZ_1\in I^{q_1-v_1-d_1(T)}\atop \bZ_1'\in I^{d_1(T)}}\notag\\
&\quad\cdot
\left(\begin{array}{c} p_1-u_1 \\ d_1(S)\end{array}\right)
\left(\begin{array}{c} q_1-v_1 \\ d_1(T)\end{array}\right)
|F_{p_1,q_1}((\bW_1,\bW_1',\bX_1),(\bZ_1,\bZ_1',\bY_1))|\notag\\
&\quad\cdot 
\prod_{j=2}^{m+1}\Bigg(
\frah^{p_j-u_j}\sum_{\bW_j\in I^{p_j-u_j-d_{s_j}(S)}\atop \bW_j'\in
 I^{d_{s_j}(S)}}
\left(\begin{array}{c} p_j-u_j \\ d_{s_j}(S)\end{array}\right)
|F_{p_j}^{s_j}(\bW_j,\bW_j',\bX_j)|\Bigg)\notag\\
&\quad\cdot 
\prod_{k=2}^{n-m}\Bigg(
\frah^{q_k-v_k}\sum_{\bZ_k\in I^{q_k-v_k-d_{t_k}(T)}\atop \bZ_k'\in
 I^{d_{t_k}(T)}}
\left(\begin{array}{c} q_k-v_k \\ d_{t_k}(T)\end{array}\right)
|F_{q_k}^{t_k}(\bZ_k,\bZ_k',\bY_k)|\Bigg)\notag\\
&\quad\cdot D^{\frac{1}{2}(\sum_{j=1}^{m+1}p_j-2m-\sum_{j=1}^{m+1}u_j)
+\frac{1}{2}(\sum_{k=1}^{n-m}q_k-2(n-m-1)-\sum_{k=1}^{n-m}v_k)
}\notag\\
&\quad\cdot 
\left(1_{m=0}+1_{m\neq 0}\left|\prod_{\{p,q\}\in
 S}\D_{\{p,q\}}(\cC)\prod_{j=1}^{m+1}\psi_{\bW_j'}^{s_j}
\right|\right)\notag\\
&\quad\cdot\left(1_{m=n-1}+1_{m\neq n-1}\left|\prod_{\{p,q\}\in
 T}\D_{\{p,q\}}(\cC)\prod_{k=1}^{n-m}\psi_{\bZ_k'}^{t_k}
\right|\right)\notag\\
&=2^{n-1}
D^{-n+1-\frac{1}{2}(\sum_{j=1}^{m+1}u_j+\sum_{k=1}^{n-m}v_k)}
\sum_{S\in\T(\{s_j\}_{j=1}^{m+1})}\sum_{T\in\T(\{t_k\}_{k=1}^{n-m})}\notag\\
&\quad\cdot
\frah^{p_1+q_1-u_1-v_1}\sum_{\bW_1\in I^{p_1-u_1-d_1(S)}\atop \bW_1'\in I^{d_1(S)}}
\sum_{\bZ_1\in I^{q_1-v_1-d_1(T)}\atop \bZ_1'\in I^{d_1(T)}}\notag\\
&\quad\cdot
\left(\begin{array}{c} p_1-u_1 \\ d_1(S)\end{array}\right)
\left(\begin{array}{c} q_1-v_1 \\ d_1(T)\end{array}\right)D^{\frac{1}{2}(p_1+q_1)}
|F_{p_1,q_1}((\bW_1,\bW_1',\bX_1),(\bZ_1,\bZ_1',\bY_1))|\notag\\
&\quad\cdot 
\prod_{j=2}^{m+1}\Bigg(
\frah^{p_j-u_j}\sum_{\bW_j\in I^{p_j-u_j-d_{s_j}(S)}\atop \bW_j'\in
 I^{d_{s_j}(S)}}
\left(\begin{array}{c} p_j-u_j \\ d_{s_j}(S)\end{array}\right)D^{\frac{p_j}{2}}
|F_{p_j}^{s_j}(\bW_j,\bW_j',\bX_j)|\Bigg)\notag\\
&\quad\cdot 
\prod_{k=2}^{n-m}\Bigg(
\frah^{q_k-v_k}\sum_{\bZ_k\in I^{q_k-v_k-d_{t_k}(T)}\atop \bZ_k'\in
 I^{d_{t_k}(T)}}
\left(\begin{array}{c} q_k-v_k \\ d_{t_k}(T)\end{array}\right)D^{\frac{q_k}{2}}
|F_{q_k}^{t_k}(\bZ_k,\bZ_k',\bY_k)|\Bigg)\notag\\
&\quad\cdot 
\left(1_{m=0}+1_{m\neq 0}\left|\prod_{\{p,q\}\in
 S}\D_{\{p,q\}}(\cC)\prod_{j=1}^{m+1}\psi_{\bW_j'}^{s_j}
\right|\right)\notag\\
&\quad\cdot \left(1_{m=n-1}+1_{m\neq n-1}\left|\prod_{\{p,q\}\in
 T}\D_{\{p,q\}}(\cC)\prod_{k=1}^{n-m}\psi_{\bZ_k'}^{t_k}
\right|\right).\notag
\end{align}
Recall that for $S\in \T(\{s_j\}_{j=1}^{m+1})$, $d_{s_j}(S)$ denotes the
 degree of the vertex $s_j$ in $S$. See \cite[\mbox{Subsection 3.1}]{K_BCS}
for the definition of the operator $\D_{\{p,q\}}(\cC)$. 
By following the tactics of
 estimation explained in and after \eqref{eq_tactic_estimation_left},
 \eqref{eq_tactic_estimation_right} we can derive that
\begin{align}
&[f_m^n(\bp,\bu,\bq,\bv),g]_1\label{eq_divided_multiplied_1_pre_pre}\\
&\le
 2^{n-1}D^{-n+1-\frac{1}{2}(\sum_{j=1}^{m+1}u_j+\sum_{k=1}^{n-m}v_k)}\|\tilde{\cC}\|_{1,\infty}^{n-2}\sum_{S\in\T(\{s_j\}_{j=1}^{m+1})}\sum_{T\in\T(\{t_k\}_{k=1}^{n-m})}\notag\\
&\quad\cdot\left(\begin{array}{c} p_1-u_1\\ d_1(S)\end{array}
\right)d_1(S)!
\left(\begin{array}{c} q_1-v_1\\ d_1(T)\end{array}
\right)d_1(T)!
D^{\frac{p_1+q_1}{2}}\notag\\
&\quad\cdot ([F_{p_1,q_1},g]_{1,\infty}\|\tilde{\cC}\|_{1,\infty}+
[F_{p_1,q_1},\tilde{\cC}]_{1,\infty}\|g\|_{1,\infty})\notag\\
&\quad\cdot\prod_{j=2}^{m+1}\Bigg(
\left(\begin{array}{c} p_j-u_j\\ d_{s_j}(S)\end{array}
\right)d_{s_j}(S)!
D^{\frac{p_j}{2}}(
1_{s_j=n}\|F_{p_j}^{s_j}\|_1+1_{s_j\neq
 n}\|F_{p_j}^{s_j}\|_{1,\infty})\Bigg)\notag\\
&\quad\cdot\prod_{k=2}^{n-m}\Bigg(
\left(\begin{array}{c} q_k-v_k\\ d_{t_k}(T)\end{array}
\right)d_{t_k}(T)!
D^{\frac{q_k}{2}}(
1_{t_k=n}\|F_{q_k}^{t_k}\|_1+1_{t_k\neq n}\|F_{q_k}^{t_k}\|_{1,\infty})\Bigg).\notag
\end{align}
In fact the inequality with the term 
\begin{align*}
\sup\{[F_{p_1,q_1},g]_{1,\infty}\|\tilde{\cC}\|_{1,\infty},\ [F_{p_1,q_1},\tilde{\cC}]_{1,\infty}\|g\|_{1,\infty}\}
\end{align*}
in place of 
\begin{align*}
[F_{p_1,q_1},g]_{1,\infty}\|\tilde{\cC}\|_{1,\infty}+
[F_{p_1,q_1},\tilde{\cC}]_{1,\infty}\|g\|_{1,\infty}
\end{align*}
can hold. However, we choose to use the above inequality for
 simplicity. Also, we took into
 account the fact that 
\begin{align*}
\prod_{\{p,q\}\in
 S}\D_{\{p,q\}}(\cC)\prod_{j=1}^{m+1}\psi_{\bW_j'}^{s_j},\quad
\prod_{\{p,q\}\in
 T}\D_{\{p,q\}}(\cC)\prod_{k=1}^{n-m}\psi_{\bZ_k'}^{t_k}
\end{align*}
create at most $\prod_{j=1}^{m+1}d_{s_j}(S)!$,
 $\prod_{k=1}^{n-m}d_{t_k}(T)!$ terms respectively. 

In order to support the readers, let us present an intermediate step
 between \eqref{eq_divided_core_kernel_inequality} and
 \eqref{eq_divided_multiplied_1_pre_pre}. Assume that $n\in
 \{s_j\}_{j=1}^{m+1}$. Take any $k_1\in
 \{1,2,\cdots,\sum_{k=1}^{n-m}v_k\}$. In this case $m\neq 0$. By
 following the strategy explained after \eqref{eq_tactic_estimation_left}, 
\begin{align*}
&\left(\frac{1}{h}\right)^{\sum_{k=1}^{n-m}v_k}\sup_{Y_0\in I}\sum_{\bY\in
 \prod_{k=1}^{n-m}I^{v_k}}|f_m^n(\bp,\bu,\bq,\bv)((\bX_1,\bX_2,\cdots,\bX_{m+1}),\bY)||g(Y_0,Y_{k_1})|\\
&\le
 2^{n-1}D^{-n+1-\frac{1}{2}(\sum_{j=1}^{m+1}u_j+\sum_{k=1}^{n-m}v_k)}
\sum_{S\in\T(\{s_j\}_{j=1}^{m+1})}\sum_{T\in\T(\{t_k\}_{k=1}^{n-m})}\notag\\
&\quad\cdot
\frah^{p_1-u_1}\sum_{\bW_1\in I^{p_1-u_1-d_1(S)}\atop \bW_1'\in
 I^{d_1(S)}}
\left(\begin{array}{c} p_1-u_1 \\ d_1(S)\end{array}\right)
\left(\begin{array}{c} q_1-v_1 \\ d_1(T)\end{array}\right)
d_1(T)!D^{\frac{1}{2}(p_1+q_1)}\\
&\quad\cdot\Bigg(
\left(\frac{1}{h}\right)^{q_1}\sup_{Y_0\in I}\sum_{\bY\in
 I^{q_1}}|F_{p_1,q_1}((\bW_1,\bW_1',\bX_1),\bY)||g(Y_0,Y_1)|\|\tilde{\cC}\|_{1,\infty}^{n-m-1}\\
&\qquad + 1_{m\neq n-1}\left(\frac{1}{h}\right)^{q_1}\sup_{Y_0\in
 I}\sum_{\bY\in
 I^{q_1}}|F_{p_1,q_1}((\bW_1,\bW_1',\bX_1),\bY)||\tilde{\cC}(Y_0,Y_1)|
\|g\|_{1,\infty}
\|\tilde{\cC}\|_{1,\infty}^{n-m-2}\Bigg)\\
&\quad\cdot
 \prod_{j=2}^{m+1}\Bigg(\left(\frac{1}{h}\right)^{p_j-u_j}
\sum_{\bW_j\in I^{p_j-u_j-d_{s_j}(S)}\atop \bW_j'\in
 I^{d_{s_j}(S)}}\left(\begin{array}{c} p_j-u_j \\
		  d_{s_j}(S)\end{array}\right)D^{\frac{p_j}{2}}
|F_{p_j}^{s_j}(\bW_j,\bW_j',\bX_j)|\Bigg)\\
&\quad\cdot
 \prod_{k=2}^{n-m}\Bigg(\left(\begin{array}{c} q_k-v_k \\
		  d_{t_k}(T)\end{array}\right)d_{t_k}(T)!
D^{\frac{q_k}{2}}
\|F_{q_k}^{t_k}\|_{1,\infty}\Bigg)\left|\prod_{\{p,q\}\in S}\D_{\{p,q\}}(\cC)\prod_{j=1}^{m+1}\psi_{\bW_j'}^{s_j}
\right|.
\end{align*}
Then by integrating with the variable $(\bX_1,\bX_2,\cdots,\bX_{m+1})$
 we obtain the right-hand side of
 \eqref{eq_divided_multiplied_1_pre_pre}. 
By following the strategy explained after
 \eqref{eq_tactic_estimation_right} we can deal with the case $n\in
 \{t_k\}_{k=1}^{n-m}$ as well. 

Now to restart with \eqref{eq_divided_multiplied_1_pre_pre}, let us
 recall the following
 estimate based on the well-known theorem on the number of trees with fixed
 degrees. 
\begin{align*}
&\sum_{S\in\T(\{s_j\}_{j=1}^{m+1})}\prod_{j=1}^{m+1}\left(\left(\begin{array}{c}
							   p_j-u_j\\
                                                           d_{s_j}(S)\end{array}
							  \right) d_{s_j}(S)!
\right)\\
&\le (1_{m=0}+1_{m\neq
 0}(m-1)!2^{-m-1})2^{2\sum_{j=1}^{m+1}(p_j-u_j)},\\
&\sum_{T\in\T(\{t_k\}_{k=1}^{n-m})}\prod_{k=1}^{n-m}\left(\left(\begin{array}{c}
							   q_k-v_k\\
                                                           d_{t_k}(T)\end{array}
							  \right) d_{t_k}(T)!
\right)\\
&\le (1_{m=n-1}+1_{m\neq
 n-1}(n-m-2)!2^{-n+m})2^{2\sum_{k=1}^{n-m}(q_k-v_k)}.
 \end{align*}
See \cite[\mbox{(3.20),(3.21)}]{K_BCS}.
By substituting these inequalities and using the inequality 
\begin{align*}
&2^{n-1}(1_{m=0}+1_{m\neq 0}(m-1)!2^{-m-1})(1_{m=n-1}+1_{m\neq
 n-1}(n-m-2)!2^{-n+m})\\
&\le (1_{m=0}+1_{m\neq 0}(m-1)!)(1_{m=n-1}+1_{m\neq n-1}(n-m-2)!)
\end{align*}
we obtain that 
\begin{align}
&[f_m^n(\bp,\bu,\bq,\bv),g]_1\label{eq_divided_multiplied_1_pre}\\
&\le (1_{m\neq 0}(m-1)!+1_{m=0})(1_{m\neq n-1}(n-m-2)!+1_{m=n-1})\notag\\
&\quad\cdot 2^{-2\sum_{j=1}^{m+1}u_j-2\sum_{k=1}^{n-m}v_k}
D^{-n+1-\frac{1}{2}(\sum_{j=1}^{m+1}u_j+\sum_{k=1}^{n-m}v_k)}
\|\tilde{\cC}\|_{1,\infty}^{n-2}\notag\\
&\quad\cdot
2^{2p_1+2q_1}D^{\frac{p_1+q_1}{2}}([F_{p_1,q_1},g]_{1,\infty}\|\tilde{\cC}\|_{1,\infty}+
[F_{p_1,q_1},\tilde{\cC}]_{1,\infty}\|g\|_{1,\infty})\notag\\
&\quad\cdot\prod_{j=2}^{m+1}(
2^{2p_j}D^{\frac{p_j}{2}}(1_{s_j=n}\|F_{p_j}^{s_j}\|_1+1_{s_j\neq
 n}\|F_{p_j}^{s_j}\|_{1,\infty}))\notag\\
&\quad\cdot\prod_{k=2}^{n-m}(
2^{2q_k}D^{\frac{q_k}{2}}(
1_{t_k=n}\|F_{q_k}^{t_k}\|_1
+1_{t_k\neq n}\|F_{q_k}^{t_k}\|_{1,\infty})).\notag
\end{align}

Let us consider $[f_m^n(\bp,\bu,\bq,\bv),g]_{1,\infty}$.  To estimate 
\begin{align*}
&\sup_{X_0\in I\atop j_1\in \{1,2,\cdots,\sum_{j=1}^{m+1}u_j\}}
\Bigg(\frah^{\sum_{j=1}^{m+1}u_j+\sum_{k=1}^{n-m}v_k-1}\sum_{\bX\in
 \prod_{j=1}^{m+1}I^{u_j}}1_{X_{j_1}=X_0}\\
&\qquad\qquad\qquad\qquad\cdot \sup_{Y_0\in I\atop 
k_1\in \{1,2,\cdots,\sum_{k=1}^{n-m}v_k\}
}\sum_{\bY\in
 \prod_{k=1}^{n-m}I^{v_k}}|f_m^n(\bp,\bu,\bq,\bv)(\bX,\bY)||g(Y_0,Y_{k_1})|\Bigg),
\end{align*}
we fix $j_1\in \{1,2,\cdots,\sum_{j=1}^{m+1}u_j\}$, $k_1\in
 \{1,2,\cdots,\sum_{k=1}^{n-m}v_k\}$. Then there uniquely exist $j_0\in \{1,2,\cdots,m+1\}$, $k_0\in
 \{1,2,\cdots,n-m\}$ such that $X_{j_1}$ is a component of the variable
 of the function $F_{p_1,q_1}$ if $j_0=1$ or $F_{p_{j_0}}^{s_{j_0}}$ if
 $j_0\neq 1$ and $Y_{k_1}$ is a component of the variable of the
 function $F_{p_1,q_1}$ if $k_0=1$ or $F_{q_{k_0}}^{t_{k_0}}$ if
 $k_0\neq 1$. We recursively estimate along the lines of $T\in
 \T(\{t_k\}_{k=1}^{n-m})$ by considering the vertex $t_{k_0}$ as the
 root in the first place. Then we recursively estimate along the lines of $S\in
 \T(\{s_j\}_{j=1}^{m+1})$ by considering the vertex $s_{j_0}$ as the
 root.  On the other hand, we estimate 
\begin{align*}
&\sup_{Y_0\in I\atop k_1\in \{1,2,\cdots,\sum_{k=1}^{n-m}v_k\}}
\Bigg(\frah^{\sum_{j=1}^{m+1}u_j+\sum_{k=1}^{n-m}v_k-1}\sum_{\bY\in
 \prod_{k=1}^{n-m}I^{v_k}}1_{Y_{k_1}=Y_0}\\
&\qquad\qquad\qquad\qquad\cdot \sup_{X_0\in I\atop 
j_1\in \{1,2,\cdots,\sum_{j=1}^{m+1}u_j\}
}\sum_{\bX\in
 \prod_{j=1}^{m+1}I^{u_j}}|f_m^n(\bp,\bu,\bq,\bv)(\bX,\bY)||g(X_0,X_{j_1})|\Bigg)
\end{align*}
by performing the recursive estimation along the lines of $S\in
 \T(\{s_j\}_{j=1}^{m+1})$ first and the recursive estimation along the
 lines of $T\in
 \T(\{t_k\}_{k=1}^{n-m})$ afterwards. Since the procedure is parallel to
 the estimation of $[f_m^n(\bp,\bu,\bq,\bv),g]_{1}$, we only state the result.
\begin{align}
&[f_m^n(\bp,\bu,\bq,\bv),g]_{1,\infty}\label{eq_divided_multiplied_1_infinity_pre}\\
&\le (1_{m\neq 0}(m-1)!+1_{m=0})(1_{m\neq n-1}(n-m-2)!+1_{m=n-1})\notag\\
&\quad\cdot 2^{-2\sum_{j=1}^{m+1}u_j-2\sum_{k=1}^{n-m}v_k}
D^{-n+1-\frac{1}{2}(\sum_{j=1}^{m+1}u_j+\sum_{k=1}^{n-m}v_k)}
\|\tilde{\cC}\|_{1,\infty}^{n-2}\notag\\
&\quad\cdot
2^{2p_1+2q_1}D^{\frac{p_1+q_1}{2}}([F_{p_1,q_1},g]_{1,\infty}\|\tilde{\cC}\|_{1,\infty}
+[F_{p_1,q_1},\tilde{\cC}]_{1,\infty}\|g\|_{1,\infty})\notag\\
&\quad\cdot\prod_{j=2}^{m+1}(2^{2p_j}D^{\frac{p_j}{2}}\|F_{p_j}^{s_j}\|_{1,\infty})
\prod_{k=2}^{n-m}(2^{2q_k}D^{\frac{q_k}{2}}\|F_{q_k}^{t_k}\|_{1,\infty}).\notag
\end{align}
Again we overestimated by replacing 
\begin{align*}
\sup\{[F_{p_1,q_1},g]_{1,\infty}\|\tilde{\cC}\|_{1,\infty},\ [F_{p_1,q_1},\tilde{\cC}]_{1,\infty}\|g\|_{1,\infty}\}
\end{align*}
by their sum for simplicity.

It follows from \eqref{eq_divided_kernel} that 
\begin{align}
&[E_{a,b}^{(n)},g]_{index}\label{eq_divided_multiplied_outside}\\
&\le \sum_{p_1,q_1=2}^{N}1_{p_1,q_1\in 2\N}\sum_{u_1=0}^{p_1}
(1_{m=0}+1_{m\neq 0}1_{u_1\le
 p_1-1})\left(\begin{array}{c} p_1 \\ u_1\end{array}\right)\notag\\
&\quad\cdot \sum_{v_1=0}^{q_1}
(1_{m=n-1}+1_{m\neq n-1}1_{v_1\le
 q_1-1})\left(\begin{array}{c} q_1 \\ v_1\end{array}\right)\notag\\
&\quad\cdot \prod_{j=2}^{m+1}\left(\sum_{p_j=2}^N\sum_{u_j=0}^{p_j-1}
\left(\begin{array}{c} p_j \\ u_j\end{array}\right)\right)
\prod_{k=2}^{n-m}\left(\sum_{q_k=2}^N\sum_{v_k=0}^{q_k-1}
\left(\begin{array}{c} q_k \\ v_k\end{array}\right)\right)
 [f_m^n(\bp,\bu,\bq,\bv),g]_{index}\notag\\
&\quad\cdot 1_{\sum_{j=1}^{m+1}u_j=a}1_{\sum_{k=1}^{n-m}v_k=b}
1_{\sum_{j=1}^{m+1}p_j-2m\ge a}1_{\sum_{k=1}^{n-m}q_k-2(n-m-1)\ge b},\notag
\end{align}
where $index=`1,\infty'$ or $index=1$. By substituting
 \eqref{eq_divided_multiplied_1_1_pre} for $index=`1,\infty'$, $1$,
 \eqref{eq_divided_multiplied_1_infinity_pre},
 \eqref{eq_divided_multiplied_1_pre} into
 \eqref{eq_divided_multiplied_outside} we obtain
 \eqref{eq_divided_multiplied_1_infinity_1}, 
 \eqref{eq_divided_multiplied_1_1},
 \eqref{eq_divided_multiplied_1_infinity},
 \eqref{eq_divided_multiplied_1} respectively.
\end{proof}

\subsection{Generalized
  covariances}\label{subsec_generalized_covariances}

We construct the general multi-scale integration process by assuming
scale-dependent bound properties of covariances. Here let us list the
properties of the generalized covariances. Assume that
$N_{\beta}<\hat{N}_{\beta}$,
$N_{\beta},\hat{N}_{\beta}\in\Z$. These numbers represent the
integration scales. We should think that at the scale
$\hat{N}_{\beta}+1$ we perform a single-scale ultra-violet (UV)
integration and from $\hat{N}_{\beta}$ to $N_{\beta}$ we perform a
multi-scale infrared (IR) integration. Let $c_0,\ M\in \R_{\ge 1}$,
$c_{end}\in \R_{>0}$. We assume that covariances $\cC_l:I_0^2\to\C$
$(l=N_{\beta},N_{\beta}+1,\cdots,\hat{N}_{\beta})$ satisfy the
following properties. 
\begin{itemize}
\item $\cC_l$ $(l=N_{\beta}+1,N_{\beta}+2,\cdots,\hat{N}_{\beta})$
      satisfy \eqref{eq_general_covariance_time_translation}.
\item 
\begin{align}
&\cC_{N_{\beta}}(\orho\rho\bx s, \oeta\eta\by t)=
 \cC_{N_{\beta}}(\orho\rho\bx 0, \oeta\eta\by
 0),\label{eq_final_covariance_time_independence}\\
&(\forall (\orho,\rho,\bx,s),(\oeta,\eta,\by,t)\in I_0).\notag
\end{align}
\item 
\begin{align}
&|\det(\<\bu_i,\bv_j\>_{\C^m}\cC_l(X_i,Y_j))_{1\le i,j\le n}|\le
 (c_0M^{\sa
 (l-\hat{N}_{\beta})})^n,\label{eq_scale_covariance_determinant_bound}\\
&(\forall m,n\in \N,\ \bu_i,\bv_i\in\C^m\text{ with
 }\|\bu_i\|_{\C^m},\|\bv_i\|_{\C^m}\le 1,\notag\\
&\quad X_i,Y_i\in I_0\ (i=1,2,\cdots,n),\ l\in
 \{N_{\beta},N_{\beta}+1,\cdots,\hat{N}_{\beta}\}).\notag
\end{align}
\item 
\begin{align}
&\|\tilde{\cC_l}\|_{1,\infty}\le c_0\left(1_{N_{\beta}+1\le l\le
 \hat{N}_{\beta}}M^{(\sa-1-\sum_{j=1}^d\frac{1}{\sn_j})(l-\hat{N}_{\beta})}
+1_{l=N_{\beta}}c_{end}\right),\label{eq_scale_covariance_decay_bound}\\
&(\forall l\in \{N_{\beta},N_{\beta}+1,\cdots,\hat{N}_{\beta}\}).\notag
\end{align}
\item 
\begin{align}
&\|\tilde{\cC_l}\|\le c_0
 M^{(\sa-1-\sum_{j=1}^d\frac{1}{\sn_j})(l-\hat{N}_{\beta})},\label{eq_scale_covariance_coupled_decay_bound}\\
&(\forall l\in \{N_{\beta}+1,N_{\beta}+2,\cdots,\hat{N}_{\beta}\}).\notag
\end{align}
\end{itemize}
Here $\tilde{\cC}_l:I^2\to \C$ is the anti-symmetric extension of
$\cC_l$ defined as in \eqref{eq_covariance_anti_symmetric_extension}. 
In Subsection \ref{subsec_real_covariance} we will explicitly define these scale-dependent covariances by decomposing the
actual covariance appearing in the formulation Lemma \ref{lem_final_Grassmann_formulation}.

\subsection{Multi-scale integration without the artificial
  term}\label{subsec_integration_without}

In the rest of this section we always extend the coupling constant 
to be a complex parameter. To distinguish, let $u$ denote the extended
coupling constant and set 
\begin{align*}
&V(u)(\psi):=\frac{u}{L^dh}\sum_{(\rho,\bx)\in\cB\times
 \G}\sum_{s\in[0,\beta)_h}\opsi_{1\rho\bx s}\psi_{1\rho\bx
 s}\\
&\qquad\qquad\quad +\frac{u}{L^dh}\sum_{(\rho,\bx),(\eta,\by)\in\cB\times \G}\sum_{s\in [0,\beta)_h}\opsi_{1\rho\bx s}\psi_{2\rho\bx
 s}\opsi_{2\eta\by s}\psi_{1\eta\by
 s},\\
&W(u)(\psi):=\frac{u}{\beta L^dh^2}\sum_{(\rho,\bx),(\eta,\by)\in\cB\times
 \G}\sum_{s,t\in[0,\beta)_h}\opsi_{1\rho\bx s}\psi_{2\rho\bx
 s}\opsi_{2\eta\by t}\psi_{1\eta\by
 t}.
\end{align*}
In this subsection we construct a multi-scale integration for the
Grassmann polynomial
\begin{align*}
\log\left(
\int
 e^{-V(u)(\psi+\psi^1)+W(u)(\psi+\psi^1)}d\mu_{\sum_{l=N_{\beta}+1}^{\hat{N}_{\beta}}\cC_l}(\psi^1)\right).
\end{align*}
The well-definedness of this Grassmann polynomial is a priori guaranteed
only for small $u$. We are going to construct an analytic continuation
of this Grassmann polynomial. Uniform boundedness of the analytically
continued polynomial is important in controlling the integrand of the
Gaussian integrals in the final formulation Lemma
\ref{lem_final_Grassmann_formulation}
\eqref{item_final_Grassmann_formulation}. In the next subsection we will
perform a multi-scale integration by adding the artificial term
$A(\psi)$ to the input polynomial. We want to prove the analyticity of
Grassmann polynomials with the variable $u$ as a result of the
multi-scale integration in this subsection. For this purpose it is natural to consider
kernels of Grassmann polynomials as elements of the Banach space
$C(\overline{D(r)},\Map(I^m,\C))$ equipped with the norm
$\|\cdot\|_{1,\infty,r}$ defined by 
$$
\|f\|_{1,\infty,r}:=\sup_{u\in\overline{D(r)}}\|f(u)\|_{1,\infty}.
$$
We also let $\|\cdot\|_{1,\infty,r}$ denote the uniform norm of
$C(\overline{D(r)},\C)$ for notational consistency. Similarly for $f\in
C(\overline{D(r)},\Map(I^m,\C))$ and an anti-symmetric function
$g:I^2\to\C$ we set
$$
[f,g]_{1,\infty,r}:=\sup_{u\in \overline{D(r)}}[f(u),g]_{1,\infty}.
$$

More generally, for any domain $D$ of $\C^n$ and finite-dimensional complex Banach space
$B$ we let $C(\overline{D},B)$, $C^{\o}(D,B)$ denote the set of all
continuous maps from $\overline{D}$ to $B$, the set of all analytic maps
from $D$ to $B$ respectively. In practice we let $B$ be
$\bigwedge_{even}\cV$ or $\Map(I^m,\C)$, even though we do not always specify a
norm on these complex vector spaces. The finite-dimensionality implies that every
norm is equivalent to each other. Normally, we use
$\|\cdot\|_{1,\infty}$ or $\|\cdot\|_1$ as a norm of $\Map(I^m,\C)$ and
induce a norm of $\bigwedge_{even}\cV$ by measuring anti-symmetric
kernels of a Grassmann polynomial by $\|\cdot\|_{1,\infty}$ or
$\|\cdot\|_1$. 
The readers should
understand which norm is being considered from the context. Observe 
that once a norm is defined in $\bigwedge_{even}\cV$, $f\in
C(\overline{D},\bigwedge_{even}\cV)$ is equivalent to $f_0\in
C(\overline{D})$, $f(\cdot)_m\in C(\overline{D},\Map(I^m,\C))$
$(m=2,4,\cdots,N)$, which is equivalent to $f_0\in C(\overline{D})$,
$f(\cdot)_m(\bX)\in C(\overline{D})$
$(\bX\in I^m,\ m=2,4,\cdots,N)$,
where $f(u)_m$ $(m=2,4,\cdots,N)$ are anti-symmetric
kernels of $f(u)(\psi)$ for $u\in \overline{D}$. The parallel statements 
can be made for $C^{\o}(D,\bigwedge_{even}\cV)$. In order to systematically
describe properties of Grassmann data in the multi-scale integration, we
define several subsets of $C(\overline{D(r)},\bigwedge_{even}\cV)$. In
the following we let $l\in
\{N_{\beta}+1,N_{\beta}+2,\cdots,\hat{N}_{\beta}\}$, $r\in\R_{>0}$ and
$\alpha\in\R_{\ge 1}$. 

We define the set $\cQ(r,l)$ as follows. $f$ belongs to $\cQ(r,l)$ if
and only if 
\begin{itemize}
\item 
$$
f\in C\left(\overline{D(r)},\bigwedge_{even}\cV\right)\cap C^{\o}\left(D(r),\bigwedge_{even}\cV\right). 
$$
\item 
For any $u\in \overline{D(r)}$ the anti-symmetric kernels $f(u)_m:I^m\to
      \C$ $(m=2,4,\cdots,N)$ satisfy
      \eqref{eq_temperature_translation_invariance} and 
\begin{align}
&\frac{h}{N}\alpha^{2}M^{(\sum_{j=1}^d\frac{1}{\sn_j}+1)(l-\hat{N}_{\beta})}\|f_0\|_{1,\infty,r}\le
 L^{-d},\label{eq_0_1_bound}\\
&\sum_{m=2}^Nc_0^{\frac{m}{2}}\alpha^mM^{\frac{m}{2}\sa(l-\hat{N}_{\beta})}\|f_m\|_{1,\infty,r}\le
 L^{-d}.\notag
\end{align}
\end{itemize}
Simply speaking, we use the set $\cQ(r,l)$ to collect Grassmann data bounded by the
      inverse volume factor.

We define the set $\cR(r,l)$ as follows. 
$f$ belongs to $\cR(r,l)$ if and only if
\begin{itemize}
\item 
$$
f\in C\left(\overline{D(r)},\bigwedge_{even}\cV\right)\cap C^{\o}\left(D(r),\bigwedge_{even}\cV\right). 
$$
\item There exist $f_{p,q}\in C(\overline{D(r)},\Map(I^p\times I^q,\C))$
$(p,q\in \{2,4,\cdots,N\})$ such that for any $u\in\overline{D(r)}$,
$p,q\in \{2,4,\cdots,N\}$, $f_{p,q}(u):I^p\times I^q\to \C$ is
bi-anti-symmetric, satisfies
\eqref{eq_temperature_translation_invariance},
\eqref{eq_temperature_integral_vanish}, 
\begin{align*}
f(u)(\psi)=\sum^{N}_{p,q=2}1_{p,q\in 2\N}\frah^{p+q}\sum_{\bX\in
 I^p\atop \bY\in I^q}f_{p,q}(u)(\bX,\bY)\psi_{\bX}\psi_{\bY}
\end{align*}
and 
\begin{align}
&M^{-(\sum_{j=1}^d\frac{1}{\sn_j}+1)(l-\hat{N}_{\beta})}\sum_{p,q=2}^Nc_0^{\frac{p+q}{2}}\alpha^{p+q}M^{\frac{p+q}{2}\sa(l-\hat{N}_{\beta})}\|f_{p,q}\|_{1,\infty,r}\le
 1,\label{eq_0_2_bound}\\
&M^{-(\sum_{j=1}^d\frac{1}{\sn_j}+1)(l-\hat{N}_{\beta})}\sum_{p,q=2}^Nc_0^{\frac{p+q}{2}}\alpha^{p+q}M^{\frac{p+q}{2}\sa(l-\hat{N}_{\beta})}[f_{p,q},g]_{1,\infty,r}\le
 L^{-d}\|g\|,\label{eq_0_2_multiplied_bound}
\end{align}
for any anti-symmetric function $g:I^2\to \C$. 
\end{itemize}
In essence the set $\cR(r,l)$ will be
used to collect Grassmann data having the vanishing property \eqref{eq_temperature_integral_vanish}.

Let us start explaining the inductive multi-scale integration process by
explicitly defining the initial Grassmann data. Define
$V_2^{0-1,\hat{N}_{\beta}}:\overline{D(r)}\to \Map(I^2,\C)$ by 
\begin{align}
&V_2^{0-1,\hat{N}_{\beta}}(u)(\orho_1\rho_1\bx_1s_1\xi_1,\orho_2\rho_2\bx_2s_2\xi_2)\label{eq_kernel_0_1_initial}\\
&:=-\frac{1}{2}uL^{-d}h1_{(\rho_1,\bx_1,s_1)=(\rho_2,\bx_2,s_2)}1_{\orho_1=\orho_2=1}(1_{(\xi_1,\xi_2)=(1,-1)}-
 1_{(\xi_1,\xi_2)=(-1,1)}).\notag
\end{align}
Then we define $V^{0-1,\hat{N}_{\beta}}\in
C(\overline{D(r)},\bigwedge_{even}\cV)$ by
\begin{align*}
V^{0-1,\hat{N}_{\beta}}(u)(\psi):=\frah^2\sum_{\bX\in
 I^2}V_2^{0-1,\hat{N}_{\beta}}(u)(\bX)\psi_{\bX}.
\end{align*}
Let us define $V_{2,2}^{0-2,\hat{N}_{\beta}}:\overline{D(r)}\to
\Map(I^2\times I^2,\C)$ by
\begin{align}
&V_{2,2}^{0-2,\hat{N}_{\beta}}(u)(\orho_1\rho_1\bx_1s_1\xi_1,\orho_2\rho_2\bx_2s_2\xi_2,\oeta_1\eta_1\by_1t_1\zeta_1,\oeta_2\eta_2\by_2t_2\zeta_2)\label{eq_kernel_0_2_initial}\\
&:=-\frac{1}{4}uL^{-d}h^2 1_{(\rho_1,\bx_1,s_1,\eta_1,\by_1,t_1)
=(\rho_2,\bx_2,s_2,\eta_2,\by_2,t_2)}(h1_{s_1=t_1}-\beta^{-1})\notag\\
&\qquad\cdot \sum_{\s,\tau\in
 \S_2}\sgn(\s)\sgn(\tau)1_{(\orho_{\s(1)},\orho_{\s(2)},\oeta_{\tau(1)},\oeta_{\tau(2)})=(1,2,2,1)}
1_{(\xi_{\s(1)},\xi_{\s(2)},\zeta_{\tau(1)},\zeta_{\tau(2)})=(1,-1,1,-1)}.\notag\end{align}
Then we define $V^{0-2,\hat{N}_{\beta}}\in
C(\overline{D(r)},\bigwedge_{even}\cV)$ by 
\begin{align*}
V^{0-2,\hat{N}_{\beta}}(u)(\psi):=\frah^4\sum_{\bX,\bY\in
 I^2}V_{2,2}^{0-2,\hat{N}_{\beta}}(u)(\bX,\bY)\psi_{\bX}\psi_{\bY}.
\end{align*}
Observe that
\begin{align*}
V^{0-1,\hat{N}_{\beta}}(u)(\psi)+V^{0-2,\hat{N}_{\beta}}(u)(\psi)=-V(u)(\psi)+W(u)(\psi).
\end{align*}
We give $V^{0-1,\hat{N}_{\beta}}+V^{0-2,\hat{N}_{\beta}}$ as the initial
data to the multi-scale integration. Using the notations introduced
above, let us inductively define $V^{0-1,l}$, $V^{0-2,l}\in$\\ 
$C(\overline{D(r)},\bigwedge_{even}\cV)$
$(l=N_{\beta}+1,N_{\beta}+2,\cdots,\hat{N}_{\beta})$ as follows. 
Assume that we have $V^{0-1,l+1}\in\cQ(r,l+1)$,
$V^{0-2,l+1}\in\cR(r,l+1)$ and
\begin{align*}
V^{0-2,l+1}(u)(\psi) =\sum_{p,q=2}^{N}1_{p,q\in 2\N}\frah^{p+q}\sum_{\bX\in
 I^p\atop \bY\in
 I^q}V_{p,q}^{0-2,l+1}(u)(\bX,\bY)\psi_{\bX}\psi_{\bY},\ (\forall u\in \overline{D(r)})
\end{align*}
with $V_{p,q}^{0-2,l+1}\in C(\overline{D(r)},\Map(I^p\times I^q,\C))$
$(p,q\in \{2,4,\cdots,N\})$ satisfying the conditions required in
$\cR(r,l+1)$.
Then, let us set for any $n\in \N_{\ge 1}$, $u\in \overline{D(r)}$,
\begin{align*}
&V^{0-1-1,l,(n)}(u)(\psi)\\
&:=\prod_{j=1}^n\Bigg(\sum_{a_j\in\{1,2\}}\Bigg)1_{\exists
 j(a_j=1)}\frac{1}{n!}Tree(\{1,2,\cdots,n\},\cC_{l+1})\notag\\
&\quad\cdot\prod_{j=1}^nV^{0-a_j,l+1}(u)(\psi^j+\psi)\Bigg|_{\psi^{j}=0\atop(\forall
 j\in\{1,2,\cdots,n\})},\notag\\
&V^{0-1-2,l,(n)}(u)(\psi)\notag\\
&:=\sum_{p,q=2}^{N}1_{p,q\in 2\N}\frah^{p+q}\sum_{\bX\in
 I^p\atop \bY\in
 I^q}V_{p,q}^{0-2,l+1}(u)(\bX,\bY)
 \frac{1}{n!}Tree(\{1,2,\cdots,n+1\},\cC_{l+1})\notag\\
&\quad\cdot (\psi^1+\psi)_{\bX}
(\psi^2+\psi)_{\bY}
\prod_{j=3}^{n+1}V^{0-2,l+1}(u)(\psi^j+\psi)\Bigg|_{\psi^{j}=0\atop(\forall
 j\in\{1,2,\cdots,n+1\})},\notag\\
&V^{0-2,l,(n)}(u)(\psi)\notag\\
&:=\frac{1}{n!}\sum_{m=0}^{n-1}\sum_{(\{s_j\}_{j=1}^{m+1},
 \{t_k\}_{k=1}^{n-m})\in S(n,m)}
\sum_{p,q=2}^{N}1_{p,q\in 2\N}\frah^{p+q}\sum_{\bX\in
 I^p\atop \bY\in
 I^q}V_{p,q}^{0-2,l+1}(u)(\bX,\bY)\notag\\
&\quad\cdot Tree(\{s_j\}_{j=1}^{m+1},\cC_{l+1})(\psi^{s_1}+\psi)_{\bX}\prod_{j=2}^{m+1}V^{0-2,l+1}(u)(\psi^{s_j}+\psi)
\Bigg|_{\psi^{s_j}=0\atop(\forall
 j\in\{1,2,\cdots,m+1\})}\notag\\
&\quad\cdot Tree(\{t_k\}_{k=1}^{n-m},\cC_{l+1})(\psi^{t_1}+\psi)_{\bY}\prod_{k=2}^{n-m}V^{0-2,l+1}(u)(\psi^{t_k}+\psi)
\Bigg|_{\psi^{t_k}=0\atop(\forall
 k\in\{1,2,\cdots,n-m\})},\notag
\end{align*}
where 
\begin{align*}
S(n,m):=\left\{(\{s_j\}_{j=1}^{m+1},
 \{t_k\}_{k=1}^{n-m})\ \Bigg|\
 \begin{array}{l}1=s_1<s_2<\cdots<s_{m+1}\le n,\\
                 1=t_1<t_2<\cdots<t_{n-m}\le n,\\
                \{s_j\}_{j=2}^{m+1}\cup
		 \{t_k\}_{k=2}^{n-m}=\{2,3,\cdots,n\},\\  
               \{s_j\}_{j=2}^{m+1}\cap \{t_k\}_{k=2}^{n-m}=\emptyset.
\end{array}
\right\}.
\end{align*}
Set 
\begin{align*}
&V^{0-1-j,l}(u)(\psi):=\sum_{n=1}^{\infty}V^{0-1-j,l,(n)}(u)(\psi),\quad
 (j=1,2),\quad V^{0-1,l}:=\sum_{j=1}^2V^{0-1-j,l},\\
&V^{0-2,l}(u)(\psi):=\sum_{n=1}^{\infty}V^{0-2,l,(n)}(u)(\psi)
\end{align*}
on the assumption that these series converge in $\bigwedge_{even} \cV$. 

Let us see how these Grassmann data are derived during the process. We give 
$V^{0-1,l+1}+V^{0-2,l+1}$ to the single-scale integration with the
covariance $\cC_{l+1}$. By applying the formula
\eqref{eq_tree_expansion} we can derive that
\begin{align}
&\frac{1}{n!}\left(\frac{d}{dz}\right)^n\log\left(\int
 e^{z\sum_{j=1}^2V^{0-j,l+1}(u)(\psi^1+\psi)}d\mu_{\cC_{l+1}}(\psi^1)\right)\Bigg|_{z=0}\label{eq_data_recursive_relation_without}\\
&=V^{0-1-1,l,(n)}(u)(\psi)\notag\\
&\quad +\sum_{p,q=2}^{N}1_{p,q\in 2\N}\frah^{p+q}\sum_{\bX\in
 I^p\atop \bY\in
 I^q}V_{p,q}^{0-2,l+1}(u)(\bX,\bY)\notag\\
&\qquad\cdot 
\frac{1}{n!}\prod_{j=2}^n\left(\frac{\partial}{\partial z_j}\right)
\int (\psi^1+\psi)_{\bX}(\psi^1+\psi)_{\bY}e^{\sum_{j=2}^n
 z_jV^{0-2,l+1}(u)(\psi^1+\psi)}d\mu_{\cC_{l+1}}(\psi^1)\notag\\
&\qquad\cdot \left(\int e^{\sum_{j=2}^n
 z_jV^{0-2,l+1}(u)(\psi^1+\psi)}d\mu_{\cC_{l+1}}(\psi^1)\right)^{-1}\Bigg|_{z_j=0\atop(\forall
 j\in \{2,3,\cdots,n\})}\notag\\
&=V^{0-1-1,l,(n)}(u)(\psi)+ V^{0-1-2,l,(n)}(u)(\psi) +
 V^{0-2,l,(n)}(u)(\psi).\notag 
\end{align}
We should remark that the above
transformation is essentially same as \cite[\mbox{(3.56)}]{K_BCS} and is based on the ideas of the earlier papers
\cite[\mbox{(3.38)}]{M}, \cite[\mbox{(IV.15)}]{L}. Also, we should
remind us that the logarithm and the inverse of the even Grassmann
polynomials are analytic with $z$, $(z_j)_{j=2}^n$ in a neighborhood of the
origin and thus the above transformation is mathematically justified.  

Let us explain the rule of the superscripts put on these Grassmann data. We
use the label $0-1$, $0-2$ as the 1st superscript of Grassmann data
independent of the artificial parameters $\la_1$, $\la_2$. In the next
subsection we will use the label $1-j$, $2$ as the 1st superscript of
Grassmann data depending on $\la_1$, $\la_2$ linearly, at least
quadratically respectively. The 2nd superscript stands for the scale of
integration. The Grassmann data with the 2nd superscript $l$ is to be
integrated with the covariance $\cC_l$. For example, $V^{0-1,l}$ is
independent of $\la_1$, $\la_2$ and to be integrated with the covariance
$\cC_l$. $V^{1-1,l+1}$ is linearly dependent on $\la_1$, $\la_2$ and to
be integrated with $\cC_{l+1}$ and so on.
 
We can describe properties of these scale-dependent Grassmann data as
follows.

\begin{lemma}\label{lem_IR_integration_without}
There exists a positive constant $c_4$ independent of any parameter such
 that if 
\begin{align}
M^{\min\{1,2\sa-1-\sum_{j=1}^d\frac{1}{\sn_j}\}}\ge c_4,\quad \alpha\ge
 c_4 M^{\frac{\sa}{2}},\quad L^d\ge
 M^{(\sum_{j=1}^d\frac{1}{\sn_j}+1)(\hat{N}_{\beta}-N_{\beta})},\label{eq_assumptions_IR_integration_without}
\end{align}
then
\begin{align*}
&V^{0-1,l}\in\cQ(b^{-1}c_0^{-2}\alpha^{-4},l),\\
&V^{0-2,l}\in \cR(b^{-1}c_0^{-2}\alpha^{-4},l),\quad (\forall l\in
 \{N_{\beta},N_{\beta}+1,\cdots,\hat{N}_{\beta}\}).
\end{align*}
\end{lemma}
 
\begin{proof} During the proof we often replace a generic positive constant
 denoted by `$c$' by a larger generic constant still denoted by the same
 symbol without commenting on the replacement. It should be understood
 in the end 
 that these replacement do not violate the validity of
 the proof of the claims. 

We can see from \eqref{eq_kernel_0_1_initial},
 \eqref{eq_kernel_0_2_initial} that
 $V_2^{0-1,\hat{N}_{\beta}}(u)(\cdot)$ is anti-symmetric, satisfies
 \eqref{eq_temperature_translation_invariance}, 
 $V^{0-2,\hat{N}_{\beta}}_{2,2}(u)(\cdot,\cdot)$ is bi-anti-symmetric, satisfies
 \eqref{eq_temperature_translation_invariance}, 
 \eqref{eq_temperature_integral_vanish} and 
\begin{align}
&\|V_2^{0-1,\hat{N}_{\beta}}(u)\|_{1,\infty}\le
 |u|L^{-d},\label{eq_0_1_2_initial}\\
&\|V_{2,2}^{0-2,\hat{N}_{\beta}}(u)\|_{1,\infty}\le
 b|u|,\label{eq_0_2_2_2_initial}\\
&[V_{2,2}^{0-2,\hat{N}_{\beta}}(u),g]_{1,\infty}\le |u| L^{-d}\|g\|
\label{eq_0_2_2_2_multiplied_initial}
\end{align}
for any anti-symmetric function $g:I^2\to \C$. Thus it follows that 
$V^{0-1,\hat{N}_{\beta}}\in
 \cQ(b^{-1}c_0^{-2}\alpha^{-4},\hat{N}_{\beta})$, 
$V^{0-2,\hat{N}_{\beta}}\in
 \cR(b^{-1}c_0^{-2}\alpha^{-4},\hat{N}_{\beta})$.

Set $r:=b^{-1}c_0^{-2}\alpha^{-4}$. Assume that $l\in
 \{N_{\beta},N_{\beta}+1,\cdots,\hat{N}_{\beta}-1\}$ and we have 
$V^{0-1,j}\in \cQ(r,j)$, $V^{0-2,j}\in \cR(r,j)$
 $(j=l+1,l+2,\cdots,\hat{N}_{\beta})$. Let us show that $V^{0-1-1,l}$,
 $V^{0-1-2,l}$, $V^{0-2,l}$ are well-defined and 
$V^{0-1-1,l}+V^{0-1-2,l}\in\cQ(r,l)$, $V^{0-2,l}\in \cR(r,l)$. 
 The following inequalities can be
 derived from the definition of $\cQ(r,l+1)$, $\cR(r,l+1)$ and the
 assumptions $\alpha\ge 2^3$, $M^{\sa}\ge 2^4$. 
\begin{align}
&\sum_{m=2}^N2^{3m}(c_0M^{\sa(l+1-\hat{N}_{\beta})})^{\frac{m}{2}}\|V_m^{0-1,l+1}\|_{1,\infty,r}\le
 c
 \alpha^{-2}L^{-d},\label{eq_0_1_bound_assumption}\\
&\sum_{m=2}^{N}2^m\alpha^m(c_0
 M^{\sa(l-\hat{N}_{\beta})})^{\frac{m}{2}}\|V_m^{0-1,l+1}\|_{1,\infty,r}\le
 c M^{-\sa}L^{-d},\label{eq_0_1_bound_weight_assumption}\\
&\sum_{m=4}^N2^{3m}(c_0M^{\sa(l+1-\hat{N}_{\beta})})^{\frac{m}{2}}\|V_m^{0-2,l+1}\|_{1,\infty,r}\le
 c
 \alpha^{-4}M^{(\sum_{j=1}^d\frac{1}{\sn_j}+1)(l+1-\hat{N}_{\beta})},\label{eq_0_2_kernel_bound_assumption}\\
&\sum_{p,q=2}^N1_{p,q\in
 2\N}2^{p+q}\alpha^{p+q}(c_0M^{\sa(l-\hat{N}_{\beta})})^{\frac{p+q}{2}}\|V_{p,q}^{0-2,l+1}\|_{1,\infty,r}\le 
c M^{-2\sa
 +(\sum_{j=1}^d\frac{1}{\sn_j}+1)(l+1-\hat{N}_{\beta})},\label{eq_0_2_bound_weight_assumption}\\
&\sum_{m=4}^N2^m\alpha^m
 (c_0M^{\sa(l-\hat{N}_{\beta})})^{\frac{m}{2}}\|V_m^{0-2,l+1}\|_{1,\infty,r}\le
 c M^{-2\sa
 +(\sum_{j=1}^d\frac{1}{\sn_j}+1)(l+1-\hat{N}_{\beta})},\label{eq_0_2_kernel_bound_weight_assumption}\\
&\sum_{p,q=2}^N1_{p,q\in
 2\N}2^{3p+3q}(c_0M^{\sa(l+1-\hat{N}_{\beta})})^{\frac{p+q}{2}}[V_{p,q}^{0-2,l+1},g]_{1,\infty,r}\label{eq_0_2_multiplied_bound_assumption}\\
&\le 
c
 \alpha^{-4}M^{(\sum_{j=1}^d\frac{1}{\sn_j}+1)(l+1-\hat{N}_{\beta})}L^{-d}\|g\|,\notag\\
&\sum_{p,q=2}^N1_{p,q\in
 2\N}2^{2p+2q}\alpha^{p+q}(c_0M^{\sa(l-\hat{N}_{\beta})})^{\frac{p+q}{2}}[V_{p,q}^{0-2,l+1},g]_{1,\infty,r}\label{eq_0_2_multiplied_bound_weight_assumption}\\
&\le 
c M^{-2\sa+(\sum_{j=1}^d\frac{1}{\sn_j}+1)(l+1-\hat{N}_{\beta})}L^{-d}\|g\|,\notag
\end{align}
for any anti-symmetric function $g:I^2\to \C$. In the derivation of
 \eqref{eq_0_2_kernel_bound_assumption},
 \eqref{eq_0_2_kernel_bound_weight_assumption} we used the inequality
\begin{align}
\|V_m^{0-2,l+1}\|_{1,\infty,r}\le \sum_{p,q=2}^{N}1_{p,q\in
 2\N}1_{p+q=m}\|V_{p,q}^{0-2,l+1}\|_{1,\infty,r},\label{eq_kernel_divided_relation}
\end{align}
which is based on the uniqueness of anti-symmetric kernel. By using
 these inequalities we will prove the claimed bound properties of the
 Grassmann data at $l$-th scale in the following. 

First of all, let us consider $V^{0-1-1,l,(n)}$. By
 \eqref{eq_simple_1_infinity_1} and
 \eqref{eq_scale_covariance_determinant_bound}, for any $m\in
 \{0,2,\cdots,N\}$ 
\begin{align*}
\|V_m^{0-1-1,l,(1)}\|_{1,\infty,r}\le
 \sum_{p=m}^N\left(\frac{N}{h}\right)^{1_{m=0}\land p\neq 0}
 \left(\begin{array}{c}p \\ m \end{array}\right)
 (c_0M^{\sa(l+1-\hat{N}_{\beta})})^{\frac{p-m}{2}}\|V_p^{0-1,l+1}\|_{1,\infty,r}.
\end{align*}
Then by \eqref{eq_0_1_bound_assumption} and \eqref{eq_0_1_bound} for
 $l+1$
\begin{align}
\|V_0^{0-1-1,l,(1)}\|_{1,\infty,r}&\le
 \|V_0^{0-1,l+1}\|_{1,\infty,r}+\frac{N}{h}c
 \alpha^{-2}L^{-d}\label{eq_0_1_1_0_1}\\
&\le
 c\frac{N}{h}\alpha^{-2}M^{-(\sum_{j=1}^d\frac{1}{\sn_j}+1)(l+1-\hat{N}_{\beta})}L^{-d}.\notag
\end{align}
Also by \eqref{eq_0_1_bound} for $l+1$ and the assumptions that
 $\alpha\ge 2$, $M^{\sa}\ge 2^4$, 
\begin{align}
&\sum_{m=2}^N\alpha^m(c_0M^{\sa(l-\hat{N}_{\beta})})^{\frac{m}{2}}\|V_m^{0-1-1,l,(1)}\|_{1,\infty,r}\label{eq_0_1_1_1}\\
&\le \sum_{m=2}^N\sum_{p=m}^N\alpha^m2^p M^{-\frac{\sa m}{2}}(c_0
 M^{\sa(l+1-\hat{N}_{\beta})})^{\frac{p}{2}}\|V_p^{0-1,l+1}\|_{1,\infty,r}\notag\\
&\le \sum_{m=2}^N2^{m}M^{-\frac{\sa
 m}{2}}\sum_{p=m}^N\alpha^p(c_0M^{\sa(l+1-\hat{N}_{\beta})})^{\frac{p}{2}}\|V_p^{0-1,l+1}\|_{1,\infty,r}\notag\\
&\le c L^{-d}M^{-\sa}.\notag
\end{align}
Assume that $n\in \N_{\ge 2}$. Observe that
\begin{align*}
&V^{0-1-1,l,(n)}(u)(\psi)\\
&=\sum_{q=1}^n\left(\begin{array}{c} n \\ q
		    \end{array}\right)\frac{1}{n!}Tree(\{1,2,\cdots,n\},\cC_{l+1})\\
&\quad\cdot \prod_{j=1}^qV^{0-1,l+1}(u)(\psi^j+\psi)\prod_{k=q+1}^nV^{0-2,l+1}(u)(\psi^k+\psi)\Bigg|_{\psi^j=0\atop (\forall j\in \{1,2,\cdots,n\})}.
\end{align*}
It follows from
 \eqref{eq_simple_1_infinity},
 \eqref{eq_scale_covariance_determinant_bound},
 \eqref{eq_scale_covariance_decay_bound} that for any $m\in \{0,2,\cdots,N\}$
\begin{align}
&\|V_m^{0-1-1,l,(n)}\|_{1,\infty,r}\label{eq_0_1_1_expansion}\\
&\le
 \left(\frac{N}{h}\right)^{1_{m=0}}2^n
\frac{(n-2)!}{n!}(c_0M^{\sa(l+1-\hat{N}_{\beta})})^{-n+1-\frac{m}{2}}2^{-2m}
(c_0M^{(\sa-1-\sum_{j=1}^d\frac{1}{\sn_j})(l+1-\hat{N}_{\beta})})^{n-1}\notag\\
&\quad\cdot
 \sum_{p_1=2}^N2^{3p_1}(c_0M^{\sa(l+1-\hat{N}_{\beta})})^{\frac{p_1}{2}}\|V_{p_1}^{0-1,l+1}\|_{1,\infty,r}\notag\\
&\quad\cdot \prod_{j=2}^n\left(\sum_{p_j=2}^N2^{3p_j}(c_0M^{\sa(l+1-\hat{N}_{\beta})})^{\frac{p_j}{2}}\sum_{\delta\in\{1,2\}}
\|V_{p_j}^{0-\delta,l+1}\|_{1,\infty,r}\right)
1_{\sum_{j=1}^np_j-2(n-1)\ge m}\notag\\
&\le 
 \left(\frac{N}{h}\right)^{1_{m=0}}2^{-2m+n}c_0^{-\frac{m}{2}}
M^{-\frac{m}{2}\sa(l+1-\hat{N}_{\beta})-(\sum_{j=1}^d\frac{1}{\sn_j}+1)(l+1-\hat{N}_{\beta})(n-1)}\notag\\
&\quad\cdot 
\sum_{p_1=2}^N2^{3p_1}(c_0M^{\sa(l+1-\hat{N}_{\beta})})^{\frac{p_1}{2}}\|V_{p_1}^{0-1,l+1}\|_{1,\infty,r}\notag\\
&\quad\cdot \prod_{j=2}^n\left(\sum_{p_j=2}^N2^{3p_j}(c_0M^{\sa(l+1-\hat{N}_{\beta})})^{\frac{p_j}{2}}\sum_{\delta\in\{1,2\}}
\|V_{p_j}^{0-\delta,l+1}\|_{1,\infty,r}\right)
1_{\sum_{j=1}^np_j-2(n-1)\ge m}.\notag
\end{align}
By \eqref{eq_0_1_bound_assumption},
 \eqref{eq_0_2_kernel_bound_assumption} and the assumption 
 $L^{d}\ge
 M^{(\sum_{j=1}^d\frac{1}{\sn_j}+1)(\hat{N}_{\beta}-N_{\beta})}$,
\begin{align*}
&\|V_0^{0-1-1,l,(n)}\|_{1,\infty,r}\\
&\le
 \frac{N}{h}M^{-(\sum_{j=1}^d\frac{1}{\sn_j}+1)(l+1-\hat{N}_{\beta})(n-1)}
c\alpha^{-2}L^{-d}\left(
c\alpha^{-2}L^{-d}
+c \alpha^{-4}M^{(\sum_{j=1}^d\frac{1}{\sn_j}+1)(l+1-\hat{N}_{\beta})} 
\right)^{n-1}\\
&\le
 \frac{N}{h}L^{-d}(c\alpha^{-2})^n,\end{align*}
or by assuming that $\alpha\ge c$, 
\begin{align}
\sum_{n=2}^{\infty}\|V_0^{0-1-1,l,(n)}\|_{1,\infty,r}\le
 c\frac{N}{h}\alpha^{-4}L^{-d}.\label{eq_0_1_1_0_higher}
\end{align}
Also, by substituting \eqref{eq_0_1_bound_weight_assumption},
 \eqref{eq_0_2_kernel_bound_weight_assumption} and the inequality 
 $L^{d}\ge
 M^{(\sum_{j=1}^d\frac{1}{\sn_j}+1)(\hat{N}_{\beta}-N_{\beta})}$ and
 using the condition $\alpha M^{-\frac{\sa}{2}}\ge 2^3$ we can derive from
 \eqref{eq_0_1_1_expansion} that
\begin{align*}
&\sum_{m=2}^N\alpha^m(c_0M^{\sa
 (l-\hat{N}_{\beta})})^{\frac{m}{2}}\|V_m^{0-1-1,l,(n)}\|_{1,\infty,r}\\
&\le c^n
 \alpha^{-2(n-1)}M^{\sa(n-1)-(\sum_{j=1}^d\frac{1}{\sn_j}+1)(l+1-\hat{N}_{\beta})(n-1)}\\
&\quad\cdot \sum_{p_1=2}^{N}2^{p_1}\alpha^{p_1}(c_0M^{\sa
 (l-\hat{N}_{\beta})})^{\frac{p_1}{2}}\|V_{p_1}^{0-1,l+1}\|_{1,\infty,r}\\
&\quad\cdot \left(\sum_{p=2}^N2^p \alpha^p (c_0M^{\sa
 (l-\hat{N}_{\beta})})^{\frac{p}{2}}\sum_{\delta\in
 \{1,2\}}\|V_p^{0-\delta,l+1}\|_{1,\infty,r}\right)^{n-1}\\
&\le c^n
 \alpha^{-2(n-1)}M^{\sa(n-1)-(\sum_{j=1}^d\frac{1}{\sn_j}+1)(l+1-\hat{N}_{\beta})(n-1)}M^{-\sa}L^{-d}\\
&\quad\cdot 
\left(c
 M^{-\sa}L^{-d}
+c M^{-2\sa+(\sum_{j=1}^d\frac{1}{\sn_j}+1)(l+1-\hat{N}_{\beta})}
\right)^{n-1}\\
&\le M^{-\sa}L^{-d}
(c\alpha^{-2})^{n-1},
\end{align*}
or by assuming that $\alpha\ge c$,
\begin{align}
\sum_{m=2}^N\alpha^m(c_0M^{\sa(l-\hat{N}_{\beta})})^{\frac{m}{2}}\sum_{n=2}^{\infty} 
\|V_m^{0-1-1,l,(n)}\|_{1,\infty,r}\le
 c\alpha^{-2}M^{-\sa}L^{-d}.\label{eq_0_1_1_higher}
\end{align}

Next let us study $V^{0-1-2,l,(n)}$. By \eqref{eq_double_1_infinity_1}
 and \eqref{eq_scale_covariance_determinant_bound}
\begin{align*}
&\|V_m^{0-1-2,l,(1)}\|_{1,\infty,r}\\
&\le \left(\frac{N}{h}\right)^{1_{m=0}}(c_0M^{\sa
 (l+1-\hat{N}_{\beta})})^{-1-\frac{m}{2}}\\
&\quad\cdot \sum_{p_1,p_2=2}^{N}1_{p_1,p_2\in 2\N}2^{2p_1+2p_2}(c_0M^{\sa(l+1-\hat{N}_{\beta})})^{\frac{p_1+p_2}{2}}[V_{p_1,p_2}^{0-2,l+1},\tilde{\cC}_{l+1}]_{1,\infty,r}1_{p_1+p_2-2\ge m}.
\end{align*}
Then by \eqref{eq_scale_covariance_coupled_decay_bound} and 
\eqref{eq_0_2_multiplied_bound_assumption}
\begin{align}
&\|V_0^{0-1-2,l,(1)}\|_{1,\infty,r}\label{eq_0_1_2_0_1}\\
&\le  \frac{N}{h}c_0^{-1}M^{-\sa (l+1-\hat{N}_{\beta})}
c
 \alpha^{-4}M^{(\sum_{j=1}^d\frac{1}{\sn_j}+1)(l+1-\hat{N}_{\beta})}L^{-d}
c_0 M^{(\sa-1-\sum_{j=1}^d\frac{1}{\sn_j})(l+1-\hat{N}_{\beta})}\notag\\
&\le c\frac{N}{h}\alpha^{-4}L^{-d}.\notag
\end{align}
Also by \eqref{eq_scale_covariance_coupled_decay_bound},
 \eqref{eq_0_2_multiplied_bound_weight_assumption} and the condition
 $\alpha M^{-\frac{\sa}{2}}\ge 2$,
\begin{align}
&\sum_{m=2}^N\alpha^m (c_0M^{\sa
 (l-\hat{N}_{\beta})})^{\frac{m}{2}}\|V_m^{0-1-2,l,(1)}\|_{1,\infty,r}\label{eq_0_1_2_1}\\
&\le c \alpha^{-2}M^{\sa} (c_0M^{\sa(l+1-\hat{N}_{\beta})})^{-1}\notag\\
&\quad\cdot  \sum_{p_1,p_2=2}^{N}1_{p_1,p_2\in2\N}2^{2p_1+2p_2}\alpha^{p_1+p_2}(c_0M^{\sa
 (l-\hat{N}_{\beta})})^{\frac{p_1+p_2}{2}}[V_{p_1,p_2}^{0-2,l+1},\tilde{\cC}_{l+1}]_{1,\infty,r}\notag\\
&\le c \alpha^{-2}M^{\sa} (c_0M^{\sa
 (l+1-\hat{N}_{\beta})})^{-1}M^{-2\sa +
 (\sum_{j=1}^d\frac{1}{\sn_j}+1)(l+1-\hat{N}_{\beta})}L^{-d} c_0 M^{(\sa-1
 -\sum_{j=1}^d\frac{1}{\sn_j})(l+1-\hat{N}_{\beta})}\notag\\
&\le c\alpha^{-2} M^{-\sa}L^{-d}.\notag
\end{align}
Let $n\in \N_{\ge 2}$. By substituting
 \eqref{eq_scale_covariance_determinant_bound}, 
\eqref{eq_scale_covariance_decay_bound} into
 \eqref{eq_double_1_infinity} we have that
\begin{align}
&\|V_m^{0-1-2,l,(n)}\|_{1,\infty,r}\label{eq_0_1_2_expansion}\\
&\le
 \left(\frac{N}{h}\right)^{1_{m=0}}(c_0M^{\sa(l+1-\hat{N}_{\beta})})^{-n-\frac{m}{2}}2^{-2m}(c_0M^{(\sa-1-\sum_{j=1}^d\frac{1}{\sn_j})(l+1-\hat{N}_{\beta})})^{n-1}\notag\\
&\quad\cdot \sum_{p_1,p_2=2}^{N}1_{p_1,p_2\in 2\N} 2^{3p_1+3p_2}(c_0
 M^{\sa(l+1-\hat{N}_{\beta})})^{\frac{p_1+p_2}{2}}[V_{p_1,p_2}^{0-2,l+1},\tilde{\cC}_{l+1}]_{1,\infty,r}\notag\\
&\quad\cdot \prod_{j=3}^{n+1}\left(
\sum_{p_j=4}^N2^{3p_j}(c_0M^{\sa(l+1-\hat{N}_{\beta})})^{\frac{p_j}{2}}\|V_{p_j}^{0-2,l+1}\|_{1,\infty,r}\right)1_{\sum_{j=1}^{n+1}p_j-2n\ge
 m}.\notag
\end{align}
Then by \eqref{eq_0_2_kernel_bound_assumption},
 \eqref{eq_0_2_multiplied_bound_assumption} and
 \eqref{eq_scale_covariance_coupled_decay_bound},
\begin{align*}
&\|V_0^{0-1-2,l,(n)}\|_{1,\infty,r}\\
&\le \frac{N}{h}c_0^{-1}M^{-\sa (l+1-\hat{N}_{\beta})n+(\sa -1 -
 \sum_{j=1}^d\frac{1}{\sn_j})(l+1-\hat{N}_{\beta})(n-1)}c \alpha^{-4}
 M^{(\sum_{j=1}^d\frac{1}{\sn_j}+1)(l+1-\hat{N}_{\beta})}L^{-d}\\
&\quad\cdot c_0
 M^{(\sa-1-\sum_{j=1}^{d}\frac{1}{\sn_j})(l+1-\hat{N}_{\beta})}(c\alpha^{-4}  M^{(\sum_{j=1}^d\frac{1}{\sn_j}+1)(l+1-\hat{N}_{\beta})})^{n-1}\\
&\le \frac{N}{h} L^{-d}(c\alpha^{-4})^n,
\end{align*}
or by assuming that $\alpha\ge c$
\begin{align}
&\sum_{n=2}^{\infty}\|V_0^{0-1-2,l,(n)}\|_{1,\infty,r}\le
 c\frac{N}{h}\alpha^{-8}L^{-d}.\label{eq_0_1_2_0_higher}
\end{align}
On the other hand, by using
 \eqref{eq_0_2_kernel_bound_weight_assumption}, 
\eqref{eq_0_2_multiplied_bound_weight_assumption},
\eqref{eq_scale_covariance_coupled_decay_bound} and the condition
 $\alpha M^{-\frac{\sa}{2}}\ge 2^3$ we can deduce from
 \eqref{eq_0_1_2_expansion} that
\begin{align*}
&\sum_{m=2}^N\alpha^m(c_0M^{\sa
 (l-\hat{N}_{\beta})})^{\frac{m}{2}}\|V_m^{0-1-2,l,(n)}\|_{1,\infty,r}\\
&\le c^n \alpha^{-2n}c_0^{-1}M^{\sa n -\sa (l+1-\hat{N}_{\beta})n + (\sa
 - 1 -\sum_{j=1}^{d}\frac{1}{\sn_j})(l+1-\hat{N}_{\beta})(n-1)}\\
&\quad\cdot \sum_{p_1,p_2=2}^{N} 1_{p_1,p_2\in
 2\N}2^{p_1+p_2}\alpha^{p_1+p_2}(c_0M^{\sa
 (l-\hat{N}_{\beta})})^{\frac{p_1+p_2}{2}}[V_{p_1,p_2}^{0-2,l+1},\tilde{\cC}_{l+1}]_{1,\infty,r}\\
&\quad\cdot \left(\sum_{p=4}^N2^p \alpha^p
 (c_0M^{\sa(l-\hat{N}_{\beta})})^{\frac{p}{2}}\|V_p^{0-2,l+1}\|_{1,\infty,r}\right)^{n-1}\\
&\le c^n \alpha^{-2n}c_0^{-1}M^{\sa n -
 \sa(l+1-\hat{N}_{\beta})n+(\sa-1-\sum_{j=1}^d\frac{1}{\sn_j})(l+1-\hat{N}_{\beta})(n-1)}
M^{-2\sa +
 (\sum_{j=1}^d\frac{1}{\sn_j}+1)(l+1-\hat{N}_{\beta})}L^{-d}\\
&\quad\cdot c_0
 M^{(\sa-1-\sum^{d}_{j=1}\frac{1}{\sn_j})(l+1-\hat{N}_{\beta})}(c M^{-2\sa
 + (\sum_{j=1}^d\frac{1}{\sn_j}+1)(l+1-\hat{N}_{\beta})})^{n-1}\\
&\le L^{-d}(c\alpha^{-2}M^{-\sa})^n.
\end{align*}
Thus on the assumption $\alpha\ge c$,
\begin{align}
\sum_{m=2}^N\alpha^m(c_0M^{\sa(l-\hat{N}_{\beta})})^{\frac{m}{2}}\sum_{n=2}^{\infty}\|V_m^{0-1-2,l,(n)}\|_{1,\infty,r}
\le c \alpha^{-4}M^{-2\sa}L^{-d}.\label{eq_0_1_2_higher}
\end{align}

By combining \eqref{eq_0_1_1_0_1}, \eqref{eq_0_1_1_1},
 \eqref{eq_0_1_1_0_higher}, \eqref{eq_0_1_1_higher},
 \eqref{eq_0_1_2_0_1}, \eqref{eq_0_1_2_1}, \eqref{eq_0_1_2_0_higher},
 \eqref{eq_0_1_2_higher} we have that
\begin{align*}
&\frac{h}{N}\sum_{n=1}^{\infty}(\|V_0^{0-1-1,l,(n)}\|_{1,\infty,r}
+ \|V_0^{0-1-2,l,(n)}\|_{1,\infty,r})\\
&\le c (\alpha^{-2}
 M^{-(\sum_{j=1}^d\frac{1}{\sn_j}+1)(l+1-\hat{N}_{\beta})}
+ \alpha^{-4})L^{-d}\\
&\le c
 (M^{-(\sum_{j=1}^d\frac{1}{\sn_j}+1)}+\alpha^{-2})\alpha^{-2}M^{-(\sum_{j=1}^d\frac{1}{\sn_j}+1)(l-\hat{N}_{\beta})}L^{-d},\\
&\sum_{m=2}^N\alpha^m(c_0M^{\sa(l-\hat{N}_{\beta})})^{\frac{m}{2}}\sum_{n=1}^{\infty}(\|V_m^{0-1-1,l,(n)}\|_{1,\infty,r}
+\|V_m^{0-1-2,l,(n)}\|_{1,\infty,r})\\
&\le c M^{-\sa}L^{-d}.
\end{align*}
The above inequalities imply that if $M \ge c$ and $\alpha\ge c$, 
$$
V^{0-1-1,l}+V^{0-1-2,l}\in \cQ(r,l).
$$
In fact Lemma \ref{lem_simple_tree_expansion}, Lemma
 \ref{lem_double_tree_expansion} ensure that the anti-symmetric kernels
 of $V^{0-1-1,l,(n)}+V^{0-1-2,l,(n)}$ satisfy
 \eqref{eq_temperature_translation_invariance} and thus so do the
 anti-symmetric kernels of $V^{0-1-1,l}+V^{0-1-2,l}$. The above uniform
 convergent property implies the claimed regularity of
 $V^{0-1-1,l}+V^{0-1-2,l}$ with $u\in \overline{D(r)}$.
Recall that we have also assumed $L^d\ge
 M^{(\sum_{j=1}^d\frac{1}{\sn_j}+1)(\hat{N}_{\beta}-N_{\beta})}$,
 $\alpha M^{-\frac{\sa}{2}}\ge 2^3$ to reach this conclusion.

Next let us deal with $V^{0-2,l,(n)}$. The analysis is based on Lemma
 \ref{lem_divided_tree_expansion}. The lemma ensures the existence of
 bi-anti-symmetric kernels satisfying
 \eqref{eq_temperature_translation_invariance},
 \eqref{eq_temperature_integral_vanish}.
We can see from \eqref{eq_divided_kernel} and the induction hypothesis
 that $V_{a,b}^{0-2,l,(n)}\in C(\overline{D(r)},\Map(I^a\times
 I^b,\C))\cap C^{\o}(D(r),\Map(I^a\times I^b,\C))$ $(n\in\N,\ a,b\in
 \{2,4,\cdots,N\})$, which implies that 
$V^{0-2,l,(n)}\in C(\overline{D(r)},\bigwedge_{even}\cV)\cap
 C^{\o}(D(r),\bigwedge_{even}\cV)$. We need to establish bound
 properties of $V^{0-2,l,(n)}$. It follows from
 \eqref{eq_divided_1_infinity_1} and
 \eqref{eq_scale_covariance_determinant_bound} that for $a,b\in
 \{2,4,\cdots,N\}$ 
\begin{align*}
&\|V_{a,b}^{0-2,l,(1)}\|_{1,\infty,r}\le
 \sum_{p=a}^{N}\sum_{q=b}^{N}
\left(\begin{array}{c}p \\ a\end{array}\right)
\left(\begin{array}{c}q \\ b\end{array}\right)
 (c_0M^{\sa(l+1-\hat{N}_{\beta})})^{\frac{1}{2}(p+q-a-b)}\|V_{p,q}^{0-2,l+1}\|_{1,\infty,r}.\end{align*}
By \eqref{eq_0_2_bound} for $l+1$ and the assumption $\alpha\ge 2$,
 $M^{\sa}\ge 2^6$,
\begin{align}
&\|V_{2,2}^{0-2,l,(1)}\|_{1,\infty,r}\label{eq_0_2_2_2_1}\\
&\le
 \|V_{2,2}^{0-2,l+1}\|_{1,\infty,r}+\sum_{p=2}^{N}\sum_{q=2}^{N}1_{p+q\ge 6} 2^{p+q}(c_0M^{\sa(l+1-\hat{N}_{\beta})})^{\frac{1}{2}(p+q-4)}\|V_{p,q}^{0-2,l+1}\|_{1,\infty,r}\notag\\&\le
 M^{(\sum_{j=1}^d\frac{1}{\sn_j}+1)(l+1-\hat{N}_{\beta})}\alpha^{-4}(c_0M^{\sa (l+1-\hat{N}_{\beta})})^{-2}\notag\\
&\quad + c \alpha^{-6}
 M^{(\sum_{j=1}^d\frac{1}{\sn_j}+1)(l+1-\hat{N}_{\beta})}(c_0 M^{\sa
 (l+1-\hat{N}_{\beta})})^{-2}\notag\\
&\le c
 M^{\sum_{j=1}^d\frac{1}{\sn_j}+1-2\sa+(\sum_{j=1}^d\frac{1}{\sn_j}+1)(l-\hat{N}_{\beta})}\alpha^{-4}(c_0M^{\sa (l-\hat{N}_{\beta})})^{-2},\notag\\
&\sum_{a,b=2}^{N}1_{a,b\in 2\N}1_{a+b\ge 6}\alpha^{a+b}(c_0M^{\sa(l-\hat{N}_{\beta})})^{\frac{a+b}{2}}\|V_{a,b}^{0-2,l,(1)}\|_{1,\infty,r}\label{eq_0_2_1_higher}\\
&\le \sum_{a,b=2}^{N}1_{a,b\in 2\N}1_{a+b\ge 6}\alpha^{a+b}
 M^{-\frac{\sa}{2}(a+b)}\sum_{p=a}^{N}\sum_{q=b}^{N}2^{p+q}(c_0M^{\sa(l+1-\hat{N}_{\beta})})^{\frac{p+q}{2}}\|V_{p,q}^{0-2,l+1}\|_{1,\infty,r}\notag\\
&\le
 M^{(\sum_{j=1}^d\frac{1}{\sn_j}+1)(l+1-\hat{N}_{\beta})}\sum_{a,b=2}^{N}1_{a+b\ge 6} 2^{a+b} M^{-\frac{\sa}{2}(a+b)}\notag\\
&\le
 M^{(\sum_{j=1}^d\frac{1}{\sn_j}+1)(l+1-\hat{N}_{\beta})}\sum_{m=6}^{\infty}2^{2m}M^{-\frac{\sa m}{2}}\notag\\
&\le c
 M^{(\sum_{j=1}^d\frac{1}{\sn_j}+1)(l+1-\hat{N}_{\beta})-3\sa}.\notag
\end{align}
By applying \eqref{eq_divided_multiplied_1_infinity_1},
 \eqref{eq_0_2_multiplied_bound} in place of
 \eqref{eq_divided_1_infinity_1}, \eqref{eq_0_2_bound} respectively and
 repeating a parallel argument to the above argument we can derive on the
 assumption $\alpha\ge 2$, $M^{\sa}\ge 2^6$ that
\begin{align}
&[V_{2,2}^{0-2,l,(1)},g]_{1,\infty,r}\le c
 M^{\sum_{j=1}^d\frac{1}{\sn_j}+1-2\sa+(\sum_{j=1}^d\frac{1}{\sn_j}+1)(l-\hat{N}_{\beta})}
 \alpha^{-4}(c_0M^{\sa
 (l-\hat{N}_{\beta})})^{-2}L^{-d}\|g\|,\label{eq_0_2_2_2_1_multiplied}\\
&\sum_{a,b=2}^{N}1_{a,b\in2\N}1_{a+b\ge
 6}\alpha^{a+b}(c_0M^{\sa(l-\hat{N}_{\beta})})^{\frac{a+b}{2}}[V_{a,b}^{0-2,l,(1)},g]_{1,\infty,r}\label{eq_0_2_1_higher_multiplied}\\ 
&\le c M^{(\sum_{j=1}^d\frac{1}{\sn_j}+1)(l+1-\hat{N}_{\beta})-3\sa}L^{-d}\|g\|, \notag
\end{align}
for any anti-symmetric function $g:I^2\to \C$. 

Let us take $n\in \N_{\ge 2}$. Observe that for any $m\in
 \{0,1,\cdots,n-1\}$
\begin{align}
\sharp S(n,m)=\left(\begin{array}{c} n-1 \\ m\end{array}\right).\label{eq_divided_tree}
\end{align}
Combination of \eqref{eq_divided_1_infinity},
 \eqref{eq_scale_covariance_determinant_bound}, 
\eqref{eq_scale_covariance_decay_bound}, 
\eqref{eq_divided_tree} 
yields that for $a,b\in \{2,4,\cdots,N\}$
\begin{align*}
&\|V_{a,b}^{0-2,l,(n)}\|_{1,\infty,r}\\
&\le \frac{1}{n!}\sum_{m=0}^{n-1}\left(\begin{array}{c}n-1\\ m
				       \end{array}\right)
(1_{m\neq 0}(m-1)!+1_{m=0})(1_{m\neq
 n-1}(n-m-2)!+1_{m=n-1})\\
&\quad \cdot
 2^{-2a-2b}(c_0M^{\sa(l+1-\hat{N}_{\beta})})^{-n+1-\frac{1}{2}(a+b)}(c_0M^{(\sa-1-\sum_{j=1}^d\frac{1}{\sn_j})(l+1-\hat{N}_{\beta})})^{n-1}\\
 &\quad\cdot \sum_{p_1,q_1=2}^{N}1_{p_1,q_1\in
 2\N}2^{3p_1+3q_1}(c_0M^{\sa(l+1-\hat{N}_{\beta})})^{\frac{p_1+q_1}{2}}\|V_{p_1,q_1}^{0-2,l+1}\|_{1,\infty,r}\\
&\quad\cdot
 \prod_{j=2}^{m+1}\left(\sum_{p_j=4}^N2^{3p_j}(c_0M^{\sa(l+1-\hat{N}_{\beta})})^{\frac{p_j}{2}}\|V_{p_j}^{0-2,l+1}\|_{1,\infty,r}\right)\\
&\quad\cdot
 \prod_{k=2}^{n-m}\left(\sum_{q_k=4}^N2^{3q_k}(c_0M^{\sa(l+1-\hat{N}_{\beta})})^{\frac{q_k}{2}}\|V_{q_k}^{0-2,l+1}\|_{1,\infty,r}\right)\\
&\quad\cdot 1_{\sum_{j=1}^{m+1}p_j-2m\ge
 a}1_{\sum_{k=1}^{n-m}q_k-2(n-m-1)\ge b}.
\end{align*}
Then by using \eqref{eq_0_2_bound_weight_assumption},
 \eqref{eq_0_2_kernel_bound_weight_assumption} and the assumption
 $\alpha M^{-\frac{\sa}{2}}\ge 2^3$,
\begin{align}
&\sum_{a,b=2}^{N}\alpha^{a+b}(c_0M^{\sa
 (l-\hat{N}_{\beta})})^{\frac{a+b}{2}}\|V_{a,b}^{0-2,l,(n)}\|_{1,\infty,r}\label{eq_0_2_higher_pre}\\
&\le c^n \alpha^{-2(n-1)}M^{\sa(n-1)-\sa(l+1-\hat{N}_{\beta})(n-1)
+(\sa-1-\sum_{j=1}^d\frac{1}{\sn_j})(l+1-\hat{N}_{\beta})(n-1)}\notag\\
&\quad\cdot 
\sum_{p_1,q_1=2}^{N}1_{p_1,q_1\in
 2\N}2^{p_1+q_1}\alpha^{p_1+q_1}
(c_0M^{\sa(l-\hat{N}_{\beta})})^{\frac{p_1+q_1}{2}}\|V_{p_1,q_1}^{0-2,l+1}\|_{1,\infty,r}\notag\\
&\quad\cdot \left(\sum_{p=4}^N 2^p\alpha^p
 (c_0M^{\sa(l-\hat{N}_{\beta})})^{\frac{p}{2}}\|V_p^{0-2,l+1}\|_{1,\infty,r}\right)^{n-1}\notag\\
&\le c^n
 \alpha^{-2(n-1)}M^{-\sa(l-\hat{N}_{\beta})(n-1)+(\sa-1-\sum_{j=1}^d\frac{1}{\sn_j})(l+1-\hat{N}_{\beta})(n-1)}M^{-2\sa
 n+ (\sum_{j=1}^d\frac{1}{\sn_j}+1)(l+1-\hat{N}_{\beta})n}\notag\\
&\le  M^{-2\sa+(\sum_{j=1}^d\frac{1}{\sn_j}+1)(l+1-\hat{N}_{\beta})}(c\alpha^{-2}M^{-\sa})^{n-1}.\notag
\end{align}
Then on the assumption $\alpha\ge c$, 
\begin{align}
&\sum_{a,b=2}^{N}\alpha^{a+b}(c_0M^{\sa
 (l-\hat{N}_{\beta})})^{\frac{a+b}{2}}\sum_{n=2}^{\infty}\|V_{a,b}^{0-2,l,(n)}\|_{1,\infty,r}\le
 c
 \alpha^{-2}M^{-3\sa+(\sum_{j=1}^d\frac{1}{\sn_j}+1)(l+1-\hat{N}_{\beta})}.\label{eq_0_2_higher}
\end{align}
Take any anti-symmetric function $g:I^2\to \C$. We can derive from 
\eqref{eq_divided_multiplied_1_infinity}, 
\eqref{eq_scale_covariance_determinant_bound}, 
\eqref{eq_scale_covariance_decay_bound} that for $a,b\in \{2,4,\cdots,N\}$
\begin{align*}
&[V_{a,b}^{0-2,l,(n)},g]_{1,\infty,r}\\
&\le \frac{1}{n!}\sum_{m=0}^{n-1}\left(\begin{array}{c}n-1\\ m
				       \end{array}\right)
(1_{m\neq 0}(m-1)!+1_{m=0})(1_{m\neq
 n-1}(n-m-2)!+1_{m=n-1})\\
&\quad \cdot
 2^{-2a-2b}(c_0M^{\sa(l+1-\hat{N}_{\beta})})^{-n+1-\frac{1}{2}(a+b)}(c_0M^{(\sa-1-\sum_{j=1}^d\frac{1}{\sn_j})(l+1-\hat{N}_{\beta})})^{n-2}\\
 &\quad\cdot \sum_{p_1,q_1=2}^{N}1_{p_1,q_1\in
 2\N}2^{3p_1+3q_1}(c_0M^{\sa(l+1-\hat{N}_{\beta})})^{\frac{p_1+q_1}{2}}\\
&\quad\cdot \left([V_{p_1,q_1}^{0-2,l+1},g]_{1,\infty,r}c_0M^{(\sa-1-\sum_{j=1}^d\frac{1}{\sn_j})(l+1-\hat{N}_{\beta})}
+
[V_{p_1,q_1}^{0-2,l+1},\tilde{\cC}]_{1,\infty,r}\|g\|_{1,\infty}
\right)\\
&\quad\cdot
 \prod_{j=2}^{m+1}\left(\sum_{p_j=4}^N2^{3p_j}(c_0M^{\sa(l+1-\hat{N}_{\beta})})^{\frac{p_j}{2}}\|V_{p_j}^{0-2,l+1}\|_{1,\infty,r}\right)\\
&\quad\cdot
 \prod_{k=2}^{n-m}\left(\sum_{q_k=4}^N2^{3q_k}(c_0M^{\sa(l+1-\hat{N}_{\beta})})^{\frac{q_k}{2}}\|V_{q_k}^{0-2,l+1}\|_{1,\infty,r}\right)\\
&\quad\cdot 1_{\sum_{j=1}^{m+1}p_j-2m\ge
 a}1_{\sum_{k=1}^{n-m}q_k-2(n-m-1)\ge b}.
\end{align*}
Then by using \eqref{eq_scale_covariance_coupled_decay_bound}, 
\eqref{eq_0_2_kernel_bound_weight_assumption},
\eqref{eq_0_2_multiplied_bound_weight_assumption},
\eqref{eq_divided_tree}, the inequality $\|g\|_{1,\infty}\le \|g\|$ and the
 assumptions $\alpha M^{-\frac{\sa}{2}}\ge
 2^3$, $\alpha \ge c$ and repeating a parallel procedure to
 \eqref{eq_0_2_higher_pre} we obtain that  
\begin{align}
\sum_{a,b=2}^{N}\alpha^{a+b}(c_0M^{\sa
 (l-\hat{N}_{\beta})})^{\frac{a+b}{2}}\sum_{n=2}^{\infty}[V_{a,b}^{0-2,l,(n)},g]_{1,\infty,r}\le
 c\alpha^{-2}
 M^{-3\sa+(\sum_{j=1}^d\frac{1}{\sn_j}+1)(l+1-\hat{N}_{\beta})}L^{-d}\|g\|.\label{eq_0_2_multiplied_higher}
\end{align}

Here we can sum up \eqref{eq_0_2_2_2_1}, \eqref{eq_0_2_1_higher},
 \eqref{eq_0_2_higher} to deduce that
\begin{align*}
&M^{-(\sum_{j=1}^d\frac{1}{\sn_j}+1)(l-\hat{N}_{\beta})}\sum_{a,b=2}^{N}\alpha^{a+b}(c_0M^{\sa(l-\hat{N}_{\beta})})^{\frac{a+b}{2}}\sum_{n=1}^{\infty}\|V_{a,b}^{0-2,l,(n)}\|_{1,\infty,r}\\
&\le
 c\left(M^{\sum_{j=1}^d\frac{1}{\sn_j}+1-2\sa}+\alpha^{-2}M^{\sum_{j=1}^d\frac{1}{\sn_j}+1-3\sa}\right).
\end{align*}
On the assumption $M^{2\sa-1-\sum_{j=1}^d\frac{1}{\sn_j}}\ge c$ the
 right-hand side of the above inequality is less than 1.
Because of the assumption \eqref{eq_dispersion_power}, the condition
 $M^{2\sa-1-\sum_{j=1}^d\frac{1}{\sn_j}}\ge c$ can be realized by taking
 $M$ large. Similarly, it follows from \eqref{eq_0_2_2_2_1_multiplied},
 \eqref{eq_0_2_1_higher_multiplied}, \eqref{eq_0_2_multiplied_higher} 
and the condition
 $M^{2\sa-1-\sum_{j=1}^d\frac{1}{\sn_j}}\ge c$ that
\begin{align*}
&M^{-(\sum_{j=1}^d\frac{1}{\sn_j}+1)(l-\hat{N}_{\beta})}\sum_{a,b=2}^{N}\alpha^{a+b}(c_0M^{\sa(l-\hat{N}_{\beta})})^{\frac{a+b}{2}}\sum_{n=1}^{\infty}[V_{a,b}^{0-2,l,(n)},g]_{1,\infty,r}\\
&\le
 c\left(M^{\sum_{j=1}^d\frac{1}{\sn_j}+1-2\sa}+\alpha^{-2}M^{\sum_{j=1}^d\frac{1}{\sn_j}+1-3\sa}\right)L^{-d}\|g\|\\
&\le L^{-d}\|g\|,
\end{align*}
for any anti-symmetric function $g:I^2\to \C$. Thus we conclude that on
 the assumption $M^{2\sa-1-\sum_{j=1}^d\frac{1}{\sn_j}}\ge c$ that 
$$
V^{0-2,l}\in\cR(r,l).
$$
We needed to assume in total that 
\begin{align*}
M \ge c,\quad M^{2\sa -1 -\sum_{j=1}^d\frac{1}{\sn_j}}\ge c,\quad
 \alpha\ge c M^{\frac{\sa}{2}},\quad L^d\ge
 M^{(\sum_{j=1}^d\frac{1}{\sn_j}+1)(\hat{N}_{\beta}-N_{\beta})}
\end{align*}
for a positive constant $c$ independent of any parameter, 
in order to conclude the $l$-th step. The above assumptions can be
 summarized as in \eqref{eq_assumptions_IR_integration_without}.

The induction with $l$ proves that the claim holds true.
\end{proof}

\begin{remark}\label{rem_dispersion_power}
In the proof of the claim $V^{0-2,l}\in\cR(r,l)$ we crucially used the
 condition \eqref{eq_dispersion_power}.
\end{remark}

\subsection{Multi-scale integration with the artificial
  term}\label{subsec_integration_with}

In this subsection we construct a multi-scale integration for 
\begin{align*}
\log\left(
\int
 e^{-V(u)(\psi+\psi^1)+W(u)(\psi+\psi^1)-A(\psi+\psi^1)}d\mu_{\sum_{l=N_{\beta}+1}^{\hat{N}_{\beta}}\cC_l}(\psi^1)\right),
\end{align*}
where $A(\psi)$ is the Grassmann polynomial defined in
\eqref{eq_Grassmann_artificial_term}. Since the artificial term
$A(\psi)$ is parameterized by $\bla=(\la_1,\la_2)\in\C^2$, the Grassmann
data in this process is parameterized by $(u,\bla)$. We will classify
them in terms of the degree with $\bla$. It is structurally
natural to measure kernels of these Grassmann data by using a variant
of the norm $\|\cdot\|_1$ defined as follows. For $f\in
C(\overline{D(r)}\times \overline{D(r')}^2, \Map(I^m,\C))$ we set 
$$
\|f\|_{1,r,r'}:=\sup_{(u,\bla)\in \overline{D(r)}\times \overline{D(r')}^2}\|f(u,\bla)\|_1.
$$
Then $C(\overline{D(r)}\times \overline{D(r')}^2,\Map(I^m,\C))$ is a
Banach space with the norm $\|\cdot\|_{1,r,r'}$. Also, to shorten
subsequent formulas, we set 
$$
\|f_0\|_{1,r,r'}:=\sup_{(u,\bla)\in \overline{D(r)}\times \overline{D(r')}^2}|f_0(u,\bla)|
$$ 
for $f_0\in C(\overline{D(r)}\times \overline{D(r')}^2,\C)$.
Moreover, we introduce a variant of the measurement $[\cdot,\cdot]_1$ as
follows. For $f\in
C(\overline{D(r)}\times\overline{D(r')}^2,\Map(I^m\times I^n,\C))$ and
an anti-symmetric function $g:I^2\to \C$,
$$
[f,g]_{1,r,r'}:=\sup_{(u,\bla)\in\overline{D(r)}\times \overline{D(r')}^2}[f(u,\bla),g]_{1}.
$$
To describe scale-dependent properties of Grassmann data during the
multi-scale integration process, we introduce sets of
$\bigwedge\cV$-valued functions. Let $l\in
\{N_{\beta},N_{\beta}+1,\cdots,\hat{N}_{\beta}\}$ and $r,r'\in
\R_{>0}$. We define the subset $\cQ'(r,r',l)$ of
$\Map(\overline{D(r)}\times \C^2,\bigwedge_{even}\cV)$ as follows.
 $f$ belongs to $\cQ'(r,r',l)$ if and only if 
\begin{itemize}
\item 
\begin{align*}
f\in C\left(\overline{D(r)}\times\C^2,\bigwedge_{even}\cV\right)\cap
 C^{\o}\left(D(r)\times \C^2,\bigwedge_{even}\cV\right).
\end{align*}
\item For any $u\in \overline{D(r)}$, $\bla\mapsto
      f(u,\bla)(\psi):\C^2\to \bigwedge_{even}\cV$ is linear. 
\item For any $(u,\bla)\in \overline{D(r)}\times \C^2$ the
      anti-symmetric kernels $f(u,\bla)_m:I^m\to \C$ $(m=2,4,\cdots,N)$
      satisfy \eqref{eq_temperature_translation_invariance} and 
\begin{align}
&\alpha^2 \|f_0\|_{1,r,r'}\le L^{-d},\label{eq_1_1_bound}\\
&\sum_{m=2}^Nc_0^{\frac{m}{2}}\alpha^mM^{\frac{m}{2}\sa(l-\hat{N}_{\beta})}\|f_m\|_{1,r,r'}\le
 L^{-d}.\notag
\end{align}
\end{itemize}
We use the set $\cQ'(r,r',l)$ to collect Grassmann data linearly
      dependent on $\bla$ and bounded by $L^{-d}$.

The set $\cR'(r,r',l)$ is defined as follows. $f$ belongs to $\cR'(r,r
',l)$ if and only if 
\begin{itemize}
\item 
\begin{align*}
f\in C\left(\overline{D(r)}\times\C^2,\bigwedge_{even}\cV\right)\cap
 C^{\o}\left(D(r)\times \C^2,\bigwedge_{even}\cV\right).
\end{align*}
\item For any $u\in \overline{D(r)}$, $\bla\mapsto f(u,\bla)(\psi)$:
      $\C^2\to \bigwedge_{even}\cV$ is linear. 
\item There exist $f_{p,q}\in C(\overline{D(r)}\times
      \C^2,\Map(I^p\times I^q,\C))$ $(p,q\in \{2,4,\cdots,N\})$ such
      that for any $(u,\bla)\in \overline{D(r)}\times \C^2$, $p,q\in
      \{2,4,\cdots, N\}$, $f_{p,q}(u,\bla):I^p\times I^q\to\C$ is
      bi-anti-symmetric, satisfies
      \eqref{eq_temperature_translation_invariance},
      \eqref{eq_temperature_integral_vanish},
\begin{align*}
f(u,\bla)(\psi)=\sum_{p,q=2}^{N}1_{p,q\in 2\N}\frah^{p+q}\sum_{\bX\in
 I^p\atop \bY\in I^q}f_{p,q}(u,\bla)(\bX,\bY)\psi_{\bX}\psi_{\bY}
\end{align*}
and
\begin{align}
&\sum_{p,q=2}^{N}1_{p,q\in
 2\N}c_0^{\frac{p+q}{2}}\alpha^{p+q}M^{\frac{p+q}{2}\sa(l-\hat{N}_{\beta})}\|f_{p,q}\|_{1,r,r'}\le
 1,\label{eq_1_2_bound}\\
&\sum_{p,q=2}^{N}1_{p,q\in
 2\N}c_0^{\frac{p+q}{2}}\alpha^{p+q}M^{\frac{p+q}{2}\sa(l-\hat{N}_{\beta})}[f_{p,q},g]_{1,r,r'}\le
 L^{-d}\|g\|,\label{eq_1_2_multiplied_bound}
\end{align}
for any anti-symmetric function $g:I^2\to \C$. 
\end{itemize}
The role of the set $\cR'(r,r',l)$ is to collect Grassmann data linearly
depending on $\bla$, having bi-anti-symmetric kernels satisfying the property
\eqref{eq_temperature_integral_vanish}.

It is also necessary to define a set which can contain descendants of the
artificial term $-A(\psi)$. $f$ belongs to $\cS(r,r',l)$ if and only if 
\begin{itemize}
\item 
\begin{align*}
f\in C\left(\overline{D(r)}\times\C^2,\bigwedge_{even}\cV\right)\cap
 C^{\o}\left(D(r)\times \C^2,\bigwedge_{even}\cV\right).
\end{align*}
\item For any $u\in \overline{D(r)}$, $\bla\mapsto
      f(u,\bla)(\psi):\C^2\to \bigwedge_{even}\cV$ is linear. 
\item For any $(u,\bla)\in \overline{D(r)}\times\C^2$ the anti-symmetric kernels      $f(u,\bla)_m:I^m\to\C$ $(m=2,4,\cdots,N)$ satisfy
      \eqref{eq_temperature_translation_invariance} and
\begin{align}
&\alpha^2\|f_0\|_{1,r,r'}\le 1,\label{eq_1_3_bound}\\
&\sum_{m=2}^Nc_0^{\frac{m}{2}}\alpha^m\|f_m\|_{1,r,r'}\le 1.\notag
\end{align}
\end{itemize}
In fact the descendants of $-A(\psi)$ are independent of $u$. Thus the
condition concerning the variable $u$ assumed in $\cS(r,r',l)$ is not
necessary. However, by defining the set as above we can avoid
introducing another norm. 

Finally we define a set of Grassmann data depending on $\bla$ at least
quadratically. $f$ belongs to $\cW(r,r',l)$ if and only if 
\begin{itemize}
 \item 
\begin{align*}
f\in C\left(\overline{D(r)}\times\overline{D(r')}^2,\bigwedge_{even}\cV\right)\cap
 C^{\o}\left(D(r)\times {D(r')}^2,\bigwedge_{even}\cV\right).
\end{align*}
\item For any $u\in D(r)$, $j\in \{1,2\}$, 
$$
f(u,\b0)(\psi)=\frac{\partial}{\partial \la_j}f(u,\b0)(\psi)=0.
$$
\item For any $(u,\bla)\in \overline{D(r)}\times \overline{D(r
')}^2$ the anti-symmetric kernels $f(u,\bla)_m:I^m\to \C$
       $(m=2,4,\cdots,N)$ satisfy
       \eqref{eq_temperature_translation_invariance} and 
\begin{align}
&\alpha^2 \|f_0\|_{1,r,r'}\le 1,\label{eq_2_bound}\\
&\sum_{m=2}^Nc_0^{\frac{m}{2}}\alpha^mM^{\frac{m}{2}\sa(l-\hat{N}_{\beta})}\|f_m\|_{1,r,r'}\le
 1.\notag
\end{align}
\end{itemize}

In the following we inductively define a family of Grassmann
polynomials, which are the scale-dependent input and output of the
multi-scale integration process from $\hat{N}_{\beta}$ to
$N_{\beta}+1$. We admit the results of Lemma
\ref{lem_IR_integration_without} stating that 
$V^{0-1,l}\in\cQ(b^{-1}c_0^{-2}\alpha^{-4},l)$,
$V^{0-2,l}\in\cR(b^{-1}c_0^{-2}\alpha^{-4},l)$ $(\forall l\in
\{N_{\beta},N_{\beta}+1,\cdots,\hat{N}_{\beta}\})$ and define
\begin{align*}
V^{0,l}\in
C\left(\overline{D(b^{-1}c_0^{-2}\alpha^{-4})},\bigwedge_{even}\cV\right)\quad
(l=N_{\beta},N_{\beta}+1,\cdots,\hat{N}_{\beta})
\end{align*}
 by $V^{0,l}:=V^{0-1,l}+V^{0-2,l}$. Recalling the definition
 \eqref{eq_Grassmann_artificial_term}, we define $V^{1-3,\hat{N}_{\beta}}\in
 C(\C^2,\bigwedge_{even}\cV)$ by
 $$V^{1-3,\hat{N}_{\beta}}(\bla)(\psi):=-A(\psi).$$ 
Moreover, set 
\begin{align*}
V^{1-1,\hat{N}_{\beta}}=V^{1-2,\hat{N}_{\beta}}:=0,\quad
V^{1,\hat{N}_{\beta}}:=\sum_{j=1}^3V^{1-j,\hat{N}_{\beta}},\quad
 V^{2,\hat{N}_{\beta}}:=0.
\end{align*}
Let us assume that $l\in \{N_{\beta},N_{\beta}+1,\cdots,\hat{N}_{\beta}-1\}$ and we have
\begin{align*}
&V^{1-1,l+1}\in \cQ'(r,r',l+1),\quad
V^{1-2,l+1}\in \cR'(r,r',l+1),\\
&V^{1-3,l+1}\in \cS(r,r',l+1),\quad
V^{2,l+1}\in \cW(r,r',l+1).
\end{align*}
Set $V^{1,l+1}:=\sum_{j=1}^3V^{1-j,l+1}$.  
By recalling the formula \eqref{eq_tree_expansion} we can observe that
\begin{align}
&\frac{1}{n!}\left(\frac{d}{dz}\right)^n\log\left(\int
 e^{z\sum_{j=0}^2V^{j,l+1}(\psi^1+\psi)} d\mu_{\cC_{l+1}}(\psi^1)
\right)\Bigg|_{z=0}\label{eq_data_recursive_relation_with}\\
&=\frac{1}{n!}Tree(\{1,2,\cdots,n\},\cC_{l+1})\prod_{j=1}^nV^{0,l+1}(\psi^j+\psi)\Bigg|_{\psi^{j}=0\atop(\forall
 j\in\{1,2,\cdots,n\})}\notag\\
&\quad+1_{n=1}Tree(\{1\},\cC_{l+1})V^{1,l+1}(\psi^1+\psi)\Big|_{\psi^{1}=0}\notag\\
&\quad +1_{n\ge 2}
 \frac{1}{(n-1)!}Tree(\{1,2,\cdots,n\},\cC_{l+1})\notag\\
&\qquad\qquad\cdot V^{1,l+1}(\psi^1+\psi)
\prod_{j=2}^n\left(\sum_{a_j\in \{1,2\}}V^{0-a_j,l+1}(\psi^j+\psi)
\right)\Bigg|_{\psi^{j}=0\atop(\forall
 j\in\{1,2,\cdots,n\})}1_{\exists j(a_j=1)}\notag\\
&\quad +1_{n\ge 2}
 \frac{1}{(n-1)!}Tree(\{1,2,\cdots,n\},\cC_{l+1})\notag\\
&\qquad\qquad\qquad\cdot V^{1,l+1}(\psi^1+\psi)
\prod_{j=2}^nV^{0-2,l+1}(\psi^j+\psi)\Bigg|_{\psi^{j}=0\atop(\forall
 j\in\{1,2,\cdots,n\})}\notag\\
&\quad +
 \frac{1}{n!}Tree(\{1,2,\cdots,n\},\cC_{l+1})
\prod_{j=1}^n\left(\sum_{b_j=0}^2V^{b_j,l+1}(\psi^j+\psi)
\right)\Bigg|_{\psi^{j}=0\atop(\forall
 j\in\{1,2,\cdots,n\})}1_{\sum_{j=1}^nb_j\ge 2}.\notag
\end{align}
Let us decompose or rename each term of the right-hand side of
\eqref{eq_data_recursive_relation_with} from top to
bottom. In Subsection \ref{subsec_integration_without} we proved that if we set 
\begin{align*}
V^{0,l,(n)}(\psi):=\frac{1}{n!}Tree(\{1,2,\cdots,n\},\cC_{l+1})\prod_{j=1}^nV^{0,l+1}(\psi^j+\psi)\Bigg|_{\psi^{j}=0\atop(\forall
 j\in\{1,2,\cdots,n\})}
\end{align*}
for $n\in \N_{\ge 1}$, 
then $V^{0,l}(\psi)=\sum_{n=1}^{\infty}V^{0,l,(n)}(\psi)$. Let us set
\begin{align}
&V^{1-1-1,l}(\psi):=Tree(\{1\},\cC_{l+1})V^{1-1,l+1}(\psi^1+\psi)\Big|_{\psi^{1}=0},\notag\\
&V^{1-1-2,l}(\psi):=
\sum_{p,q=2}^{N}1_{p,q\in 2\N}\frah^{p+q}\sum_{\bX\in I^p\atop \bY\in
 I^q}V_{p,q}^{1-2,l+1}(\bX,\bY)\notag\\
&\qquad\qquad\qquad\quad\cdot Tree(\{1,2\},\cC_{l+1})(\psi^1+\psi)_{\bX}(\psi^2+\psi)_{\bY}\Bigg|_{\psi^{1}=\psi^2=0},\notag\\
&V^{1-2-1,l}(\psi):=\sum_{p,q=2}^{N}1_{p,q\in 2\N}\frah^{p+q}\sum_{\bX\in I^p\atop \bY\in
 I^q}V_{p,q}^{1-2,l+1}(\bX,\bY)\notag\\
&\qquad\qquad\qquad\quad\cdot Tree(\{1\},\cC_{l+1})(\psi^1+\psi)_{\bX}\Big|_{\psi^1=0}
Tree(\{1\},\cC_{l+1})(\psi^1+\psi)_{\bY}\Big|_{\psi^1=0},\notag\\
&V^{1-3,l}(\psi):=Tree(\{1\},\cC_{l+1})V^{1-3,l+1}(\psi^1+\psi)\Big|_{\psi^{1}=0}.\label{eq_linear_artificial_free_part}
\end{align}
Then, for the same reason that the transformation
\eqref{eq_data_recursive_relation_without} is valid, the following equality holds.
\begin{align*}
\frac{d}{dz}\log\left(\int
 e^{zV^{1-2,l+1}(\psi^1+\psi)} d\mu_{\cC_{l+1}}(\psi^1)
\right)\Bigg|_{z=0}=V^{1-1-2,l}(\psi)+V^{1-2-1,l}(\psi).
\end{align*}
Thus it follows that 
\begin{align*}
&Tree(\{1\},\cC_{l+1})V^{1,l+1}(\psi^1+\psi)_{\bX}\Big|_{\psi^1=0}\\
&=
V^{1-1-1,l}(\psi)+V^{1-1-2,l}(\psi)+V^{1-2-1,l}(\psi)+V^{1-3,l}(\psi).
\end{align*}
Moreover, set for $n\in \N_{\ge 2}$,
\begin{align}
&V^{1-1-3,l,(n)}(\psi)\label{eq_1_1_3_definition}\\
&:=
 \frac{1}{(n-1)!}Tree(\{1,2,\cdots,n\},\cC_{l+1})\notag\\
&\quad\cdot V^{1,l+1}(\psi^1+\psi)
\prod_{j=2}^n\left(\sum_{a_j\in \{1,2\}}V^{0-a_j,l+1}(\psi^j+\psi)
\right)\Bigg|_{\psi^{j}=0\atop(\forall
 j\in\{1,2,\cdots,n\})}1_{\exists j(a_j=1)},\notag\\
&V^{1-1-4,l,(n)}(\psi)\label{eq_1_1_4_definition}\\
&:=
 \sum_{p,q=2}^{N}1_{p,q\in 2\N}\frah^{p+q}\sum_{\bX\in I^p\atop \bY\in
 I^q} V_{p,q}^{0-2,l+1}(\bX,\bY)\frac{1}{(n-1)!}Tree(\{1,2,\cdots,n+1\},\cC_{l+1})\notag\\
&\quad\cdot (\psi^1+\psi)_{\bX}(\psi^2+\psi)_{\bY}
\prod_{j=3}^nV^{0-2,l+1}(\psi^j+\psi)\cdot V^{1,l+1}(\psi^{n+1}+\psi)\Bigg|_{\psi^{j}=0\atop(\forall
 j\in\{1,2,\cdots,n+1\})},\notag\\
&V^{1-2-2,l,(n)}(\psi)\label{eq_1_2_2_definition}\\
&:=\frac{1}{(n-1)!}\sum_{m=0}^{n-1}\sum_{(\{s_j\}_{j=1}^{m+1},\{t_k\}_{k=1}^{n-m})\in
 S(n,m)}\sum_{p,q=2}^{N}
1_{p,q\in 2\N}\frah^{p+q}\sum_{\bX\in I^p\atop \bY\in
 I^q} V_{p,q}^{0-2,l+1}(\bX,\bY)\notag\\
&\quad\cdot Tree(\{s_j\}_{j=1}^{m+1},\cC_{l+1})(\psi^{s_1}+\psi)_{\bX}\notag\\
&\quad\cdot\prod_{j=2}^{m+1}(1_{s_j\neq n}V^{0-2,l+1}(\psi^{s_j}+\psi)+
1_{s_j= n}V^{1,l+1}(\psi^{s_j}+\psi))\Bigg|_{\psi^{s_j}=0\atop(\forall
 j\in\{1,2,\cdots,m+1\})}\notag\\
&\quad\cdot Tree(\{t_k\}_{k=1}^{n-m},\cC_{l+1})(\psi^{t_1}+\psi)_{\bY}\notag\\
&\quad\cdot\prod_{k=2}^{n-m}(1_{t_k\neq n}V^{0-2,l+1}(\psi^{t_k}+\psi)+
1_{t_k= n}V^{1,l+1}(\psi^{t_k}+\psi))\Bigg|_{\psi^{t_k}=0\atop(\forall
 k\in\{1,2,\cdots,n-m\})}.\notag
\end{align}
It follows from the same argument as in
\eqref{eq_data_recursive_relation_without} that 
\begin{align*}
&\frac{1}{(n-1)!}Tree(\{1,2,\cdots,n\},\cC_{l+1}) V^{1,l+1}(\psi^1+\psi)
\prod_{j=2}^n V^{0-2,l+1}(\psi^j+\psi)\Bigg|_{\psi^{j}=0\atop(\forall
 j\in\{1,2,\cdots,n\})}\\
&=\frac{1}{(n-1)!}Tree(\{1,2,\cdots,n\},\cC_{l+1})\\
&\quad\cdot \prod_{j=1}^{n-1}
 V^{0-2,l+1}(\psi^j+\psi)\cdot 
 V^{1,l+1}(\psi^n+\psi)\Bigg|_{\psi^{j}=0\atop(\forall
 j\in\{1,2,\cdots,n\})}\\
&=V^{1-1-4,l,(n)}(\psi)+V^{1-2-2,l,(n)}(\psi).
\end{align*}
Finally we set for $n\in \N_{\ge 1}$,
\begin{align}
&V^{2,l,(n)}(\psi)\label{eq_2_definition}\\
&:=\frac{1}{n!}Tree(\{1,2,\cdots,n\},\cC_{l+1})
\prod_{j=1}^n\left(\sum_{b_j=0}^2V^{b_j,l+1}(\psi^j+\psi)
\right)\Bigg|_{\psi^{j}=0\atop(\forall
 j\in\{1,2,\cdots,n\})}1_{\sum_{j=1}^nb_j\ge 2}.\notag
\end{align}
By giving back these Grassmann polynomials to the expansion
\eqref{eq_data_recursive_relation_with} we see that the following
equality holds.
\begin{align*}
&\frac{1}{n!}\left(\frac{d}{dz}\right)^n\log\left(\int
 e^{z\sum_{j=0}^2V^{j,l+1}(\psi^1+\psi)} d\mu_{\cC_{l+1}}(\psi^1)
\right)\Bigg|_{z=0}\\
&=V^{0,l,(n)}(\psi)+1_{n=1}(V^{1-1-1,l}(\psi)+V^{1-1-2,l}(\psi)+V^{1-2-1,l}(\psi)+V^{1-3,l}(\psi))\\
&\quad + 1_{n\ge
 2}(V^{1-1-3,l,(n)}(\psi)+V^{1-1-4,l,(n)}(\psi)+V^{1-2-2,l,(n)}(\psi))+V^{2,l,(n)}(\psi).
\end{align*}
By assuming that these are convergent let us set
\begin{align}
&V^{1-1-3,l}(\psi):=\sum_{n=2}^{\infty}V^{1-1-3,l,(n)}(\psi),\quad
V^{1-1-4,l}(\psi):=\sum_{n=2}^{\infty}V^{1-1-4,l,(n)}(\psi),\label{eq_formal_sum_definition}\\
&V^{1-2-2,l}(\psi):=\sum_{n=2}^{\infty}V^{1-2-2,l,(n)}(\psi),\quad
V^{2,l}(\psi):=\sum_{n=1}^{\infty}V^{2,l,(n)}(\psi),\notag\\
&V^{1-1,l}(\psi):=\sum_{j=1}^4V^{1-1-j,l}(\psi),\quad 
V^{1-2,l}(\psi):=\sum_{j=1}^2V^{1-2-j,l}(\psi).\notag
\end{align}
We are going to prove the convergence, regularity and bound properties
of these Grassmann polynomials. Remind us that the data $V^{0,l}$ is
independent of the artificial parameter $\bla$, the data $V^{1-j,l}$
$(j=1,2,3)$ are linear with $\bla$ and $V^{2,l}$ depends on $\bla$ at
least quadratically. The 2nd superscript $l$ indicates that these are to
be integrated with the covariance $\cC_l$.

Set 
\begin{align}
\eps_{\beta}:=M^{(\sum_{j=1}^d\frac{1}{\sn_j}+1)(N_{\beta}-\hat{N}_{\beta})}.\label{eq_scaling_factor}
\end{align}

\begin{lemma}\label{lem_IR_integration_with}
Let $c_4$ be the constant appearing in Lemma
 \ref{lem_IR_integration_without}. Then there exists a constant $c_5\in
 [c_4,\infty)$ independent of any other parameters such that if 
\begin{align}
M^{\min\{1,2\sa-1-\sum_{j=1}^d\frac{1}{\sn_j}\}}\ge c_5,\quad \alpha\ge
 c_5 M^{\frac{\sa}{2}},\quad L^d\ge
 M^{(\sum_{j=1}^d\frac{1}{\sn_j}+1)(\hat{N}_{\beta}-N_{\beta})},\quad
 h\ge 2,\label{eq_assumptions_IR_integration_with}
\end{align}
then
\begin{align*}
&V^{1-1,l}\in\cQ'(b^{-1}c_0^{-2}\alpha^{-4},
 c_5^{-1}\eps_{\beta}^{\hat{N}_{\beta}-l}\beta^{-1}c_0^{-2}\alpha^{-4},l),\\
&V^{1-2,l}\in\cR'(b^{-1}c_0^{-2}\alpha^{-4},
 c_5^{-1}\eps_{\beta}^{\hat{N}_{\beta}-l}\beta^{-1}c_0^{-2}\alpha^{-4},l),\\
&V^{1-3,l}\in\cS(b^{-1}c_0^{-2}\alpha^{-4},
 c_5^{-1}\eps_{\beta}^{\hat{N}_{\beta}-l}\beta^{-1}c_0^{-2}\alpha^{-4},l),\\
&V^{2,l}\in\cW(b^{-1}c_0^{-2}\alpha^{-4},
 c_5^{-1}h^{l-\hat{N}_{\beta}}\eps_{\beta}^{\hat{N}_{\beta}-l}\beta^{-1}c_0^{-2}\alpha^{-4},l),\\
&(\forall l\in \{N_{\beta},N_{\beta}+1,\cdots,\hat{N}_{\beta}\}).
\end{align*}
\end{lemma}

\begin{remark}
The radius
 $c_5^{-1}\eps_{\beta}^{\hat{N}_{\beta}-l}\beta^{-1}c_0^{-2}\alpha^{-4}$
 of $\bla$ assumed on $V^{1-1,l}$, $V^{1-2,l}$, $V^{1-3,l}$ amounts to
 heavy $\beta$-dependent bounds on these Grassmann data. Also, the
 radius of analyticity of $V^{2,l}$ with $\bla$ depends not only on
 $\beta$ but on $h$ heavily. While we have to make best efforts to
 improve $\beta$-dependency of the possible magnitude of the variable
 $u$ as the main focus of this paper, the $\beta$-dependency of the
 magnitude of $\bla$ does not affect our main results. Therefore we
 choose to simplify the following inductive estimation procedure at the
 expense of $(\beta,h)$-dependency of the magnitude of $\bla$ rather
 than to optimize it with some complications. 
\end{remark}

In the proof of Lemma \ref{lem_IR_integration_with} we will use the
following lemma. 

\begin{lemma}\label{lem_artificial_scaling_properties}
Assume that $m,\ p,\ q\in \N_{\ge 2}$,  
\begin{align*}
&f_0^1\in
 C(\overline{D(r)}\times\C^2)\cap C^{\o}(D(r)\times \C^2),\\
&f_m^1\in C(\overline{D(r)}\times \C^2,\Map(I^m,\C))\cap
 C^{\o}(D(r)\times \C^2,\Map(I^m,\C)),\\ 
&f_{p,q}^1\in  C(\overline{D(r)}\times \C^2,\Map(I^p\times I^q,\C))\cap
 C^{\o}(D(r)\times \C^2,\Map(I^p\times I^q,\C)),\\
&f_0^2\in C(\overline{D(r)}\times\overline{D(r')}^2)\cap
 C^{\o}(D(r)\times D(r')^2),\\
&f_m^2\in C(\overline{D(r)}\times\overline{D(r')}^2,\Map(I^m,\C))\cap
C^{\o}(D(r)\times D(r')^2,\Map(I^m,\C)),\\ 
&\bla\mapsto f_0^1(u,\bla):\C^2\to \C,\quad\bla\mapsto
 f_m^1(u,\bla)(\bX):\C^2\to \C,\\
&\bla\mapsto f_{p,q}^1(u,\bla)(\bY,\bZ):\C^2\to\C \text{ are linear and }\\
&f_0^2(u,\b0)=\frac{\partial}{\partial \la_j}f_0^2(u,\b0)=0,\quad 
f_m^2(u,\b0)(\bX)=\frac{\partial}{\partial \la_j}f_m^2(u,\b0)(\bX)=0,\\
&(\forall u\in D(r),\ \bX\in I^m,\ (\bY,\bZ)\in I^p\times I^q,\ j\in \{1,2\}).
\end{align*}
Then, 
\begin{align}
&\|f_n^j(u,\eps\bla)\|_1\le \eps \|f_n^j\|_{1,r,r'},\quad
 \|f_n^j(u,\eps\bla)\|_{1,\infty}\le
 h\eps\|f_n^j\|_{1,r,r'},\label{eq_artificial_scaling_properties}\\
&[f_{p,q}^1(u,\eps \bla),g]_1\le \eps [f_{p,q}^1,g]_{1,r,r'},\notag\\
&(\forall u\in \overline{D(r)},\ \bla\in \overline{D(r')}^2,\ \eps\in
 [0,1/2],\ j\in \{1,2\},\ n\in \{0,m\}).\notag
\end{align}
Here $g:I^2\to \C$ is any anti-symmetric function.
\end{lemma}

\begin{proof}
These are essentially same as the inequalities
 \cite[\mbox{(3.91),(3.92)}]{K_BCS}. The inequalities
 \eqref{eq_artificial_scaling_properties} for $j=1$ are trivial because
 of the linearity. To
 prove the inequalities for $j=2$, one can use the following equality.
\begin{align*}
&f_m^2(u,\eps\bla)(\bX)=\frac{1}{2\pi i}\oint_{|z|=\delta}dz
 f_m^2(u,z\bla)(\bX)\frac{\eps^2}{z^2(z-\eps)},\\
&(\forall u\in \overline{D(r)},\ \bla\in \overline{D(r')}^2,\ \eps\in
 [0,1/2],\ \delta\in (1/2,1),\ \bX\in I^m).
\end{align*}
\end{proof}

\begin{proof}[Proof of Lemma \ref{lem_IR_integration_with}] 
During the proof we often omit the sign of dependency on the parameter
 $(u,\bla)$ to shorten formulas. Since $V^{1-3,l}(\psi)$
 $(l=N_{\beta},N_{\beta}+1,\cdots,\hat{N}_{\beta})$ are defined
 independently of other polynomials, we can readily summarize their
 properties. Since $V_4^{1-3,l}(\psi)=-\la_2 A_4^2(\psi)$ for any $l\in
 \{N_{\beta},N_{\beta}+1,\cdots,\hat{N}_{\beta}\}$, 
\begin{align*}
&V_4^{1-3,l}(\orho_1\rho_1\bx_1s_1\xi_1,\orho_2\rho_2\bx_2s_2\xi_2,\orho_3\rho_3\bx_3s_3\xi_3,\orho_4\rho_4\bx_4s_4\xi_4)\\
&=-\frac{\la_2h^3}{4!}1_{s_1=s_2=s_3=s_4}\sum_{\s\in\S_4}\sgn(\s)\\
&\qquad\cdot 
1_{((\orho_{\s(1)},\rho_{\s(1)},\bx_{\s(1)},\xi_{\s(1)}),
(\orho_{\s(2)},\rho_{\s(2)},\bx_{\s(2)},\xi_{\s(2)}),
(\orho_{\s(3)},\rho_{\s(3)},\bx_{\s(3)},\xi_{\s(3)}),
(\orho_{\s(4)},\rho_{\s(4)},\bx_{\s(4)},\xi_{\s(4)}))\atop =
((1,\hrho,r_L(\hbx),1),(2,\hrho,r_L(\hbx),-1),(2,\heta,r_L(\hby),1),
(1,\heta,r_L(\hby),-1))},\\
&(\forall (\orho_j,\rho_j,\bx_j,s_j,\xi_j)\in I\ (j=1,2,3,4)).
\end{align*}
Thus
\begin{align}
\|V_4^{1-3,l}\|_{1,r,r'}\le \beta r',\quad (\forall l\in
 \{N_{\beta},N_{\beta}+1,\cdots,\hat{N}_{\beta}\}).\label{eq_1_3_4}
\end{align}
Also, 
\begin{align*}
&V_2^{1-3,\hat{N}_{\beta}}(\orho_1\rho_1\bx_1s_1\xi_1,\orho_2\rho_2\bx_2s_2\xi_2)\\
&=-\frac{h}{2}1_{s_1=s_2}\sum_{\s\in\S_2}\sgn(\s)
\Big(1_{((\orho_{\s(1)},\rho_{\s(1)},\bx_{\s(1)},\xi_{\s(1)}),
(\orho_{\s(2)},\rho_{\s(2)},\bx_{\s(2)},\xi_{\s(2)}))\atop =
((1,\hrho,r_L(\hbx),1),(2,\hrho,r_L(\hbx),-1))}\la_1\\
&\qquad\qquad\qquad\qquad\qquad\quad
+1_{((\orho_{\s(1)},\rho_{\s(1)},\bx_{\s(1)},\xi_{\s(1)}),
(\orho_{\s(2)},\rho_{\s(2)},\bx_{\s(2)},\xi_{\s(2)}))\atop =
((1,\hrho,r_L(\hbx),1),(1,\hrho,r_L(\hbx),-1))}
1_{(\hrho,r_L(\hbx))=(\heta,r_L(\hby))}\la_2\Big),\\
&\quad(\forall (\orho_j,\rho_j,\bx_j,s_j,\xi_j)\in I\ (j=1,2)).
\end{align*}
Thus
\begin{align}
\|V_2^{1-3,\hat{N}_{\beta}}\|_{1,r,r'}\le 2\beta
 r'.\label{eq_1_3_2_first}
\end{align}
We can see from the definition that for $l\in
 \{N_{\beta},N_{\beta}+1,\cdots,\hat{N}_{\beta}-1\}$
\begin{align}
&V_2^{1-3,l}(\psi)\label{eq_1_3_2_expansion}\\
&=V_2^{1-3,l+1}(\psi)\notag\\
&\quad+\frah^2\sum_{\bX\in I^2}\left(\left(\begin{array}{c}4 \\ 2\end{array}\right)\frah^2\sum_{\bY\in
 I^2} V_4^{1-3,l+1}(\bX,\bY)Tree(\{1\},\cC_{l+1})\psi^1_{\bY}\Big|_{\psi^1=0}
\right)\psi_{\bX}\notag\\
&=V_2^{1-3,\hat{N}_{\beta}}(\psi)\notag\\
&\quad+\frah^2\sum_{\bX\in I^2}
\left(\left(\begin{array}{c}4 \\ 2\end{array}\right)\frah^2\sum_{\bY\in I^2}V_4^{1-3,\hat{N}_{\beta}}(\bX,\bY)\sum_{j=l+1}^{\hat{N}_{\beta}}Tree(\{1\},\cC_{j})\psi^1_{\bY}\Big|_{\psi^1=0}\right)\psi_{\bX}.\notag
\end{align}
By using \eqref{eq_scale_covariance_determinant_bound},
 \eqref{eq_1_3_4}, \eqref{eq_1_3_2_first} and the assumptions $c_0\ge
 1$, $M\ge 2$, 
\begin{align}
\|V_2^{1-3,l}\|_{1,r,r'}\le
 \|V_2^{1-3,\hat{N}_{\beta}}\|_{1,r,r'}+\left(\begin{array}{c}4\\ 2
					      \end{array}\right)
 \|V_4^{1-3,\hat{N}_{\beta}}\|_{1,r,r'}\sum_{j=l+1}^{\hat{N}_{\beta}}c_0M^{\sa(j-\hat{N}_{\beta})}\le
 c \beta c_0r'.\label{eq_1_3_2}
\end{align}
Moreover it follows from the definition and \eqref{eq_1_3_2_expansion} that
\begin{align*}
&V_0^{1-3,l}\\
&=V_0^{1-3,l+1}+Tree(\{1\},\cC_{l+1})V_2^{1-3,l+1}(\psi^1)\Big|_{\psi^1=0}+
Tree(\{1\},\cC_{l+1})V_4^{1-3,l+1}(\psi^1)\Big|_{\psi^1=0}\\
&=V_0^{1-3,l+1}+Tree(\{1\},\cC_{l+1})V_2^{1-3,\hat{N}_{\beta}}(\psi^1)\Big|_{\psi^1=0}\\
&\quad+ 1_{l\le \hat{N}_{\beta}-2}\frah^{2}\sum_{\bX\in I^2}\\
&\qquad\quad\cdot \left(
\left(\begin{array}{c}4 \\ 2 \end{array}\right)\frah^2\sum_{\bY\in I^2}V_4^{1-3,\hat{N}_{\beta}}(\bX,\bY)\sum_{j=l+2}^{\hat{N}_{\beta}}Tree(\{1\},\cC_j)\psi_{\bY}^1\Big|_{\psi^1=0}
\right)\\
&\qquad\quad\cdot Tree(\{1\},\cC_{l+1})\psi_{\bX}^1\Big|_{\psi^1=0}\\
&\quad +Tree(\{1\},\cC_{l+1})V_4^{1-3,\hat{N}_{\beta}}(\psi^1)\Big|_{\psi^1=0}\\
&=\sum_{k=l+1}^{\hat{N}_{\beta}}Tree(\{1\},\cC_k)V_2^{1-3,\hat{N}_{\beta}}(\psi^1)\Big|_{\psi^1=0}\\
&\quad 
+1_{l\le \hat{N}_{\beta}-2}\sum_{k=l+1}^{\hat{N}_{\beta}-1}\left(\frac{1}{h}\right)^2
\sum_{\bX\in I^2}\\
&\qquad\quad\cdot\left(
\left(\begin{array}{c}4 \\ 2 \end{array}\right)\frah^2\sum_{\bY\in I^2}V_4^{1-3,\hat{N}_{\beta}}(\bX,\bY)\sum_{j=k+1}^{\hat{N}_{\beta}}Tree(\{1\},\cC_j)\psi_{\bY}^1\Big|_{\psi^1=0}
\right)\\
&\qquad\quad\cdot Tree(\{1\},\cC_{k})\psi_{\bX}^1\Big|_{\psi^1=0}\\
&\quad +\sum_{k=l+1}^{\hat{N}_{\beta}}Tree(\{1\},\cC_k)V_4^{1-3,\hat{N}_{\beta}}(\psi^1)\Big|_{\psi^1=0}.
\end{align*}
Thus by \eqref{eq_scale_covariance_determinant_bound}, \eqref{eq_1_3_4},
 \eqref{eq_1_3_2_first} and $c_0\ge 1$, $M\ge 2$, 
\begin{align}
&\|V_0^{1-3,l}\|_{1,r,r'}\label{eq_1_3_0}\\
&\le 2\beta r'\sum_{k=l+1}^{\hat{N}_{\beta}}c_0M^{\sa(k-\hat{N}_{\beta})}
 + c\beta r' 1_{l\le
 \hat{N}_{\beta}-2}\sum_{k=l+1}^{\hat{N}_{\beta}-1}c_0M^{\sa
 (k-\hat{N}_{\beta})}\sum_{j=k+1}^{\hat{N}_{\beta}}c_0M^{\sa
 (j-\hat{N}_{\beta})}\notag\\
&\quad +\beta r' \sum_{k=l+1}^{\hat{N}_{\beta}}c_0^2M^{2\sa (k-\hat{N}_{\beta})}\notag\\
&\le c\beta c_0^2r'.\notag
\end{align}
The inequalities \eqref{eq_1_3_4}, \eqref{eq_1_3_2}, \eqref{eq_1_3_0}
 and $c_0\ge 1$, $\alpha\ge 1$
 result in 
\begin{align*}
&\alpha^2 \|V_0^{1-3,l}\|_{1,r,r'}\le c\beta c_0^2\alpha^2r',\quad \sum_{m=2}^Nc_0^{\frac{m}{2}}\alpha^m\|V_m^{1-3,l}\|_{1,r,r'}\le c\beta
 c_0^2\alpha^4 r'.
\end{align*}
By definition $\bla\mapsto V^{1-3,l}(\bla)(\psi)$ is linear for any
 $l\in\{N_{\beta},N_{\beta}+1,\cdots,\hat{N}_{\beta}\}$. The statement
 of Lemma \ref{lem_simple_tree_expansion} and the induction with $l$
 ensure that $V_m^{1-3,l}:I^m\to \C$ $(m=2,4,\
 l=N_{\beta},N_{\beta}+1,\cdots,\hat{N}_{\beta})$ satisfy
 \eqref{eq_temperature_translation_invariance}. Combined with these
 basic properties, the above inequalities conclude that there exists a
 generic positive constant $c'$ independent of any parameter such that 
\begin{align}
V^{1-3,l}\in
 \cS(r,{c'}^{-1}\eps_{\beta}^{\hat{N}_{\beta}-l}\beta^{-1}c_0^{-2}\alpha^{-4},l),\quad(\forall
 l\in \{N_{\beta},N_{\beta}+1,\cdots,\hat{N}_{\beta}\})\label{eq_1_3_inclusion}
\end{align}
for any $\alpha\in \R_{\ge 1}$, $r\in \R_{>0}$. Here we also used that
 $\eps_{\beta}\le 1$.

Let us set 
$r:=b^{-1}c_0^{-2}\alpha^{-4}$, $r':={c'}^{-1}\beta^{-1}c_0^{-2}\alpha^{-4}$
with the constant $c'$ appearing in \eqref{eq_1_3_inclusion}. Let $l\in
 \{N_{\beta},N_{\beta}+1,\cdots,\hat{N}_{\beta}-1\}$ and assume that 
\begin{align*}
&V^{1-1,l+1}\in \cQ'(r,\eps_{\beta}^{\hat{N}_{\beta}-l-1}r',l+1),\quad V^{1-2,l+1}\in\cR'(r,\eps_{\beta}^{\hat{N}_{\beta}-l-1}r',l+1),\\
&V^{2,l+1}\in
 \cW(r,h^{l+1-\hat{N}_{\beta}}\eps_{\beta}^{\hat{N}_{\beta}-l-1}r',l+1)
\end{align*}
as the induction hypothesis. Note that these inclusion trivially hold
 for $l=\hat{N}_{\beta}-1$. 
Check that if $M\ge 2$, $\eps_{\beta}\le M^{-1}\le\frac{1}{2}$. Thus we
 can apply the inequalities \eqref{eq_artificial_scaling_properties} with
 $\eps=\eps_{\beta}$. Let us list useful inequalities derived from the
 induction hypothesis,
 \eqref{eq_artificial_scaling_properties}, \eqref{eq_1_3_inclusion} and the conditions $M^{\sa}\ge 2^4$,
 $\alpha \ge 2^3$.
\begin{align}
&\sum_{m=2}^{N}2^{3m}(c_0M^{\sa
 (l+1-\hat{N}_{\beta})})^{\frac{m}{2}}\|V_m^{1-1,l+1}\|_{1,r,\eps_{\beta}^{\hat{N}_{\beta}-l}r'}\le
 c\eps_{\beta}\alpha^{-2}L^{-d},\label{eq_1_1_assumption}\\
&\sum_{m=2}^{N}2^{m}\alpha^m(c_0M^{\sa
 (l-\hat{N}_{\beta})})^{\frac{m}{2}}\|V_m^{1-1,l+1}\|_{1,r,\eps_{\beta}^{\hat{N}_{\beta}-l}r'}\le
 c\eps_{\beta}M^{-\sa}L^{-d},\label{eq_1_1_weight_assumption}\\
&\sum_{m=4}^{N}2^{3m}(c_0M^{\sa
 (l+1-\hat{N}_{\beta})})^{\frac{m}{2}}\|V_m^{1-2,l+1}\|_{1,r,\eps_{\beta}^{\hat{N}_{\beta}-l}r'}\le
 c\eps_{\beta}\alpha^{-4},\label{eq_1_2_kernel_assumption}\\
&\sum_{m=4}^{N}2^{m}\alpha^m(c_0M^{\sa
 (l-\hat{N}_{\beta})})^{\frac{m}{2}}\|V_m^{1-2,l+1}\|_{1,r,\eps_{\beta}^{\hat{N}_{\beta}-l}r'}\le
 c\eps_{\beta}M^{-2\sa},\label{eq_1_2_kernel_weight_assumption}\\
 &\sum_{p,q=2}^{N}1_{p,q\in
 2\N}2^{3p+3q}(c_0M^{\sa(l+1-\hat{N}_{\beta})})^{\frac{p+q}{2}}[V_{p,q}^{1-2,l+1},g]_{1,r,\eps_{\beta}^{\hat{N}_{\beta}-l}r'}\le c\eps_{\beta}\alpha^{-4}L^{-d}\|g\|,\label{eq_1_2_multiplied_assumption}\\
&\sum_{p,q=2}^{N}1_{p,q\in
 2\N}2^{2p+2q}\alpha^{p+q}(c_0M^{\sa(l-\hat{N}_{\beta})})^{\frac{p+q}{2}}[V_{p,q}^{1-2,l+1},g]_{1,r,\eps_{\beta}^{\hat{N}_{\beta}-l}r'}\le
 c\eps_{\beta}M^{-2\sa}L^{-d}\|g\|,\label{eq_1_2_multiplied_weight_assumption}
\end{align}
for any anti-symmetric function $g:I^2\to \C$. 
\begin{align}
&\sum_{m=2}^N2^{3m}(c_0M^{\sa(l+1-\hat{N}_{\beta})})^{\frac{m}{2}}\|V_m^{1-3,l+1}\|_{1,r,\eps_{\beta}^{\hat{N}_{\beta}-l}r'}\le
 c \eps_{\beta}\alpha^{-2},\label{eq_1_3_assumption}\\
&\sum_{m=2}^{N}2^{m}\alpha^m(c_0M^{\sa
 (l-\hat{N}_{\beta})})^{\frac{m}{2}}\|V_m^{1-3,l+1}\|_{1,r,\eps_{\beta}^{\hat{N}_{\beta}-l}r'}\le
 c\eps_{\beta}M^{-\sa},\label{eq_1_3_weight_assumption}\\
&\sum_{m=2}^N2^{3m}(c_0M^{\sa(l+1-\hat{N}_{\beta})})^{\frac{m}{2}}\|V_m^{2,l+1}\|_{1,r,\eps_{\beta}^{\hat{N}_{\beta}-l}r'}\le
 c \eps_{\beta}\alpha^{-2},\label{eq_2_assumption}\\
&\sum_{m=2}^N2^{m}\alpha^m(c_0M^{\sa(l-\hat{N}_{\beta})})^{\frac{m}{2}}\|V_m^{2,l+1}\|_{1,r,\eps_{\beta}^{\hat{N}_{\beta}-l}r'}\le
 c \eps_{\beta} M^{-\sa}.\label{eq_2_weight_assumption}
\end{align}
To derive \eqref{eq_1_2_kernel_assumption},
 \eqref{eq_1_2_kernel_weight_assumption}, we used a variant of the
 inequality \eqref{eq_kernel_divided_relation}. We can derive from
 \eqref{eq_1_1_assumption}, \eqref{eq_1_1_weight_assumption},
 \eqref{eq_1_2_kernel_assumption},
 \eqref{eq_1_2_kernel_weight_assumption}, \eqref{eq_1_3_assumption},
 \eqref{eq_1_3_weight_assumption} that
\begin{align}
&\sum_{m=2}^N2^{3m}(c_0M^{\sa(l+1-\hat{N}_{\beta})})^{\frac{m}{2}}\|V_m^{1,l+1}\|_{1,r,\eps_{\beta}^{\hat{N}_{\beta}-l}r'}\le
 c \eps_{\beta}\alpha^{-2},\label{eq_1_assumption}\\
&\sum_{m=2}^{N}2^{m}\alpha^m(c_0M^{\sa
 (l-\hat{N}_{\beta})})^{\frac{m}{2}}\|V_m^{1,l+1}\|_{1,r,\eps_{\beta}^{\hat{N}_{\beta}-l}r'}\le
 c\eps_{\beta}M^{-\sa}.\label{eq_1_weight_assumption}
\end{align}

Let us start the analysis of $l$-th step by studying $V^{1-1-1,l}$. By
 Lemma \ref{lem_simple_tree_expansion} its kernels satisfy
 \eqref{eq_temperature_translation_invariance}. By \eqref{eq_simple_1_1},
 \eqref{eq_scale_covariance_determinant_bound}, 
  \eqref{eq_1_1_bound} for
 $l+1$, \eqref{eq_artificial_scaling_properties} and the conditions $\alpha\ge 2$, $M^{\sa}\ge 2^4$,
\begin{align}
&\|V_0^{1-1-1,l}\|_{1,r,\eps_{\beta}^{\hat{N}_{\beta}-l}r'}\label{eq_1_1_1_0}\\
&\le
 \eps_{\beta}\|V_0^{1-1,l+1}\|_{1,r,\eps_{\beta}^{\hat{N}_{\beta}-l-1}r'}
+\eps_{\beta}\sum_{p=2}^N(c_0M^{\sa
 (l+1-\hat{N}_{\beta})})^{\frac{p}{2}}\|V_p^{1-1,l+1}\|_{1,r,\eps_{\beta}^{\hat{N}_{\beta}-l-1}r'}\notag\\
&\le 2\eps_{\beta}\alpha^{-2}L^{-d},\notag\\
&\sum_{m=2}^N\alpha^m(c_0M^{\sa(l-\hat{N}_{\beta})})^{\frac{m}{2}}\|V_m^{1-1-1,l}\|_{1,r,\eps_{\beta}^{\hat{N}_{\beta}-l}r'}\label{eq_1_1_1}\\
&\le
 \sum_{m=2}^N\alpha^m(c_0M^{\sa(l-\hat{N}_{\beta})})^{\frac{m}{2}}\sum_{p=m}^N2^p(c_0M^{\sa(l+1-\hat{N}_{\beta})})^{\frac{p-m}{2}}\|V_p^{1-1,l+1}\|_{1,r,\eps_{\beta}^{\hat{N}_{\beta}-l}r'}\notag\\
&\le \eps_{\beta}L^{-d}\sum_{m=2}^N2^mM^{-\frac{\sa}{2}m}\le
 c\eps_{\beta}M^{-\sa}L^{-d}.\notag
\end{align}

Next let us consider $V^{1-1-2,l}$. By Lemma
 \ref{lem_double_tree_expansion} the kernels of $V^{1-1-2,l}$ satisfy
 \eqref{eq_temperature_translation_invariance}. By \eqref{eq_double_1_1}
 and \eqref{eq_scale_covariance_determinant_bound}, 
\begin{align*}
&\|V^{1-1-2,l}_m\|_{1,r,\eps_{\beta}^{\hat{N}_{\beta}-l}r'}\\
&\le (c_0 M^{\sa
 (l+1-\hat{N}_{\beta})})^{-1-\frac{m}{2}}\\
&\quad\cdot \sum_{p_1,p_2=2}^{N}1_{p_1,p_2\in 2\N} 2^{2p_1+2p_2}(c_0 M^{\sa(l+1-\hat{N}_{\beta})})^{\frac{p_1+p_2}{2}}[V_{p_1,p_2}^{1-2,l+1},\tilde{\cC}_{l+1}]_{1,r,\eps_{\beta}^{\hat{N}_{\beta}-l}r'}1_{p_1+p_2-2\ge m}.
\end{align*}
Then by using  \eqref{eq_scale_covariance_coupled_decay_bound},
\eqref{eq_1_2_multiplied_assumption},
 \eqref{eq_1_2_multiplied_weight_assumption}
 and the condition
 $\alpha M^{-\frac{\sa}{2}}\ge 2$ we observe that
\begin{align}
&\|V_0^{1-1-2,l}\|_{1,r,\eps_{\beta}^{\hat{N}_{\beta}-l}r'}\label{eq_1_1_2_0}\\
&\le (c_0M^{\sa
 (l+1-\hat{N}_{\beta})})^{-1}c\eps_{\beta}\alpha^{-4}c_0
M^{(\sa-1-\sum^{d}_{j=1}\frac{1}{\sn_j})(l+1-\hat{N}_{\beta})} L^{-d}\notag\\
&\le c\eps_{\beta}\alpha^{-4} M^{-(\sum^d_{j=1}\frac{1}{\sn_j}+1)(l+1-\hat{N}_{\beta})}L^{-d},\notag\\
&\sum_{m=2}^N\alpha^m(c_0 M^{\sa
 (l-\hat{N}_{\beta})})^{\frac{m}{2}}\|V_m^{1-1-2,l}\|_{1,r,\eps_{\beta}^{\hat{N}_{\beta}-l}r'}\label{eq_1_1_2}\\
&\le c (c_0 M^{\sa
 (l+1-\hat{N}_{\beta})})^{-1}M^{\sa}\alpha^{-2}\notag\\
&\quad\cdot \sum_{p_1,p_2=2}^{N}1_{p_1,p_2\in 2\N} 2^{2p_1+2p_2}\alpha^{p_1+p_2}(c_0M^{\sa (l-\hat{N}_{\beta})})^{\frac{p_1+p_2}{2}}[V_{p_1,p_2}^{1-2,l+1},\tilde{\cC}_{l+1}]_{1,r,\eps_{\beta}^{\hat{N}_{\beta}-l}r'}\notag\\
&\le c
 \eps_{\beta}M^{-\sa-(\sum_{j=1}^d\frac{1}{\sn_j}+1)(l+1-\hat{N}_{\beta})}\alpha^{-2}L^{-d}.\notag
\end{align}

Next let us consider $V^{1-1-3,l,(n)}$ $(n\in \N_{\ge 2})$. By Lemma
 \ref{lem_simple_tree_expansion} the anti-symmetric kernels of
 $V^{1-1-3,l,(n)}$ satisfy
 \eqref{eq_temperature_translation_invariance}. Thus if
 $\sum_{n=2}^{\infty}V^{1-1-3,l,(n)}$ converges, those of $V^{1-1-3,l}$
 satisfy \eqref{eq_temperature_translation_invariance} too. Let us
 establish bound properties. Observe that
\begin{align*}
&V^{1-1-3,l,(n)}(\psi)\\
&=\sum_{q=1}^{n-1}\left(\begin{array}{c} n-1  \\ q \end{array}\right)
\frac{1}{(n-1)!}Tree(\{1,2,\cdots,n\},\cC_{l+1})V^{1,l+1}(\psi^1+\psi)\\
&\quad\cdot \prod_{j=2}^{q+1}V^{0-1,l+1}(\psi^j+\psi)\prod_{k=q+2}^nV^{0-2,l+1}(\psi^k+\psi)\Bigg|_{\psi^j=0\atop
 (\forall j\in \{1,2,\cdots,n\})}.
\end{align*}
By \eqref{eq_simple_1},
 \eqref{eq_scale_covariance_determinant_bound} and
 \eqref{eq_scale_covariance_decay_bound}, 
\begin{align}
&\|V_m^{1-1-3,l,(n)}\|_{1,r,\eps_{\beta}^{\hat{N}_{\beta}-l}r'}\label{eq_1_1_3_expansion}\\
&\le
 (c_0M^{\sa
 (l+1-\hat{N}_{\beta})})^{-n+1-\frac{m}{2}}2^{-2m+n-1}(c_0M^{(\sa-1-\sum_{j=1}^d\frac{1}{\sn_j})(l+1-\hat{N}_{\beta})})^{n-1}\notag\\
&\quad\cdot \sum_{p_1=2}^N2^{3p_1}(c_0M^{\sa
 (l+1-\hat{N}_{\beta})})^{\frac{p_1}{2}}\|V_{p_1}^{1,l+1}\|_{1,r,\eps_{\beta}^{\hat{N}_{\beta}-l}r'}\notag\\
&\quad\cdot\sum_{p_2=2}^N2^{3p_2}(c_0M^{\sa
 (l+1-\hat{N}_{\beta})})^{\frac{p_2}{2}}\|V_{p_2}^{0-1,l+1}\|_{0,\infty,r}\notag\\
&\quad\cdot \prod_{j=3}^n\left(\sum_{p_j=2}^N 2^{3p_j}(c_0M^{\sa
 (l+1-\hat{N}_{\beta})})^{\frac{p_j}{2}}(\|V_{p_j}^{0-1,l+1}\|_{0,\infty,r}+
\|V_{p_j}^{0-2,l+1}\|_{0,\infty,r})\right)\notag\\ 
&\quad\cdot 1_{\sum_{j=1}^{n}p_j-2(n-1)\ge m}.\notag
\end{align}
Substitution of 
 \eqref{eq_0_1_bound_assumption},
 \eqref{eq_0_2_kernel_bound_assumption},
\eqref{eq_1_assumption}
 and the condition $L^d\ge
 M^{(\sum_{j=1}^d\frac{1}{\sn_j}+1)(\hat{N}_{\beta}-N_{\beta})}$ yields
 that
\begin{align*}
&\|V_0^{1-1-3,l,(n)}\|_{1,r,\eps_{\beta}^{\hat{N}_{\beta}-l}r'}\\
&\le
 M^{-(\sum_{j=1}^d\frac{1}{\sn_j}+1)(l+1-\hat{N}_{\beta})(n-1)}c\eps_{\beta}\alpha^{-4}L^{-d}\left(c\alpha^{-2}L^{-d}
+c\alpha^{-4}M^{(\sum_{j=1}^d\frac{1}{\sn_j}+1)(l+1-\hat{N}_{\beta})}\right)^{n-2}\notag\\
&\le
 \eps_{\beta}M^{-(\sum_{j=1}^d\frac{1}{\sn_j}+1)(l+1-\hat{N}_{\beta})}L^{-d}(c\alpha^{-2})^n.
\end{align*}
Thus on the assumption $\alpha\ge c$,
\begin{align}
\sum_{n=2}^{\infty}\|V_0^{1-1-3,l,(n)}\|_{1,r,\eps_{\beta}^{\hat{N}_{\beta}-l}r'}\le
c
 \alpha^{-4}\eps_{\beta}M^{-(\sum_{j=1}^d\frac{1}{\sn_j}+1)(l+1-\hat{N}_{\beta})}L^{-d}.\label{eq_1_1_3_0}
\end{align}
Also by using  \eqref{eq_0_1_bound_weight_assumption},
 \eqref{eq_0_2_kernel_bound_weight_assumption},
\eqref{eq_1_weight_assumption}
and the assumptions $\alpha M^{-\frac{\sa}{2}}\ge 2^3$, $L^d\ge
 M^{(\sum_{j=1}^d\frac{1}{\sn_j}+1)(\hat{N}_{\beta}-N_{\beta})}$ we
 obtain from \eqref{eq_1_1_3_expansion} that 
\begin{align*}
&\sum_{m=2}^N\alpha^m(c_0 M^{\sa(l-\hat{N}_{\beta})})^{\frac{m}{2}}\|V_m^{1-1-3,l,(n)}\|_{1,r,\eps_{\beta}^{\hat{N}_{\beta}-l}r'}\\
&\le c^n \alpha^{-2(n-1)}M^{\sa
 (n-1)}M^{-(\sum_{j=1}^d\frac{1}{\sn_j}+1)(l+1-\hat{N}_{\beta})(n-1)}\eps_{\beta}M^{-2\sa}L^{-d}\\&\quad\cdot 
\left(c M^{-\sa}L^{-d}
+c
 M^{-2\sa+(\sum_{j=1}^d\frac{1}{\sn_j}+1)(l+1-\hat{N}_{\beta})}\right)^{n-2}\\
&\le
 M^{-\sa-(\sum_{j=1}^d\frac{1}{\sn_j}+1)(l+1-\hat{N}_{\beta})}\eps_{\beta}L^{-d}(c\alpha^{-2})^{n-1}.
\end{align*}
Therefore by assuming that $\alpha\ge c$,
\begin{align}
&\sum_{m=2}^N\alpha^m(c_0M^{\sa(l-\hat{N}_{\beta})})^{\frac{m}{2}}\sum_{n=2}^{\infty}\|V_m^{1-1-3,l,(n)}\|_{1,r,\eps_{\beta}^{\hat{N}_{\beta}-l}r'}\label{eq_1_1_3}\\
&\le
 c\alpha^{-2}M^{-\sa-(\sum_{j=1}^d\frac{1}{\sn_j}+1)(l+1-\hat{N}_{\beta})}\eps_{\beta}L^{-d}.\notag
\end{align}

Next let us deal with $V^{1-1-4,l,(n)}$ $(n\in \N_{\ge 2})$. By Lemma
 \ref{lem_double_tree_expansion} its anti-symmetric kernels satisfy
 \eqref{eq_temperature_translation_invariance}. Thus it suffices to
 prove the convergence of $\sum_{n=2}^{\infty}V^{1-1-4,l,(n)}$ in order
 to prove that the kernels of $V^{1-1-4,l}$ satisfy
 \eqref{eq_temperature_translation_invariance} as well. By \eqref{eq_double_1},
 \eqref{eq_scale_covariance_determinant_bound} and
 \eqref{eq_scale_covariance_decay_bound},
for $m\in \{0,2,\cdots,N\}$,
\begin{align}
&\|V_m^{1-1-4,l,(n)}\|_{1,r,\eps_{\beta}^{\hat{N}_{\beta}-l}r'}\label{eq_1_1_4_expansion}\\&\le
 (c_0 M^{\sa (l+1-\hat{N}_{\beta})})^{-n-\frac{m}{2}}
2^{-2m} (c_0
 M^{(\sa-1-\sum_{j=1}^d\frac{1}{\sn_j})(l+1-\hat{N}_{\beta})})^{n-1}\notag\\
&\quad\cdot \sum_{p_1,p_2=2}^{N}1_{p_1,p_2\in
 2\N}2^{3p_1+3p_2}(c_0M^{\sa
 (l+1-\hat{N}_{\beta})})^{\frac{p_1+p_2}{2}}[V_{p_1,p_2}^{0-2,l+1},\tilde{\cC}_{l+1}]_{1,\infty,r}\notag\\
&\quad\cdot \prod_{j=3}^n\left(\sum_{p_j=4}^N2^{3p_j}(c_0M^{\sa
 (l+1-\hat{N}_{\beta})})^{\frac{p_j}{2}}\|V_{p_j}^{0-2,l+1}\|_{1,\infty,r}
\right)\notag\\
&\quad\cdot
 \sum_{p_{n+1}=2}^N2^{3p_{n+1}}(c_0M^{\sa(l+1-\hat{N}_{\beta})})^{\frac{p_{n+1}}{2}}\|V_{p_{n+1}}^{1,l+1}\|_{1,r,\eps_{\beta}^{\hat{N}_{\beta}-l}r'}1_{\sum_{j=1}^{n+1}p_j-2n\ge
 m}.\notag
\end{align}
Then by substituting \eqref{eq_scale_covariance_coupled_decay_bound}, \eqref{eq_0_2_kernel_bound_assumption},
 \eqref{eq_0_2_multiplied_bound_assumption}, \eqref{eq_1_assumption} 
  we observe that
\begin{align*}
&\|V_0^{1-1-4,l,(n)}\|_{1,r,\eps_{\beta}^{\hat{N}_{\beta}-l}r'}\\
&\le M^{-\sa
 (l+1-\hat{N}_{\beta})-(\sum_{j=1}^d\frac{1}{\sn_j}+1)(l+1-\hat{N}_{\beta})(n-1)}
 c\alpha^{-4}M^{(\sum_{j=1}^d\frac{1}{\sn_j}+1)(l+1-\hat{N}_{\beta})}L^{-d}\\
&\quad\cdot 
 M^{(\sa-1-\sum_{j=1}^d\frac{1}{\sn_j})(l+1-\hat{N}_{\beta})}(c\alpha^{-4}M^{(\sum_{j=1}^d\frac{1}{\sn_j}+1)(l+1-\hat{N}_{\beta})})^{n-2}\eps_{\beta}\alpha^{-2}\\
&\le
 M^{-(\sum_{j=1}^d\frac{1}{\sn_j}+1)(l+1-\hat{N}_{\beta})}\eps_{\beta}L^{-d}\alpha^2(c\alpha^{-4})^n.
\end{align*}
Thus on the assumption $\alpha\ge c$,
\begin{align}
\sum_{n=2}^{\infty}\|V_0^{1-1-4,l,(n)}\|_{1,r,\eps_{\beta}^{\hat{N}_{\beta}-l}r'}\le
 c
 \alpha^{-6}M^{-(\sum_{j=1}^d\frac{1}{\sn_j}+1)(l+1-\hat{N}_{\beta})}\eps_{\beta}L^{-d}.\label{eq_1_1_4_0}
\end{align}
Also by using \eqref{eq_scale_covariance_coupled_decay_bound},
\eqref{eq_0_2_kernel_bound_weight_assumption},
 \eqref{eq_0_2_multiplied_bound_weight_assumption},
 \eqref{eq_1_weight_assumption} and the assumption $\alpha
 M^{-\frac{\sa}{2}}\ge 2^3$ we can derive from
 \eqref{eq_1_1_4_expansion} that
\begin{align*}
&\sum_{m=2}^{N}\alpha^m(
c_0M^{\sa(l-\hat{N}_{\beta})})^{\frac{m}{2}}\|V_m^{1-1-4,l,(n)}\|_{1,r,\eps_{\beta}^{\hat{N}_{\beta}-l}r'}\\
&\le c \alpha^{-2n}M^{\sa n- \sa
 (l+1-\hat{N}_{\beta})n+(\sa-1-\sum_{j=1}^d\frac{1}{\sn_j})(l+1-\hat{N}_{\beta})(n-1)}M^{-2\sa+(\sum_{j=1}^d\frac{1}{\sn_j}+1)(l+1-\hat{N}_{\beta})}L^{-d}\\
&\quad\cdot M^{(\sa-1-\sum_{j=1}^d\frac{1}{\sn_j})(l+1-\hat{N}_{\beta})}(cM^{-2\sa+(\sum_{j=1}^d\frac{1}{\sn_j}+1)(l+1-\hat{N}_{\beta})})^{n-2}\eps_{\beta}M^{-\sa}\\
&\le
 M^{\sa-(\sum_{j=1}^d\frac{1}{\sn_j}+1)(l+1-\hat{N}_{\beta})}\eps_{\beta}L^{-d}(c\alpha^{-2}M^{-\sa})^n,
\end{align*}
or on the assumption $\alpha\ge c$,
\begin{align}
&\sum_{m=2}^{N}\alpha^m(c_0M^{\sa(l-\hat{N}_{\beta})})^{\frac{m}{2}}
\sum_{n=2}^{\infty}\|V_m^{1-1-4,l,(n)}\|_{1,r,\eps_{\beta}^{\hat{N}_{\beta}-l}r'}\label{eq_1_1_4}\\
&\le c \alpha^{-4}M^{-\sa
 -(\sum_{j=1}^d\frac{1}{\sn_j}+1)(l+1-\hat{N}_{\beta})}\eps_{\beta}L^{-d}.\notag
\end{align}

Let us sum up \eqref{eq_1_1_1_0}, \eqref{eq_1_1_1}, \eqref{eq_1_1_2_0},
 \eqref{eq_1_1_2}, \eqref{eq_1_1_3_0}, \eqref{eq_1_1_3},
 \eqref{eq_1_1_4_0}, \eqref{eq_1_1_4}.
\begin{align}
&\|V_0^{1-1-1,l}\|_{1,r,\eps_{\beta}^{\hat{N}_{\beta}-l}r'}+
\|V_0^{1-1-2,l}\|_{1,r,\eps_{\beta}^{\hat{N}_{\beta}-l}r'}
+\sum_{n=2}^{\infty}\|V_0^{1-1-3,l,(n)}\|_{1,r,\eps_{\beta}^{\hat{N}_{\beta}-l}r'}\label{eq_1_1_0_sum}\\
&\quad+\sum_{n=2}^{\infty}\|V_0^{1-1-4,l,(n)}\|_{1,r,\eps_{\beta}^{\hat{N}_{\beta}-l}r'}\notag\\
&\le c \left(
1+\alpha^{-2}M^{-(\sum_{j=1}^d\frac{1}{\sn_j}+1)(l+1-\hat{N}_{\beta})} 
\right)\eps_{\beta}\alpha^{-2}L^{-d},\notag\\
&\sum_{m=2}^N \alpha^m(c_0 M^{\sa(l-\hat{N}_{\beta})})^{\frac{m}{2}}\Bigg(
\|V_m^{1-1-1,l}\|_{1,r,\eps_{\beta}^{\hat{N}_{\beta}-l}r'}+
\|V_m^{1-1-2,l}\|_{1,r,\eps_{\beta}^{\hat{N}_{\beta}-l}r'}\label{eq_1_1_sum}\\
&\quad+ \sum_{n=2}^{\infty}\|V_m^{1-1-3,l,(n)}\|_{1,r,\eps_{\beta}^{\hat{N}_{\beta}-l}r'}+
\sum_{n=2}^{\infty}\|V_m^{1-1-4,l,(n)}\|_{1,r,\eps_{\beta}^{\hat{N}_{\beta}-l}r'}
\Bigg)\notag\\
&\le
 c (M^{-\sa}+\alpha^{-2}M^{-\sa-(\sum_{j=1}^d\frac{1}{\sn_j}+1)(l+1-\hat{N}_{\beta})})\eps_{\beta}L^{-d}.\notag
\end{align}
Recalling \eqref{eq_scaling_factor}, one can see that under the
 assumptions $\alpha\ge c$, $M\ge c$ the right-hand side of
 \eqref{eq_1_1_0_sum}, \eqref{eq_1_1_sum} is less than
 $\alpha^{-2}L^{-d}$, $L^{-d}$ respectively. By setting
 $V^{1-1,l}:=\sum_{j=1}^4V^{1-1-j,l}$ we conclude that 
$$
V^{1-1,l}\in \cQ'(r,\eps_{\beta}^{\hat{N}_{\beta}-l}r',l).
$$
Recall that we have also assumed that $\alpha M^{-\frac{\sa}{2}}\ge
 2^3$, $L^{d}\ge
 M^{(\sum_{j=1}^d\frac{1}{\sn_j}+1)(\hat{N}_{\beta}-N_{\beta})}$ to
 reach this conclusion.

Let us study $V^{1-2-1,l}$. By Lemma \ref{lem_divided_tree_expansion}
 there exist $V_{a,b}^{1-2-1,l}\in
 \Map(\overline{D(r)}\times\C^2,\Map(I^a\times I^b,\C))$
 $(a,b=2,4,\cdots,N)$
 such that for any $(u,\bla)\in
 \overline{D(r)}\times \C^2$, $V_{a,b}^{1-2-1,l}(u,\bla):I^a\times
 I^b\to \C$ is bi-anti-symmetric and satisfies
 \eqref{eq_temperature_translation_invariance},
 \eqref{eq_temperature_integral_vanish} and 
\begin{align*}
V^{1-2-1,l}(u,\bla)(\psi)=\sum_{a,b=2}^{N}1_{a,b\in 2\N}\frah^{a+b}\sum_{\bX\in
 I^a\atop \bY\in I^b}V_{a,b}^{1-2-1,l}(u,\bla)(\bX,\bY)\psi_{\bX}\psi_{\bY}.
\end{align*}
Moreover it follows from the definition and the induction hypothesis
 that 
\begin{align*}
V^{1-2-1,l}\in C\left(\overline{D(r)}\times
 \C^2,\bigwedge_{even}\cV\right)\cap C^{\o}\left(D(r)\times
 \C^2,\bigwedge_{even}\cV\right)
\end{align*}
and $\bla\mapsto V^{1-2-1,l}(u,\bla)(\psi):\C^2\mapsto
 \bigwedge_{even}\cV$ is linear for any $u\in \overline{D(r)}$. Let us
 establish bound properties. By \eqref{eq_divided_1_1} and
 \eqref{eq_scale_covariance_determinant_bound}, for any $a,b\in
 \{2,4,\cdots,N\}$,
\begin{align*}
&\|V_{a,b}^{1-2-1,l}\|_{1,r,\eps_{\beta}^{\hat{N}_{\beta}-l}r'}\\
&\le \sum_{p=a}^{N}\sum_{q=b}^{N}1_{p,q\in
 2\N}\left(\begin{array}{c} p \\
	   a\end{array}\right)\left(\begin{array}{c} q \\
				    b\end{array}\right)(c_0M^{\sa(l+1-\hat{N}_{\beta})})^{\frac{1}{2}(p+q-a-b)}\|V_{p,q}^{1-2,l+1}\|_{1,r,\eps_{\beta}^{\hat{N}_{\beta}-l}r'}.
\end{align*}
Then by using \eqref{eq_1_2_bound} for $l+1$,
 \eqref{eq_artificial_scaling_properties} and the conditions $\alpha\ge
 2$, $M^{\sa}\ge 2^4$ we have that
\begin{align}
&\sum_{a,b=2}^{N}1_{a,b\in
 2\N}\alpha^{a+b}(c_0 M^{\sa(l-\hat{N}_{\beta})})^{\frac{a+b}{2}}\|V_{a,b}^{1-2-1,l}\|_{1,r,\eps_{\beta}^{\hat{N}_{\beta}-l}r'}\label{eq_1_2_1}\\
&\le \sum_{a,b=2}^{N}1_{a,b\in
 2\N}M^{-\frac{a+b}{2}\sa}\alpha^{a+b}\eps_{\beta}
\sum_{p=a}^{N}\sum_{q=b}^{N}
1_{p,q\in
 2\N}2^{p+q}(c_0M^{\sa(l+1-\hat{N}_{\beta})})^{\frac{p+q}{2}}\|V_{p,q}^{1-2,l+1}\|_{1,r,\eps_{\beta}^{\hat{N}_{\beta}-l-1}r'}\notag\\
&\le \sum_{a,b=2}^{N}1_{a,b\in
 2\N}2^{a+b}
M^{-\frac{a+b}{2}\sa}\eps_{\beta}
\le c M^{-2\sa}\eps_{\beta}.\notag
\end{align}
Based on \eqref{eq_divided_multiplied_1_1},
 \eqref{eq_scale_covariance_determinant_bound},
 \eqref{eq_1_2_multiplied_bound} for $l+1$,
 \eqref{eq_artificial_scaling_properties} and the conditions $\alpha\ge
 2$, $M^{\sa}\ge 2^4$, an argument parallel to the above shows that 
\begin{align}
&\sum_{a,b=2}^{N}1_{a,b\in
 2\N}\alpha^{a+b}(c_0M^{\sa(l-\hat{N}_{\beta})})^{\frac{a+b}{2}}[V_{a,b}^{1-2-1,l},g]_{1,r,\eps_{\beta}^{\hat{N}_{\beta}-l}r'}\le
 c M^{-2\sa}\eps_{\beta}L^{-d}\|g\|,
\label{eq_1_2_1_multiplied}
 \end{align}
for any anti-symmetric function $g:I^2\to \C$.

Next let us consider $V^{1-2-2,l,(n)}$ $(n\in \N_{\ge 2})$. Lemma
 \ref{lem_divided_tree_expansion} ensures that $V^{1-2-2,l,(n)}$ can be
 written with bi-anti-symmetric kernels $V_{a,b}^{1-2-2,l,(n)}\in
 \Map(\overline{D(r)}\times \C^2,$\\
 $\Map(I^a\times I^b,\C))$
 $(a,b=2,4,\cdots,N)$ and for any $(u,\bla)\in
 \overline{D(r)}\times \C^2$, $a,b\in \{2,4,\cdots,N\}$ the
 kernel $V_{a,b}^{1-2-2,l,(n)}(u,\bla)$ satisfies
 \eqref{eq_temperature_translation_invariance} and
 \eqref{eq_temperature_integral_vanish}. Moreover we can deduce from
 \eqref{eq_divided_kernel} and the induction hypothesis that
\begin{align*} 
V_{a,b}^{1-2-2,l,(n)}\in
 C(\overline{D(r)}\times\C^2,\Map(I^a\times I^b,\C))\cap 
C^{\o}(D(r)\times\C^2,\Map(I^a\times I^b,\C))
\end{align*}
and $\bla\mapsto V^{1-2-2,l,(n)}(u,\bla)(\psi):\C^2\to
 \bigwedge_{even}\cV$ is linear for any $u\in \overline{D(r)}$. These
 properties must hold true for $V^{1-2-2,l}$, once
 $\sum_{n=2}^{\infty}V^{1-2-2,l,(n)}$ is proved to be uniformly convergent.
By using \eqref{eq_divided_1},
 \eqref{eq_scale_covariance_determinant_bound},
 \eqref{eq_scale_covariance_decay_bound} we obtain that
\begin{align*}
&\|V_{a,b}^{1-2-2,l,(n)}\|_{1,r,\eps_{\beta}^{\hat{N}_{\beta}-l}r'}\\
&\le
 \frac{1}{(n-1)!}\sum_{m=0}^{n-1}\sum_{(\{s_j\}_{j=1}^{m+1},\{t_k\}_{k=1}^{n-m})\in S(n,m)}\\
&\quad\cdot (1_{m\neq 0}(m-1)!+1_{m=0})(1_{m\neq n-1}(n-m-2)!+1_{m=n-1})\\
&\quad\cdot 2^{-2a-2b}(c_0M^{\sa(l+1-\hat{N}_{\beta})})^{-n+1-\frac{1}{2}(a+b)}(c_0M^{(\sa-1-\sum_{j=1}^d\frac{1}{\sn_j})(l+1-\hat{N}_{\beta})})^{n-1}\\
&\quad\cdot \sum_{p_1,q_1=2}^{N}1_{p_1,q_1\in
 2\N}2^{3p_1+3q_1}
(c_0M^{\sa(l+1-\hat{N}_{\beta})})^{\frac{p_1+q_1}{2}}\|V_{p_1,q_1}^{0-2,l+1}\|_{1,\infty,r}\\
&\quad\cdot \prod_{j=2}^{m+1}\left(\sum_{p_j=2}^N2^{3p_j}(c_0M^{\sa(l+1-\hat{N}_{\beta})})^{\frac{p_j}{2}}
\left(1_{s_j\neq n}\|V_{p_j}^{0-2,l+1}\|_{1,\infty,r}+
1_{s_j=
 n}\|V_{p_j}^{1,l+1}\|_{1,r,\eps_{\beta}^{\hat{N}_{\beta}-l}r'}\right)\right)\\
&\quad \cdot 
\prod_{k=2}^{n-m}\left(\sum_{q_k=2}^N2^{3q_k}(c_0M^{\sa(l+1-\hat{N}_{\beta})})^{\frac{q_k}{2}}
\left(1_{t_k\neq n}\|V_{q_k}^{0-2,l+1}\|_{1,\infty,r}+
1_{t_k=
 n}\|V_{q_k}^{1,l+1}\|_{1,r,\eps_{\beta}^{\hat{N}_{\beta}-l}r'}\right)
\right)\\
&\quad\cdot 
1_{\sum_{j=1}^{m+1}p_j-2m\ge a}
1_{\sum_{k=1}^{n-m}q_k-2(n-m-1)\ge b}.
\end{align*}
Then by \eqref{eq_0_2_bound_weight_assumption},
 \eqref{eq_0_2_kernel_bound_weight_assumption}, \eqref{eq_divided_tree},
 \eqref{eq_1_weight_assumption} and the
 assumption $\alpha M^{-\frac{\sa}{2}}\ge 2^3$,
\begin{align}
&\sum_{a,b=2}^{N}1_{a,b\in
 2\N}\alpha^{a+b}(c_0 M^{\sa(l-\hat{N}_{\beta})})^{\frac{a+b}{2}}\|V_{a,b}^{1-2-2,l,(n)}\|_{1,r,\eps_{\beta}^{\hat{N}_{\beta}-l}r'}\label{eq_1_2_2_pre}\\
&\le c^n
 \alpha^{-2(n-1)}M^{\sa(n-1)}M^{-\sa(l+1-\hat{N}_{\beta})(n-1)+(\sa-1-\sum_{j=1}^d\frac{1}{\sn_j})(l+1-\hat{N}_{\beta})(n-1)}\notag\\
&\quad\cdot
 \frac{1}{(n-1)!}\sum_{m=0}^{n-1}\sum_{(\{s_j\}_{j=1}^m,\{t_k\}_{k=1}^{n-m})\in
 S(n,m)}\notag\\
&\quad\cdot (1_{m\neq 0}(m-1)!+1_{m=0})
(1_{m\neq n-1}(n-m-2)!+1_{m=n-1})\notag\\
&\quad\cdot \sum_{p_1,q_1=2}^{N}1_{p_1,q_1\in
 2\N}2^{p_1+q_1}\alpha^{p_1+q_1}
(c_0M^{\sa(l-\hat{N}_{\beta})})^{\frac{p_1+q_1}{2}}\|V_{p_1,q_1}^{0-2,l+1}\|_{1,\infty,r}\notag\\
&\quad\cdot \prod_{j=2}^{m+1}\left(\sum_{p_j=2}^N2^{p_j}\alpha^{p_j}(c_0M^{\sa(l-\hat{N}_{\beta})})^{\frac{p_j}{2}}
\left(1_{s_j\neq n}\|V_{p_j}^{0-2,l+1}\|_{1,\infty,r}+
1_{s_j=
 n}\|V_{p_j}^{1,l+1}\|_{1,r,\eps_{\beta}^{\hat{N}_{\beta}-l}r'}\right)\right)\notag\\
&\quad \cdot 
\prod_{k=2}^{n-m}\left(\sum_{q_k=2}^N2^{q_k}\alpha^{q_k}
(c_0M^{\sa(l-\hat{N}_{\beta})})^{\frac{q_k}{2}}
\left(1_{t_k\neq n}\|V_{q_k}^{0-2,l+1}\|_{1,\infty,r}+
1_{t_k=
 n}\|V_{q_k}^{1,l+1}\|_{1,r,\eps_{\beta}^{\hat{N}_{\beta}-l}r'}\right)
\right)\notag\\
&\le
 \alpha^{-2(n-1)}M^{\sa(n-1)-(\sum_{j=1}^d\frac{1}{\sn_j}+1)(l+1-\hat{N}_{\beta})(n-1)}(c
 M^{-2\sa+(\sum_{j=1}^d\frac{1}{\sn_j}+1)(l+1-\hat{N}_{\beta})})^{n-1}\eps_{\beta}M^{-\sa}\notag\\
&\le (c\alpha^{-2}M^{-\sa})^{n-1}\eps_{\beta}M^{-\sa}.\notag
\end{align}
Thus on the assumption $\alpha\ge c$,
\begin{align}
\sum_{a,b=2}^{N}1_{a,b\in
 2\N}\alpha^{a+b}(c_0M^{\sa(l-\hat{N}_{\beta})})^{\frac{a+b}{2}}\sum_{n=2}^{\infty}\|V_{a,b}^{1-2-2,l,(n)}\|_{1,r,\eps_{\beta}^{\hat{N}_{\beta}-l}r'}\le
 c\alpha^{-2}M^{-2\sa}\eps_{\beta}.\label{eq_1_2_2}
\end{align}
On the other hand, one can apply \eqref{eq_divided_multiplied_1},
 \eqref{eq_scale_covariance_determinant_bound},
 \eqref{eq_scale_covariance_decay_bound} to derive that for any
 anti-symmetric function $g:I^2\to \C$,
\begin{align*}
&[V_{a,b}^{1-2-2,l,(n)},g]_{1,r,\eps_{\beta}^{\hat{N}_{\beta}-l}r'}\\
&\le
 \frac{1}{(n-1)!}\sum_{m=0}^{n-1}\sum_{(\{s_j\}_{j=1}^{m+1},\{t_k\}_{k=1}^{n-m})\in S(n,m)}\\
&\quad\cdot (1_{m\neq 0}(m-1)!+1_{m=0})(1_{m\neq n-1}(n-m-2)!+1_{m=n-1})\\
&\quad\cdot 2^{-2a-2b}
(c_0M^{\sa(l+1-\hat{N}_{\beta})})^{-n+1-\frac{1}{2}(a+b)}(c_0M^{(\sa-1-\sum_{j=1}^d\frac{1}{\sn_j})(l+1-\hat{N}_{\beta})})^{n-2}\\
&\quad\cdot \sum_{p_1,q_1=2}^{N}1_{p_1,q_1\in
 2\N}2^{3p_1+3q_1}
(c_0M^{\sa(l+1-\hat{N}_{\beta})})^{\frac{p_1+q_1}{2}}\\
&\quad\cdot 
\left([V_{p_1,q_1}^{0-2,l+1},g]_{1,\infty,r}c_0
 M^{(\sa-1-\sum_{j=1}^d\frac{1}{\sn_j})(l+1-\hat{N}_{\beta})}+
[V_{p_1,q_1}^{0-2,l+1},\tilde{\cC}_{l+1}]_{1,\infty,r}\|g\|_{1,\infty}\right)\\
&\quad\cdot \prod_{j=2}^{m+1}\Bigg(\sum_{p_j=2}^N2^{3p_j}(c_0M^{\sa(l+1-\hat{N}_{\beta})})^{\frac{p_j}{2}}
\left(1_{s_j\neq  n}\|V_{p_j}^{0-2,l+1}\|_{1,\infty,r}+1_{s_j=
 n}\|V_{p_j}^{1,l+1}\|_{1,r,\eps_{\beta}^{\hat{N}_{\beta}-l}r'}
 \right)\Bigg)\\
&\quad \cdot 
\prod_{k=2}^{n-m}\Bigg(\sum_{q_k=2}^N2^{3q_k}(c_0M^{\sa(l+1-\hat{N}_{\beta})})^{\frac{q_k}{2}}
\left(1_{t_k\neq n}\|V_{q_k}^{0-2,l+1}\|_{1,\infty,r}+1_{t_k=
 n}\|V_{q_k}^{1,l+1}\|_{1,r,\eps_{\beta}^{\hat{N}_{\beta}-l}r'}\right)\Bigg)\\
&\quad\cdot 
1_{\sum_{j=1}^{m+1}p_j-2m\ge a}
1_{\sum_{k=1}^{n-m}q_k-2(n-m-1)\ge b}.
\end{align*}
Then by substituting 
\eqref{eq_scale_covariance_coupled_decay_bound},
\eqref{eq_0_2_kernel_bound_weight_assumption},
\eqref{eq_0_2_multiplied_bound_weight_assumption}, \eqref{eq_divided_tree},
 \eqref{eq_1_weight_assumption},  using the inequality
 $\|g\|_{1,\infty}\le \|g\|$ and 
the assumption $\alpha M^{-\frac{\sa}{2}}\ge 2^3$ and computing in a
 parallel way to \eqref{eq_1_2_2_pre} one reaches that
\begin{align*}
&\sum_{a,b=2}^{N}1_{a,b\in
 2\N}\alpha^{a+b}(c_0M^{\sa(l-\hat{N}_{\beta})})^{\frac{a+b}{2}}[V_{a,b}^{1-2-2,l,(n)},g]_{1,r,\eps_{\beta}^{\hat{N}_{\beta}-l}r'}\\
&\le
(c\alpha^{-2}M^{-\sa})^{n-1}\eps_{\beta}M^{-\sa}L^{-d}\|g\|, 
\end{align*}
or on the assumption $\alpha\ge c$,
\begin{align}
&\sum_{a,b=2}^{N}1_{a,b\in
 2\N}\alpha^{a+b}(c_0M^{\sa(l-\hat{N}_{\beta})})^{\frac{a+b}{2}}
\sum_{n=2}^{\infty}
[V_{a,b}^{1-2-2,l,(n)},g]_{1,r,\eps_{\beta}^{\hat{N}_{\beta}-l}r'}\label{eq_1_2_2_multiplied}\\
&\le c \alpha^{-2}M^{-2\sa} \eps_{\beta}L^{-d}\|g\|.\notag
\end{align}

Here let us combine \eqref{eq_1_2_1}, \eqref{eq_1_2_1_multiplied} with
 \eqref{eq_1_2_2}, \eqref{eq_1_2_2_multiplied} respectively to derive that 
\begin{align}
&\sum_{a,b=2}^{N}1_{a,b\in
 2\N}\alpha^{a+b}(c_0 M^{\sa(l-\hat{N}_{\beta})})^{\frac{a+b}{2}}\label{eq_1_2}\\&\quad\cdot\left(\|V_{a,b}^{1-2-1,l}\|_{1,r,\eps_{\beta}^{\hat{N}_{\beta}-l}r'}
+\sum_{n=2}^{\infty}\|V_{a,b}^{1-2-2,l,(n)}\|_{1,r,\eps_{\beta}^{\hat{N}_{\beta}-l}r'}
\right)\notag\\
&\le c(M^{-2\sa}+\alpha^{-2}M^{-2\sa})\eps_{\beta},\notag\\
&\sum_{a,b=2}^{N}1_{a,b\in
 2\N}\alpha^{a+b}(c_0M^{\sa(l-\hat{N}_{\beta})})^{\frac{a+b}{2}}\label{eq_1_2_multiplied}\\
&\quad \cdot\left([V_{a,b}^{1-2-1,l},g]_{1,r,\eps_{\beta}^{\hat{N}_{\beta}-l}r'}
+\sum_{n=2}^{\infty}[V_{a,b}^{1-2-2,l,(n)},g]_{1,r,\eps_{\beta}^{\hat{N}_{\beta}-l}r'}\right)\notag\\
&\le c(M^{-2\sa}+\alpha^{-2}M^{-2\sa})\eps_{\beta} L^{-d}\|g\|.\notag
\end{align}
If we assume that $M \ge c$, the right-hand side of \eqref{eq_1_2},
 \eqref{eq_1_2_multiplied} is less than 1, $L^{-d}\|g\|$
 respectively. By setting
 $V^{1-2,l}:=V^{1-2-1,l}+\sum_{n=2}^{\infty}V^{1-2-2,l,(n)}$ we conclude
 that 
$$V^{1-2,l}\in \cR'(r,\eps_{\beta}^{\hat{N}_{\beta}-l}r',l).$$
Remind us that we have also assumed $\alpha\ge c$,
$\alpha M^{-\frac{\sa}{2}}\ge 2^3$
 on the way to this result.

It remains to analyze $V^{2,l,(n)}$ $(n\in \N_{\ge 1})$. It can be seen
 from the definition, the induction hypothesis and the conditions $h\ge
 1$, $\eps_{\beta}\le 1$ that 
\begin{align}
&V^{2,l,(n)}\label{eq_temporal_quadratic_properties}\\
&\in C\left(\overline{D(r)}\times
 \overline{D(h^{l-\hat{N}_{\beta}}\eps_{\beta}^{\hat{N}_{\beta}-l}r')}^2,\bigwedge_{even}\cV\right)\cap 
C^{\o}\left(D(r)\times
 D(h^{l-\hat{N}_{\beta}}\eps_{\beta}^{\hat{N}_{\beta}-l}r')^2,\bigwedge_{even}\cV\right),\notag\\
&V^{2,l,(n)}(u,\b0)(\psi)=\frac{\partial }{\partial
 \la_j}V^{2,l,(n)}(u,\b0)(\psi)=0,\quad (\forall j\in \{1,2\},\ u\in
 D(r)).\notag
\end{align}
For $(u,\bla)\in \overline{D(r)}\times
 \overline{D(h^{l-\hat{N}_{\beta}}\eps_{\beta}^{\hat{N}_{\beta}-l}r')}^2$
 Lemma \ref{lem_simple_tree_expansion} implies that the anti-symmetric
 kernels of $V^{2,l,(n)}(u,\bla)(\psi)$ satisfy
 \eqref{eq_temperature_translation_invariance}. If
 $\sum_{n=1}^{\infty}V^{2,l,(n)}$ is uniformly convergent, then
 $V^{2,l}$ and its kernels automatically satisfy the properties
 \eqref{eq_temporal_quadratic_properties},
 \eqref{eq_temperature_translation_invariance}. By \eqref{eq_simple_1_1} and
 \eqref{eq_scale_covariance_determinant_bound}, for any $m\in
 \{0,2,\cdots,N\}$,
\begin{align*}
\|V_m^{2,l,(1)}\|_{1,r,h^{l-\hat{N}_{\beta}}\eps_{\beta}^{\hat{N}_{\beta}-l}r'}
\le \sum_{p=m}^{N}\left(\begin{array}{c}p\\ m\end{array}\right)
(c_0M^{\sa(l+1-\hat{N}_{\beta})})^{\frac{p-m}{2}}\|V_p^{2,l+1}\|_{1,r,h^{l-\hat{N}_{\beta}}\eps_{\beta}^{\hat{N}_{\beta}-l}r'}.
\end{align*}
By \eqref{eq_2_bound} for $l+1$,
 \eqref{eq_artificial_scaling_properties} and the conditions $h\ge 1$,
 $M^{\sa}\ge 2^4$, $\alpha\ge 2$,
\begin{align}
&\|V_0^{2,l,(1)}\|_{1,r,h^{l-\hat{N}_{\beta}}\eps_{\beta}^{\hat{N}_{\beta}-l}r'}\le \sum_{p=0}^{N}
(c_0M^{\sa(l+1-\hat{N}_{\beta})})^{\frac{p}{2}}\|V_p^{2,l+1}\|_{1,r,h^{l-\hat{N}_{\beta}}\eps_{\beta}^{\hat{N}_{\beta}-l}r'}\le
 c \eps_{\beta}\alpha^{-2},\label{eq_2_0_1}\\
&\sum_{m=2}^{N}\alpha^m(c_0M^{\sa
 (l-\hat{N}_{\beta})})^{\frac{m}{2}}\|V_m^{2,l,(1)}\|_{1,r,h^{l-\hat{N}_{\beta}}\eps_{\beta}^{\hat{N}_{\beta}-l}r'}\label{eq_2_1}\\
&\le \sum_{m=2}^N\alpha^mM^{-\frac{\sa}{2}m}
\sum_{p=m}^N2^p(c_0M^{\sa(l+1-\hat{N}_{\beta})})^{\frac{p}{2}}
\|V_p^{2,l+1}\|_{1,r,h^{l-\hat{N}_{\beta}}\eps_{\beta}^{\hat{N}_{\beta}-l}r'}\notag\\
&\le \eps_{\beta}\sum_{m=2}^N2^mM^{-\frac{\sa}{2}m}\le
 c\eps_{\beta}M^{-\sa}.\notag
\end{align}

Let $n\in \N_{\ge 2}$. Take $b_j\in \{0,1,2\}$ $(j=1,2,\cdots,n)$
 satisfying $\sum_{j=1}^nb_j\ge 2$. 
Set 
\begin{align*}
&V^{2,l,(n),(b_j)_{j=1}^n}(\psi):=\frac{1}{n!}Tree(\{1,2,\cdots,n\},\cC_{l+1})\prod_{j=1}^nV^{b_j,l+1}(\psi^j+\psi)\Bigg|_{\psi^j=0\atop
 (\forall j\in \{1,2,\cdots,n\})}.
\end{align*}
There exists $\s\in \S_n$ such that $b_{\s(1)}\neq 0$. It follows from 
the definition
 of the tree expansion \eqref{eq_tree_expansion} that  
\begin{align*}
&V^{2,l,(n),(b_j)_{j=1}^n}(\psi)=\frac{1}{n!}Tree(\{1,2,\cdots,n\},\cC_{l+1})\prod_{j=1}^nV^{b_{\s(j)},l+1}(\psi^j+\psi)\Bigg|_{\psi^j=0\atop
 (\forall j\in \{1,2,\cdots,n\})}.
\end{align*}
We can apply \eqref{eq_simple_1},
 \eqref{eq_scale_covariance_determinant_bound},
 \eqref{eq_scale_covariance_decay_bound} and
 \eqref{eq_artificial_scaling_properties} to derive that for any $m\in
 \{0,2,\cdots,N\}$
\begin{align}
&\|V_m^{2,l,(n),(b_j)_{j=1}^n}\|_{1,r,h^{l-\hat{N}_{\beta}}
\eps_{\beta}^{\hat{N}_{\beta}-l}r'}\label{eq_2_expansion_pre}\\
&\le 
 (c_0 M^{\sa (l+1-\hat{N}_{\beta})})^{-n+1-\frac{m}{2}}
2^{-2m} (c_0
 M^{(\sa-1-\sum_{j=1}^d\frac{1}{\sn_j})(l+1-\hat{N}_{\beta})})^{n-1}\notag\\
&\quad\cdot \sum_{p_1=2}^{N}2^{3p_1}(c_0M^{\sa
 (l+1-\hat{N}_{\beta})})^{\frac{p_1}{2}}\sum_{a_1=1}^2\|V_{p_1}^{a_1,l+1}\|_{1,r,h^{l-\hat{N}_{\beta}}
\eps_{\beta}^{\hat{N}_{\beta}-l}r'}\notag\\
&\quad\cdot \prod_{j=2}^n\left(\sum_{p_j=2}^N2^{3p_j}(c_0M^{\sa
 (l+1-\hat{N}_{\beta})})^{\frac{p_j}{2}}
\left(\|V_{p_j}^{0,l+1}\|_{1,\infty,r}+\sum_{a_j=1}^2
\|V_{p_j}^{a_j,l+1}\|_{1,r,h^{l+1-\hat{N}_{\beta}}
\eps_{\beta}^{\hat{N}_{\beta}-l}r'}
\right)\right)\notag\\
&\quad\cdot 1_{\sum_{j=1}^np_j-2(n-1)\ge m}.\notag
\end{align}
Here we used the condition $h\ge 2$ to apply
 \eqref{eq_artificial_scaling_properties} with $\eps=1/h$.
Since $V^{2,l,(n)}$ is the sum of $V^{2,l,(n),(b_j)_{j=1}^n}$ over
 possible $(b_j)_{j=1}^n$,
\begin{align}
\|V_m^{2,l,(n)}\|_{1,r,h^{l-\hat{N}_{\beta}}
\eps_{\beta}^{\hat{N}_{\beta}-l}r'}
\le 3^n\cdot(\text{R. H. S of }\eqref{eq_2_expansion_pre}).
\label{eq_2_expansion}
\end{align}
 We need to
 use the
 following inequalities which are derived from
 \eqref{eq_0_1_bound_assumption},
 \eqref{eq_0_1_bound_weight_assumption},
 \eqref{eq_0_2_kernel_bound_assumption},
 \eqref{eq_0_2_kernel_bound_weight_assumption} and the assumption 
$L^{d}\ge
 M^{(\sum_{j=1}^d\frac{1}{\sn_j}+1)(\hat{N}_{\beta}-N_{\beta})}$.
\begin{align}
&\sum_{m=2}^N2^{3m}(c_0M^{\sa(l+1-\hat{N}_{\beta})})^{\frac{m}{2}}\|V_m^{0,l+1}\|_{1,\infty,r}
\le c\alpha^{-2}M^{(\sum_{j=1}^d\frac{1}{\sn_j}+1)(l+1-\hat{N}_{\beta})},\label{eq_0_assumption}\\
&\sum_{m=2}^N2^m\alpha^m(c_0M^{\sa(l-\hat{N}_{\beta})})^{\frac{m}{2}}\|V_m^{0,l+1}\|_{1,\infty,r}
\le c M^{-\sa
 +(\sum_{j=1}^d\frac{1}{\sn_j}+1)(l+1-\hat{N}_{\beta})}.\label{eq_0_weight_assumption}\end{align}
By inserting \eqref{eq_2_assumption}, \eqref{eq_1_assumption}, \eqref{eq_0_assumption}
 into \eqref{eq_2_expansion} and recalling
 \eqref{eq_scaling_factor}, $h\ge 1$ and $l+1-\hat{N}_{\beta}\le 0$ we have that
\begin{align*}
&\|V_0^{2,l,(n)}\|_{1,r,h^{l-\hat{N}_{\beta}}
\eps_{\beta}^{\hat{N}_{\beta}-l}r'}\\
&\le c
 M^{-(\sum_{j=1}^d\frac{1}{\sn_j}+1)(l+1-\hat{N}_{\beta})(n-1)}\eps_{\beta}\alpha^{-2}\left(c\alpha^{-2}M^{(\sum_{j=1}^d\frac{1}{\sn_j}+1)(l+1-\hat{N}_{\beta})}+c \eps_{\beta}\alpha^{-2}\right)^{n-1}\\
&\le \eps_{\beta}(c\alpha^{-2})^n,
\end{align*}
or by assuming that $\alpha\ge c$,
\begin{align}
\sum_{n=2}^{\infty}\|V_0^{2,l,(n)}\|_{1,r,h^{l-\hat{N}_{\beta}}
\eps_{\beta}^{\hat{N}_{\beta}-l}r'}\le
 c\eps_{\beta}\alpha^{-4}.\label{eq_2_0_higher}
\end{align}
Also by combining 
 \eqref{eq_2_weight_assumption}, \eqref{eq_1_weight_assumption}, \eqref{eq_0_weight_assumption} with
 \eqref{eq_2_expansion}, using the condition $\alpha
 M^{-\frac{\sa}{2}}\ge 2^3$, $h\ge 1$, $l+1-\hat{N}_{\beta}\le 0$ and recalling
 \eqref{eq_scaling_factor} we see that  
\begin{align*}
&\sum_{m=2}^N\alpha^m(c_0M^{\sa(l-\hat{N}_{\beta})})^{\frac{m}{2}}
\|V_m^{2,l,(n)}\|_{1,r,h^{l-\hat{N}_{\beta}}
\eps_{\beta}^{\hat{N}_{\beta}-l}r'}\\
&\le c
 \alpha^{-2(n-1)}M^{\sa(n-1)-(\sum_{j=1}^d\frac{1}{\sn_j}+1)(l+1-\hat{N}_{\beta})(n-1)}\eps_{\beta}M^{-\sa}\\
&\quad\cdot \left(c M^{-\sa+(\sum_{j=1}^d\frac{1}{\sn_j}+1)(l+1-\hat{N}_{\beta})}+c\eps_{\beta}M^{-\sa}\right)^{n-1}\\
&\le \eps_{\beta}M^{-\sa}(c\alpha^{-2})^{n-1},
\end{align*}
or by the condition $\alpha\ge c$, 
\begin{align}
\sum_{m=2}^N\alpha^m(c_0M^{\sa(l-\hat{N}_{\beta})})^{\frac{m}{2}}
\sum_{n=2}^{\infty}\|V_m^{2,l,(n)}\|_{1,r,h^{l-\hat{N}_{\beta}}
\eps_{\beta}^{\hat{N}_{\beta}-l}r'}\le c \eps_{\beta}
 M^{-\sa}\alpha^{-2}.\label{eq_2_higher}
\end{align}
By summing up \eqref{eq_2_0_1}, \eqref{eq_2_1}, \eqref{eq_2_0_higher},
 \eqref{eq_2_higher} we obtain that 
\begin{align}
&\sum_{n=1}^{\infty}\|V_0^{2,l,(n)}\|_{1,r,h^{l-\hat{N}_{\beta}}
\eps_{\beta}^{\hat{N}_{\beta}-l}r'}\le
 c\eps_{\beta}\alpha^{-2},\label{eq_2_0}\\
&\sum_{m=2}^N\alpha^m(c_0M^{\sa(l-\hat{N}_{\beta})})^{\frac{m}{2}}
\sum_{n=1}^{\infty}\|V_m^{2,l,(n)}\|_{1,r,h^{l-\hat{N}_{\beta}}
\eps_{\beta}^{\hat{N}_{\beta}-l}r'}\le c \eps_{\beta}
 M^{-\sa}.\label{eq_2}
\end{align}
Under the assumption $M\ge c$, the right-hand side of \eqref{eq_2_0},
 \eqref{eq_2} becomes less than $\alpha^{-2}$, 1 respectively. Thus by
 setting $V^{2,l}:=\sum_{n=1}^{\infty}V^{2,l,(n)}$ we conclude that 
$$
V^{2,l}\in \cW(r,h^{l-\hat{N}_{\beta}}\eps_{\beta}^{\hat{N}_{\beta}-l}r',l).
$$

In the $l$-th step we needed the conditions 
\begin{align*}
M\ge c,\quad \alpha\ge cM^{\frac{\sa}{2}},\quad L^{d}\ge
 M^{(\sum_{j=1}^d\frac{1}{\sn_j}+1)(\hat{N}_{\beta}-N_{\beta})},\quad
h\ge 2
\end{align*}
for a positive constant $c$ independent of any parameter.  
Since we have admitted the results of Lemma \ref{lem_IR_integration_without}, we have to
 combine the above conditions with the conditions \eqref{eq_assumptions_IR_integration_without}. 
All the conditions are summarized as in \eqref{eq_assumptions_IR_integration_with}.
The induction with $l$ ensures that the claim holds true.
\end{proof}

\subsection{The final integration}\label{subsec_final_integration}

Here we study  properties of an analytic continuation of the function
\begin{align*}
(u,\bla)\mapsto \log\left(\int
 e^{-V(u)(\psi)+W(u)(\psi)-A(\psi)}
d\mu_{\sum_{l=N_{\beta}}^{\hat{N}_{\beta}}\cC_l}(\psi)
\right).
\end{align*}
Since we have constructed an analytic continuation of the Grassmann
polynomial
\begin{align*}
\log\left(\int
 e^{-V(u)(\psi+\psi^1)+W(u)(\psi+\psi^1)-A(\psi+\psi^1)}
d\mu_{\sum_{l=N_{\beta}+1}^{\hat{N}_{\beta}}\cC_l}(\psi^1)
\right),
\end{align*}
we can use it as the input to the single-scale integration with the
covariance $\cC_{N_{\beta}}$. Since the constant $c_5$ is not less than
$c_4$, we can deduce from Lemma \ref{lem_IR_integration_without}, Lemma
\ref{lem_IR_integration_with} that under the assumptions of Lemma
\ref{lem_IR_integration_with}, 
\begin{align}
&V^{0-1,N_{\beta}}\in
 \cQ(b^{-1}c_0^{-2}\alpha^{-4},N_{\beta}),\label{eq_list_IR_with_without}\\
&V^{0-2,N_{\beta}}\in \cR(b^{-1}c_0^{-2}\alpha^{-4},N_{\beta}),\notag\\
&V^{1-1,N_{\beta}}\in
 \cQ'(b^{-1}c_0^{-2}\alpha^{-4},c_5^{-1}\eps_{\beta}^{\hat{N}_{\beta}-N_{\beta}}\beta^{-1}c_0^{-2}\alpha^{-4},N_{\beta}),\notag\\
&V^{1-2,N_{\beta}}\in
 \cR'(b^{-1}c_0^{-2}\alpha^{-4},c_5^{-1}\eps_{\beta}^{\hat{N}_{\beta}-N_{\beta}}\beta^{-1}c_0^{-2}\alpha^{-4},N_{\beta}),\notag\\
&V^{1-3,N_{\beta}}\in
 \cS(b^{-1}c_0^{-2}\alpha^{-4},c_5^{-1}\eps_{\beta}^{\hat{N}_{\beta}-N_{\beta}}\beta^{-1}c_0^{-2}\alpha^{-4},N_{\beta}),\notag\\ 
&V^{2,N_{\beta}}\in \cW(b^{-1}c_0^{-2}\alpha^{-4},c_5^{-1}h^{N_{\beta}-\hat{N}_{\beta}}
\eps_{\beta}^{\hat{N}_{\beta}-N_{\beta}}\beta^{-1}c_0^{-2}\alpha^{-4},N_{\beta}).\notag
\end{align}
Set 
\begin{align*}
r:=b^{-1}c_0^{-2}\alpha^{-4},\quad
 \hat{r}:=c_5^{-1}h^{N_{\beta}-\hat{N}_{\beta}}\eps_{\beta}^{\hat{N}_{\beta}-N_{\beta}}\beta^{-1}c_0^{-2}\alpha^{-4}.
\end{align*}
Then let us define the functions $V^{end,(n)}:\overline{D(r)}\times
\overline{D(\hat{r})}^2\to \C$ $(n\in \N_{\ge 1})$,
$V^{1-3,end}:\overline{D(\hat{r})}^2\to\C$ by 
\begin{align*}
&V^{end,(n)}:=\frac{1}{n!}Tree(\{1,2,\cdots,n\},\cC_{N_{\beta}})\\
&\qquad\qquad\quad\cdot \prod_{j=1}^n\left(\sum_{p=1}^2V^{0-p,N_{\beta}}(\psi^j)+\sum_{q=1}^3V^{1-q,N_{\beta}}(\psi^j)+V^{2,N_{\beta}}(\psi^j)\right)\Bigg|_{\psi^j=0\atop
 (\forall j\in \{1,2,\cdots,n\})},\\
&V^{1-3,end}:=Tree(\{1\},\cC_{N_{\beta}})V^{1-3,N_{\beta}}(\psi^1)\Big|_{\psi^1=0}.
\end{align*}
We set $V^{end}:=\sum_{n=1}^{\infty}V^{end,(n)}$ if it converges. We
conclude this section by summarizing properties of $V^{end}$ in a
convenient way for applications in the next section.

\begin{lemma}\label{lem_final_integration}
Let $c_5$ be the constant appearing in Lemma
 \ref{lem_IR_integration_with}. Then there exists a constant $c_6\in
 [c_5,\infty)$ independent of any other parameters such that if 
\begin{align}
&M^{\min\{1,2\sa-1-\sum_{j=1}^d\frac{1}{\sn_j}\}}\ge c_6,\quad \alpha\ge
 c_6 M^{\frac{\sa}{2}},\quad L^d\ge (c_{end}+1)
 M^{(\sa+\sum_{j=1}^d\frac{1}{\sn_j}+1)(\hat{N}_{\beta}-N_{\beta})},\label{eq_final_integration_parameter_assumption}\\
&h\ge 2,\notag
\end{align}
the following statements hold true.
\begin{enumerate}[(i)]
\item\label{item_final_regularity_general}
\begin{align*}
V^{end}\in &C\left(
\overline{D(b^{-1}c_0^{-2}\alpha^{-4})}\times \overline{D(c_6^{-1}L^{-d}
h^{N_{\beta}-\hat{N}_{\beta}-1}
\eps_{\beta}^{\hat{N}_{\beta}-N_{\beta}}\beta^{-1}c_0^{-2}\alpha^{-4})}^2\right)\\
&\cap C^{\o}\left(
D(b^{-1}c_0^{-2}\alpha^{-4})\times D(c_6^{-1}L^{-d}
h^{N_{\beta}-\hat{N}_{\beta}-1}
\eps_{\beta}^{\hat{N}_{\beta}-N_{\beta}}\beta^{-1}c_0^{-2}\alpha^{-4})^2\right).\end{align*}
\item\label{item_final_bound_general}
\begin{align*}
\frac{h}{N}\sup_{u\in\overline{D(b^{-1}c_0^{-2}\alpha^{-4})}}|V^{end}(u,\b0)|
 \le c_6(1+N^{-1})\alpha^{-2}
 M^{(\sum_{j=1}^d\frac{1}{\sn_j}+1)(\hat{N}_{\beta}-N_{\beta})}L^{-d}.
\end{align*}
\item\label{item_final_difference_general} For $j\in \{1,2\}$
\begin{align*}
&\sup_{u\in D(b^{-1}c_0^{-2}\alpha^{-4})}\left|
\frac{\partial}{\partial \la_j}V^{end}(u,\b0)-
\frac{\partial}{\partial \la_j}V^{1-3,end}(u,\b0)
\right|\\
&\le c_6\eps_{\beta}^{N_{\beta}-\hat{N}_{\beta}}\beta c_0^{2}\left(\alpha^2+
c_{end}M^{\sa(\hat{N}_{\beta}-N_{\beta})}\right)L^{-d}.
\end{align*}
\end{enumerate}
\end{lemma}

\begin{proof}
The claims can be proved in the same way as the proof of 
\cite[\mbox{Lemma 3.8}]{K_BCS}. However, we provide a sketch of the
 proof for completeness of the paper. 
Observe that for any $c_6\in [c_5,\infty)$ the conditions
 \eqref{eq_final_integration_parameter_assumption} imply
 \eqref{eq_assumptions_IR_integration_with}. It follows from the property
 \eqref{eq_final_covariance_time_independence} and the property
 \eqref{eq_temperature_integral_vanish} of the kernels of
 $V^{0-2,N_{\beta}}$, $V^{1-2,N_{\beta}}$ that for any
 $z\in\C$
\begin{align*}
&\int
 e^{z\sum_{j=1}^2V^{0-j,N_{\beta}}(\psi)+z\sum_{k=1}^3V^{1-k,N_{\beta}}(\psi)+zV^{2,N_{\beta}}(\psi)}d\mu_{\cC_{N_{\beta}}}(\psi)\\
&=\int
 e^{zV^{0-1,N_{\beta}}(\psi)+zV^{1-1,N_{\beta}}(\psi)+zV^{1-3,N_{\beta}}(\psi)+
zV^{2,N_{\beta}}(\psi)}d\mu_{\cC_{N_{\beta}}}(\psi).
\end{align*}
Thus by recalling the formula \eqref{eq_tree_expansion},
\begin{align}
&V^{end,(n)}\label{eq_final_term_drop}\\
&=\frac{1}{n!}Tree(\{1,2,\cdots,n\},\cC_{N_{\beta}})\notag\\
&\quad\cdot
 \prod_{j=1}^n\left(V^{0-1,N_{\beta}}(\psi^j)+V^{1-1,N_{\beta}}(\psi^j)+
V^{1-3,N_{\beta}}(\psi^j)
+V^{2,N_{\beta}}(\psi^j)\right)\Bigg|_{\psi^j=0\atop
 (\forall j\in \{1,2,\cdots,n\})}.\notag
\end{align}
We can use \eqref{eq_simple_1_1}, \eqref{eq_simple_1}, \eqref{eq_scale_covariance_determinant_bound},
 \eqref{eq_scale_covariance_decay_bound},  \eqref{eq_0_1_bound}, \eqref{eq_1_1_bound}, \eqref{eq_1_3_bound},
 \eqref{eq_2_bound}, 
 \eqref{eq_artificial_scaling_properties} with $\eps=h^{-1}L^{-d}(\le
 1/2)$ 
 and the condition $\alpha\ge 2^3$ to derive that
\begin{align*}
&\|V^{end,(1)}\|_{1,r,h^{-1}L^{-d}\hat{r}}\\
&\le \|V_0^{0-1,N_{\beta}}\|_{1,\infty,r}+h^{-1}L^{-d}\sum_{\delta\in
 \{1-1,1-3,2\}}\|V_0^{\delta,N_{\beta}}\|_{1,r,\hat{r}}\\
&\quad +\sum_{m=2}^{N}(c_0M^{\sa (N_{\beta}-\hat{N}_{\beta})})^{\frac{m}{2}}\left(
\frac{N}{h}\|V_m^{0-1,N_{\beta}}\|_{1,\infty,r}+h^{-1}L^{-d}
\sum_{\delta\in
 \{1-1,1-3,2\}}\|V_m^{\delta,N_{\beta}}\|_{1,r,\hat{r}}
\right)\\
&\le
 c\left(\frac{N}{h}\alpha^{-2}M^{-(\sum_{j=1}^d\frac{1}{\sn_j}+1)(N_{\beta}-\hat{N}_{\beta})}L^{-d}+h^{-1}L^{-d}\alpha^{-2}\right)\\
&\le c h^{-1}(N+1)\alpha^{-2}M^{-(\sum_{j=1}^d\frac{1}{\sn_j}+1)(N_{\beta}-\hat{N}_{\beta})}L^{-d},\\
&\|V^{end,(n)}\|_{1,r,h^{-1}L^{-d}\hat{r}}\\
&\le (c_0 M^{\sa(N_{\beta}-\hat{N}_{\beta})})^{-n+1}(c_0c_{end})^{n-1}\\
&\quad\cdot \sum_{p_1=2}^N2^{3p_1}(c_0M^{\sa (N_{\beta}-\hat{N}_{\beta})})^{\frac{p_1}{2}}
\left(
\frac{N}{h}\|V_{p_1}^{0-1,N_{\beta}}\|_{1,\infty,r}+h^{-1}L^{-d}
\sum_{\delta\in
 \{1-1,1-3,2\}}\|V_{p_1}^{\delta,N_{\beta}}\|_{1,r,\hat{r}}
\right)\\
&\quad\cdot\left(\sum_{p=2}^N2^{3p}(c_0M^{\sa(N_{\beta}-\hat{N}_{\beta})})^{\frac{p}{2}}
\left(\|V_{p}^{0-1,N_{\beta}}\|_{1,\infty,r}+L^{-d}
\sum_{\delta\in
 \{1-1,1-3,2\}}\|V_{p}^{\delta,N_{\beta}}\|_{1,r,\hat{r}}
\right)
\right)^{n-1}\\
&\le
 M^{-\sa(N_{\beta}-\hat{N}_{\beta})(n-1)}c_{end}^{n-1}(c\alpha^{-2})^n
\left(\frac{N}{h}L^{-d}+h^{-1}L^{-d}
\right)L^{-d(n-1)}\\
&\le c
 h^{-1}(N+1)\alpha^{-2}L^{-d}(c c_{end}
 \alpha^{-2}M^{-\sa(N_{\beta}-\hat{N}_{\beta})}L^{-d})^{n-1}
\end{align*}
for $n\in \N_{\ge 2}$, or under the assumptions of the lemma,
\begin{align*}
&\frac{h}{N}\sum_{n=1}^{\infty}\|V^{end,(n)}\|_{1,r,h^{-1}L^{-d}\hat{r}}
\le
 c(1+N^{-1})\alpha^{-2}M^{-(\sum_{j=1}^d\frac{1}{\sn_j}+1)(N_{\beta}-\hat{N}_{\beta})}L^{-d}.
\end{align*}
This implies the claims \eqref{item_final_regularity_general}, \eqref{item_final_bound_general}.

To prove the claim \eqref{item_final_difference_general}, let us set
 $$
\overline{r}:=c_5^{-1}\eps_{\beta}^{\hat{N}_{\beta}-N_{\beta}}\beta^{-1}c_0^{-2}\alpha^{-4}.
$$
Recalling \eqref{eq_list_IR_with_without}, we see that
\begin{align}
\frac{\partial}{\partial
 \la_j}V^{\delta,N_{\beta}}(u,\b0)=\frac{1}{\overline{r}}V^{\delta,N_{\beta}}(u,\overline{r}\be_j),\quad
 (\forall j\in\{1,2\},\ \delta\in \{1-1,1-3\},\ u\in
 D(r)),\label{eq_application_linearity}
\end{align}
where $\be_1$, $\be_2$ are the canonical basis of $\C^2$.
We can deduce from \eqref{eq_simple_1_1}, \eqref{eq_simple_1}, \eqref{eq_scale_covariance_determinant_bound},
 \eqref{eq_scale_covariance_decay_bound}, 
 \eqref{eq_0_1_bound},
 \eqref{eq_1_1_bound}, \eqref{eq_1_3_bound}, \eqref{eq_final_term_drop},
 \eqref{eq_application_linearity} and the assumptions of the
 lemma that for $j\in \{1,2\}$, $u\in D(r)$
\begin{align*}
&\left|
\frac{\partial}{\partial \la_j}V^{end}(u,\b0)-
\frac{\partial}{\partial \la_j}V^{1-3,end}(u,\b0)
\right|\\
&\le \frac{1}{\overline{r}}\left|Tree(\{1\},\cC_{N_{\beta}})V^{1-1,N_{\beta}}(u,\overline{r}\be_j)(\psi^1)\Big|_{\psi^1=0}
\right|\\
&\quad +  \frac{1}{\overline{r}}\sum_{n=2}^{\infty}\frac{1}{(n-1)!}
\Bigg|Tree(\{1,2,\cdots,n\},\cC_{N_{\beta}})\\
&\qquad\qquad\qquad\qquad\quad\cdot \sum_{\delta \in
 \{1-1,1-3\}}V^{\delta,N_{\beta}}(u,\overline{r}\be_j)(\psi^1)
\prod_{k=2}^nV^{0-1,N_{\beta}}(u)(\psi^k)\Bigg|_{\psi^k=0\atop (\forall
 k\in \{1,2,\cdots,n\})}\Bigg|\\
&\le
 \frac{1}{\overline{r}}\left(\|V_0^{1-1,N_{\beta}}\|_{1,r,\overline{r}}+
\sum_{m=2}^N(c_0M^{\sa(N_{\beta}-\hat{N}_{\beta})})^{\frac{m}{2}}\|V_m^{1-1,N_{\beta}}\|_{1,r,\overline{r}}\right)\\
&\quad
 +\frac{1}{\overline{r}}\sum_{n=2}^{\infty}M^{-\sa(N_{\beta}-\hat{N}_{\beta})(n-1)}c_{end}^{n-1}\\
&\qquad\quad\cdot \sum_{p_1=2}^N2^{3p_1}(c_0M^{\sa(N_{\beta}-\hat{N}_{\beta})})^{\frac{p_1}{2}}(\|V_{p_1}^{1-1,N_{\beta}}\|_{1,r,\overline{r}}+\|V_{p_1}^{1-3,N_{\beta}}\|_{1,r,\overline{r}})\\
&\qquad\quad \cdot \left(
\sum_{p=2}^N2^{3p}(c_0M^{\sa(N_{\beta}-\hat{N}_{\beta})})^{\frac{p}{2}}\|V_p^{0-1,N_{\beta}}\|_{0,\infty,r}
\right)^{n-1}\\
&\le
 \frac{c}{\overline{r}}\alpha^{-2}L^{-d}+\frac{c}{\overline{r}}\alpha^{-2}\sum_{n=2}^{\infty}(c c_{end}\alpha^{-2}M^{-\sa(N_{\beta}-\hat{N}_{\beta})}L^{-d})^{n-1}\\
&\le
 \frac{c}{\overline{r}}\left(\alpha^{-2}+\alpha^{-4}c_{end}M^{-\sa(N_{\beta}-\hat{N}_{\beta})}\right)L^{-d}.
\end{align*}
Thus the claim holds true.
\end{proof}

\section{Proof of the theorem}\label{sec_proof_theorem}

In this section we complete the proof of Theorem \ref{thm_main_theorem}
and Corollary \ref{cor_zero_temperature_limit}. Since we have developed
the multi-scale integration scheme in the previous section and we plan
to apply the convergence result \cite[\mbox{Proposition 4.16}]{K_BCS},
we have the main general tools at hand. We need to confirm that our
actual covariance appearing in the formulation Lemma
\ref{lem_final_Grassmann_formulation} can be decomposed into a family
of covariances which fit in our framework. The way to complete the proof
of Theorem \ref{thm_main_theorem} after the confirmation is essentially parallel to
the proof of \cite[\mbox{Theorem 1.3}]{K_BCS}. Proving Corollary
\ref{cor_zero_temperature_limit} requires some additional arguments
which we provide in the end of this section. 

From here we assume that 
\begin{align}
h\ge \max\{2,\sc\}.\label{eq_first_h_condition}
\end{align}
Since we send $h$ to infinity first, we can assume that $h$ is larger
than any other parameter. As we proceed, we will replace \eqref{eq_first_h_condition}
by stricter conditions.

\subsection{Decomposition of the
  covariance}\label{subsec_real_covariance}

Let us decompose the covariance characterized in
\eqref{eq_characterization_covariance} into a sum of scale-dependent
covariances. We begin with discretizing the time-variables. Let $\cM_h$
denote the set of Matsubara frequency  with cut-off
\begin{align*}
\left\{\o\in \frac{\pi}{\beta}(2\Z+1)\ \Big|\ |\o|<\pi h \right\}.
\end{align*}

\begin{lemma}\label{lem_time_discrete_covariance}
For any $(\orho,\rho,\bx,s)$, $(\oeta,\eta,\by,t)$ $\in
 \{1,2\}\times\cB\times\G\times[0,\beta)_h$, $\phi\in\C$,
\begin{align*}
&C(\phi)(\orho\rho\bx s,\oeta\eta\by t)\\
&=\frac{1}{\beta L^{d}}\sum_{\bk\in\G^*}\sum_{\o\in
 \cM_h}e^{i\<\bk,\bx-\by\>+i\o(s-t)}\\
&\qquad\qquad \cdot h^{-1}(I_{2b}-e^{-\frac{i}{h}(\o-\frac{\theta(\beta)}{2})I_{2b}+\frac{1}{h}E(\phi)(\bk)})^{-1}((\orho-1)b+\rho,(\oeta-1)b+\eta).
\end{align*}
\end{lemma}

\begin{proof}
According to \cite[\mbox{Lemma C.3}]{K_9},
\begin{align*}
&e^{sA}\left(\frac{1_{s\ge 0}}{1+e^{\beta A}}- \frac{1_{s< 0}}{1+e^{-\beta A}}
\right)
=\frac{1}{\beta}\sum_{\o\in \cM_h}\frac{e^{i\o
 s}}{h(1-e^{-i\frac{\o}{h}+\frac{A}{h}})},\\
&(\forall s\in \{-\beta,-\beta+1/h,\cdots,\beta-1/h\},\ A\in\C\backslash
 i(\pi/\beta)(2\Z+1)).
\end{align*}
By diagonalizing $E(\phi)(\bk)$ in
 \eqref{eq_characterization_covariance} by a unitary matrix and
 substituting this formula we can derive the claimed equality.
\end{proof}

Next let us introduce a cut-off function. Let $\chi$ be a real-valued
function on $\R$ satisfying the following properties. 
\begin{align*}
&\chi\in C^{\infty}(\R),\\
&\chi(x)=1,\quad (\forall x\in (-\infty,8/5]),\\
&\chi(x)\in (0,1),\quad (\forall x\in (8/5,2)),\\
&\chi(x)=0,\quad (\forall x\in [2,\infty)),\\
&\frac{d}{d x}\chi(x)\le 0,\quad (\forall x\in \R).
\end{align*}
Using the function $\chi$, we construct scale-dependent cut-off
functions. We use the parameter $M\in \R_{\ge 2}$ to control the support
size of the cut-off functions. Here we give explicit definitions of the
numbers $N_{\beta}$, $\hat{N}_{\beta}$ introduced in Subsection
\ref{subsec_generalized_covariances}. Let 
\begin{align*}
N_{\beta}:=\left\lfloor \frac{\log(1/\beta)}{\log M}\right\rfloor,\quad
 \hat{N}_{\beta}:=1_{\beta \le 1}(N_{\beta}+1).
\end{align*}
Moreover, set
\begin{align*}
N_h:=\left\lfloor \frac{\log (h)}{\log M}\right\rfloor+2.
\end{align*}
Then by \eqref{eq_first_h_condition} and the condition $h\ge 1/\beta$
implied by $h\in \frac{2}{\beta}\N$, 
\begin{align}
&N_{\beta}<\hat{N}_{\beta}<N_h,\notag\\
&M^{-1}\beta^{-1}\le M^{N_{\beta}}\le
 \beta^{-1}.\label{eq_beta_relation}
\end{align}
Set $A(\beta,M):=\beta^{-1}M^{-N_{\beta}}$. It follows that $1\le
A(\beta,M)\le M$. With the function $e(\cdot):\R^d\to \R$ we define the
functions $\chi_l:\R^{d+1}\to \R$
$(l=N_{\beta},N_{\beta}+1,\cdots,N_h)$ by 
\begin{align*}
&\chi_{N_{\beta}}(\o,\bk):=\chi\left(M^{-N_{\beta}}A(\beta,M)^{-1}\sqrt{h^2\sin^2\left(\frac{\o-\pi/\beta}{2h}\right)+e(\bk)^2}\right),\\
&\chi_{l}(\o,\bk):=\chi\left(M^{-l}A(\beta,M)^{-1}\sqrt{h^2\sin^2\left(\frac{\o-\pi/\beta}{2h}\right)+e(\bk)^2}\right)\\
&\qquad\qquad\quad-\chi\left(M^{-(l-1)}A(\beta,M)^{-1}\sqrt{h^2\sin^2\left(\frac{\o-\pi/\beta}{2h}\right)+e(\bk)^2}\right),\\
&((\o,\bk)\in\R\times\R^d,\ l\in
 \{N_{\beta}+1,N_{\beta}+2,\cdots,N_h\}).
\end{align*}
Keeping in mind that $2A(\beta,M)M^{l-1}<\frac{8}{5}A(\beta,M)M^l$, 
we observe that
\begin{align}
&\chi_{N_{\beta}}(\o,\bk)\label{eq_support_description}\\
&=\left\{
\begin{array}{ll} 1 & \text{if
 }\sqrt{h^2\sin^2\left(\frac{\o-\pi/\beta}{2h}\right)+e(\bk)^2}\le
  \frac{8}{5}A(\beta,M)M^{N_{\beta}},\\
\in (0,1) & \text{if }\frac{8}{5}A(\beta,M)M^{N_{\beta}}<
 \sqrt{h^2\sin^2\left(\frac{\o-\pi/\beta}{2h}\right)+e(\bk)^2} <2
 A(\beta,M) M^{N_{\beta}},\\
 0 & \text{if }\sqrt{h^2\sin^2\left(\frac{\o-\pi/\beta}{2h}\right)+e(\bk)^2}\ge
  2 A(\beta,M)M^{N_{\beta}},
\end{array}\right.\notag\\
&\chi_{l}(\o,\bk)\notag\\
&=\left\{
\begin{array}{ll} 0 & \text{if
 }\sqrt{h^2\sin^2\left(\frac{\o-\pi/\beta}{2h}\right)+e(\bk)^2}\le
  \frac{8}{5}A(\beta,M)M^{l-1},\\
\in (0,1] & \text{if }\frac{8}{5}A(\beta,M)M^{l-1}<
 \sqrt{h^2\sin^2\left(\frac{\o-\pi/\beta}{2h}\right)+e(\bk)^2} <2
 A(\beta,M) M^{l},\\
 0 & \text{if }\sqrt{h^2\sin^2\left(\frac{\o-\pi/\beta}{2h}\right)+e(\bk)^2}\ge
  2 A(\beta,M)M^{l},
\end{array}\right.\notag\\
&(\forall (\o,\bk)\in\R^{d+1},\quad l\in
 \{N_{\beta}+1,N_{\beta}+2,\cdots,N_h\}).\notag
\end{align}
Basic properties of $\chi_l$ are summarized as follows.

\begin{lemma}\label{lem_properties_cutoff}
Assume that 
\begin{align}
L\ge \frac{\max_{j\in
 \{1,2,\cdots,d\}}M^{-N_{\beta}/\sn_j}}{\min\{M^{\sa
 N_{\beta}},1\}}.\label{eq_large_L_condition}
\end{align}
Then there exists a positive constant $\hat{c}$ depending only on $d$,
 $M$, $\chi$, $\sc$, $\sa$, $(\hbv_j)_{j=1}^d$ such
 that the following statements hold.
\begin{enumerate}[(i)]
\item\label{item_regularity_cutoff}
\begin{align*}
\chi_l\in C^{\infty}(\R^{d+1}),\quad (\forall l\in
 \{N_{\beta},N_{\beta}+1,\cdots,N_h\}).
\end{align*}
\item\label{item_unity_cutoff}
\begin{align*}
\sum_{l=N_{\beta}}^{N_h}\chi_l(\o,\bk)=1,\quad (\forall
 (\o,\bk)\in\R^{d+1}).
\end{align*}
\item\label{item_derivative_cutoff}
\begin{align*}
&\left|\left(\frac{\partial}{\partial
 \o}\right)^n\chi_l\left(\o,\sum_{j=1}^d \hat{k}_j\hbv_j\right)\right|
\le \hat{c} M^{-nl},\\
&\left|\left(\frac{\partial}{\partial \hat{k}_i}\right)^n\chi_l\left(\o,\sum_{j=1}^d \hat{k}_j\hbv_j\right)\right|
\le \hat{c}(1_{l\ge
 0}M^{-\frac{l}{\sn_i}}+1_{l<0}M^{-\frac{n}{\sn_i}l}),\\
&(\forall n\in \{1,2,\cdots,d+2\},\ i\in \{1,2,\cdots,d\},\
 (\o,\hat{k}_1,\cdots,\hat{k}_d)\in\R^{d+1},\\
&\quad l\in \{N_{\beta},N_{\beta}+1,\cdots,N_h\}).
\end{align*}
\item\label{item_support_size_cutoff}
\begin{align*}
&\frac{1}{\beta}\sup_{\bk\in\R^d}\sum_{\o\in
 \cM_h}1_{\sum_{j=N_{\beta}}^l\chi_j(\o,\bk)\neq 0}\le \hat{c}M^l,\\
&\frac{1}{\beta L^d}\sup_{x\in\R\atop
 \bp\in\R^d}\sum_{\bk\in\G^*}\sum_{\o\in\cM_h}
1_{\chi_l(\o+x,\bk+\bp)\neq 0}\le \hat{c}M^l\min\{M^{\sa l},1\},\\
&(\forall l\in \{N_{\beta},N_{\beta}+1,\cdots,N_h\}).
\end{align*}
\item\label{item_final_cutoff}
\begin{align*}
\chi_{N_{\beta}}(\o,\bk)=0,\quad (\forall \o\in
 \cM_h\backslash\{\pi/\beta\},\ \bk\in\R^d).
\end{align*}
\item\label{item_support_shift}
Let $l\in \{N_{\beta}+1,N_{\beta}+2,\cdots,N_h\}$, $(\o,\bk)\in\R^{d+1}$. If 
\begin{align*}
\frac{8}{5}A(\beta,M)M^{l-1}\le
 \sqrt{h^2\sin^2\left(\frac{\o-\pi/\beta}{2h}\right)+e(\bk)^2},
\end{align*}
then
\begin{align}
\sqrt{\left(\o-\frac{\theta(\beta)}{2}\right)^2+e(\bk)^2}\ge
 \left(\frac{8}{5}-\frac{\pi}{2}\right)A(\beta,M)M^{l-1}>0.\label{eq_regularization_lower_bound}
\end{align}
Especially if $\chi_l(\o,\bk)\neq 0$,
     \eqref{eq_regularization_lower_bound} holds.
\end{enumerate}
\end{lemma}

Let us prepare a useful inequality beforehand. 

\begin{lemma}\label{lem_discrete_continuous_estimate}
Let $f\in C^{1}(\R^d,\C)$ satisfy 
\begin{align*}
\sup_{j\in\{1,2,\cdots,d\}}\sup_{\bk\in\R^d}\left|\frac{\partial}{\partial
 k_j}f(\bk)\right|<\infty.
\end{align*}
Then the following inequality holds.
\begin{align*}
&\left|
D_d\int_{\G_{\infty}^*}d\bp f(\bp)-\frac{1}{L^d}\sum_{\bk\in\G^*}f(\bk)\right|\\
&\le L^{-1}2\pi d^2\sup_{i\in\{1,2,\cdots,d\}}\|\hbv_i\|_{\R^d}\sup_{j\in
 \{1,2,\cdots,d\}}\sup_{\bk\in \R^d}\left|\frac{\partial}{\partial
 k_j}f(\bk)\right|.
\end{align*}
\end{lemma}

\begin{proof}
Observe that 
\begin{align*}
&\Bigg| \int_0^{2\pi}dk_1\int_{[0,2\pi]^{d-1}}dk_2\cdots dk_d
 f\left(\sum_{j=1}^dk_j\hbv_j\right)\\
&\quad -\frac{2\pi}{L}\sum_{m_1=0}^{L-1}\int_{[0,2\pi]^{d-1}}dk_2\cdots
 dk_df\left(\frac{2\pi}{L}m_1\hbv_1+\sum_{j=2}^dk_j\hbv_j\right)\Bigg|\\
&\le
 \sum_{m_1=0}^{L-1}\int_{\frac{2\pi}{L}m_1}^{\frac{2\pi}{L}(m_1+1)}dk_1
\int_{[0,2\pi]^{d-1}}dk_2\cdots dk_d
 \int_{\frac{2\pi}{L}m_1}^{k_1}dq_1\left|
\frac{\partial}{\partial q_1}f\left(q_1\hbv_1+\sum_{j=2}^dk_j\hbv_j\right)
\right|\\
&\le L^{-1}(2\pi)^{d+1}d\sup_{i\in \{1,2,\cdots,d\}}\|\hbv_i\|_{\R^d}
\sup_{j\in
 \{1,2,\cdots,d\}}\sup_{\bk\in\R^d}\left|\frac{\partial}{\partial
 k_j}f(\bk)\right|.
\end{align*}
By repeating this type of estimate $d$ times we have that
\begin{align*}
&\left|\int_{[0,2\pi]^d}d\bk
 f\left(\sum_{j=1}^dk_j\hbv_j\right)-\left(\frac{2\pi}{L}\right)^d\sum_{\bk\in\G^*}f(\bk)\right|\\
&\le L^{-1}(2\pi)^{d+1}d^2\sup_{i\in \{1,2,\cdots,d\}}\|\hbv_i\|_{\R^d}
\sup_{j\in
 \{1,2,\cdots,d\}}\sup_{\bk\in\R^d}\left|\frac{\partial}{\partial
 k_j}f(\bk)\right|.
\end{align*}
By combining this inequality with the equality 
\begin{align*}
\frac{1}{(2\pi)^d}\int_{[0,2\pi]^d}d\bk
 f\left(\sum_{j=1}^dk_j\hbv_j\right)=D_d\int_{\G_{\infty}^*}d\bp f(\bp)
\end{align*}
we obtain the result.
\end{proof}

\begin{proof}[Proof of Lemma \ref{lem_properties_cutoff}]
\eqref{item_regularity_cutoff}: The claim follows from the assumptions
 $e^2\in C^{\infty}(\R^d)$, $\chi\in C^{\infty}(\R)$ and that 
$\chi(\cdot)$ is constant in a neighborhood of the origin.

\eqref{item_unity_cutoff}: By \eqref{eq_dispersion_upper_bound} and 
the assumption
 \eqref{eq_first_h_condition}, for any $(\o,\bk)\in\R^{d+1}$ 
\begin{align*}
\sqrt{h^2\sin^2\left(\frac{\o-\pi/\beta}{2h}\right)+e(\bk)^2}\le
 \sqrt{2}h\le 
 \frac{8}{5}A(\beta,M)M^{N_h}.
\end{align*}
Thus
\begin{align*}
\sum_{l=N_{\beta}}^{N_h}\chi_l(\o,\bk)
=\chi\left(M^{-N_h}A(\beta,M)^{-1}\sqrt{h^2\sin^2\left(\frac{\o-\pi/\beta}{2h}\right)+e(\bk)^2}\right)=1.
\end{align*}

\eqref{item_derivative_cutoff}: We use the following formula. See
 e.g. \cite[\mbox{Lemma C.1}]{K_RG} for the proof. Let $\O_1$, $\O_2$ be
 open sets of $\R$. Let $f_j\in C^{\infty}(\O_j,\R)$ $(j=1,2)$ and
 $f_1(\O_1)\subset \O_2$. Then for $x_0\in \O_1$, $n\in \N$,
\begin{align}
\left(\frac{d}{dx}\right)^nf_2(f_1(x))\Big|_{x=x_0}=
\sum_{m=1}^n\frac{n!}{m!}f_2^{(m)}(f_1(x_0))\prod_{j=1}^m\left(\sum_{l_j=1}^n\frac{1}{l_j!}f_1^{(l_j)}(x_0)\right)1_{\sum_{j=1}^ml_j=n}.\label{eq_composition_derivative}
\end{align}
Take $l\in \Z$ satisfying $l\le N_h$. Define the function
 $g:\R^{d+1}\to\R$ by 
\begin{align*}
g(\o,\bk):=h^2\sin^2\left(\frac{\o-\pi/\beta}{2h}\right)+e(\bk)^2.
\end{align*}
Since $g$ is continuous, $g^{-1}(((\frac{8}{5}A(\beta,M)M^l)^2,
(2A(\beta,M)M^l)^2))$ is an open set of $\R^{d+1}$. 
Assume that $(\o,\bk)\in g^{-1}(((\frac{8}{5}A(\beta,M)M^l)^2,
(2A(\beta,M)M^l)^2))$.
Take any $i\in \{1,2,\cdots,d\}$, $n\in
 \{1,2,\cdots,d+2\}$. Let us estimate $|(\frac{\partial}{\partial
 \hat{k}_i})^n\sqrt{g(\o,\bk)}|$, $|(\frac{\partial}{\partial
 \o})^n\sqrt{g(\o,\bk)}|$ where $\bk=\sum_{j=1}^d\hat{k}_j\hbv_j$. Since $h^{-1}\le M^{-N_h+2}\le M^{-l+2}$, 
\begin{align*}
\left|
\left(\frac{\partial}{\partial \o}\right)^ng(\o,\bk)\right| \le
 c(d,M)M^{l(2-n)}.
\end{align*}
Then by applying \eqref{eq_composition_derivative} we have that
\begin{align}
\left|\left(\frac{\partial}{\partial
 \o}\right)^n\sqrt{g(\o,\bk)}\right|
&\le
 c(d,M)\sum_{m=1}^nM^{2l(\frac{1}{2}-m)}\prod_{j=1}^m\left(\sum_{l_j=1}^n
 M^{l(2-l_j)}\right)1_{\sum_{j=1}^ml_j=n}\label{eq_time_derivative_inside}\\
&\le c(d,M) M^{l(1-n)}.\notag
\end{align}
Note that by \eqref{eq_dispersion_derivative}
\begin{align*}
\left|\left(\frac{\partial}{\partial \hat{k}_i}\right)^ng(\o,\bk)\right|
\le c(M,\sc)\left(1_{n\le 2\sn_i} M^{l(2-\frac{n}{\sn_i})}+1_{n>2\sn_i}\right).
\end{align*}
Thus we can apply \eqref{eq_composition_derivative} to derive that 
\begin{align}
&\left|\left(\frac{\partial}{\partial
 \hat{k}_i}\right)^n\sqrt{g(\o,\bk)}\right|\label{eq_space_derivative_inside}\\
&\le c(d,M,\sc) \sum_{m=1}^nM^{2l(\frac{1}{2}-m)}\prod_{j=1}^m\left(
\sum_{l_j=1}^n\left(1_{l_j\le 2\sn_i}M^{l(2-\frac{l_j}{\sn_i})}
+1_{l_j> 2\sn_i}\right)\right) 1_{\sum_{j=1}^ml_j=n}\notag\\
&\le c(d,M,\sc) \sum_{m=1}^nM^{2l(\frac{1}{2}-m)}\prod_{j=1}^m\left(
\sum_{l_j=1}^n\right) 1_{\sum_{j=1}^ml_j=n}\notag\\
&\quad\cdot \left(1_{l\ge
 0}M^{2ml-\frac{m}{\sn_i}l}+1_{l<0}M^{l(2m-\frac{n}{\sn_i})-\sum_{j=1}^m1_{l_j>2\sn_i}l(2-\frac{l_j}{\sn_i})}\right)\notag\\
&\le c(d,M,\sc) \sum_{m=1}^nM^{2l(\frac{1}{2}-m)}
 \left(1_{l\ge
 0}M^{2ml-\frac{m}{\sn_i}l}+1_{l<0}M^{l(2m-\frac{n}{\sn_i})}\right)\notag\\
&\le c(d,M,\sc)
\left(1_{l\ge
 0}M^{(1-\frac{1}{\sn_i})l}+1_{l<0}M^{(1-\frac{n}{\sn_i})l}\right).\notag
\end{align}
Moreover we can use \eqref{eq_composition_derivative},
 \eqref{eq_time_derivative_inside}, \eqref{eq_space_derivative_inside}
 to deduce that
\begin{align*}
&\left|\left(\frac{\partial}{\partial
 \o}\right)^n\chi(M^{-l}A(\beta,M)^{-1}\sqrt{g(\o,\bk)})\right|\\
&\le c(d,M,\chi) \sum_{m=1}^nM^{-ml}\prod_{j=1}^m\left(
\sum_{l_j=1}^nM^{l(1-l_j)}\right)1_{\sum_{j=1}^ml_j=n}\\
&\le c(d,M,\chi) M^{-nl},\\
&\left|\left(\frac{\partial}{\partial
 \hat{k}_i}\right)^n\chi(M^{-l}A(\beta,M)^{-1}\sqrt{g(\o,\bk)})\right|\\
&\le  c(d,M,\chi,\sc) 
\sum_{m=1}^nM^{-ml}\prod_{j=1}^m\left(\sum_{l_j=1}^n
 \left(1_{l\ge
 0}M^{(1-\frac{1}{\sn_i})l}+1_{l<0}M^{(1-\frac{l_j}{\sn_i})l}\right)\right)
1_{\sum_{j=1}^ml_j=n}\\
&\le  c(d,M,\chi,\sc) \left(1_{l\ge
 0}M^{-\frac{l}{\sn_i}}+1_{l<0}M^{-\frac{n}{\sn_i}l}\right).
\end{align*}
On the other hand, if $(\o,\bk)\in\R^{d+1}\backslash
 g^{-1}(((\frac{8}{5}A(\beta,M)M^l)^2,(2A(\beta,M)M^l)^2))$,
\begin{align*}
\left(\frac{\partial}{\partial
 \o}\right)^n\chi(M^{-l}A(\beta,M)^{-1}\sqrt{g(\o,\bk)})=
\left(\frac{\partial}{\partial
 \hat{k}_i}\right)^n\chi(M^{-l}A(\beta,M)^{-1}\sqrt{g(\o,\bk)})=0.
\end{align*}
Thus by summing up,
\begin{align}
&\left|\left(\frac{\partial}{\partial
 \o}\right)^n\chi\left(M^{-l}A(\beta,M)^{-1}\sqrt{g\left(\o,\sum_{j=1}^d\hat{k}_j\hbv_j\right)}\right)\right|\le
 c(d,M,\chi)M^{-nl},\label{eq_derivative_cutoff_essence}\\
&\left|\left(\frac{\partial}{\partial
 \hat{k}_i}\right)^n\chi\left(M^{-l}A(\beta,M)^{-1}\sqrt{g\left(\o,\sum_{j=1}^d\hat{k}_j\hbv_j\right)}\right)\right|\notag\\
&\le
 c(d,M,\chi,\sc)\left(1_{l\ge 0}M^{-\frac{l}{\sn_i}}+1_{l<0}M^{-\frac{n}{\sn_i}l}
\right),\notag\\
&(\forall (\o,\hat{k}_1,\cdots,\hat{k}_d)\in\R^{d+1},\ i\in \{1,2,\cdots,d\},\ n\in
 \{1,2,\cdots,d+2\}).\notag
\end{align}
This implies the claimed results. 

\eqref{item_support_size_cutoff}: By \eqref{eq_beta_relation},
 the inequality $A(\beta,M)\le M$ and the support property of $\chi(\cdot)$,
\begin{align*}
\frac{1}{\beta}\sup_{\bk\in \R^d}\sum_{\o\in
 \cM_h}1_{\sum_{j=N_{\beta}}^l\chi_j(\o,\bk)\neq 0}\le
 \frac{1}{\beta}\sum_{\o\in \cM_h}1_{|\o-\pi/\beta|\le 2\pi M^{l+1}}\le
 c(M)M^l,
\end{align*}
which is the first inequality.
Note that by \eqref{eq_support_description}, the support property of
 $\chi(\cdot)$, $A(\beta,M)\le M$
 and Lemma
 \ref{lem_discrete_continuous_estimate}
\begin{align*}
&\frac{1}{\beta
 L^d}\sum_{\bk\in\G^*}\sum_{\o\in\cM_h}1_{\chi_l(\o+x,\bk+\bp)\neq 0}\\
&\le \frac{1}{\beta
 L^d}\sum_{\bk\in\G^*}\sum_{\o\in\cM_h}\chi(2^{-1}M^{-l}A(\beta,M)^{-1}\sqrt{g(\o+x,\bk+\bp)})\\
&\le \frac{D_d}{\beta}\sum_{\o\in \cM_h}\int_{\G_{\infty}^*}d\bk 
\chi(2^{-1}M^{-l}A(\beta,M)^{-1}\sqrt{g(\o+x,\bk+\bp)})\\
&\quad +\frac{1}{\beta}\sum_{\o\in \cM_h}
\Bigg|D_d\int_{\G_{\infty}^*}d\bk\chi(2^{-1}M^{-l}A(\beta,M)^{-1}\sqrt{g(\o+x,\bk+\bp)})\\
&\qquad\qquad\qquad -\frac{1}{L^d}\sum_{\bk\in\G^*}\chi(2^{-1}M^{-l}A(\beta,M)^{-1}\sqrt{g(\o+x,\bk+\bp)})\Bigg|\\
&\le
 c(d,D_d,(\hbv_j)_{j=1}^d)\frac{1}{\beta}\sum_{\o\in
 \cM_h}1_{h|\sin(\frac{1}{2h}(\o-\frac{\pi}{\beta}+x))|\le 2^2M^{l+1}}\\
&\quad\cdot \Bigg(\int_{\G_{\infty}^*}d\bk 1_{e(\bk+\bp)\le
 2^2M^{l+1}}\\
&\qquad\quad + L^{-1}\sup_{j\in
 \{1,2,\cdots,d\}}\sup_{(\o,\bk)\in\R^{d+1}}\left|\frac{\partial}{\partial
 k_j}\chi(2^{-1}M^{-l}A(\beta,M)^{-1}\sqrt{g(\o,\bk)})\right|\Bigg).
\end{align*}
Moreover by using \eqref{eq_dispersion_measure},
 \eqref{eq_beta_relation}, 
 \eqref{eq_large_L_condition} and a simple variant of \eqref{eq_derivative_cutoff_essence}
 having $2^{-2}g(\cdot)$ in place of $g(\cdot)$
 we have that
\begin{align*}
&\frac{1}{\beta
 L^d}\sum_{\bk\in\G^*}\sum_{\o\in\cM_h}1_{\chi_l(\o+x,\bk+\bp)\neq 0}\\
&\le c(d,M,\chi,\sc,D_d,(\hbv_j)_{j=1}^d,\sa)M^l\left(
\min\{M^{\sa l},1\}+L^{-1}\max_{j\in \{1,2,\cdots,d\}}M^{-\frac{N_{\beta}}{\sn_j}}
\right)\\
&\le
 c(d,M,\chi,\sc,D_d,(\hbv_j)_{j=1}^d,\sa)M^l\min\{M^{\sa
 l},1\},
\end{align*}
which is the second inequality.

\eqref{item_final_cutoff}: Take any $\o\in
 \cM_h\backslash\{\pi/\beta\}$, $\bk\in\R^d$. It follows from the
 inequality $|\sin x|\ge \frac{2}{\pi}|x|$, $(x\in
 [-\frac{\pi}{2},\frac{\pi}{2}])$ and the definition of $A(\beta,M)$
 that
\begin{align*}
M^{-N_{\beta}}A(\beta,M)^{-1}\sqrt{g(\o,\bk)}\ge \beta
 h\left|\sin\left(\frac{\o-\pi/\beta}{2h}\right)\right|\ge 2.
\end{align*}
Then we can deduce the claim from \eqref{eq_support_description}.

\eqref{item_support_shift}: By the assumption, the definition of
 $A(\beta,M)$ and the triangle
 inequality of the norm $\|\cdot\|_{\R^2}$, 
\begin{align*}
&\frac{8}{5}A(\beta,M)M^{l-1}\le \sqrt{g(\o,\bk)}\\
&\le
 \left(\frac{1}{4}\left(\o-\frac{\pi}{\beta}\right)^2+e(\bk)^2\right)^{\frac{1}{2}}
\le
 \left(\frac{1}{4}\left(\o-\frac{\theta(\beta)}{2}\right)^2+e(\bk)^2\right)^{\frac{1}{2}}+\frac{1}{2}\left|\frac{\pi}{\beta}-\frac{\theta(\beta)}{2}\right|\\
&\le
 \left(\left(\o-\frac{\theta(\beta)}{2}\right)^2+e(\bk)^2\right)^{\frac{1}{2}}
+\frac{\pi}{2\beta}\\
&\le
 \left(\left(\o-\frac{\theta(\beta)}{2}\right)^2+e(\bk)^2\right)^{\frac{1}{2}}
+\frac{\pi}{2}A(\beta,M)M^{l-1},
\end{align*}
which implies the result.
\end{proof}

\begin{remark} In fact we set up the support of the function $\chi$ in order
 that we can explicitly prove Lemma \ref{lem_properties_cutoff}
 \eqref{item_final_cutoff},\eqref{item_support_shift}.
\end{remark}

We fix $\phi\in\C$ in the following unless otherwise stated. Using the
cut-off functions $\chi_l$ $(l=N_{\beta},N_{\beta}+1,\cdots,N_h)$, let us
define the covariances $C_l:I_0^2\to\C$
$(l=N_{\beta},N_{\beta}+1,\cdots,N_h)$ by 
\begin{align*}
C_l(\orho\rho\bx s,\oeta\eta\by t)
&:=\frac{1}{\beta L^d}\sum_{\bk\in\G^*}\sum_{\o\in
 \cM_h}e^{i\<\bk,\bx-\by\>+i(\o-\frac{\pi}{\beta})(s-t)}\chi_l(\o,\bk)\\
&\quad\cdot h^{-1}\left(I_{2b}-e^{-\frac{i}{h}(\o-\frac{\theta(\beta)}{2})I_{2b}+\frac{1}{h}E(\phi)(\bk)}\right)^{-1}((\orho-1)b+\rho,(\oeta-1)b+\eta).
\end{align*}
By Lemma \ref{lem_time_discrete_covariance} and Lemma \ref{lem_properties_cutoff} \eqref{item_unity_cutoff}, 
\begin{align}
&\sum_{l=N_{\beta}}^{N_h}C_l(\orho\rho\bx s, \oeta\eta\by
 t)=e^{-i\frac{\pi}{\beta}(s-t)}C(\phi)(\orho\rho\bx s, \oeta\eta\by
 t),\quad (\forall (\orho,\rho,\bx,s), (\oeta,\eta,\by,t)\in I_0).\label{eq_covariance_sum_unity}
\end{align}
We collect basic properties of $C_l$ in the next lemma.
During the proof and in subsequent arguments we will need to consider a
function on $(\{1,2\}\times\cB)^2$ as a $2b\times 2b$ matrix and measure
the function by using the norm $\|\cdot\|_{2b\times 2b}$. Let us set the
rule for this identification. For any $j\in \{1,2,\cdots,2b\}$ there
uniquely exists $(\orho,\rho)\in\{1,2\}\times\cB$ such that
$j=(\orho-1)b+\rho$. This defines the bijection
$\varphi:\{1,2,\cdots,2b\}\to \{1,2\}\times \cB$. 
We identify a function $f:(\{1,2\}\times \cB)^2\to\C$ with the $2b\times
2b$ matrix $(f(\varphi(i),\varphi(j)))_{1\le i,j\le 2b}$. We will apply this
rule to $C_l(\cdot \bx s,\cdot \by t):(\{1,2\}\times\cB)^2\to\C$ for
fixed $(\bx,s),(\by,t)\in\G\times [0,\beta)_h$ in particular.

\begin{lemma}\label{lem_real_covariances_properties}
Assume that 
\begin{align}
&h\ge \max\left\{2,\sc,\sup_{\bk\in\R^d}\|E(\phi)(\bk)\|_{2b\times
 2b}\right\},\label{eq_huge_h_condition}\\
&L\ge\label{eq_huge_L_condition}\\
&\max\left\{\frac{\max_{j\in\{1,2,\cdots,d\}}M^{-N_{\beta}/\sn_j}}{\min\{M^{\sa
 N_{\beta}},1\}},\frac{\max_{j\in\{1,2,\cdots,d\}}M^{-N_{\beta}/\sn_j}|\frac{\pi}{\beta}-\frac{\theta(\beta)}{2}|^{-1}+
|\frac{\pi}{\beta}-\frac{\theta(\beta)}{2}|^{-3}}{\min\{M^{(\sa-1)
 N_{\beta}},1,|\frac{\pi}{\beta}-\frac{\theta(\beta)}{2}|^{-1}\}}\right\}.\notag
\end{align}
Then there exists a positive constant $\hat{c}$ depending only on $d$,
 $b$, $(\hbv_j)_{j=1}^d$, $\sa$, $(\sn_j)_{j=1}^d$, $\sc$, $M$, $\chi$
 such that the following statements hold true.
\begin{enumerate}[(i)]
\item\label{item_preliminary_covariance_determinant_bound}
\begin{align}
&|\det(\<\bu_i,\bv_j\>_{\C^m}C_l(X_i,Y_j))_{1\le i,j\le n}|\label{eq_covariance_determinant_bound_pre}\\
&\le \left(\hat{c} \left(1_{l\neq N_{\beta}}\min\{M^{\sa l},1\}+
1_{l= N_{\beta}}\beta^{-1}\min\left\{M^{(\sa-1)N_{\beta}},1,
\left|\frac{\pi}{\beta}-\frac{\theta(\beta)}{2}\right|^{-1}\right\}\right)
\right)^n,\notag\\
&\left|\det\left(\<\bu_i,\bv_j\>_{\C^m}\sum_{p=N_{\beta}}^{\hat{N}_{\beta}-1}C_p(X_i,Y_j)\right)_{1\le
 i,j\le n}\right|
\le \left(\hat{c}M^{\hat{N}_{\beta}}\min\left\{1,
\left|\frac{\pi}{\beta}-\frac{\theta(\beta)}{2}\right|^{-1}\right\}\right)^n,
\label{eq_covariance_determinant_bound_lower_sum}\\
&(\forall m,n\in\N,\ \bu_i,\bv_i\in\C^m\text{ with
 }\|\bu_i\|_{\C^m},\|\bv_i\|_{\C^m}\le 1,\ X_i,Y_i\in I_0\
 (i=1,2,\cdots,n),\notag\\
&\quad l\in \{N_{\beta},N_{\beta}+1,\cdots,N_h\}).\notag
\end{align}
\item\label{item_preliminary_covariance_decay_bound}
\begin{align}
&\|\tilde{C}_l\|_{1,\infty}\le \hat{c} \left(1_{l\ge 0}M^{-l}+1_{l<
 0}M^{(\sa-1-\sum_{j=1}^d\frac{1}{\sn_j})l}\right),\label{eq_covariance_decay_bound_pre}\\
&(\forall l\in \{N_{\beta}+1,N_{\beta}+2,\cdots,N_h\}),\notag\\
&\|\tilde{C}_{N_{\beta}}\|_{1,\infty}\le
 \hat{c}\left|\frac{\pi}{\beta}-\frac{\theta(\beta)}{2}\right|^{-1}\prod_{j=1}^d\left(1+\left|\frac{\pi}{\beta}-\frac{\theta(\beta)}{2}\right|^{-\frac{1}{\sn_j}}
\right).\label{eq_covariance_decay_bound_end_pre}
\end{align}
\item\label{item_preliminary_covariance_coupled_decay_bound}
\begin{align}
&\|\tilde{C}_l\|\le \hat{c}\left(1_{l\ge 0}+1_{l<
 0}M^{(\sa-1-\sum^d_{j=1}\frac{1}{\sn_j})l}\right),\quad (\forall l\in \{N_{\beta}+1,N_{\beta}+2,\cdots,N_h\}).\label{eq_covariance_coupled_decay_bound_pre}\\
&\left\|\sum_{p=l'}^{N_h}\tilde{C}_p\right\|\le \hat{c}\left(
\sup_{(\orho,\rho,s,\xi),(\oeta,\eta,t,\zeta)\atop
\in \{1,2\}\times\cB\times [0,\beta)_h\times
 \{1,-1\}}\left|\sum_{p=l'}^{N_h}\tilde{C}_p(\orho\rho \b0 s \xi,\oeta\eta \b0 t \zeta)
\right|+1\right),\label{eq_covariance_coupled_decay_bound_sum_pre}\\
&(\forall l'\in \{N_{\beta}+1,N_{\beta}+2,\cdots,N_h\}\cap\Z_{\ge 0}).\notag
\end{align}
\end{enumerate}
\end{lemma}

\begin{remark} Here we need to assume that $h$ is large depending on
 $\phi$ as stated in \eqref{eq_huge_h_condition}. This is a notable
 difference from \cite[\mbox{Lemma 4.10}]{K_BCS} where we had no
 condition depending on $\phi$. We assume the $\phi$-dependent condition
 \eqref{eq_huge_h_condition} in order to simplify the proof of the
 lemma. We can see from Lemma \ref{lem_final_Grassmann_formulation}
 \eqref{item_final_Grassmann_formulation} that this condition does not
affect our goal since we take the limit $h\to \infty$ before the
 integration with $\phi$ in the final formulation.
\end{remark}

\begin{proof}[Proof of Lemma \ref{lem_real_covariances_properties}]
Let $x\in [-\pi h,\pi h]$, $\bk\in \R^d$, $\delta\in \{1,-1\}$ and
 $e_{\rho}(\bk)$ be an eigenvalue of $E(\bk)$. By the condition
 \eqref{eq_huge_h_condition}, $h\ge \sqrt{e_{\rho}(\bk)^2+|\phi|^2}$. Also,
$|\frac{1}{2h}(x-\frac{\theta(\beta)}{2})|\le
 \frac{\pi}{2}+\frac{\pi}{2h\beta}\le \frac{3\pi}{4}$. By using these
 inequalities and \eqref{eq_dispersion_upper_lower_bound},
\begin{align}
&\left|h\left(1-e^{-\frac{i}{h}(x-\frac{\theta(\beta)}{2})+\frac{\delta}{h}\sqrt{e_{\rho}(\bk)^2+|\phi|^2}}\right)\right|^2\label{eq_integrand_lower_bound}\\
&\ge
 h^2\left(1-e^{-\frac{1}{h}\sqrt{e_{\rho}(\bk)^2+|\phi|^2}}\right)^2+4h^2 
e^{-\frac{1}{h}\sqrt{e_{\rho}(\bk)^2+|\phi|^2}}
\sin^{2}\left(\frac{1}{2h}\left(x-\frac{\theta(\beta)}{2}\right)\right)\notag\\
&\ge 1_{\sqrt{e_{\rho}(\bk)^2+|\phi|^2}>h} h^2(1-e^{-1})^2\notag\\
&\quad+1_{\sqrt{e_{\rho}(\bk)^2+|\phi|^2}\le h} 
\left(e^{-2}(e_{\rho}(\bk)^2+|\phi|^2)+4h^2e^{-1}\sin^2\left(\frac{1}{2h}\left(x-\frac{\theta(\beta)}{2}\right)\right)\right)\notag\\
&\ge
 c\left(e(\bk)^2+|\phi|^2+\left(x-\frac{\theta(\beta)}{2}\right)^2\right).\notag
\end{align}
It follows from this inequality
 and \eqref{eq_hopping_matrix_derivative} that 
\begin{align}
&\left\|h^{-1}\left(I_{2b}-e^{-\frac{i}{h}(x-\frac{\theta(\beta)}{2})I_{2b}+\frac{1}{h}E(\phi)(\bk)}\right)^{-1}\right\|_{2b\times
 2b}\label{eq_integrand_bare_bound}\\
&\le
 c\left(\left(x-\frac{\theta(\beta)}{2}\right)^2+e(\bk)^2+|\phi|^2\right)^{-\frac{1}{2}},\notag\\
&(\forall (x,\bk)\in [-\pi h,\pi h]\times \R^d\text{ satisfying
 }(x-\theta(\beta)/2)^2+e(\bk)^2\neq 0),\notag\\
&\left\|\left(\frac{\partial}{\partial
 \hat{k}_j}\right)^nE(\phi)\left(\sum_{i=1}^d\hat{k}_i\hbv_i\right)\right\|_{2b\times
 2b}
\le \sc\left( 1_{n\le \sn_j}e\left(\sum_{i=1}^d\hat{k}_i\hbv_i\right)^{1-\frac{n}{\sn_j}}
+ 1_{\sn_j<n}\right),\label{eq_hopping_matrix_field_derivative}\\
&(\forall (\hat{k}_1,\cdots,\hat{k}_d)\in \R^d,\ n\in \{1,2,\cdots,d+2\},\ j\in
 \{1,2,\cdots,d\}).\notag
\end{align}
The following inequality will also be useful. 
\begin{align}
\left|\o-\frac{\theta(\beta)}{2}\right|\ge
 h\left|\sin\left(\frac{\o-\pi/\beta}{2h}\right)\right|,\quad (\forall
 \o\in \cM_h).\label{eq_matsubara_shift}
\end{align}

\eqref{item_preliminary_covariance_determinant_bound}: Set the Hilbert
 space $\cH$ by 
 $\cH:=L^2(\{1,2\}\times \cB\times \G^*\times \cM_h)$. The inner product
 $\<\cdot,\cdot\>_{\cH}$ of $\cH$ is defined by 
\begin{align*}
\<f,g\>_{\cH}:=\frac{1}{\beta L^d}\sum_{K\in
 \{1,2\}\times\cB\times\G^*\times\cM_h}\overline{f(K)}g(K),\quad (f,g\in
 \cH).
\end{align*}
We are going to apply Gram's inequality. 
Let us define $f_{X}^l$, $g_X^l$, $f_X^{>}$, $g_X^{>}\in \cH$ $(X\in
 I_0,\ l\in \{N_{\beta},N_{\beta}+1,\cdots,N_h\})$ by 
\begin{align*}
&f_{\orho\rho \bx s}^{\xi}(\oeta,\eta,\bk,\o)\\
&:=e^{-i\<\bk,\bx\>-is(\o-\frac{\pi}{\beta})}
\left(1_{\xi\in \{N_{\beta},N_{\beta}+1,\cdots,N_h\}}
 \chi_{\xi}(\o,\bk)^{\frac{1}{2}}
+1_{\xi=>}\left(\sum_{j=N_{\beta}}^{\hat{N}_{\beta}-1}\chi_j(\o,\bk)
\right)^{\frac{1}{2}}\right)\\
&\quad\cdot 1_{(\orho,\rho)=(\oeta,\eta)}
\left(
\left(\o-\frac{\theta(\beta)}{2}\right)^2
+e(\bk)^2+|\phi|^2\right)^{-\frac{1}{4}},\\
&g_{\orho\rho \bx s}^{\xi}(\oeta,\eta,\bk,\o)\\
&:=e^{-i\<\bk,\bx\>-is(\o-\frac{\pi}{\beta})}
\left(1_{\xi\in \{N_{\beta},N_{\beta}+1,\cdots,N_h\}}
 \chi_{\xi}(\o,\bk)^{\frac{1}{2}}
+1_{\xi=>}\left(\sum_{j=N_{\beta}}^{\hat{N}_{\beta}-1}\chi_j(\o,\bk)
\right)^{\frac{1}{2}}\right)\\
&\quad \cdot 
\left(
\left(\o-\frac{\theta(\beta)}{2}\right)^2
+e(\bk)^2+|\phi|^2\right)^{\frac{1}{4}}\\
&\quad\cdot 
h^{-1}\left(I_{2b}-e^{-\frac{i}{h}(\o-\frac{\theta(\beta)}{2})I_{2b}+\frac{1}{h}E(\phi)(\bk)}\right)^{-1}((\oeta-1)b+\eta,(\orho-1)b+\rho),\\
&(\xi\in \{N_{\beta},N_{\beta}+1,\cdots,N_h,>\}).
\end{align*}
Observe that
\begin{align*}
&\<f_X^l,g_Y^l\>_{\cH}=C_l(X,Y),\quad
 \<f_X^{>},g_Y^{>}\>_{\cH}=\sum_{p=N_{\beta}}^{\hat{N}_{\beta}-1}C_p(X,Y),\\
&(\forall X,Y\in I_0,\ l\in \{N_{\beta},N_{\beta}+1,\cdots,N_h\}).
\end{align*}
It follows from Lemma \ref{lem_properties_cutoff}
 \eqref{item_support_size_cutoff},\eqref{item_final_cutoff},  
\eqref{eq_support_description},
 \eqref{eq_integrand_bare_bound}, \eqref{eq_matsubara_shift}
 that for any $X\in I_0$, $l\in
 \{N_{\beta}+1,N_{\beta}+2,\cdots,N_h\}$, 
\begin{align}
&\|f_X^l\|_{\cH}^2,\ \|g_X^l\|_{\cH}^2\le
 c(d,M,\chi,\sc,\sa,(\hbv_j)_{j=1}^d)\min\{M^{\sa
 l},1\},\label{eq_gram_inequality_pre}\\
&\|f_X^{N_{\beta}}\|_{\cH}^2,\ \|g_X^{N_{\beta}}\|_{\cH}^2\notag\\
&\le
  \frac{c}{\beta
 L^d}\sum_{\bk\in\G^*}\chi(M^{-N_{\beta}}A(\beta,M)^{-1}e(\bk))
\left(\left(\frac{\pi}{\beta}-\frac{\theta(\beta)}{2}\right)^2+e(\bk)^2+|\phi|^2\right)^{-\frac{1}{2}},\notag\\
&\|f_X^{>}\|_{\cH}^2,\ \|g_X^{>}\|_{\cH}^2\notag\\
&\le
 c(d,M,\chi,\sc,\sa,(\hbv_j)_{j=1}^d) M^{\hat{N}_{\beta}}
\frac{1}{L^d}\sum_{\bk\in\G^*}
\left(\left(\frac{\pi}{\beta}-\frac{\theta(\beta)}{2}\right)^2+e(\bk)^2+|\phi|^2\right)^{-\frac{1}{2}}.\notag
\end{align}
By \eqref{eq_dispersion_upper_bound}, \eqref{eq_dispersion_derivative},
\eqref{eq_dispersion_measure_divided}, the support property of $\chi$,
 $A(\beta,M)\le M$,
 \eqref{eq_derivative_cutoff_essence}, the assumption
 \eqref{eq_huge_L_condition} and Lemma
 \ref{lem_discrete_continuous_estimate},
\begin{align}
&\frac{1}{
 L^d}\sum_{\bk\in\G^*}\chi(M^{-N_{\beta}}A(\beta,M)^{-1}e(\bk))
\left(\left(\frac{\pi}{\beta}-\frac{\theta(\beta)}{2}\right)^2+e(\bk)^2+|\phi|^2\right)^{-\frac{1}{2}}\label{eq_gram_inequality_inside}\\
&\le D_d\int_{\G_{\infty}^*}d\bk\chi(M^{-N_{\beta}}A(\beta,M)^{-1}e(\bk))
\left(\left(\frac{\pi}{\beta}-\frac{\theta(\beta)}{2}\right)^2+e(\bk)^2+|\phi|^2\right)^{-\frac{1}{2}}\notag\\
&\quad + c(d,(\hbv_j)_{j=1}^d)L^{-1}\sup_{j\in
 \{1,2,\cdots,d\}}\sup_{\bk\in \R^d}\notag\\
&\qquad\quad\cdot 
\left|\frac{\partial}{\partial k_j}\left(\chi(M^{-N_{\beta}}A(\beta,M)^{-1}e(\bk))
\left(\left(\frac{\pi}{\beta}-\frac{\theta(\beta)}{2}\right)^2+e(\bk)^2+|\phi|^2\right)^{-\frac{1}{2}}\right)\right|\notag\\
&\le c(D_d)\min\left\{\int_{\G_{\infty}^*}d\bk1_{e(\bk)\le
 2M^{N_{\beta}+1}}e(\bk)^{-1},\quad
 \left|\frac{\pi}{\beta}-\frac{\theta(\beta)}{2}\right|^{-1}\right\}\notag\\
&\quad +c(d,M,\chi,\sc,(\hbv_j)_{j=1}^d)
L^{-1}\left(\max_{j\in \{1,2,\cdots,d\}}M^{-\frac{N_{\beta}}{\sn_j}}
\left|\frac{\pi}{\beta}-\frac{\theta(\beta)}{2}\right|^{-1}+
\left|\frac{\pi}{\beta}-\frac{\theta(\beta)}{2}\right|^{-3}\right)\notag\\
&\le c(d,M,\chi,\sc,(\hbv_j)_{j=1}^d,D_d,\sa)\min\left\{
M^{(\sa-1)N_{\beta}},1,\left|\frac{\pi}{\beta}-\frac{\theta(\beta)}{2}\right|^{-1}
\right\},\notag\\
&\frac{1}{L^d}\sum_{\bk\in\G^*}
\left(\left(\frac{\pi}{\beta}-\frac{\theta(\beta)}{2}\right)^2+e(\bk)^2+|\phi|^2\right)^{-\frac{1}{2}}\notag\\
&\le D_d\int_{\G_{\infty}^*}d\bk\left(\left(\frac{\pi}{\beta}-\frac{\theta(\beta)}{2}\right)^2+e(\bk)^2+|\phi|^2\right)^{-\frac{1}{2}}\notag\\
&\quad + c(d,(\hbv_j)_{j=1}^d)L^{-1}\sup_{j\in
 \{1,2,\cdots,d\}}\sup_{\bk\in \R^d}
\left|\frac{\partial}{\partial k_j}
\left(\left(\frac{\pi}{\beta}-\frac{\theta(\beta)}{2}\right)^2+e(\bk)^2+|\phi|^2\right)^{-\frac{1}{2}}\right|\notag\\
&\le c(d,D_d,\sc,(\hbv_j)_{j=1}^d)\left(
\min\left\{1,\left|\frac{\pi}{\beta}-\frac{\theta(\beta)}{2}\right|^{-1}\right\}
+L^{-1}\left|\frac{\pi}{\beta}-\frac{\theta(\beta)}{2}\right|^{-3}\right)\notag\\
&\le
 c(d,D_d,\sc,(\hbv_j)_{j=1}^d)\min\left\{1,\left|\frac{\pi}{\beta}-\frac{\theta(\beta)}{2}\right|^{-1}\right\}.\notag
\end{align}
We can apply Gram's inequality in the Hilbert space $\C^m\otimes \cH$
 together with \eqref{eq_gram_inequality_pre},
 \eqref{eq_gram_inequality_inside} to derive the claimed bounds.

\eqref{item_preliminary_covariance_decay_bound}: By Lemma
 \ref{lem_properties_cutoff}
 \eqref{item_regularity_cutoff},\eqref{item_support_shift} and
 \eqref{eq_integrand_lower_bound} the matrix-valued functions 
 \begin{align*}
&(x,\bk)\mapsto
  \chi_l(x,\bk)\left(I_{2b}-e^{-\frac{i}{h}(x-\frac{\theta(\beta)}{2})I_{2b}+\frac{1}{h}E(\phi)(\bk)}\right)^{-1}:\R^{d+1}\to
  \Mat(2b,\C),\\
&(l=N_{\beta}+1,N_{\beta}+2,\cdots,N_h)
\end{align*}
are well-defined and $C^{\infty}$-class. Indeed, the matrix-valued
 function with $l$ is identically zero in the open set 
\begin{align}
\left\{(x,\bk)\in\R^{d+1}\ \Big|\
 \sqrt{h^2\sin^2\left(\frac{x-\pi/\beta}{2h}\right)+e(\bk)^2}<\frac{8}{5}A(\beta,M)M^{l-1}\right\}.
\label{eq_confirmation_matrix_open_set}
\end{align}
By Lemma \ref{lem_properties_cutoff} \eqref{item_support_shift},
 \eqref{eq_integrand_lower_bound} and the periodicity, for $(x,\bk)\in\R^{d+1}$ satisfying 
\begin{align*}
\sqrt{h^2\sin^2\left(\frac{x-\pi/\beta}{2h}\right)+e(\bk)^2}\ge
 \frac{8}{5}A(\beta,M)M^{l-1},
\end{align*}
\begin{align*}
\left|\det\left(I_{2b}-e^{-\frac{i}{h}(x-\frac{\theta(\beta)}{2})I_{2b}+\frac{1}{h}E(\phi)(\bk)}\right)\right|\ge
 \left(c\left(\frac{8}{5}-\frac{\pi}{2}\right)A(\beta,M)M^{l-1}\right)^{2b}>0.
\end{align*}
This implies that at any point belonging to the complement of the set
 \eqref{eq_confirmation_matrix_open_set} the matrix
 $(I_{2b}-e^{-\frac{i}{h}(x-\frac{\theta(\beta)}{2})I_{2b}+\frac{1}{h}E(\phi)(\bk)})^{-1}$ is well-defined and infinitely differentiable. Thus the claim follows.
For any $\o\in \cM_h$,
 $\o-\theta(\beta)/2\neq 0$ (mod $2\pi h$) and thus the matrix-valued
 function 
 \begin{align*}
\bk\mapsto
  \chi_{N_{\beta}}(\o,\bk)\left(I_{2b}-e^{-\frac{i}{h}(\o-\frac{\theta(\beta)}{2})I_{2b}+\frac{1}{h}E(\phi)(\bk)}\right)^{-1}:\R^{d}\to
  \Mat(2b,\C)
\end{align*}
is well-defined and $C^{\infty}$-class. By keeping these basic facts in
 mind and using the periodicity we can derive that
for $n\in \N$, $j\in \{1,2,\cdots,d\}$, 
\begin{align}
&\left(\frac{\beta}{2\pi}(e^{-i\frac{2\pi}{\beta}(s-t)}-1)\right)^nC_l(\cdot\bx
 s,\cdot \by t)\label{eq_time_periodicity_application}\\
&=\frac{1}{\beta
 L^d}\sum_{\bk\in\G^*}\sum_{\o\in\cM_h}e^{i\<\bk,\bx-\by\>+i(\o-\frac{\pi}{\beta})(s-t)} \prod_{m=1}^n\left(\frac{\beta}{2\pi}\int_0^{\frac{2\pi}{\beta}}dr_m
\right)\notag\\
&\quad\cdot \left(\frac{\partial}{\partial
 r}\right)^n\chi_l(r,\bk)h^{-1}\left(I_{2b}-e^{-\frac{i}{h}(r-\frac{\theta(\beta)}{2})I_{2b}+\frac{1}{h}E(\phi)(\bk)}\right)^{-1}\Bigg|_{r=\o+\sum_{m=1}^nr_m},\notag\\
&\quad (l\in \{N_{\beta}+1,N_{\beta}+2,\cdots,N_h\}),\notag\\
&\left(\frac{L}{2\pi}(e^{-i\frac{2\pi}{L}\<\bx-\by,\hbv_j\>}-1)\right)^nC_{l'}(\cdot\bx
 s,\cdot \by t)\label{eq_space_periodicity_application}\\
&=\frac{1}{\beta
 L^d}\sum_{\bk\in\G^*}\sum_{\o\in\cM_h}e^{i\<\bk,\bx-\by\>+i(\o-\frac{\pi}{\beta})(s-t)}\prod_{m=1}^n\left(\frac{L}{2\pi}\int_0^{\frac{2\pi}{L}}dp_m
\right)\notag\\
&\quad\cdot \left(\frac{\partial}{\partial
 \hat{k}_j}\right)^n\chi_{l'}(\o,\bk+\hat{k}_j\hbv_j)h^{-1}\left(I_{2b}-e^{-\frac{i}{h}(\o-\frac{\theta(\beta)}{2})I_{2b}+\frac{1}{h}E(\phi)(\bk+\hat{k}_j\hbv_j)}\right)^{-1}\Bigg|_{\hat{k}_j=\sum_{m=1}^np_m},\notag\\
&\quad (l'\in \{N_{\beta},N_{\beta}+1,\cdots,N_h\}).\notag
\end{align}

For $(\o,\bk)\in\R^{d+1}$, set 
\begin{align*}
B(\o,\bk):=h\left(I_{2b}-e^{-\frac{i}{h}(\o-\frac{\theta(\beta)}{2})I_{2b}+\frac{1}{h}E(\phi)(\bk)}\right).
\end{align*}
Observe that
\begin{align}
&\left\|\left(\frac{\partial}{\partial
 \hat{k}_j}\right)^nB(\o,\bk)^{-1}\right\|_{2b\times 2b}\label{eq_inverse_matrix_derivative_bound}\\ 
&\le
 c(d)\sum_{m=1}^n\prod_{u=1}^m\left(\sum_{l_u=1}^n\right)1_{\sum_{u=1}^ml_u=n}\notag\\
&\quad\cdot \prod_{p=1}^m\left\|
B(\o,\bk)^{-1}\left(\frac{\partial}{\partial \hat{k}_j}\right)^{l_p}B(\o,\bk)
\right\|_{2b\times 2b}\|B(\o,\bk)^{-1}\|_{2b\times 2b},\notag\\
&(\forall n\in \{1,2,\cdots,d+2\},\ \o\in [-\pi h,\pi h],\
 (\hat{k}_1,\cdots,\hat{k}_d)\in\R^d\text{ satisfying }\notag\\
&\quad (\o-\theta(\beta)/2)^2+e(\bk)^2\neq 0,\  j\in \{0,1,\cdots,d\}),\notag
\end{align}
where we set $\bk:=\sum_{j=1}^d\hat{k}_j\hbv_j$,
$\partial/\partial \hat{k}_0:=\partial /\partial \o$. This
 inequality follows from e.g. the formula \cite[\mbox{(C.1)}]{K_RG}. By
 using the assumption $h\ge \sup_{\bk\in\R^d}\|E(\phi)(\bk)\|_{2b\times
 2b}$ and \eqref{eq_integrand_bare_bound} we can derive from
 \eqref{eq_inverse_matrix_derivative_bound} that
\begin{align}
&\left\|\left(\frac{\partial}{\partial
 \o}\right)^nB(\o,\bk)^{-1}\right\|_{2b\times
 2b}
\le c(d)\sum_{m=1}^nh^{m-n}\|B(\o,\bk)^{-1}\|_{2b\times 2b}^{m+1}
\label{eq_integrand_matsubara_derivative}\\
&\le
 c(d)\sum_{m=1}^nh^{m-n}\left(\left(\o-\frac{\theta(\beta)}{2}\right)^2+e(\bk)^2+|\phi|^2\right)^{-\frac{m+1}{2}},\notag\\
&(\forall n\in \{1,2,\cdots,d+2\},\ \o\in [-\pi h,\pi h],\
 \bk\in\R^d\text{ satisfying }\notag\\
&\quad (\o-\theta(\beta)/2)^2+e(\bk)^2\neq
 0).\notag
\end{align}
On the other hand, it follows from the inequality 
\begin{align*}
\left(\frac{e(\bk)}{\sc}\right)^s\le
 \left(\frac{e(\bk)}{\sc}\right)^t,\quad (\forall \bk\in \R^d,\ s,t\in
 \R_{\ge 0}\text{ satisfying }t\le s)
\end{align*}
and \eqref{eq_hopping_matrix_field_derivative} that
\begin{align}
&\left\|\left(\frac{\partial}{\partial
 \hat{k}_j}\right)^mE(\phi)\left(\sum_{i=1}^d\hat{k}_i\hbv_i\right)\right\|_{2b\times 2b}\le c(\sc)\left(1_{n\le
 \sn_j}e\left(\sum_{i=1}^d\hat{k}_i\hbv_i\right)^{1-\frac{n}{\sn_j}}+1_{\sn_j<n}\right),\label{eq_hopping_matrix_field_derivative_application}\\
&(\forall (\hat{k}_1,\cdots,\hat{k}_d)\in\R^d,\ m,n\in \{1,2,\cdots,d+2\}\text{ satisfying }m\le
 n,\ j\in \{1,2,\cdots,d\}).\notag
\end{align}
Let us admit that $\bk=\sum_{j=1}^d\hat{k}_j\hbv_j$,
 $(\hat{k}_1,\cdots,\hat{k}_d)\in\R^d$ in the following arguments. 
By using \eqref{eq_hopping_matrix_field_derivative_application}, the
 assumption \eqref{eq_huge_h_condition} and the formula
\begin{align*}
\frac{\partial}{\partial
 \hat{k}_j}e^{\frac{1}{h}E(\phi)(\bk)}=\frac{1}{h}\int_0^1ds
 e^{\frac{s}{h}E(\phi)(\bk)}\frac{\partial}{\partial
 \hat{k}_j}E(\phi)(\bk)e^{\frac{1-s}{h}E(\phi)(\bk)} 
\end{align*}
repeatedly we obtain that 
\begin{align*}
&\left\|\left(\frac{\partial}{\partial
 \hat{k}_j}\right)^nB(\o,\bk)\right\|_{2b\times 2b}\\
&\le
 c(d,\sc)\sum_{m=1}^n\left\|\left(\frac{\partial}{\partial \hat{k}_j}\right)^mE(\phi)(\bk)
\right\|_{2b\times 2b}
\le c(d,\sc)\left(
1_{n\le \sn_j}e(\bk)^{1-\frac{n}{\sn_j}}+1_{\sn_j<n}\right)\\
&\le c(d,\sc)
 \left(\left(\o-\frac{\theta(\beta)}{2}\right)^2+e(\bk)^2+|\phi|^2\right)^{\frac{1}{2}(1-\frac{n}{\sn_j})-1_{\sn_j<n}\frac{1}{2}(1-\frac{n}{\sn_j})},\\
&(\forall j\in \{1,2,\cdots,d\},\ n\in \{1,2,\cdots,d+2\}).
\end{align*}
By substituting this inequality and \eqref{eq_integrand_bare_bound} into
 \eqref{eq_inverse_matrix_derivative_bound} we have that
\begin{align}
&\left\|\left(\frac{\partial}{\partial
 \hat{k}_j}\right)^nB(\o,\bk)^{-1}\right\|_{2b\times 2b}\label{eq_integrand_momentum_derivative}\\
&\le
 c(d,\sc)\sum_{m=1}^n\prod_{u=1}^m\left(\sum_{l_u=1}^n\right)1_{\sum_{u=1}^ml_u=n}\notag\\
&\quad\cdot 
\left(\left(\o-\frac{\theta(\beta)}{2}\right)^2+e(\bk)^2+|\phi|^2\right)^{-\frac{1}{2}(m+1)+\sum_{p=1}^m(\frac{1}{2}(1-\frac{l_p}{\sn_j})-1_{\sn_j<l_p}\frac{1}{2}(1-\frac{l_p}{\sn_j}))}\notag\\
&=c(d,\sc)\sum_{m=1}^n\prod_{u=1}^m
\left(\sum_{l_u=1}^n\right)1_{\sum_{u=1}^ml_u=n}\notag\\
&\quad\cdot 
\left(\left(\o-\frac{\theta(\beta)}{2}\right)^2+e(\bk)^2+|\phi|^2\right)^{-\frac{1}{2}
-\frac{n}{2\sn_j}+\sum_{p=1}^m1_{\sn_j<l_p}\frac{1}{2}(\frac{l_p}{\sn_j}-1)}\notag\\
&=c(d,\sc)\sum_{m=1}^n\prod_{u=1}^m\left(\sum_{l_u=1}^n\right)1_{\sum_{u=1}^ml_u=n}\notag\\
&\quad\cdot 
\left(\left(\o-\frac{\theta(\beta)}{2}\right)^2+e(\bk)^2+|\phi|^2\right)^{-\frac{1}{2}-\sum_{p=1}^m(1_{l_p\le\sn_j}\frac{l_p}{2\sn_j}+1_{\sn_j<l_p}\frac{1}{2})},\notag\\
&(\forall n\in \{1,2,\cdots,d+2\},\ \o\in [-\pi h,\pi h],\ \bk\in
 \R^d\text{ satisfying }\notag\\
&\quad (\o-\theta(\beta)/2)^2+e(\bk)^2\neq
 0,\ j\in \{1,2,\cdots,d\}).\notag
\end{align}

By using Lemma
 \ref{lem_properties_cutoff}
 \eqref{item_derivative_cutoff},\eqref{item_support_size_cutoff},\eqref{item_support_shift},
 \eqref{eq_integrand_bare_bound},
 \eqref{eq_integrand_matsubara_derivative} we can derive from
 \eqref{eq_time_periodicity_application} that for $l\in
 \{N_{\beta}+1,N_{\beta}+2,\cdots,N_h\}$ 
\begin{align}
&\left\|\left(\frac{\beta}{2\pi}(e^{-i\frac{2\pi}{\beta}(s-t)}-1)\right)^{d+2}C_l(\cdot\bx
 s,\cdot \by t)\right\|_{2b\times 2b}\label{eq_covariance_time_decay}\\
&\le \frac{1}{\beta
 L^d}\sup_{x\in\R}\sum_{\bk\in\G^*}\sum_{\o\in\cM_h}1_{\chi_l(\o+x,\bk)\neq
 0}\sup_{r\in [-\pi h,\pi h]\atop \bp\in \R^d}\left\|
\left(\frac{\partial}{\partial r}\right)^{d+2}\chi_l(r,\bp)B(r,\bp)^{-1}
\right\|_{2b\times 2b}\notag\\
&\le
 c(d,M,\chi,\sc,\sa,(\hbv_j)_{j=1}^d)M^{l}\min\{M^{\sa
 l},1\}\notag\\
&\quad\cdot \sum_{m=0}^{d+2}\sup_{r\in [-\pi h,\pi h]\atop \bp\in \R^d}\left|
\left(\frac{\partial}{\partial r}\right)^{m}\chi_l(r,\bp)\right|
\left\|\left(\frac{\partial}{\partial r}\right)^{d+2-m}
B(r,\bp)^{-1}\right\|_{2b\times 2b}\notag\\
&\le c(d,M,\chi,\sc,\sa,(\hbv_j)_{j=1}^d)M^{l}\min\{M^{\sa
 l},1\}\notag\\
&\quad \cdot \Bigg(\sum_{m=0}^{d+1}\sup_{r\in [-\pi h,\pi h]\atop
 \bp\in \R^d}\notag\\
&\qquad\cdot 
\left(\left|\left(\frac{\partial}{\partial r}\right)^{m}\chi_l(r,\bp)\right|
\sum_{u=1}^{d+2-m}h^{u-(d+2-m)}
\left(\left(r-\frac{\theta(\beta)}{2}\right)^2+e(\bp)^2+|\phi|^2\right)^{-\frac{u+1}{2}}\right)\notag\\
&\qquad+\sup_{r\in [-\pi h,\pi h]\atop \bp\in \R^d}
\left|\left(\frac{\partial}{\partial r}\right)^{d+2}\chi_l(r,\bp)\right|
\left(\left(r-\frac{\theta(\beta)}{2}\right)^2+e(\bp)^2+|\phi|^2\right)^{-\frac{1}{2}}
\Bigg)\notag\\
&\le  c(d,M,\chi,\sc,\sa,(\hbv_j)_{j=1}^d)M^{l}\min\{M^{\sa
 l},1\}\notag\\
&\quad\cdot \left(\sum_{m=0}^{d+1}M^{-ml}\sum_{u=1}^{d+2-m}h^{u-(d+2-m)}M^{-(u+1)l}+ M^{-(d+2)l-l}
\right)\notag\\
&\le  c(d,M,\chi,\sc,\sa,(\hbv_j)_{j=1}^d)M^{l}\min\{M^{\sa
 l},1\}\notag\\
&\quad\cdot \left(\sum_{m=0}^{d+1}M^{-ml}M^{-(d+3-m)l}+M^{-(d+3)l}\right)\notag\\
&\le c(d,M,\chi,\sc,\sa,(\hbv_j)_{j=1}^d) M^{-(d+2)l}\min\{M^{\sa
 l},1\},\notag
\end{align}
where we also used that $h\ge M^{N_h-2}$. By combining
 Lemma \ref{lem_properties_cutoff}
 \eqref{item_derivative_cutoff},\eqref{item_support_size_cutoff},\eqref{item_support_shift}, 
 \eqref{eq_integrand_bare_bound},
 \eqref{eq_integrand_momentum_derivative} with
 \eqref{eq_space_periodicity_application} we deduce that for $l\in
 \{N_{\beta}+1,N_{\beta}+2,\cdots,N_h\}$, $n\in \{1,2,\cdots,d+2\}$,
\begin{align}
&\left\|\left(\frac{L}{2\pi}(e^{-i\frac{2\pi}{L}\<\bx-\by,\hbv_j\>}-1)\right)^n
C_l(\cdot\bx s,\cdot \by t)\right\|_{2b\times
 2b}\label{eq_covariance_space_decay}\\
&\le \frac{1}{\beta L^d}\sup_{\bp\in
 \R^d}\sum_{\bk\in\G^*}\sum_{\o\in \cM_h}1_{\chi_l(\o,\bk+\bp)\neq
 0}\sup_{\o\in \cM_h\atop \bk\in\R^d}
\left\|\left(\frac{\partial}{\partial \hat{k}_j}\right)^n\chi_l(\o,\bk)B(\o,\bk)^{-1}
\right\|_{2b\times 2b}\notag\\
&\le c(d,M,\chi,\sc, \sa,
 (\hbv_j)_{j=1}^d)M^l\min\{M^{\sa l},l\}\notag\\
&\quad\cdot \Bigg(\sup_{\o\in\cM_h\atop
 \bk\in\R^d}\left|\left(\frac{\partial}{\partial \hat{k}_j}\right)^n\chi_l(w,\bk)\right|
\|B(\o,\bk)^{-1}\|_{2b\times 2b}\notag\\
&\qquad\quad+\sum_{m=0}^{n-1}\sup_{\o\in\cM_h\atop
 \bk\in\R^d}\left|\left(\frac{\partial}{\partial \hat{k}_j}\right)^m\chi_l(w,\bk)\right|
\left\|\left(\frac{\partial}{\partial
 \hat{k}_j}\right)^{n-m}B(\o,\bk)^{-1}\right\|_{2b\times 2b}\Bigg)\notag\\
&\le c(d,M,\chi,\sc,\sa,(\hbv_j)_{j=1}^d)M^l\min\{M^{\sa
 l},1\}\notag\\
&\quad\cdot\Bigg(1_{l\ge
 0}\Bigg(M^{-(\frac{1}{\sn_j}+1)l}+\sum_{m=0}^{n-1}\left(
1_{m=0}+1_{m\ge
 1}M^{-\frac{l}{\sn_j}}\right)\sum_{k=1}^{n-m}\prod_{u=1}^k\left(
\sum_{l_u=1}^{n-m}\right)
1_{\sum_{u=1}^kl_u=n-m}\notag\\
&\qquad\qquad\qquad\qquad\qquad\qquad \cdot M^{2l(-\frac{1}{2}-\sum_{p=1}^k(1_{l_p\le
 \sn_j}\frac{l_p}{2\sn_j}+1_{\sn_j<l_p}\frac{1}{2}))}\Bigg)\notag\\
&\qquad\quad + 1_{l< 0}\Bigg(M^{-(\frac{n}{\sn_j}+1)l}+\sum_{m=0}^{n-1}
M^{-\frac{m}{\sn_j}l}\sum_{k=1}^{n-m}\prod_{u=1}^k\left(
\sum_{l_u=1}^{n-m}\right)
1_{\sum_{u=1}^kl_u=n-m}\notag\\
&\qquad\qquad\qquad\qquad\qquad\qquad\quad \cdot M^{2l(-\frac{1}{2}-\frac{n-m}{2\sn_j}+
\sum_{p=1}^k
 1_{\sn_j<l_p}\frac{1}{2}(\frac{l_p}{\sn_j}-1))}\Bigg)\Bigg)\notag\\
&\le c(d,M,\chi,\sc,\sa,(\hbv_j)_{j=1}^d)M^l\min\{M^{\sa
 l},1\}\notag\\
&\quad\cdot\Bigg(1_{l\ge
 0}\left(M^{-(\frac{1}{\sn_j}+1)l}+\sum_{m=0}^{n-1}\left(
1_{m=0}+1_{m\ge
 1}M^{-\frac{l}{\sn_j}}\right)M^{2l(-\frac{1}{2}-\frac{1}{2\sn_j})}\right)\notag\\
&\qquad\quad+1_{l<
 0}\left(M^{-(\frac{n}{\sn_j}+1)l}+\sum_{m=0}^{n-1}M^{-\frac{m}{\sn_j}l}M^{2l(-\frac{1}{2}-\frac{n-m}{2\sn_j})}\right)\Bigg)\notag\\
&\le  c(d,M,\chi,\sc,\sa,(\hbv_j)_{j=1}^d)\min\{M^{\sa l},1\}
\left(1_{l\ge
 0}M^{-\frac{l}{\sn_j}}+1_{l<0}M^{-\frac{n}{\sn_j}l}\right).\notag
\end{align}
Here we estimated for all $n\in\{1,2,\cdots,d+2\}$ not only for $n=d+2$ so
 that we can use the result to prove Lemma
 \ref{lem_covariance_construction}
 \eqref{item_covariance_construction_ODLRO} later.
Moreover by combining Lemma \ref{lem_properties_cutoff}
 \eqref{item_derivative_cutoff},\eqref{item_final_cutoff},
 \eqref{eq_integrand_bare_bound},
 \eqref{eq_integrand_momentum_derivative} with
 \eqref{eq_space_periodicity_application} and using that
 $M^{-N_{\beta}}\le M\beta\le \pi M |\frac{\pi}{\beta}-\frac{\theta(\beta)}{2}|^{-1}$ we have that
\begin{align}
&\left\|\left(\frac{L}{2\pi}(e^{-i\frac{2\pi}{L}\<\bx-\by,\hbv_j\>}-1)\right)^{d+1}C_{N_{\beta}}(\cdot\bx s,\cdot \by t)\right\|_{2b\times
 2b}\label{eq_covariance_space_decay_end}\\
&\le \frac{1}{\beta}\sup_{\bk\in\R^d}\left\|\left(\frac{\partial}{\partial
 \hat{k}_j}\right)^{d+1}\chi_{N_{\beta}}\left(\frac{\pi}{\beta},\bk\right)
B\left(\frac{\pi}{\beta},\bk\right)^{-1}\right\|_{2b\times 2b}\notag\\
&\le  \frac{c(d)}{\beta}\Bigg(
\sup_{\bk\in \R^d}\left|
\left(\frac{\partial}{\partial
 \hat{k}_j}\right)^{d+1}\chi_{N_{\beta}}\left(\frac{\pi}{\beta},\bk\right)\right|\left|
\frac{\pi}{\beta}-\frac{\theta(\beta)}{2}\right|^{-1}\notag\\
&\qquad\qquad +\sum_{m=0}^d\sup_{\bk\in\R^d}\left|
\left(\frac{\partial}{\partial
 \hat{k}_j}\right)^{m}\chi_{N_{\beta}}\left(\frac{\pi}{\beta},\bk\right)
\right|\left\|\left(\frac{\partial}{\partial
 \hat{k}_j}\right)^{d+1-m}B\left(\frac{\pi}{\beta},\bk\right)^{-1}\right\|_{2b\times
 2b}\Bigg)\notag\\
&\le
 c(d,M,\chi,\sc,\sa,(\hbv_j)_{j=1}^d)\beta^{-1}\notag\\
&\quad\cdot \Bigg(
\left(1_{N_{\beta}\ge
 0}M^{-\frac{N_{\beta}}{\sn_j}}+1_{N_{\beta}<0}M^{-\frac{d+1}{\sn_j}N_{\beta}}\right)\left|\frac{\pi}{\beta}-\frac{\theta(\beta)}{2}\right|^{-1}\notag\\
&\qquad+ \sum_{m=0}^d \left(1_{m=0}+1_{m\ge 1} \left(1_{N_{\beta}\ge 0}M^{-\frac{N_{\beta}}{\sn_j}}+1_{N_{\beta}<0}M^{-\frac{m}{\sn_j}N_{\beta}}
\right)\right)\notag\\
&\qquad\quad\cdot \sum_{k=1}^{d+1-m}\prod_{u=1}^k\left(\sum_{l_u=1}^{d+1-m}\right)1_{\sum_{u=1}^kl_u=d+1-m}\notag\\
&\qquad\quad\cdot\Bigg(
1_{|\pi/\beta-\theta(\beta)/2|>1}\left|\frac{\pi}{\beta}-\frac{\theta(\beta)}{2}\right|^{-1-\sum_{p=1}^k(1_{l_p\le
 \sn_j}\frac{l_p}{\sn_j}+1_{\sn_j<l_p})}\notag\\
&\qquad\qquad\quad+1_{|\pi/\beta-\theta(\beta)/2|\le 1}
\left|\frac{\pi}{\beta}-\frac{\theta(\beta)}{2}\right|^{-1-\frac{d+1-m}{\sn_j}
+\sum_{p=1}^k1_{\sn_j<l_p}(\frac{l_p}{\sn_j}-1)}
\Bigg)\Bigg)\notag\\
&\le
 c(d,M,\chi,\sc,\sa,(\hbv_j)_{j=1}^d)\beta^{-1}\notag\\
&\quad\cdot
 \Bigg(\left(1+
\left|\frac{\pi}{\beta}-\frac{\theta(\beta)}{2}\right|^{-\frac{d+1}{\sn_j}}
\right)\left|\frac{\pi}{\beta}-\frac{\theta(\beta)}{2}\right|^{-1}\notag\\
&\qquad +\sum_{m=0}^d\left(1_{m=0}+1_{m\ge
 1}\left(1+\left|\frac{\pi}{\beta}-\frac{\theta(\beta)}{2}\right|^{-\frac{m}{\sn_j}}
\right)
\right)\notag\\
&\qquad \quad\cdot
 \left(1_{|\pi/\beta-\theta(\beta)/2|>1}\left|\frac{\pi}{\beta}-\frac{\theta(\beta)}{2}\right|^{-1}
+1_{|\pi/\beta-\theta(\beta)/2|\le
 1}\left|\frac{\pi}{\beta}-\frac{\theta(\beta)}{2}\right|^{-1-\frac{d+1-m}{\sn_j}}\right)\Bigg)\notag\\
&\le
 c(d,M,\chi,\sc,\sa,(\hbv_j)_{j=1}^d)\beta^{-1}\left|\frac{\pi}{\beta}-\frac{\theta(\beta)}{2}\right|^{-1}\left(
1+\left|\frac{\pi}{\beta}-\frac{\theta(\beta)}{2}\right|^{-\frac{1}{\sn_j}}
\right)^{d+1}.\notag
\end{align}

By summing up \eqref{eq_covariance_determinant_bound_pre} for $n=1$,
 \eqref{eq_covariance_time_decay}, \eqref{eq_covariance_space_decay},
 \eqref{eq_covariance_space_decay_end} we reach that for $l\in
 \{N_{\beta}+1,N_{\beta}+2,\cdots,N_h\}$, $\bx,\by\in\G$, $s,t\in
 [0,\beta)_h$,
\begin{align}
&\|C_l(\cdot\bx s,\cdot \by t)\|_{2b\times
 2b}\label{eq_covariance_full_decay}\\
&\le
 c(d,M,\chi,\sc,\sa,(\hbv_j)_{j=1}^d)(1_{l\ge 0}+1_{l< 0
 }M^{\sa
 l})\notag\\
&\quad\cdot 
\Bigg(1+M^{(d+2)l}\left|\frac{\beta}{2\pi}\left(e^{i\frac{2\pi}{\beta}(s-t)}-1\right)\right|^{d+2}\notag\\
&\qquad\quad+\sum_{j=1}^d\left(1_{l\ge 0}M^{\frac{l}{\sn_j}}+1_{l<
 0}M^{\frac{d+2}{\sn_j}l}\right)\left|\frac{L}{2\pi}(e^{i\frac{2\pi}{L}\<\bx-\by,\hbv_j\>}-1)\right|^{d+2}\Bigg)^{-1},\notag\\
&\|C_{N_{\beta}}(\cdot\bx s,\cdot \by t)\|_{2b\times
 2b}\notag\\
&\le
 c(d,M,\chi,\sc,\sa,(\hbv_j)_{j=1}^d)\beta^{-1}\left|\frac{\pi}{\beta}-\frac{\theta(\beta)}{2}\right|^{-1}\notag\\
&\quad\cdot \Bigg(1+\sum_{j=1}^d\left(1+\left|\frac{\pi}{\beta}-\frac{\theta(\beta)}{2}\right|^{-\frac{1}{\sn_j}}\right)^{-d-1}
\left|\frac{L}{2\pi}(e^{i\frac{2\pi}{L}\<\bx-\by,\hbv_j\>}-1)\right|^{d+1}\Bigg)^{-1}.\notag
\end{align}
These inequalities together with the fact $h\ge M^{N_h-2}$ imply the
 claimed bounds.

\eqref{item_preliminary_covariance_coupled_decay_bound}: 
By \eqref{eq_beta_relation} and \eqref{eq_covariance_full_decay}, for $l\in
 \{N_{\beta}+1,N_{\beta}+2,\cdots,N_h\}$, 
\begin{align*}
&\|\tilde{C}_l\|\\
&\le
 c(d,M,\chi,\sc,\sa,(\hbv_j)_{j=1}^d,b)\\
&\quad\cdot \left(1_{l\ge 0}+1_{l<
 0}M^{(\sa-\sum_{j=1}^d\frac{1}{\sn_j})l}+(1+M^{N_{\beta}})\left(1_{l\ge
 0}M^{-l}+1_{l<
 0}M^{(\sa-1-\sum_{j=1}^d\frac{1}{\sn_j})l}\right)\right)\\
&\le  c(d,M,\chi,\sc,\sa,(\hbv_j)_{j=1}^d,b)\left(1_{l\ge 0}+1_{l<
 0}M^{(\sa-1-\sum_{j=1}^d\frac{1}{\sn_j})l}\right),
\end{align*}
which is \eqref{eq_covariance_coupled_decay_bound_pre}.
Note
 that by \eqref{eq_beta_relation} and \eqref{eq_covariance_full_decay}, for any $l'\in
 \{N_{\beta}+1,N_{\beta}+2,\cdots,N_h\}\cap \Z_{\ge 0}$,
\begin{align*}
\left\|\sum_{p=l'}^{N_h}\tilde{C}_p\right\|
&\le 4b
 \sup_{(\orho,\rho,s,\xi),(\oeta,\eta,t,\zeta)\atop \in \{1,2\}\times
 \cB\times [0,\beta)_h\times
 \{1,-1\}}\left|\sum_{p=l'}^{N_h}\tilde{C}_p(\orho\rho \b0 s\xi,\oeta \eta
 \b0 t\zeta)\right|\\
&\quad + \sum_{p=l'}^{N_h}\sum_{\bx\in\G\atop \bx\neq
 \b0}\frac{c(d,M,\chi,\sc,\sa,(\hbv_j)_{j=1}^d,b)}{\sum_{j=1}^dM^{\frac{p}{\sn_j}}|\frac{L}{2\pi}(e^{i\frac{2\pi}{L}\<\bx,\hbv_j\>}-1)|^{d+2}}\\
&\quad +
 c(d,M,\chi,\sc,\sa,(\hbv_j)_{j=1}^d,b)(1+M^{N_{\beta}})\sum_{p=l'}^{N_h}M^{-p}\\
&\le 4b
 \sup_{(\orho,\rho,s,\xi),(\oeta,\eta,t,\zeta)\atop \in \{1,2\}\times
 \cB\times [0,\beta)_h\times
 \{1,-1\}}\left|\sum_{p=l'}^{N_h}\tilde{C}_p(\orho\rho \b0 s\xi,\oeta \eta
 \b0 t\zeta)\right|\\
&\quad +
 c(d,M,\chi,\sc,(\sn_j)_{j=1}^d,\sa,(\hbv_j)_{j=1}^d,b)(1+(1+M^{N_{\beta}})M^{-l'}),
\end{align*}
which gives \eqref{eq_covariance_coupled_decay_bound_sum_pre}.
\end{proof}

Using the covariances $C_l$ $(l=N_{\beta},N_{\beta}+1,\cdots,N_h)$, we
 define the covariances $\cC_l$
$(l=N_{\beta},N_{\beta}+1,\cdots,\hat{N}_{\beta})$ as follows. 
\begin{align*}
&\cC_{\hat{N}_{\beta}}:=\sum_{l=\hat{N}_{\beta}}^{N_h}C_l,\\
&\cC_l:=C_l,\quad (\forall l\in
 \{N_{\beta},N_{\beta}+1,\cdots,\hat{N}_{\beta}-1\}).
\end{align*}
The claim \eqref{item_covariance_construction_necessary} of the next
lemma states that the covariances $\cC_l$
$(l=N_{\beta},N_{\beta}+1,\cdots,\hat{N}_{\beta})$ satisfy the
conditions assumed in Subsection
\ref{subsec_generalized_covariances}. The claim
\eqref{item_covariance_construction_ODLRO} of the lemma will be used to
prove Corollary \ref{cor_zero_temperature_limit} \eqref{item_ODLRO_zero} in
the next subsection. 

\begin{lemma}\label{lem_covariance_construction}
Assume that \eqref{eq_huge_h_condition}, \eqref{eq_huge_L_condition}
 hold. Then there exists a constant $\hat{c}(\in\R_{\ge 1})$ depending only on
 $d$, $b$, $(\hbv_j)_{j=1}^d$, $\sa$, $(\sn_j)_{j=1}^d$, $\sc$, $M$,
 $\chi$ such that the following statements hold true. 
\begin{enumerate}[(i)]
\item\label{item_covariance_construction_necessary}
$\cC_l$ $(l=N_{\beta}+1,N_{\beta}+2,\cdots,\hat{N}_{\beta})$ satisfy
     \eqref{eq_general_covariance_time_translation} and
     $\cC_{N_{\beta}}$ satisfies
     \eqref{eq_final_covariance_time_independence}. Moreover
$\cC_l$ $(l=N_{\beta},N_{\beta}+1,\cdots,\hat{N}_{\beta})$ satisfy
\eqref{eq_scale_covariance_determinant_bound}, 
\eqref{eq_scale_covariance_decay_bound},
\eqref{eq_scale_covariance_coupled_decay_bound} with 
\begin{align*}
&c_0=\hat{c}\left(\beta^{-1}\min\left\{1,\left|\frac{\pi}{\beta}-\frac{\theta(\beta)}{2}\right|^{-1}\right\}+1\right),\\
&c_{end}=\left|\frac{\pi}{\beta}-\frac{\theta(\beta)}{2}\right|^{-1}\prod_{j=1}^d\left(1+\left|\frac{\pi}{\beta}-\frac{\theta(\beta)}{2}\right|^{-\frac{1}{\sn_j}}\right).
\end{align*}
\item\label{item_covariance_construction_ODLRO}
\begin{align*}
&\sum_{j=1}^d\frac{L}{2\pi}|e^{i\frac{2\pi}{L}\<\bx-\by,\hbv_j\>}-1|\left\|
\sum_{l=N_{\beta}+1}^{\hat{N}_{\beta}}\cC_l(\cdot\bx s,\cdot\by
 t)\right\|_{2b\times 2b}\le \hat{c},\\
&(\forall \bx,\by\in \G,\
 s,t\in [0,\beta)_h).
\end{align*}
\end{enumerate}
\end{lemma}

\begin{proof} 
\eqref{item_covariance_construction_necessary}: The claim concerning
 \eqref{eq_general_covariance_time_translation} is clear. Lemma
 \ref{lem_properties_cutoff} \eqref{item_final_cutoff} ensures that
$\cC_{N_{\beta}}$ satisfies
 \eqref{eq_final_covariance_time_independence}. 
Let us derive the
 determinant bound on $\cC_{\hat{N}_{\beta}}$. Let us improve the second
 inequality in Lemma \ref{lem_characterization_covariance}
 \eqref{item_P_S_determinant_bound}. 
For any $x\in \R_{\ge 0}$, 
\begin{align}
&\left(1+2\cos\left(\frac{\beta\theta(\beta)}{2}\right) e^{-\beta
 x}+e^{-2\beta
 x}\right)^{-\frac{1}{2}}\label{eq_determinant_bound_improvement}\\
&\le c \left(1_{\beta x>1}+1_{\beta x\le
 1}\beta^{-1}\left(\left(\frac{\pi}{\beta}-\frac{\theta(\beta)}{2}\right)^2+x^2\right)^{-\frac{1}{2}}\right)\notag\\
&\le c  \left(1 + \beta^{-1}\left(\left(\frac{\pi}{\beta}-\frac{\theta(\beta)}{2}\right)^2+x^2\right)^{-\frac{1}{2}}\right).\notag
\end{align}
Thus by \eqref{eq_dispersion_upper_lower_bound},
 \eqref{eq_dispersion_measure_divided}, \eqref{eq_huge_L_condition} and
 Lemma \ref{lem_discrete_continuous_estimate}, 
\begin{align*}
&\frac{1}{L^d}\sum_{\bk\in\G^*}\Tr \left(
1+2\cos\left(\frac{\beta\theta(\beta)}{2}\right)e^{-\beta\sqrt{E(\bk)^2+|\phi|^2}}+e^{-2\beta\sqrt{E(\bk)^2+|\phi|^2}}\right)^{-\frac{1}{2}}\\
&\le \frac{c(b) \beta^{-1}}{L^d}\sum_{\bk\in
 \G^*}\left(\left(\frac{\pi}{\beta}-\frac{\theta(\beta)}{2}\right)^2+e(\bk)^2\right)^{-\frac{1}{2}}+c(b)\notag\\
&\le
 c(D_d,b)\beta^{-1}\int_{\G_{\infty}^*}\left(\left(\frac{\pi}{\beta}-\frac{\theta(\beta)}{2}\right)^2+e(\bk)^2\right)^{-\frac{1}{2}}\notag\\
&\quad +c(d,(\hbv_j)_{j=1}^d,\sc,b)\left|
\frac{\pi}{\beta}-\frac{\theta(\beta)}{2}\right|^{-3}L^{-1}\beta^{-1}+c(b)\notag\\
&\le
 c(d,(\hbv_j)_{j=1}^d,\sc,b)\left(\beta^{-1}\min\left\{1,\left|\frac{\pi}{\beta}-\frac{\theta(\beta)}{2}\right|^{-1}\right\}+1\right).\notag
\end{align*}
Then by this inequality, Lemma \ref{lem_characterization_covariance}
 \eqref{item_P_S_determinant_bound} and \eqref{eq_covariance_sum_unity},
\begin{align}
&\left|\det\left(\<\bu_i,\bv_j\>_{\C^m}\sum_{l=N_{\beta}}^{N_h}C_l(X_i,Y_j)
\right)_{1\le i,j\le n}\right|\label{eq_P_S_bound_improvement}\\
&\le
 \left(c(d,(\hbv_j)_{j=1}^d,\sc,b)\left(\beta^{-1}\min\left\{1,\left|\frac{\pi}{\beta}-\frac{\theta(\beta)}{2}\right|^{-1}\right\}+1\right)\right)^n,\notag\\
&(\forall m,n\in \N,\ \bu_i,\bv_i\in\C^m\ (i=1,2,\cdots,n)\text{ with
 }\|\bu_i\|_{\C^m},\ \|\bv_i\|_{\C^m}\le 1,\notag\\
&\quad  X_i,Y_i\in I_0\
 (i=1,2,\cdots,n)).\notag
\end{align}
We can apply \cite[\mbox{Lemma A.1}]{K_BCS}, which is based on the
 Cauchy-Binet formula, together with
 \eqref{eq_covariance_determinant_bound_lower_sum},
 \eqref{eq_P_S_bound_improvement}, the equality
 $\cC_{\hat{N}_{\beta}}=\sum_{l=N_{\beta}}^{N_h}C_l-\sum_{l=N_{\beta}}^{\hat{N}_{\beta}-1}C_l$
 and the inequality $M^{\hat{N}_{\beta}}\le M\beta^{-1}+1$ to derive the
 claim determinant bound on $\cC_{\hat{N}_{\beta}}$. If 
$N_{\beta}+1\le \hat{N}_{\beta}-1$, then
 $\hat{N}_{\beta}=0$. Thus the claimed determinant bound on $\cC_l$
 $(l=N_{\beta}+1,N_{\beta}+2,\cdots,\hat{N}_{\beta}-1)$ follows from
 \eqref{eq_covariance_determinant_bound_pre}. If $\beta >1$,
 $N_{\beta}<0=\hat{N}_{\beta}$. Thus by \eqref{eq_beta_relation}
\begin{align*}
\beta^{-1}\min\left\{M^{(\sa-1)N_{\beta}},1,\left|\frac{\pi}{\beta}-\frac{\theta(\beta)}{2}\right|^{-1}\right\}\le
 \beta^{-1}M^{(\sa-1)N_{\beta}}\le M^{1+\sa(N_{\beta}-\hat{N}_{\beta})}.
\end{align*}
If $\beta\le 1$,
\begin{align*}
\beta^{-1}\min\left\{M^{(\sa-1)N_{\beta}},1,\left|\frac{\pi}{\beta}-\frac{\theta(\beta)}{2}\right|^{-1}\right\}\le
 \beta^{-1}\min\left\{1,\left|\frac{\pi}{\beta}-\frac{\theta(\beta)}{2}\right|^{-1}\right\}M^{\sa+\sa(N_{\beta}-\hat{N}_{\beta})}.
\end{align*}
Thus the claimed determinant bound on $\cC_{N_{\beta}}$ follows from
 \eqref{eq_covariance_determinant_bound_pre}.

Note that $\hat{N}_{\beta}\ge 0$ $(\forall \beta>0)$. Thus by
 \eqref{eq_covariance_decay_bound_pre},
\begin{align*}
&\|\tilde{\cC}_{\hat{N}_{\beta}}\|_{1,\infty}\le
 \sum_{l=\hat{N}_{\beta}}^{N_h}\|\tilde{C}_l\|_{1,\infty}\le
 \hat{c}\sum_{l=0}^{\infty}M^{-l}\le 2\hat{c}.
\end{align*}
The inclusion $l\in
 \{N_{\beta}+1,N_{\beta}+2,\cdots,\hat{N}_{\beta}-1\}$ implies that
 $\hat{N}_{\beta}=0$ and $l\le -1$. Thus by
 \eqref{eq_covariance_decay_bound_pre} the claimed bound on
 $\|\tilde{\cC}_l\|_{1,\infty}$ holds true for $l\in \{N_{\beta}+1,N_{\beta}+2,\cdots,\hat{N}_{\beta}\}$. The claimed bound on
 $\|\tilde{\cC}_{N_{\beta}}\|_{1,\infty}$ follows from
 \eqref{eq_covariance_decay_bound_end_pre}. For the same reason as above
 \eqref{eq_covariance_coupled_decay_bound_pre} gives the claimed bound
 on $\|\tilde{\cC}_l\|$ for $l\in
 \{N_{\beta}+1,N_{\beta}+2,\cdots,\hat{N}_{\beta}-1\}$. Since
 $\hat{N}_{\beta}\in \{N_{\beta}+1,\cdots,N_h\}\cap\Z_{\ge 0}$, we can
 derive the claimed bound on $\|\tilde{\cC}_{\hat{N}_{\beta}}\|$ by
 combining \eqref{eq_covariance_coupled_decay_bound_sum_pre} with the
 determinant bound on $\cC_{\hat{N}_{\beta}}$ for $n=1$.

\eqref{item_covariance_construction_ODLRO}: By \eqref{eq_covariance_space_decay} for $n=1$ and the inequality
 $\sa>1/\sn_j$ $(j=1,2,\cdots,d)$ implied by the original assumptions
 $\sa>1$, $\sn_j\ge 1$ $(j=1,2,\cdots,d)$,
\begin{align*}
&\sum_{j=1}^d\frac{L}{2\pi}|e^{i\frac{2\pi}{L}\<\bx-\by,\hbv_j\>}-1|\left\|
\sum_{l=N_{\beta}+1}^{\hat{N}_{\beta}}\cC_l(\cdot\bx s,\cdot \by
 t)\right\|_{2b\times 2b}\\
&\le
 c(d,M,\chi,\sc,\sa,(\hbv_j)_{j=1}^d)\sum_{j=1}^d\left(\sum_{l=0}^{\infty}M^{-\frac{l}{\sn_j}}+\sum_{l=-1}^{-\infty}M^{(\sa-\frac{1}{\sn_j})l}\right)\\
&\le  c(d,M,\chi,\sc,\sa,(\hbv_j)_{j=1}^d,(\sn_j)_{j=1}^d),
\end{align*}
which is the claimed bound.
\end{proof}

\begin{remark}\label{rem_strictly_larger_one}
 We use the condition $\sa>1$ to prove Lemma
 \ref{lem_covariance_construction}
 \eqref{item_covariance_construction_ODLRO}, which will be used only to prove
 Corollary \ref{cor_zero_temperature_limit} \eqref{item_ODLRO_zero}. 
\end{remark}

Since we have confirmed that the real covariance derived from the free
Hamiltonian can be decomposed into a family of covariances
satisfying the desired properties, we can apply the general result Lemma
\ref{lem_final_integration} to analyze the Grassmann Gaussian integral
appearing in the formulation Lemma \ref{lem_final_Grassmann_formulation}.

\begin{proposition}\label{prop_general_integration_application}
Let $c_6$ be the positive constant appearing in Lemma
 \ref{lem_final_integration}. Let $c_0$, $c_{end}$ be those set in Lemma
 \ref{lem_covariance_construction} \eqref{item_covariance_construction_necessary}.
Fix $M$, $\alpha\in\R_{\ge 1}$ satisfying 
\begin{align*}
M^{\min\{1,2\sa-1-\sum_{j=1}^d\frac{1}{\sn_j}\}}\ge c_6,\quad \alpha\ge
 c_6 M^{\frac{\sa}{2}}.
\end{align*}
Then the following statements hold for any $h\in
\frac{2}{\beta}\N$, $L\in \N$ satisfying \eqref{eq_huge_h_condition} and 
\begin{align}
&L\ge
 \max\Bigg\{(c_{end}+1)^{\frac{1}{d}}M^{\frac{1}{d}(\sa+\sum_{j=1}^d\frac{1}{\sn_j}+1)(\hat{N}_{\beta}-N_{\beta})},\frac{\max_{j\in\{1,2,\cdots,d\}}M^{-N_{\beta}/\sn_j}}{\min\{M^{\sa
 N_{\beta}},1\}},\label{eq_final_huge_L_condition}\\
&\qquad\qquad\quad 
\frac{\max_{j\in \{1,2,\cdots,d\}}
 M^{-N_{\beta}/\sn_j}|\frac{\pi}{\beta}-\frac{\theta(\beta)}{2}|^{-1}+
|\frac{\pi}{\beta}-\frac{\theta(\beta)}{2}|^{-3}}{\min\{M^{(\sa-1)N_{\beta}},1,|\frac{\pi}{\beta}-\frac{\theta(\beta)}{2}|^{-1}\}}\Bigg\}.\notag
\end{align}
\begin{enumerate}[(i)]
\item\label{item_integral_formulation_lower_upper_bound}
\begin{align*}
&e^{-8c_6b\beta
 \alpha^{-2}M^{(\sum_{j=1}^d\frac{1}{\sn_j}+1)(\hat{N}_{\beta}-N_{\beta})}}\le \left|\int e^{-V(u)(\psi)+W(u)(\psi)}d\mu_{C(\phi)}(\psi)\right|\\
&\qquad\qquad\qquad\qquad\qquad\quad\ \ \le e^{8c_6b\beta
 \alpha^{-2}M^{(\sum_{j=1}^d\frac{1}{\sn_j}+1)(\hat{N}_{\beta}-N_{\beta})}},\\
&(\forall u\in \overline{D(b^{-1}c_0^{-2}\alpha^{-4})}).
\end{align*}
\item\label{item_integral_formulation_difference}
\begin{align*}
&\sup_{u\in\overline{D(b^{-1}c_0^{-2}\alpha^{-4})}}\left|
\frac{\int e^{-V(u)(\psi)+W(u)(\psi)}A^j(\psi)d\mu_{C(\phi)}(\psi)}{\int
 e^{-V(u)(\psi)+W(u)(\psi)}d\mu_{C(\phi)}(\psi)}-\int
 A^j(\psi)d\mu_{C(\phi)}(\psi)\right|\\
&\le c_6 \eps_{\beta}^{N_{\beta}-\hat{N}_{\beta}}\beta c_0^2\left(
\alpha^2+c_{end}M^{\sa(\hat{N}_{\beta}-N_{\beta})}
\right)L^{-d},\quad (\forall j\in \{1,2\}).
\end{align*}
\end{enumerate}
\end{proposition}

\begin{proof}
Observe that under the condition of this proposition the claims of Lemma
 \ref{lem_final_integration} and Lemma \ref{lem_covariance_construction}
 hold true. By the definition and \eqref{eq_covariance_sum_unity}, 
\begin{align}
\sum_{l=N_{\beta}}^{\hat{N}_{\beta}}\cC_l(\orho\rho\bx s,\oeta\eta \by
 t)=e^{-i\frac{\pi}{\beta}(s-t)}C(\phi)(\orho\rho\bx s,\oeta\eta \by
 t),\quad (\forall (\orho,\rho,\bx, s),(\oeta,\eta, \by, t)\in
 I_0).\label{eq_covariance_full_summation}
\end{align}
Thus by the gauge transform $\psi_{\orho\rho\bx s \xi}\to e^{-\xi
 i\frac{\pi}{\beta}s}\psi_{\orho\rho\bx s \xi}$, 
\begin{align*}
\int e^{-V(u)(\psi)+W(u)(\psi)-A(\psi)}d\mu_{C(\phi)}(\psi)
=\int
 e^{-V(u)(\psi)+W(u)(\psi)-A(\psi)}d\mu_{\sum_{l=N_{\beta}}^{\hat{N}_{\beta}}\cC_l}(\psi).
\end{align*}
Thus the function $V^{end}$ studied in Lemma \ref{lem_final_integration}
 coincides with the function 
\begin{align*}
(u,\bla)\mapsto \log\left(\int
 e^{-V(u)(\psi)+W(u)(\psi)-A(\psi)}d\mu_{C(\phi)}(\psi)\right)
\end{align*}
if $|u|$, $\|\bla\|_{\C^2}$ are sufficiently small. 
This claim can be confirmed by an elementary argument close to a part
 of the proof of \cite[\mbox{Lemma 4.13}]{K_BCS} or the proof of
 \cite[\mbox{Proposition 6.4 (3)}]{K_RG}. However, we provide a sketch
 for the readers' convenience. With the constant $c_6$ appearing in
 Lemma \ref{lem_final_integration}, set $r:=b^{-1}c_0^{-2}\alpha^{-4}$,
 $r'':=c_6^{-1}L^{-d}h^{N_{\beta}-\hat{N}_{\beta}-1}\eps_{\beta}^{\hat{N}_{\beta}-N_{\beta}}\beta^{-1}c_0^{-2}\alpha^{-4}$. For $l\in \{N_{\beta},N_{\beta}+1,\cdots,\hat{N}_{\beta}\}$ we define $V^l\in C(\overline{D(r)}\times \overline{D(r'')}^2,\bigwedge_{even}\cV)\cap C^{\o}(D(r)\times D(r'')^2,\bigwedge_{even}\cV)$ by $V^l:=\sum_{j=1}^2V^{0-j,l}+\sum_{j=1}^3V^{1-j,l}+V^{2,l}$. By Lemma \ref{lem_IR_integration_without} and Lemma \ref{lem_IR_integration_with}, 
\begin{align*}
&\sum_{m=0}^N\|V_m^{l}\|_{1,r,r''}\le
 c(\beta,L^d,h,(\sn_j)_{j=1}^d,b,M,\sa)\alpha^{-2}, \quad(\forall l\in \{N_{\beta},N_{\beta}+1,\cdots,\hat{N}_{\beta}\}).
\end{align*}
This together with \eqref{eq_scale_covariance_determinant_bound} implies
 that there exists a positive constant\\
 $c'(\beta,L^d,h,(\sn_j)_{j=1}^d,b,M,\sa,c_0)$ such that if $\alpha
 \ge c'(\beta,L^d,h,(\sn_j)_{j=1}^d,b,M,\sa,c_0)$, 
\begin{align}
&\Re \int e^{z
 V^l(u,\bla)(\psi)}d\mu_{\cC_l}(\psi)>0,\label{eq_each_integral_crude_positivity}\\
&(\forall z\in \overline{D(2)},\ (u,\bla)\in \overline{D(r)}\times
 \overline{D(r'')}^2,\ l\in
 \{N_{\beta},N_{\beta}+1,\cdots,\hat{N}_{\beta}\}).\notag
\end{align}
Let us fix such a large $\alpha$. 
Then for any $(u,\bla)\in\overline{D(r)}\times\overline{D(r'')}^2$,
 $l\in \{N_{\beta},N_{\beta}+1,\cdots,\hat{N}_{\beta}\}$
\begin{align*}
&z\mapsto \log\left(\int e^{z
 V^{l}(u,\bla)(\psi+\psi^1)}d\mu_{\cC_l}(\psi^1)\right),\quad z\mapsto \log\left(\int e^{z V^{N_{\beta}}(u,\bla)(\psi)}d\mu_{\cC_{N_{\beta}}}(\psi)
\right)
\end{align*}
are analytic in $D(2)$. Then by definition, 
\begin{align}
&V^l(u,\bla)(\psi)=\sum_{n=1}^{\infty}\frac{1}{n!}\left(\frac{d}{dz}\right)^n\Big|_{z=0}\log\left(\int
 e^{z
 V^{l+1}(u,\bla)(\psi+\psi^1)}d\mu_{\cC_{l+1}}(\psi^1)\right)\label{eq_each_integral_connection}\\
&\qquad\qquad\quad =\log\left(\int
 e^{
 V^{l+1}(u,\bla)(\psi+\psi^1)}d\mu_{\cC_{l+1}}(\psi^1)\right),\notag\\
&V^{end}(u,\bla)=\sum_{n=1}^{\infty}\frac{1}{n!}\left(\frac{d}{dz}\right)^n\Big|_{z=0}\log\left(\int
 e^{z
 V^{N_{\beta}}(u,\bla)(\psi)}d\mu_{\cC_{N_{\beta}}}(\psi)\right)\label{eq_last_integral_connection}\\
&\qquad\qquad\ \  =\log\left(\int
 e^{
 V^{N_{\beta}}(u,\bla)(\psi)}d\mu_{\cC_{N_{\beta}}}(\psi)\right),\notag\\
&(\forall (u,\bla)\in\overline{D(r)}\times\overline{D(r'')}^2,\ l\in
 \{N_{\beta},N_{\beta}+1,\cdots,\hat{N}_{\beta}-1\}).\notag
\end{align}
By \eqref{eq_each_integral_crude_positivity} for $l=\hat{N}_{\beta}$, 
\begin{align*}
&\int
 e^{V^{\hat{N}_{\beta}-1}(u,\bla)(\psi+\psi^1)}d\mu_{\cC_{\hat{N}_{\beta}-1}}(\psi^1)
=\int \int
 e^{V^{\hat{N}_{\beta}}(u,\bla)(\psi+\psi^1+\psi^2)}d\mu_{\cC_{\hat{N}_{\beta}}}(\psi^2)d\mu_{\cC_{\hat{N}_{\beta}-1}}(\psi^1)\\
&\qquad\qquad\qquad\qquad\qquad\qquad\quad\ \ =\int
 e^{V^{\hat{N}_{\beta}}(u,\bla)(\psi+\psi^1)}d\mu_{\cC_{\hat{N}_{\beta}-1}+\cC_{\hat{N}_{\beta}}}(\psi^1),\\
& (\forall (u,\bla)\in \overline{D(r)}\times
 \overline{D(r'')}^2).
\end{align*}
Here we used the fact that $e^{\log f(\psi)}=f(\psi)$ $(\forall f\in\bigwedge\cV$ with $\Re
 f_0>0)$ (see e.g. \cite[\mbox{Lemma C.2}]{K_14}) and the division
 formula of Grassmann Gaussian integral (see
 e.g. \cite[\mbox{Proposition I.21}]{FKT}). Assume that $l\in
 \{N_{\beta}+2,N_{\beta}+3,\cdots,\hat{N}_{\beta}\}$ and 
\begin{align}
&\int
 e^{V^{l-1}(u,\bla)(\psi+\psi^1)}d\mu_{\cC_{l-1}}(\psi^1)
=\int
 e^{V^{\hat{N}_{\beta}}(u,\bla)(\psi+\psi^1)}d\mu_{\sum_{j=l-1}^{\hat{N}_{\beta}}\cC_{j}}(\psi^1),\label{eq_sum_integral_induction_hypothesis}\\
&(\forall (u,\bla)\in \overline{D(r)}\times
 \overline{D(r'')}^2).\notag
\end{align}
Then by using \eqref{eq_each_integral_connection},
 \eqref{eq_each_integral_crude_positivity},
 \eqref{eq_sum_integral_induction_hypothesis} in this order,
\begin{align*}
&\int
 e^{V^{l-2}(u,\bla)(\psi+\psi^2)}d\mu_{\cC_{l-2}}(\psi^2)=
\int\int
 e^{V^{l-1}(u,\bla)(\psi+\psi^1+\psi^2)}d\mu_{\cC_{l-1}}(\psi^1)d\mu_{\cC_{l-2}}(\psi^2)\\
&\qquad\qquad\qquad\qquad\qquad\qquad\ = \int
 e^{V^{\hat{N}_{\beta}}(u,\bla)(\psi+\psi^1)}d\mu_{\sum_{j=l-2}^{\hat{N}_{\beta}}\cC_{j}}(\psi^1),\\
&(\forall (u,\bla)\in \overline{D(r)}\times
 \overline{D(r'')}^2).
\end{align*}
Thus by induction with $l$ we have that 
\begin{align*}
&\int
 e^{V^{N_{\beta}}(u,\bla)(\psi+\psi^1)}d\mu_{\cC_{N_{\beta}}}(\psi^1)=
\int
 e^{V^{\hat{N}_{\beta}}(u,\bla)(\psi+\psi^1)}d\mu_{\sum_{j=N_{\beta}}^{\hat{N}_{\beta}}\cC_{j}}(\psi^1),\\
&(\forall (u,\bla)\in \overline{D(r)}\times
 \overline{D(r'')}^2).
\end{align*}
By combining this equality with \eqref{eq_last_integral_connection} we
 obtain that 
\begin{align*}
&V^{end}(u,\bla)=\log\left(\int
 e^{V^{\hat{N}_{\beta}}(u,\bla)(\psi^1)}d\mu_{\sum_{j=N_{\beta}}^{\hat{N}_{\beta}}\cC_j}(\psi^1)\right),\\
&(\forall (u,\bla)\in \overline{D(r)}\times
 \overline{D(r'')}^2),
\end{align*}
which implies the claim.

By the identity theorem and continuity,
\begin{align}
&\int
 e^{-V(u)(\psi)+W(u)(\psi)-A(\psi)}d\mu_{C(\phi)}(\psi)=e^{V^{end}(u,\bla)},\label{eq_analytic_continued_equality}\\
&\left(\forall (u,\bla)\in \overline{D(b^{-1}c_0^{-2}\alpha^{-4})}\times
 \overline{D(c_6^{-1}L^{-d}h^{N_{\beta}-\hat{N}_{\beta}-1}\eps_{\beta}^{\hat{N}_{\beta}-N_{\beta}}\beta^{-1}c_0^{-2}\alpha^{-4})}^2\right).\notag
\end{align}
By Lemma \ref{lem_final_integration}
 \eqref{item_final_bound_general},
\begin{align}
\sup_{u\in\overline{D(b^{-1}c_0^{-2}\alpha^{-4})}}|V^{end}(u,\b0)|\le
 8c_6 b\beta
 \alpha^{-2}M^{(\sum_{j=1}^d\frac{1}{\sn_j}+1)(\hat{N}_{\beta}-N_{\beta})}.\label{eq_final_bound_application}
\end{align}
Then the claim \eqref{item_integral_formulation_lower_upper_bound} follows
 from \eqref{eq_analytic_continued_equality} and
 \eqref{eq_final_bound_application}. By the same gauge transform as
 above and the division formula of Grassmann Gaussian integral,
\begin{align*}
-\int A(\psi)d\mu_{C(\phi)}(\psi)= -\int
 A(\psi)d\mu_{\sum_{l=N_{\beta}}^{\hat{N}_{\beta}}\cC_l}(\psi)=V^{1-3,end}.
\end{align*}
Then Lemma \ref{lem_final_integration} \eqref{item_final_difference_general} and
 \eqref{eq_analytic_continued_equality} ensure the claim \eqref{item_integral_formulation_difference}.
\end{proof}

\subsection{Proof of the main results}\label{subsec_proof_theorem}

Let us move toward the proofs of Theorem \ref{thm_main_theorem} and
Corollary \ref{cor_zero_temperature_limit}.
In order to prove the existence of the limit $L\to \infty$ claimed in
Theorem \ref{thm_main_theorem}, we use the following proposition.

\begin{proposition}\label{prop_infinite_volume_limit}
Assume that $M$, $\alpha$ satisfy the same conditions as in Proposition
 \ref{prop_general_integration_application} and $L$ $(\in \N)$ satisfies
 \eqref{eq_final_huge_L_condition}. Then for any non-empty compact set
 $Q$ of $\C$
\begin{align*}
\lim_{h\to \infty\atop h\in \frac{2}{\beta}\N}\int e^{-V(u)(\psi)+W(u)(\psi)}
d\mu_{C(\phi)}(\psi),\quad
\lim_{L\to \infty\atop L\in \N}
\lim_{h\to \infty\atop h\in \frac{2}{\beta}\N}\int e^{-V(u)(\psi)+W(u)(\psi)}
d\mu_{C(\phi)}(\psi)
\end{align*}
converge in $C(Q\times \overline{D(2^{-1}b^{-1}c_0^{-2}\alpha^{-4})})$
 as sequences of function of the variable $(\phi,u)$.
\end{proposition}

\begin{proof}
In essence the claim can be proved in the same way as the proof of
 \cite[\mbox{Proposition 4.16}]{K_BCS}. Our model depends on the band
 index $\rho\in \cB$, while the model in \cite{K_BCS} does not. However,
 the band index makes no essential difference to prove the claim. We
 have to comment on one notable difference from the situation of
 \cite[\mbox{Proposition 4.16}]{K_BCS}. Here we do not know whether the
 function 
\begin{align*}
u\mapsto \log\left(\int e^{-V(u)(\psi)+W(u)(\psi)}
d\mu_{C(\phi)}(\psi)\right)
\end{align*}
is analytic in $D(b^{-1}c_0^{-2}\alpha^{-4})$, while we knew the
 analyticity of the function in the domain of interest
in \cite[\mbox{Proposition 4.16}]{K_BCS}. 
In the following we outline the proof of the proposition without using
 the above-mentioned analyticity.

For $\phi\in\C$ we define the function $\cG(\phi):(\{1,2\}\times
 \cB\times \G_{\infty}\times [0,\beta))^2\to\C$ by 
\begin{align*}
\cG(\phi)(\orho\rho\bx s,\oeta \eta \by
 t):=e^{-i\frac{\pi}{\beta}(s-t)}C(\phi)(\orho\rho\bx s,\oeta \eta \by
 t).
\end{align*}
It follows from the gauge transform that
\begin{align*}
\int e^{-V(u)(\psi)+W(u)(\psi)}
d\mu_{\cG(\phi)}(\psi)=\int e^{-V(u)(\psi)+W(u)(\psi)}
d\mu_{C(\phi)}(\psi),\quad (\forall u,\phi\in \C).
\end{align*}
For $n\in \N$, $\phi\in Q$, set 
\begin{align*}
\alpha_{n,L,h}(\phi):=\frac{1}{n!}\left(\frac{d}{du}\right)^n\log\left(
\int e^{-V(u)(\psi)+W(u)(\psi)}
d\mu_{\cG(\phi)}(\psi)\right)\Bigg|_{u=0}.
\end{align*}
We consider $\alpha_{n,L,h}$ as a function of $\phi$ on $Q$.
The major part of the proof of \cite[\mbox{Proposition 4.16}]{K_BCS} was
 devoted to proving the uniform convergence of the function denoted by
 the same symbol `$\alpha_{n,L,h}$'. Despite the presence of the band
 index, the essentially same argument as the corresponding part of the
 proof of \cite[\mbox{Proposition 4.16}]{K_BCS} shows that for any $n\in
 \N$
\begin{align*}
\lim_{h\to \infty\atop h\in\frac{2}{\beta}\N}\alpha_{n,L,h},\quad 
 \lim_{L\to \infty\atop L\in \N}\lim_{h\to \infty\atop
 h\in\frac{2}{\beta}\N}\alpha_{n,L,h}
\end{align*}
converge in $C(Q)$. 
In place of \cite[\mbox{(4.50)}]{K_BCS}, here we need to use a spatial decay property of
 $\cG(\phi)$ such as
\begin{align}
&\|\cG(\phi)(\cdot\bx s, \cdot \by t)\|_{2b\times 2b}\le
 \frac{c(d,\beta,\theta,b,E,Q)}{1+\sum_{j=1}^d|\frac{L}{2\pi}(e^{i\frac{2\pi}{L}\<\bx-\by,\hbv_j\>}-1)|^{d+1}},\label{eq_full_covariance_total_decay}\\
&(\forall \bx,\by\in\G_{\infty},\ s,t\in [0,\beta),\ \phi\in Q)\notag
\end{align}
with a positive constant $c(d,\beta,\theta,b,E,Q)$
 depending only on $d$, $\beta$, $\theta$, $b$, $E$,
  $Q$. The
 above inequality can be directly derived from
 \eqref{eq_characterization_covariance}. 
Let $V^{end}(\phi)$ be the function studied in Lemma
 \ref{lem_final_integration}. Here we indicate that it depends on the
 variable $\phi$. Set $r:=b^{-1}c_0^{-2}\alpha^{-4}$. 
Take any $h\in \frac{2}{\beta}\N$ satisfying 
\begin{align}
h\ge \max\left\{2,\sc,\sup_{\phi\in
 Q}\sup_{\bk\in\R^d}\|E(\phi)(\bk)\|_{2b\times 2b}\right\}.\label{eq_final_huge_h_condition}
\end{align}
Since 
\begin{align*}
V^{end}(\phi)(u,\b0)=\log\left(\int e^{-V(u)(\psi)+W(u)(\psi)}
d\mu_{\cG(\phi)}(\psi)\right)
\end{align*}
for small $u$ (see the proof of Proposition
 \ref{prop_general_integration_application} for this claim)
and $u\mapsto V^{end}(\phi)(u,\b0)$ is analytic in $D(r)$, 
\begin{align}
&V^{end}(\phi)(u,\b0)=\sum_{n=1}^{\infty}\alpha_{n,L,h}(\phi)u^n,\quad
 (\forall u\in
 \overline{D(r/2)}),\label{eq_final_term_taylor_expansion}\\
&\alpha_{n,L,h}(\phi)=\frac{1}{2\pi i}\oint_{|z|=(1+\eps)2^{-1}r}dz
 \frac{V^{end}(\phi)(z,\b0)}{z^{n+1}},\quad (\forall\eps \in
 (0,1)).\notag
\end{align}
Here we should remark that the  condition \eqref{eq_huge_h_condition}
 was replaced by \eqref{eq_final_huge_h_condition}. By
 \eqref{eq_final_bound_application},
\begin{align}
\sup_{\phi\in Q}|\alpha_{n,L,h}(\phi)|\le \frac{8 c_6 b \beta
 \alpha^{-2}
 M^{(\sum_{j=1}^d\frac{1}{\sn_j}+1)(\hat{N}_{\beta}-N_{\beta})}}{(1+\eps)^n2^{-n}r^n}\label{eq_taylor_coefficient_uniform_bound}
\end{align}
for $h\in \frac{2}{\beta}\N$ satisfying
 \eqref{eq_final_huge_h_condition} and $L\in \N$ satisfying
 \eqref{eq_final_huge_L_condition}. By the convergent properties of
 $\alpha_{n,L,h}$, \eqref{eq_final_term_taylor_expansion},
 \eqref{eq_taylor_coefficient_uniform_bound} and the dominated
 convergence theorem in $l^1(\N,C(Q\times \overline{D(r/2)}))$
we conclude that
\begin{align*}
\lim_{h\to \infty\atop h\in
 \frac{2}{\beta}\N}V^{end}(\cdot)(\cdot,\b0),\quad 
\lim_{L\to \infty\atop L\in \N}
\lim_{h\to \infty\atop h\in
 \frac{2}{\beta}\N}V^{end}(\cdot)(\cdot,\b0)
\end{align*}
converge in $C(Q\times \overline{D(\frac{r}{2})})$.
Then the equality \eqref{eq_analytic_continued_equality} implies the
 claim. 
\end{proof}

Here let us characterize the covariance with zero time-variable so that
it can be used to compute the thermal expectations of our interest.

\begin{lemma}\label{lem_covariance_final_characterization}
For any $\bx,\by\in\G$, $\rho\in\cB$, $\orho,\oeta\in \{1,2\}$ with
 $\orho\neq \oeta$, 
\begin{align}
&C(\phi)(\orho\cdot\bx 0,\orho\cdot \by
 0)\label{eq_covariance_final_characterization}\\
&=\frac{1}{2L^d}\sum_{\bk\in \G^*}e^{i\<\bx-\by,\bk\>}
\Bigg(\frac{e^{-i\frac{\beta\theta(\beta)}{2}}+\cosh(\beta\sqrt{E(\bk)^2+|\phi|^2})}{\cos(\beta\theta(\beta)/2)+\cosh(\beta\sqrt{E(\bk)^2+|\phi|^2})}\notag\\
&\qquad\qquad\qquad\qquad\quad  +
\frac{(-1)^{\orho}\sinh(\beta\sqrt{E(\bk)^2+|\phi|^2})E(\bk)}{(\cos(\beta\theta(\beta)/2)+\cosh(\beta\sqrt{E(\bk)^2+|\phi|^2}))\sqrt{E(\bk)^2+|\phi|^2}}\Bigg),\notag\\
&C(\phi)(\orho\cdot\bx 0,\oeta\cdot \by
 0)\notag\\
&=\frac{-1}{2L^d}(1_{(\orho,\oeta)=(1,2)}\overline{\phi}+1_{(\orho,\oeta)=(2,1)}\phi)\notag\\
&\quad\cdot \sum_{\bk\in \G^*}e^{i\<\bx-\by,\bk\>}\frac{\sinh(\beta\sqrt{E(\bk)^2+|\phi|^2})}{(\cos(\beta\theta(\beta)/2)+\cosh(\beta\sqrt{E(\bk)^2+|\phi|^2}))\sqrt{E(\bk)^2+|\phi|^2}},\notag\\
&|C(\phi)(\orho\rho \b0 0, \oeta \rho \b0
 0)|\label{eq_covariance_final_lower_bound}\\
&\ge |\phi|\sinh\left(\beta\sqrt{\sup_{\bk\in\R^d}\|E(\bk)\|_{b\times
 b}^2+|\phi|^2}\right)\notag\\
&\quad\cdot\left(2\left(\cos\left(\frac{\beta\theta(\beta)}{2}\right)+\cosh\left(\beta\sqrt{\sup_{\bk\in\R^d}\|E(\bk)\|_{b\times
 b}^2+|\phi|^2}\right)\right)\right)^{-1}\notag\\
&\quad\cdot \left(\sqrt{\sup_{\bk\in\R^d}\|E(\bk)\|_{b\times
 b}^2+|\phi|^2}\right)^{-1}.\notag
\end{align}
\end{lemma}

\begin{proof}
For any $A\in \Mat(2,\C)$, $B\in \Mat(b,\C)$ let us define $A\otimes
 B\in \Mat(2b,\C)$ by 
\begin{align*}
A\otimes B:=\left(\begin{array}{cc}A(1,1)B & A(1,2)B \\
                                   A(2,1)B & A(2,2)B \end{array}\right).
\end{align*}
Fix $\bk\in \R^d$, $\phi\in \C$. Using the notations used in the
 spectral decomposition \eqref{eq_spectral_decomposition}, let us set 
\begin{align*}
E_j(\phi)(\bk):=\left(\begin{array}{cc} e_j(\bk) & \ophi \\
                                   \phi & -e_j(\bk)
		      \end{array}\right),\quad (j=1,2,\cdots,b')
\end{align*}
and observe that
\begin{align*}
E(\phi)(\bk)=\sum_{j=1}^{b'}E_j(\phi)(\bk)\otimes P_j(\bk).
\end{align*}
Then it follows from \eqref{eq_characterization_covariance} that
\begin{align}
&C(\phi)(\orho\rho\bx 0, \oeta\eta \by
 0)\label{eq_spectral_decomposition_tensor}\\
&=\frac{1}{L^d}\sum_{\bk\in\G^*}e^{i\<\bk,\bx-\by\>}\sum_{j=1}^{b'}\notag\\
&\quad \cdot \left(I_2+e^{\beta(i\frac{\theta(\beta)}{2}I_2+E_j(\phi)(\bk))}\right)^{-1}\otimes
 P_j(\bk)((\orho-1)b+\rho,(\oeta-1)b+\eta).\notag
\end{align}
The following equalities are essentially same as what we computed in\\ 
\cite[\mbox{Lemma 4.20}]{K_BCS}. For $\orho,\oeta\in \{1,2\}$ with
 $\orho\neq \oeta$,
\begin{align}
&\left(I_2+e^{\beta(i\frac{\theta(\beta)}{2}I_2+E_j(\phi)(\bk))}\right)^{-1}(\orho,\orho)\label{eq_covariance_final_characterization_core}\\
&=\frac{e^{-i\frac{\beta\theta(\beta)}{2}}+\cosh(\beta\sqrt{e_j(\bk)^2+|\phi|^2})}{2(\cos(\beta\theta(\beta)/2)+\cosh(\beta\sqrt{e_j(\bk)^2+|\phi|^2}))}\notag\\
&\quad +
\frac{(-1)^{\orho}\sinh(\beta\sqrt{e_j(\bk)^2+|\phi|^2})e_j(\bk)}{2(\cos(\beta\theta(\beta)/2)+\cosh(\beta\sqrt{e_j(\bk)^2+|\phi|^2}))\sqrt{e_j(\bk)^2+|\phi|^2}},\notag\\
&\left(I_2+e^{\beta(i\frac{\theta(\beta)}{2}I_2+E_j(\phi)(\bk))}\right)^{-1}(\orho,\oeta)\notag\\
&=\frac{-(1_{(\orho,\oeta)=(1,2)}\overline{\phi}+1_{(\orho,\oeta)=(2,1)}\phi)\sinh(\beta\sqrt{e_j(\bk)^2+|\phi|^2})}{2(\cos(\beta\theta(\beta)/2)+\cosh(\beta\sqrt{e_j(\bk)^2+|\phi|^2}))\sqrt{e_j(\bk)^2+|\phi|^2}}.\notag
\end{align}
By combining \eqref{eq_spectral_decomposition_tensor} with
 \eqref{eq_covariance_final_characterization_core} we can derive \eqref{eq_covariance_final_characterization}.

Note that
\begin{align*}
&|C(\phi)(\orho\rho\b0 0,\oeta\rho \b0 0)|\notag\\
&=\frac{|\phi|}{2L^d}\sum_{\bk\in
 \G^*}\sum_{j=1}^{b'}\frac{\sinh(\beta\sqrt{e_j(\bk)^2+|\phi|^2})}{(\cos(\beta\theta(\beta)/2)+\cosh(\beta\sqrt{e_j(\bk)^2+|\phi|^2}))\sqrt{e_j(\bk)^2+|\phi|^2}}P_j(\bk)(\rho,\rho).
\end{align*}
Then the inequality \eqref{eq_covariance_final_lower_bound} follows from
 that $P_j(\bk)(\rho,\rho)\ge 0$, $\sum_{j=1}^{b'}P_j(\bk)(\rho,\rho)=1$ and the fact that 
\begin{align*}
x\mapsto
 \frac{\sinh x}{(\cos(\beta\theta(\beta)/2)+\cosh x)x}:[0,\infty)\to \R
\end{align*}
is strictly monotone decreasing. See e.g. \cite[\mbox{Lemma
 4.19}]{K_BCS} for the proof of this fact.
\end{proof}

By admitting general lemmas proved in Appendix
\ref{app_infinite_volume_limit} we can prove Theorem
\ref{thm_main_theorem} here. 

\begin{proof}[Proof of Theorem \ref{thm_main_theorem}]
On the whole, the structure of the proof is parallel to the proof of
 \cite[\mbox{Theorem 1.3}]{K_BCS}. However, we should keep in mind that
 in the present case the parameter $h$ must be taken large depending on
 $\phi$ and for this reason we cannot change the order of the limit
 operation $h\to \infty$ and the integral with the variable $\phi$,
 while they were interchangeable in the proof of \cite[\mbox{Theorem
 1.3}]{K_BCS}.

Note that
\begin{align*}
\left|\frac{\pi}{\beta}-\frac{\theta(\beta)}{2}\right|=\min_{m\in
 \Z}\left|\frac{\theta}{2}-\frac{\pi (2m+1)}{\beta}\right|.
\end{align*}
If $\beta <1$, 
\begin{align*}
\beta^{-1}\min\left\{1,\left|\frac{\pi}{\beta}-\frac{\theta(\beta)}{2}\right|^{-1}\right\}\ge
 \pi^{-1}.
\end{align*}
By using these properties,
\begin{align*}
&\left(\beta^{-1}\min\left\{1,\left|\frac{\pi}{\beta}-\frac{\theta(\beta)}{2}\right|^{-1}\right\}+1\right)^{-2}\\
&\ge 1_{\beta\ge 1}
 2^{-2}+1_{\beta<1}(1+\pi)^{-2}\beta^2\max\left\{1,\left|\frac{\pi}{\beta}-\frac{\theta(\beta)}{2}\right|^2\right\}\\
&=1_{\beta\ge 1}2^{-2}+1_{\beta <1}(1+\pi)^{-2}\max\left\{\beta^2,
\min_{m\in\Z}\left|\frac{\beta\theta}{2}-\pi(2m+1)\right|^2\right\}.
\end{align*}
Thus by recalling the definition of $c_0$ stated in Lemma
 \ref{lem_covariance_construction} we see that
there exists $c'\in (0,1]$ depending only on $d$, $b$,
 $(\hbv_j)_{j=1}^d$, $\sa$, $(\sn_j)_{j=1}^d$, $\sc$, $M$, $\chi$,
 $\alpha$ such that 
\begin{align*}
c'\left(1_{\beta\ge 1}+1_{\beta<1}\max\left\{
\beta^2,\min_{m\in \Z}\left|\frac{\beta\theta}{2}-\pi(2m+1)\right|^2\right\}\right)
\le 2^{-1}b^{-1}c_0^{-2}\alpha^{-4}.
\end{align*}
In the following we always assume that $U\in \R_{<0}$ and 
\begin{align*}
|U|<c'\left(1_{\beta \ge 1} +
 1_{\beta<1}\max\left\{\beta^2,\min_{m\in\Z}\left|
\frac{\beta\theta}{2}-\pi(2m+1)\right|^2\right\}
\right)
\end{align*}
so that the claims of Proposition
 \ref{prop_general_integration_application}, Proposition
 \ref{prop_infinite_volume_limit} hold with this $U$ and $h\in
 \frac{2}{\beta}\N$, $L\in \N$ satisfying \eqref{eq_huge_h_condition},
 \eqref{eq_final_huge_L_condition}. By Lemma
 \ref{lem_final_Grassmann_formulation}
 \eqref{item_formulation_uniform_convergence} and 
Proposition
 \ref{prop_general_integration_application}
 \eqref{item_integral_formulation_lower_upper_bound}, for any
 $\phi\in\C$, $L\in \N$ satisfying \eqref{eq_final_huge_L_condition}
 and $u\in [-2^{-1}b^{-1}c_0^{-2}\alpha^{-4},2^{-1}b^{-1}c_0^{-2}\alpha^{-4}]$
\begin{align*}
\left|\lim_{h\to\infty\atop h\in\frac{2}{\beta}\N}\int
 e^{-V(u)(\psi)+W(u)(\psi)}d\mu_{C(\phi)}(\psi)\right|
\ge e^{-8c_6b\beta
 \alpha^{-2}M^{(\sum_{j=1}^d\frac{1}{\sn_j}+1)(\hat{N}_{\beta}-N_{\beta})}}.
\end{align*}
By Lemma \ref{lem_final_Grassmann_formulation}
 \eqref{item_formulation_reality} and Proposition
 \ref{prop_infinite_volume_limit}, the real-valued function 
\begin{align*}
u\mapsto \lim_{h\to\infty\atop h\in\frac{2}{\beta}\N}\int
 e^{-V(u)(\psi)+W(u)(\psi)}d\mu_{C(\phi)}(\psi):[-2^{-1}b^{-1}c_0^{-2}\alpha^{-4},2^{-1}b^{-1}c_0^{-2}\alpha^{-4}]\to\R
\end{align*}
is continuous. Since this function takes 1 at $u=0$, we conclude that
\begin{align}
\lim_{h\to\infty\atop h\in\frac{2}{\beta}\N}\int
 e^{-V(\psi)+W(\psi)}d\mu_{C(\phi)}(\psi)
\ge e^{-8c_6b\beta
 \alpha^{-2}M^{(\sum_{j=1}^d\frac{1}{\sn_j}+1)(\hat{N}_{\beta}-N_{\beta})}}
\label{eq_uniform_real_positivity}
\end{align}
for any $\phi\in\C$ and $L\in \N$ satisfying
 \eqref{eq_final_huge_L_condition}. Therefore we see from Lemma
 \ref{lem_restriction_theta}, Lemma \ref{lem_free_partition_function}
 and Lemma \ref{lem_final_Grassmann_formulation}
 \eqref{item_final_Grassmann_formulation} that the claim
 \eqref{item_partition_positivity} holds. 

Let us prove the claim \eqref{item_SSB}. Assume that $\g\in (0,1]$.
Let us define the functions $F:\R^2\to \R$, $F_L:\R^2\to\R$ by 
\begin{align*}
&F(\bx):=-\frac{1}{|U|}((x_1-\g)^2+x_2^2)\\
&\qquad\qquad +\frac{D_d}{\beta}\int_{\G_{\infty}^*}d\bk
 \Tr
 \log\left(\cos\left(\frac{\beta\theta(\beta)}{2}\right)+\cosh\left(\beta\sqrt{E(\bk)^2+\|\bx\|_{\R^2}^2}\right)\right)\\
&\qquad\qquad  -
\frac{D_d}{\beta}\int_{\G_{\infty}^*}d\bk
 \Tr
 \log\left(\cos\left(\frac{\beta\theta(\beta)}{2}\right)+\cosh\left(\beta
 E(\bk)\right)\right),\\
&F_L(\bx):=-\frac{1}{|U|}((x_1-\g)^2+x_2^2)\\
&\qquad\qquad+\frac{1}{\beta
 L^d}\sum_{\bk\in \G^*}
 \Tr
 \log\left(\cos\left(\frac{\beta\theta(\beta)}{2}\right)+\cosh\left(\beta\sqrt{E(\bk)^2+\|\bx\|_{\R^2}^2}\right)\right)\\
&\qquad\qquad  -\frac{1}{\beta
 L^d}\sum_{\bk\in \G^*}
 \Tr
 \log\left(\cos\left(\frac{\beta\theta(\beta)}{2}\right)+\cosh\left(\beta
 E(\bk)\right)\right).
\end{align*}
Let us recall the definition \eqref{eq_matrix_valued_notation} of the
 matrix-valued function $G_{x,y,z}(\cdot)$. By making use of the monotone
 decreasing property of the function 
\begin{align*}
x\mapsto \frac{\sinh x}{(\cos(\beta\theta(\beta)/2)+\cosh x)x}:[0,\infty)\to\R
\end{align*}
we can prove that there uniquely exist $a(\g)\in (\D,\infty)$,
 $a_L(\g)\in (0,\infty)$ such that
\begin{align}
&a(\g)\left(-\frac{2}{|U|}+D_d\int_{\G_{\infty}^*}d\bk\Tr
 G_{\beta,\theta(\beta),a(\g)}(\bk)\right)=-\frac{2\g}{|U|},\label{eq_perturbed_E_L_equation}\\
&a_L(\g)\left(-\frac{2}{|U|}+\frac{1}{L^d}\sum_{\bk\in \G^*}\Tr
 G_{\beta,\theta(\beta),a_L(\g)}(\bk)\right)=-\frac{2\g}{|U|},\notag\\
&\lim_{\g\searrow 0\atop \g\in
 (0,1]}a(\g)=\D.\label{eq_perturbed_order_parameter}
\end{align}
Set $\ba_L:=(a_L(\g),0)$, $\ba:=(a(\g),0)$. By computing Hessians one
 can check that $\ba_L$, $\ba$ are the unique global maximum point of
 $F_L$, $F$ respectively.

Let us define the functions $g_L$, $u_{1,L}:\R^2\to \C$ by
\begin{align*}
&g_L(\bx):=\lim_{h\to\infty\atop h\in \frac{2}{\beta}\N}\int
 e^{-V(\psi)+W(\psi)}d\mu_{C(x_1+ix_2)}(\psi),\\
&u_{1,L}(\bx):=\lim_{h\to\infty\atop h\in \frac{2}{\beta}\N}\int
 e^{-V(\psi)+W(\psi)}A^1(\psi)d\mu_{C(x_1+ix_2)}(\psi).
\end{align*}
It follows from Lemma \ref{lem_final_Grassmann_formulation}
 \eqref{item_formulation_uniform_convergence} that $g_L$, $u_{1,L}\in
 C(\R^2)$. Moreover by Proposition
 \ref{prop_general_integration_application}
 \eqref{item_integral_formulation_lower_upper_bound},\eqref{item_integral_formulation_difference} and the determinant bound
 Lemma \ref{lem_characterization_covariance}
 \eqref{item_P_S_determinant_bound},
\begin{align*}
&\sup_{L\in \N\atop\text{satisfying
 }\eqref{eq_final_huge_L_condition}}\sup_{\bx\in
 \R^2}|g_L(\bx)|<\infty,\quad
\sup_{L\in \N\atop\text{satisfying
 }\eqref{eq_final_huge_L_condition}}\sup_{\bx\in
 \R^2}|u_{1,L}(\bx)|<\infty.
\end{align*}
Furthermore by Proposition \ref{prop_infinite_volume_limit} and
 Proposition \ref{prop_general_integration_application}
 \eqref{item_integral_formulation_difference} there exists $g\in
 C(\R^2)$ such that $g_L$ converges to $g$ locally uniformly as $L\to \infty$ $(L\in
 \N)$ and if we
 set 
\begin{align*}
u_1(x_1,x_2):=\beta g(x_1,x_2)\lim_{L\to \infty\atop L\in
 \N}C(x_1+ix_2)(1\hat{\rho}\b0 0, 2\hat{\rho}\b0 0),\quad
 (x_1,x_2\in\R),
\end{align*}
$u_{1,L}$ converges to $u_1$ locally uniformly as $L\to \infty$ $(L\in
 \N)$. Also, let us remark that by Proposition
 \ref{prop_general_integration_application}
 \eqref{item_integral_formulation_lower_upper_bound}, $g(\ba)\neq
 0$. We can check that the assumptions of Lemma
 \ref{lem_infinite_volume_expectation} are satisfied. We can apply
 the lemma to ensure that
\begin{align}
\lim_{L\to \infty\atop L\in\N}\frac{\int_{\R^2}d\bx e^{\beta
 L^dF_L(\bx)} u_{1,L}(\bx)}{\int_{\R^2}d\bx e^{\beta
 L^dF_L(\bx)} g_{L}(\bx)}&=\frac{\beta g(\ba)\lim_{L\to \infty\atop
 L\in\N}C(a(\g))(1\hat{\rho}\b0 0,2\hat{\rho}\b0 0)}{g(\ba)}\label{eq_application_general_lemma_2_point}\\
&=\beta\lim_{L\to \infty\atop L\in \N}C(a(\g))(1\hat{\rho}\b0 0,
 2\hat{\rho}\b0 0).\notag
\end{align}
By Lemma \ref{lem_covariance_final_characterization}, 
\begin{align}
\lim_{L\to \infty\atop L\in \N}C(a(\g))(1\hat{\rho}\b0 0,2\hat{\rho}\b0
 0)=-\frac{a(\g)D_d}{2}\int_{\G_{\infty}^*}d\bk
 G_{\beta,\theta(\beta),a(\g)}(\bk)(\hat{\rho},\hat{\rho}).\label{eq_application_covariance_characterization_2_point}
\end{align}
By combining \eqref{eq_perturbed_order_parameter},
 \eqref{eq_application_general_lemma_2_point},
 \eqref{eq_application_covariance_characterization_2_point}, Lemma
 \ref{lem_restriction_theta} with Lemma
 \ref{lem_final_Grassmann_formulation}
 \eqref{item_final_Grassmann_formulation} we have that
\begin{align}
&\lim_{L\to \infty\atop L\in \N}\frac{\Tr (e^{-\beta(\sH+i\theta
 \sS_z+\sF)}\sA_1)}{\Tr e^{-\beta(\sH+i\theta
 \sS_z+\sF)}}=-\frac{a(\g)D_d}{2}\int_{\G_{\infty}^*}d\bk
 G_{\beta,\theta(\beta),a(\g)}(\bk)(\hat{\rho},\hat{\rho}),\label{eq_2_point_correlation_with_field}\\
&\lim_{\g\searrow 0\atop \g\in (0,1]}\lim_{L\to \infty\atop L\in \N}\frac{\Tr (e^{-\beta(\sH+i\theta
 \sS_z+\sF)}\sA_1)}{\Tr e^{-\beta(\sH+i\theta
 \sS_z+\sF)}}=-\frac{\D D_d}{2}\int_{\G_{\infty}^*}d\bk
 G_{\beta,\theta,\D}(\bk)(\hat{\rho},\hat{\rho}).\notag
\end{align}
This concludes the proof of the claim \eqref{item_SSB}.

Let us show the claim \eqref{item_ODLRO}. Define the
 functions $f:\R\to \R$, $f_L:\R\to\R$ by
\begin{align*}
&f(x):=-\frac{x^2}{|U|} +\frac{D_d}{\beta}\int_{\G_{\infty}^*}d\bk
 \Tr
 \log\left(\cos\left(\frac{\beta\theta(\beta)}{2}\right)+\cosh\left(\beta\sqrt{E(\bk)^2+x^2}\right)\right)\\
&\qquad\qquad  -
\frac{D_d}{\beta}\int_{\G_{\infty}^*}d\bk
 \Tr
 \log\left(\cos\left(\frac{\beta\theta(\beta)}{2}\right)+\cosh\left(\beta
 E(\bk)\right)\right),\\
&f_L(x):=-\frac{x^2}{|U|}+\frac{1}{\beta
 L^d}\sum_{\bk\in \G^*}
 \Tr
 \log\left(\cos\left(\frac{\beta\theta(\beta)}{2}\right)+\cosh\left(\beta\sqrt{E(\bk)^2+x^2}\right)\right)\\
&\qquad\qquad  -\frac{1}{\beta
 L^d}\sum_{\bk\in \G^*}
 \Tr
 \log\left(\cos\left(\frac{\beta\theta(\beta)}{2}\right)+\cosh\left(\beta
 E(\bk)\right)\right).
\end{align*}
We let $\D_L$ $(\in[0,\infty))$ be the solution to 
\begin{align}
-\frac{2}{|U|}+\frac{1}{L^d}\sum_{\bk\in\G^*}\Tr
 G_{\beta,\theta(\beta),\D_L}(\bk)=0,\label{eq_discrete_gap_equation}
\end{align}
if 
\begin{align*}
-\frac{2}{|U|}+\frac{1}{L^d}\sum_{\bk\in\G^*}\Tr
 G_{\beta,\theta(\beta),0}(\bk)\ge 0.
\end{align*}
We let $\D_L:=0$ if 
\begin{align*}
-\frac{2}{|U|}+\frac{1}{L^d}\sum_{\bk\in\G^*}\Tr
 G_{\beta,\theta(\beta),0}(\bk)< 0.
\end{align*}
The well-definedness of $\D_L$ is guaranteed by the parallel
 consideration to Lemma \ref{lem_gap_equation_solvability}. Note that
$\D$, $\D_L$ are the unique maximum point of $f|_{\R_{\ge 0}}$,
 $f_L|_{\R_{\ge 0}}$ respectively. 

First let us consider the case that
\begin{align}
|U|\neq \left(\frac{D_d}{2}\int_{\G_{\infty}^*}d\bk \Tr
 G_{\beta,\theta(\beta),0}(\bk)\right)^{-1}.\label{eq_not_phase_boundary}
\end{align}
It follows that 
\begin{align*}
-\frac{2}{|U|}+\frac{1}{L^d}\sum_{\bk\in\G^*}\Tr
 G_{\beta,\theta(\beta),0}(\bk)\neq 0 
\end{align*}
for sufficiently large $L\in\N$. Moreover, $\frac{d^2f}{dx^2}(\D)<0$. If
 $\D=0$, $\D_L=0$ for sufficiently large $L\in\N$. Let us define the
 functions $v_L:\R\to\C$, $u_{2,L}:\R\to\C$ by 
\begin{align*}
&v_L(x):=\int_0^{2\pi}d\xi \lim_{h\to \infty\atop
 h\in\frac{2}{\beta}\N}\int
 e^{-V(\psi)+W(\psi)}d\mu_{C(xe^{i\xi})}(\psi)=\int_0^{2\pi}d\xi
 g_L(x\cos\xi,x\sin\xi),\\
&u_{2,L}(x):=\int_0^{2\pi}d\xi \lim_{h\to \infty\atop
 h\in\frac{2}{\beta}\N}\int
 e^{-V(\psi)+W(\psi)}A^2(\psi)d\mu_{C(xe^{i\xi})}(\psi).
\end{align*}
By Lemma \ref{lem_final_Grassmann_formulation}
 \eqref{item_formulation_uniform_convergence}, $v_L$, $u_{2,L}\in
 C(\R)$. By Lemma \ref{lem_characterization_covariance}
 \eqref{item_P_S_determinant_bound}, Proposition
 \ref{prop_general_integration_application}
 \eqref{item_integral_formulation_lower_upper_bound},\eqref{item_integral_formulation_difference}
 and Proposition \ref{prop_infinite_volume_limit}, for any $r\in\R_{>0}$
\begin{align*}
&\sup_{L\in \N\atop \text{satisfying
 }\eqref{eq_final_huge_L_condition}}\sup_{x\in\R}|v_L(x)|<\infty,\quad
\sup_{L\in \N\atop \text{satisfying
 }\eqref{eq_final_huge_L_condition}}\sup_{x\in\R}|u_{2,L}(x)|<\infty,\\
&\lim_{L\to\infty\atop L\in \N}\sup_{x\in[-r,r]}\left|
v_L(x)-\int_0^{2\pi}d\xi g(x\cos\xi,x\sin\xi)\right|=0,\\
&\lim_{L\to\infty\atop L\in \N}\sup_{x\in[-r,r]}\Bigg|
u_{2,L}(x)\\
&\quad -\beta \int_0^{2\pi}d\xi g(x\cos \xi, x \sin \xi)
 \Bigg(1_{(\hat{\rho},\hbx)=(\hat{\eta},\hby)}\lim_{L\to\infty\atop
 L\in\N}C(x e^{i\xi})(1\hrho\b0 0,1\hrho\b0 0)\\
&\qquad\qquad\qquad  -\lim_{L\to \infty\atop
 L\in \N}\det\left(\begin{array}{cc} C(xe^{i\xi})(1\hrho\hbx 0,1\heta
	      \hby 0) & C(xe^{i\xi})(1\hrho\b0 0,2\hrho\b0 0) \\
               C(xe^{i\xi})(2\heta\b0 0,1\heta \b0 0) & C(xe^{i\xi})(2\heta\hby 0,2\hrho
	      \hbx 0)\end{array}\right)\Bigg)\Bigg|=0.
\end{align*}
Also, note that by \eqref{eq_uniform_real_positivity}
$$
\int_0^{2\pi}d\xi g(\D \cos\xi,\D\sin \xi)>0.
$$
Moreover by changing variables we can see that 
\begin{align*}
&\int_{\R^2}d\phi_1d\phi_2 e^{\beta L^d
 f_L(|\phi|)}g_L(\phi_1,\phi_2)=\int_0^{\infty}dx x e^{\beta L^d
 f_L(x)}v_L(x),\\
&\int_{\R^2}d\phi_1d\phi_2 e^{\beta L^d
 f_L(|\phi|)}\lim_{h\to\infty\atop h\in \frac{2}{\beta}\N}\int
 e^{-V(\psi)+W(\psi)}A^2(\psi)d\mu_{C(\phi)}(\psi)=\int_0^{\infty}dx x e^{\beta L^d
 f_L(x)}u_{2,L}(x).
\end{align*}
In this situation we can apply Lemma
 \ref{lem_infinite_volume_correlation}. As the result,
\begin{align*}
&\lim_{L\to \infty\atop L\in\N}\frac{\int_{0}^{\infty}dx x e^{\beta
 L^df_L(x)} u_{2,L}(x)}{\int_{0}^{\infty}dx x e^{\beta
 L^df_L(x)} v_{L}(x)}=\frac{\lim_{L\to\infty\atop L\in
 \N}u_{2,L}(\D)}{\lim_{L\to\infty\atop L\in \N}v_{L}(\D)}\\
&=\beta \int_0^{2\pi}d\xi g(\D\cos \xi, \D \sin \xi)
 \Bigg(1_{(\hat{\rho},\hbx)=(\hat{\eta},\hby)}\lim_{L\to\infty\atop
 L\in\N}C(\D e^{i\xi})(1\hrho\b0 0,1\hrho\b0 0)\\
&\qquad\qquad\qquad  -\lim_{L\to \infty\atop
 L\in \N}\det\left(\begin{array}{cc} C(\D e^{i\xi})(1\hrho\hbx 0,1\heta
	      \hby 0) & C(\D e^{i\xi})(1\hrho\b0 0,2\hrho\b0 0) \\
               C(\D e^{i\xi})(2\heta\b0 0,1\heta \b0 0) & C(\D e^{i\xi})(2\heta\hby 0,2\hrho
	      \hbx 0)\end{array}\right)\Bigg)\\
&\quad\cdot \left(
\int_0^{2\pi}d\xi g(\D\cos\xi,\D\sin\xi)\right)^{-1}\\
&=\beta\Bigg(1_{(\hat{\rho},\hbx)=(\hat{\eta},\hby)}\lim_{L\to\infty\atop
 L\in\N}C(\D)(1\hrho\b0 0,1\hrho\b0 0)\\
&\qquad\quad -\lim_{L\to \infty\atop
 L\in \N}\det\left(\begin{array}{cc} C(\D)(1\hrho\hbx 0,1\heta
	      \hby 0) & C(\D)(1\hrho\b0 0,2\hrho\b0 0) \\
               C(\D)(2\heta\b0 0,1\heta \b0 0) & C(\D)(2\heta\hby 0,2\hrho
	      \hbx 0)\end{array}\right)\Bigg).
\end{align*}
We derived the last equality by recalling Lemma
 \ref{lem_covariance_final_characterization}. Therefore by Lemma
 \ref{lem_restriction_theta} and Lemma
 \ref{lem_final_Grassmann_formulation}
 \eqref{item_final_Grassmann_formulation},
\begin{align}
&\lim_{L\to \infty\atop L\in \N}\frac{\Tr (e^{-\beta(\sH+i\theta
 \sS_z)}\sA_2)}{\Tr e^{-\beta(\sH+i\theta
 \sS_z)}}\label{eq_4_point_function_characterization}\\
&=1_{(\hat{\rho},\hbx)=(\hat{\eta},\hby)}\lim_{L\to\infty\atop
 L\in\N}C(\D)(1\hrho\b0 0,1\hrho\b0 0)\notag\\
&\quad -\lim_{L\to \infty\atop
 L\in \N}\det\left(\begin{array}{cc} C(\D )(1\hrho\hbx 0,1\heta
	      \hby 0) & C(\D)(1\hrho\b0 0,2\hrho\b0 0) \\
               C(\D)(2\heta\b0 0,1\heta \b0 0) & C(\D)(2\heta\hby 0,2\hrho
	      \hbx 0)\end{array}\right).\notag
\end{align}
By using the fact that for any compact set $K$ of $\C$ and
 $(\orho,\rho)$, $(\oeta,\eta)\in\{1,2\}\times\cB$
\begin{align}
\lim_{\|\bx-\by\|_{\R^d}\to \infty}\sup_{\phi\in K}\left|\lim_{L\to
 \infty\atop L\in\N}C(\phi)(\orho\rho\bx 0, \oeta \eta \by
 0)\right|=0\label{eq_total_covariance_decay}
\end{align}
and recalling Lemma \ref{lem_covariance_final_characterization} we
 observe that 
\begin{align*}
\lim_{\|\hbx-\hby\|_{\R^d}\to \infty}\lim_{L\to \infty\atop L\in \N}\frac{\Tr (e^{-\beta(\sH+i\theta
 \sS_z)}\sA_2)}{\Tr e^{-\beta(\sH+i\theta
 \sS_z)}}&=\lim_{L\to \infty\atop L\in \N} C(\D)(1\hrho\b0 0,2\hrho\b0 0)
 C(\D)(2\heta\b0 0,1\heta \b0 0)\\
&=(\text{R. H. S of
 }\eqref{eq_ODLRO_explicit}).
\end{align*}
We can show the property \eqref{eq_total_covariance_decay} by
 establishing a decay bound such as
 \eqref{eq_full_covariance_total_decay}.

Let us assume that
\begin{align}
|U|= \left(\frac{D_d}{2}\int_{\G_{\infty}^*}d\bk \Tr
 G_{\beta,\theta(\beta),0}(\bk)\right)^{-1}.\label{eq_on_phase_boundary}
\end{align}
In this case we apply Lemma
 \ref{lem_infinite_volume_correlation_estimate} to prove the claim. By
 \eqref{eq_uniform_real_positivity} $v_L\in C(\R,\R)$ and 
\begin{align}
\inf_{L\in \N\atop\text{satisfying
 }\eqref{eq_final_huge_L_condition}}\inf_{x\in\R}v_L(x)>0.\label{eq_continuum_global_lower_bound}
\end{align}
Let us define the function $u_{3,L}:\R\to \C$ by
\begin{align*}
u_{3,L}(x):=&\beta 1_{(\hrho,r_L(\hbx))=(\heta,r_L(\hby))}C(x)(1\hrho \b0 0, 1\hrho
 \b0 0)\\
&-\beta \det\left(\begin{array}{cc} C(x )(1\hrho\hbx 0,1\heta
	      \hby 0) & C(x)(1\hrho\b0 0,2\hrho\b0 0) \\
               C(x)(2\heta\b0 0,1\heta \b0 0) & C(x)(2\heta\hby 0,2\hrho
	      \hbx 0)\end{array}\right).\notag
\end{align*}
We can see from \eqref{eq_covariance_final_characterization} that
\begin{align*}
\int A^2(\psi)d\mu_{C(x e^{i\xi})}(\psi)=u_{3,L}(x),\quad (\forall
 x,\xi\in\R).
\end{align*}
We have to prove that there exists $n_0\in 2\N$ such that
\begin{align}
\frac{d^n f}{d x^n}(0)=0,\quad (\forall n\in
 \{1,2,\cdots,n_0-1\}),\quad
\frac{d^{n_0} f}{d
 x^{n_0}}(0)<0.\label{eq_effective_potential_at_origin}
\end{align}
We define the function $q$ in a neighborhood of the origin by
\begin{align*}
q(z):=&-\frac{z^2}{|U|}+\frac{D_d}{\beta}\int_{\G_{\infty}^*}d\bk \Tr
 \log\left(\cos\left(\frac{\beta
 \theta(\beta)}{2}\right)+\sum_{n=0}^{\infty}\frac{\beta^{2n}}{(2n)!}(E(\bk)^2+z^2)^n\right)\\
&-\frac{D_d}{\beta}\int_{\G_{\infty}^*}d\bk \Tr
 \log\left(\cos\left(\frac{\beta
 \theta(\beta)}{2}\right)+\cosh(\beta E(\bk))\right).
\end{align*}
Since $\cos(\beta\theta(\beta)/2)>-1$, $q$ is analytic in a neighborhood
 of the origin. Moreover $q(x)=f(x)$ if $x\in \R$.
Since $0=f(0)>f(x)$ $(\forall x\in \R\backslash\{0\})$, $q$ is not
 identically 0. Thus there exists $n_0\in\N$ such that $q^{(n_0)}(0)\neq
 0$ and $q(z)=\sum_{n=n_0}^{\infty}\frac{1}{n!}q^{(n)}(0)z^n$ in a
 neighborhood of the origin. Thus $f^{(n_0)}(0)\neq 0$ and
 $f(x)=\sum_{n=n_0}^{\infty}\frac{1}{n!}f^{(n)}(0)x^n$ for any $x\in \R$
 close to 0. Since $f$ takes the maximum value 0 at $x=0$, $n_0$ must be
 even and $f^{(n_0)}(0)<0$. Therefore the claim
 \eqref{eq_effective_potential_at_origin} holds true. We can  check that
 all the other conditions required in Lemma
 \ref{lem_infinite_volume_correlation_estimate} are satisfied by the
 functions $f_L$, $f$, $v_L$, $u_{3,L}$. Thus the lemma ensures that 
\begin{align*}
\limsup_{L\to\infty\atop L\in\N}\left|\frac{\int_0^{\infty}dx x e^{\beta
 L^d f_L(x)}v_L(x)u_{3,L}(x)}{\int_0^{\infty}dx x e^{\beta
 L^d f_L(x)}v_L(x)}\right|\le \left|
\lim_{L\to \infty\atop L\in \N}u_{3,L}(0)\right|.
\end{align*}
Moreover by Lemma \ref{lem_restriction_theta}, Lemma
 \ref{lem_final_Grassmann_formulation}
 \eqref{item_final_Grassmann_formulation}, Proposition
 \ref{prop_general_integration_application}
\eqref{item_integral_formulation_lower_upper_bound},\eqref{item_integral_formulation_difference}
 and \eqref{eq_continuum_global_lower_bound},
\begin{align}
&\limsup_{L\to \infty\atop L\in\N}\left|
\frac{\Tr(e^{-\beta (\sH+i\theta\sS_z)}\sA_2)}{\Tr e^{-\beta (\sH+i\theta
 \sS_z)}}\right|\label{eq_4_point_function_estimate_boundary}\\
&\le\frac{1}{\beta}\limsup_{L\to \infty\atop L\in\N}\left|\frac{\int_0^{\infty}dx x e^{\beta
 L^d f_L(x)}u_{2,L}(x)}{\int_0^{\infty}dx x e^{\beta
 L^d f_L(x)}v_L(x)}\right|
\le \frac{1}{\beta}\limsup_{L\to \infty\atop L\in\N}\left|\frac{\int_0^{\infty}dx x e^{\beta
 L^d f_L(x)}v_L(x) u_{3,L}(x)}{\int_0^{\infty}dx x e^{\beta
 L^d f_L(x)}v_L(x)}\right|\notag\\
&\le \left|1_{(\hrho,\hbx)=(\heta,\hby)}\lim_{L\to \infty\atop L\in
 \N}C(0)(1\hrho\b0 0,1\hrho\b0 0)-\lim_{L\to \infty\atop L\in
 \N}C(0)(1\hrho\hbx 0,1\heta\hby 0)C(0)(2\heta\hby 0,2\hrho\hbx
 0)\right|.\notag
\end{align}
Then we can apply \eqref{eq_total_covariance_decay} to conclude the
 claimed convergent property. 

Let us prove the claim \eqref{item_CPD}. First let us consider the case
 that \eqref{eq_not_phase_boundary} holds. Define the function
 $u_{4,L}:\R\to\C$ by 
\begin{align*}
u_{4,L}(x):=\frac{1}{L^{2d}}\sum_{(\hrho,\hbx),(\heta,\hby)\in\cB\times
 \G}\int_0^{2\pi}d\xi \lim_{h\to\infty\atop h\in \frac{2}{\beta}\N}\int
 e^{-V(\psi)+W(\psi)}A^{2}(\psi)d\mu_{C(x e^{i\xi})}(\psi).
\end{align*}
By Lemma \ref{lem_final_Grassmann_formulation}
 \eqref{item_formulation_uniform_convergence}, $u_{4,L}\in
 C(\R)$. Moreover by Lemma \ref{lem_characterization_covariance}
 \eqref{item_P_S_determinant_bound}, Proposition
 \ref{prop_general_integration_application}
 \eqref{item_integral_formulation_lower_upper_bound},\eqref{item_integral_formulation_difference},
 Proposition \ref{prop_infinite_volume_limit} and Lemma
 \ref{lem_covariance_final_characterization}, for any $r\in \R_{>0}$ 
\begin{align*}
&\sup_{L\in\N\atop \text{satisfying
 }\eqref{eq_final_huge_L_condition}}\sup_{x\in\R}|u_{4,L}(x)|<\infty,\\
&\lim_{L\to \infty\atop L\in \N}\sup_{x\in [-r,r]}\Bigg|
u_{4,L}(x)\\
&\qquad -\frac{1}{L^{2d}}\sum_{(\hrho,\hbx),(\heta,\hby)\in \cB\times
 \G}\int_0^{2\pi}d\xi \lim_{h\to\infty\atop h\in \frac{2}{\beta}\N}\int
 A^2(\psi)d\mu_{C(xe^{i\xi})}(\psi)g_L(x\cos\xi, x\sin \xi)\Bigg|=0,\\
&\lim_{L\to \infty\atop L\in \N}\sup_{x\in[-r,r]}
\Bigg|\frac{1}{L^{2d}}\sum_{(\hrho,\hbx),(\heta,\hby)\in \cB\times
 \G}\int_0^{2\pi}d\xi \lim_{h\to\infty\atop h\in \frac{2}{\beta}\N}\int
 A^2(\psi)d\mu_{C(xe^{i\xi})}(\psi)g_L(x\cos\xi, x\sin \xi)\\
&\qquad\qquad\qquad -\beta x^2 \left(\frac{D_d}{2}\int_{\G_{\infty}^*}d\bk \Tr
 G_{\beta,\theta(\beta),x}(\bk)\right)^2\int_0^{2\pi}d\xi
 g(x\cos\xi,x\sin \xi)\Bigg|=0.
\end{align*}
Here we can apply Lemma \ref{lem_infinite_volume_correlation} to derive
 that
\begin{align*}
&\lim_{L\to \infty\atop L\in \N}\frac{\int_0^{\infty}dx x e^{\beta
 L^d f_L(x)}u_{4,L}(x)}{\int_0^{\infty}dx x e^{\beta
 L^df_L(x)} v_L(x)}=\frac{\lim_{L\to\infty\atop L\in
 \N}u_{4,L}(\D)}{\lim_{L\to\infty\atop
 L\in\N}v_L(\D)}\\
&=\beta\D^2\left(\frac{D_d}{2}\int_{\G_{\infty}^*}d\bk
 \Tr G_{\beta,\theta(\beta),\D}(\bk)\right)^2=\frac{\beta \D^2}{U^2},
\end{align*}
which combined with Lemma \ref{lem_restriction_theta}, Lemma
 \ref{lem_final_Grassmann_formulation}
 \eqref{item_final_Grassmann_formulation} ensures the claimed result in
 this case.

Next let us assume that \eqref{eq_on_phase_boundary} holds. Define the
 function $u_{5,L}:\R\to\C$ by 
\begin{align*}
u_{5,L}(x):=\frac{1}{L^{2d}}\sum_{(\hrho,\hbx),(\heta,\hby)\in\cB\times
 \G}\lim_{h\to \infty\atop h\in \frac{2}{\beta}\N}\int
 A^2(\psi)d\mu_{C(x)}(\psi).
\end{align*}
Then by Lemma \ref{lem_characterization_covariance}
 \eqref{item_P_S_determinant_bound} and Lemma
 \ref{lem_covariance_final_characterization}, for any $r\in \R_{>0}$
\begin{align*}
&\sup_{L\in\N}\sup_{x\in\R}|u_{5,L}(x)|<\infty,\\
&\lim_{L\to\infty\atop L\in \N}\sup_{x\in [-r,r]}\left|
u_{5,L}(x)-\beta x^2\left(\frac{D_d}{2}\int_{\G_{\infty}^*}d\bk \Tr
 G_{\beta,\theta(\beta),x}(\bk)\right)^2\right|=0.
\end{align*}
Thus by Lemma \ref{lem_infinite_volume_correlation_estimate}
\begin{align*}
&\limsup_{L\to\infty\atop L\in \N}\left|
\frac{\int_0^{\infty}dx x e^{\beta L^d
 f_L(x)}v_L(x)u_{5,L}(x)}{\int_0^{\infty}dx x e^{\beta L^d f_L(x)}v_L(x)}
\right|\le \left|\lim_{L\to \infty\atop L\in\N}u_{5,L}(0)\right|=0,
\end{align*}
which together with Lemma \ref{lem_restriction_theta}, Lemma
 \ref{lem_final_Grassmann_formulation}
 \eqref{item_final_Grassmann_formulation}, Proposition
 \ref{prop_general_integration_application} \eqref{item_integral_formulation_lower_upper_bound},\eqref{item_integral_formulation_difference} and \eqref{eq_continuum_global_lower_bound}
gives that
\begin{align*}
\limsup_{L\to\infty\atop L\in \N}\left|
\frac{1}{L^{2d}}\sum_{(\hat{\rho},\hbx),(\hat{\eta},\hby)\in\cB\times
 \G}\frac{\Tr(e^{-\beta (\sH+i\theta \sS_z)}\sA_2)}{\Tr e^{-\beta (\sH+i\theta \sS_z)}}
\right|=0.
\end{align*}
This implies the claim in this case as well.

Finally let us prove the claim \eqref{item_free_energy_density}. Whether
 \eqref{eq_not_phase_boundary} holds or not, we can readily apply Lemma
 \ref{lem_infinite_volume_logarithm} to derive that
\begin{align} 
\lim_{L\to \infty\atop L\in \N}\frac{1}{L^d}\log\left(\int_0^{\infty}dx
 x e^{\beta L^d f_L(x)}v_L(x)\right) =\beta
 f(\D).\label{eq_free_energy_one}
\end{align}
On the other hand, by Lemma \ref{lem_restriction_theta}, Lemma
 \ref{lem_final_Grassmann_formulation}
 \eqref{item_final_Grassmann_formulation} and Lemma
 \ref{lem_free_partition_function}
\begin{align}
&\lim_{L\to\infty\atop L\in \N}\frac{1}{L^d}\log\left(\int_0^{\infty}dx x
 e^{\beta L^d f_L(x)}v_L(x)\right)=\lim_{L\to\infty\atop L\in \N}\frac{1}{L^d}\log\left(\frac{\pi
 |U|}{\beta L^d}\frac{\Tr e^{-\beta (\sH+i\theta \sS_z)}}{\Tr e^{-\beta
 (\sH_0+i\theta \sS_z)}}\right)\label{eq_free_energy_two}\\
&=\lim_{L\to\infty\atop L\in \N}\frac{1}{L^d}\log(\Tr e^{-\beta
 (\sH+i\theta \sS_z)})\notag\\
&\quad -D_d\int_{\G_{\infty}^*}d\bk \Tr
 \log\left(1+2\cos\left(\frac{\beta\theta}{2}\right)e^{-\beta
 E(\bk)}+e^{-2\beta E(\bk)}\right).\notag
\end{align}
By coupling \eqref{eq_free_energy_one} with \eqref{eq_free_energy_two}
 we obtain the claimed equality.

Since the claim \eqref{item_gap_equation_solvable} has been proved right
 after the statement of Theorem \ref{thm_main_theorem}, the proof of the
 theorem is complete.
\end{proof}

In the rest of this section we prove Corollary \ref{cor_zero_temperature_limit}.

\begin{proof}[Proof of Corollary \ref{cor_zero_temperature_limit}]
Let $c_1$ be the constant introduced in Theorem
 \ref{thm_main_theorem}. Let us assume that 
\begin{align*}
U\in \left(-\min\left\{c_1,\frac{2(\cosh(1)-1)}{\cosh(1)D_d
 b\sc}\right\},0\right)
\end{align*}
in the following.

\eqref{item_gap_bound}: Assume that there exists $\beta\in\R_{\ge 1}$
 with $\beta\theta/2\notin \pi (2\Z+1)$ such that $\D>1/\beta$. Then by
 \eqref{eq_dispersion_measure_divided}, 
\begin{align*}
D_d\int_{\G_{\infty}^*}d\bk \Tr G_{\beta,\theta,\D}(\bk)&\le
 \left(1-\frac{1}{\cosh(1)}\right)^{-1}D_d b\int_{\G_{\infty}^*}d\bk
 \frac{1}{e(\bk)}\\
&\le
 \left(1-\frac{1}{\cosh(1)}\right)^{-1}D_d b\sc.
\end{align*}
Thus
$$
|U|\ge \frac{2(\cosh(1)-1)}{\cosh(1) D_d b\sc},
$$
which contradicts the assumption.
Thus the claim \eqref{item_gap_bound} holds with 
$$
c_2=\min\left\{c_1,\frac{2(\cosh(1)-1)}{\cosh(1)D_db\sc}\right\}.
$$

\eqref{item_CPD_zero}: The claim follows from Theorem
 \ref{thm_main_theorem} \eqref{item_CPD} and the claim
 \eqref{item_gap_bound} of this corollary. 

\eqref{item_ground_state_energy}: Let $\beta\ge 1$. Observe that
\begin{align*}
&(\text{R. H. S of }\eqref{eq_free_energy_density})\\
&=\frac{\D^2}{|U|}-D_d\int_{\G_{\infty}^*}d\bk
 \Tr(\sqrt{E(\bk)^2+\D^2}-E(\bk))\\
&\quad -\frac{D_d}{\beta}\int_{\G_{\infty}^*}d\bk \Tr \log\left(
1+2\cos\left(\frac{\beta\theta(\beta)}{2}\right)e^{-\beta\sqrt{E(\bk)^2+\D^2}}+e^{-2\beta\sqrt{E(\bk)^2+\D^2}}\right).
\end{align*}
By a calculation similar to \eqref{eq_determinant_bound_improvement}
\begin{align*}
b\log 4 &\ge \Tr \log\left(1 +
 2\cos\left(\frac{\beta\theta(\beta)}{2}\right)
e^{-\beta\sqrt{E(\bk)^2+\D^2}}+ e^{-2\beta\sqrt{E(\bk)^2+\D^2}}
\right)\\
&\ge
 b\log\left(c\min\left\{1,\beta^2\left(e(\bk)^2+\left(\frac{\theta(\beta)}{2}-\frac{\pi}{\beta}\right)^2\right)\right\}\right),\quad(\forall
 \bk\in \R^d).
\end{align*}
Thus
\begin{align*}
&\left|\frac{D_d}{\beta}\int_{\G_{\infty}^*}d\bk \Tr \log\left( 1+
 2\cos\left(\frac{\beta\theta(\beta)}{2}\right)e^{-\beta\sqrt{E(\bk)^2+\D^2}}+
e^{-2\beta\sqrt{E(\bk)^2+\D^2}}\right)\right|\\
&\le \frac{c(D_d,b)}{\beta}\left(1+\int_{\G_{\infty}^*}d\bk
 \left|\log\left(\beta^2
 \left(e(\bk)^2+\left(\frac{\theta(\beta)}{2}-\frac{\pi}{\beta}\right)^2\right)\right)\right|\right)\\
&\le \frac{c(D_d,b)}{\beta}\Bigg(1+\log\beta +\int_{\G_{\infty}^*}d\bk
 \left(e(\bk)^2+\left(\frac{\theta(\beta)}{2}-\frac{\pi}{\beta}\right)^2
\right)^{\frac{1}{2}}\\
&\qquad\qquad\qquad +
\int_{\G_{\infty}^*}d\bk
 \left(e(\bk)^2+\left(\frac{\theta(\beta)}{2}-\frac{\pi}{\beta}\right)^2
\right)^{-\frac{1}{2}}\Bigg)\\
&\le  \frac{c(D_d,b,\sc)}{\beta}(1+\log\beta).
\end{align*}
In the last inequality we used
 \eqref{eq_dispersion_measure_divided}. Then by using the claim
 \eqref{item_gap_bound} of this corollary we can deduce the first
 convergent property. The second convergent property can be derived
 from Lemma \ref{lem_free_partition_function} and the same calculation
 as above.

\eqref{item_SSB_zero}: Observe that the modulus of the right-hand side
 of \eqref{eq_SSB} is less than or equal to $\D/|U|$. Thus it is clear from the claim \eqref{item_gap_bound}
 of this corollary that the expectation value converges to zero if we
 take the limit $\beta\to \infty$ after sending $\g$ to 0. Let us prove
 the claims concerning the limit $\g\searrow 0$ after sending $\beta\to
 \infty$. Recall the equality \eqref{eq_perturbed_E_L_equation}.
To make clear the dependency on $\beta$, let us write $a(\beta,\g)$
 instead of $a(\g)$. Let us define the function
 $f:\R_{>0}\times[-1,1]\times\R_{>0}\to \R$ by
\begin{align*}
&f(x,y,z)\\
&:=z\left(-\frac{2}{|U|}+D_d\int_{\G_{\infty}^*}d\bk\Tr\left(\frac{\sinh(x\sqrt{E(\bk)^2+z^2})}{(y+\cosh(x\sqrt{E(\bk)^2+z^2}))\sqrt{E(\bk)^2+z^2}}\right)\right)+\frac{2\g}{|U|}.
\end{align*}
For any $(x,y)\in \R_{>0}\times[-1,1]$ there uniquely exists
 $z(x,y)\in\R_{>0}$ such that $f(x,y,z(x,y))=0$. Set 
\begin{align*}
S:=\Bigg\{&(x,y,z)\in\R_{>0}\times (-1,1)\times\R_{>0}\ \Big|\\
&  -\frac{2}{|U|}+D_d\int_{\G_{\infty}^*}d\bk\Tr \left(\frac{\sinh(x
 \sqrt{E(\bk)^2+z^2})}{(y+\cosh(x\sqrt{E(\bk)^2+z^2}))\sqrt{E(\bk)^2+z^2}}\right)<0\Bigg\}.
\end{align*}
The set $S$ is an open set of $\R^3$ and $f\in C^{\infty}(S)$. If
 $(x,y)\in\R_{>0}\times (-1,1)$ and $f(x,y,z(x,y))=0$, then
 $(x,y,z(x,y))\in S$. Observe that for any $(x,y,z)\in S$
 $\frac{\partial f}{\partial y}(x,y,z)<0$,  $\frac{\partial f}{\partial
 z}(x,y,z)<0$.
Thus by the implicit function theorem, $z(\cdot,\cdot)\in
 C^{\infty}(\R_{>0}\times (-1,1))$ and 
\begin{align*}
\frac{\partial z}{\partial y}(x,y)=-\frac{\frac{\partial f}{\partial
 y}(x,y,z(x,y))}{\frac{\partial f}{\partial z}(x,y,z(x,y))}<0,\quad
 (\forall (x,y)\in \R_{>0}\times (-1,1)).
\end{align*}
Fix $x\in \R_{>0}$. Since $y\mapsto z(x,y):(-1,1)\to\R_{>0}$ is monotone decreasing and
 bounded from below, $\lim_{y\nearrow 1}z(x,y)$ exists in $\R_{\ge
 0}$. We can take the limit $y\nearrow 1$ in the equality
 $f(x,y,z(x,y))=0$. Then by the uniqueness of the solution to the
 equation $f(x,1,z)=0$, $\lim_{y\nearrow 1}z(x,y)=z(x,1)$. Since 
$\lim_{z\to \infty}\sup_{y\in [-1,1]}f(x,y,z)=-\infty$, $y\mapsto
 z(x,y):(-1,1)\to\R_{>0}$ is bounded from above. Thus $\lim_{y\searrow
 -1}z(x,y)$ exists in $\R_{\ge 0}$. Since $\lim_{y\searrow -1}z(x,y)\ge
 z(x,1)>0$, we can take the limit $y\searrow -1$ in the equality
 $f(x,y,z(x,y))=0$ and by the uniqueness of the solution we conclude
 that $\lim_{y\searrow -1}z(x,y)=z(x,-1)$. Thus we have proved that 
\begin{align}
0<z(x,1)\le z(x,y)\le z(x,-1),\quad (\forall
 (x,y)\in\R_{>0}\times[-1,1]).\label{eq_gap_function_2nd_argument}
\end{align}
For $\delta \in \{1,-1\}$, set 
\begin{align*}
S_{\delta}:=\Bigg\{&(x,z)\in\R_{>0}\times \R_{>0}\ \Big|\\
&  -\frac{2}{|U|}+D_d\int_{\G_{\infty}^*}d\bk\Tr \left(\frac{\sinh(x
 \sqrt{E(\bk)^2+z^2})}{(\delta+\cosh(x\sqrt{E(\bk)^2+z^2}))\sqrt{E(\bk)^2+z^2}}\right)<0\Bigg\}.
\end{align*}
The set $S_{\delta}$ is an open set of $\R^2$,
 $f(\cdot,\delta,\cdot)\in C^{\infty}(S_{\delta})$ and $(x,z(x,\delta))\in
 S_{\delta}$ for any $x\in \R_{>0}$, $\delta\in \{1,-1\}$. Bearing in mind the fact that the
 functions 
\begin{align*}
x\mapsto \frac{\sinh x}{-1+\cosh x}:\R_{>0}\to \R,\quad
x\mapsto \frac{\sinh x}{1+\cosh x}:\R_{>0}\to \R
\end{align*}
are strictly monotone decreasing, strictly monotone increasing
 respectively, we see that $\frac{\partial f}{\partial
 x}(x,-1,z(x,-1))<0$, $\frac{\partial
 f}{\partial x}(x,1,z(x,1))>0$, $(\forall x\in\R_{>0})$.
As we considered in the proof of Lemma
 \ref{lem_gap_equation_solvability}, the functions 
\begin{align*}
x\mapsto \frac{\sinh x}{(\delta + \cosh x)x}:\R_{>0}\to\R,\quad
 (\delta\in \{1,-1\}) 
\end{align*}
are strictly monotone decreasing. Based on this fact, we can also verify
 that $\frac{\partial f}{\partial z}(x,\delta, z(x,\delta))<0$,
 $(\forall x\in\R_{>0},\ \delta \in \{1,-1\})$. Therefore by the
 implicit function theorem, $z(\cdot,\delta)\in C^{\infty}(\R_{>0})$
 $(\delta=1,-1)$ and 
\begin{align*}
&\frac{\partial z}{\partial x}(x,1)=-\frac{\frac{\partial f}{\partial
 x}(x,1,z(x,1))}{\frac{\partial f}{\partial
 z}(x,1,z(x,1))}>0,\quad
\frac{\partial z}{\partial x}(x,-1)=-\frac{\frac{\partial f}{\partial
 x}(x,-1,z(x,-1))}{\frac{\partial f}{\partial
 z}(x,-1,z(x,-1))}<0,\\
&(\forall x\in\R_{>0}),
\end{align*}
which implies that the functions $x\mapsto z(x,1):\R_{>0}\to \R_{>0}$, 
$x\mapsto z(x,-1):\R_{>0}\to \R_{>0}$ are strictly monotone increasing,
 strictly monotone decreasing respectively. Then by the boundedness
 \eqref{eq_gap_function_2nd_argument} we see that
$\lim_{x\to \infty}z(x,1)$, $\lim_{x\to
 \infty}z(x,-1)$ converge in $\R_{>0}$. Set
 $z_{\infty}(\delta):=\lim_{x\to
 \infty}z(x,\delta)$ for $\delta\in \{1,-1\}$. We
 can take the limit $x\to \infty$ in the equality
 $f(x,\delta,z(x,\delta))=0$ to derive that
\begin{align}
z_{\infty}(\delta)\left(-\frac{2}{|U|}+D_d\int_{\G_{\infty}^*}d\bk
 \Tr\left(\frac{1}{\sqrt{E(\bk)^2+z_{\infty}(\delta)^2}}\right)\right)=-\frac{2\g}{|U|}\label{eq_limit_equation_with_field}
\end{align}
for $\delta\in\{1,-1\}$. Since the solution to this equation is unique
 in $\R_{>0}$, we have that $z_{\infty}(1)=z_{\infty}(-1)$. We can read
 from \eqref{eq_gap_function_2nd_argument} that 
\begin{align*}
&z(\beta,1)\le
 z\left(\beta,\cos\left(\frac{\beta\theta(\beta)}{2}\right)\right)=a(\beta,\g)\le
 z(\beta,-1),\\
&(\forall \beta \in\R_{>0}\text{ with
 }\beta\theta/2\notin \pi (2\Z+1)).
\end{align*}
Thus it follows that $a(\beta,\g)$ converges to the unique positive
 solution of the equation \eqref{eq_limit_equation_with_field} as
 $\beta\to \infty$ with $\beta\in \R_{>0}$ satisfying
 $\beta\theta/2\notin \pi (2\Z+1)$.
Set $a_{\infty}(\g):=\lim_{\beta\to \infty,\beta\in \R_{>0}\text{ with
 }\frac{\beta\theta}{2}\notin \pi (2\Z+1)}a(\beta,\g)$. We can
 derive from \eqref{eq_limit_equation_with_field},
 \eqref{eq_dispersion_upper_lower_bound} and
 \eqref{eq_dispersion_measure_divided} that
$$
-\frac{2\g}{|U|}\le a_{\infty}(\g)\left(-\frac{2}{|U|}+D_db\sc\right),
$$ 
which combined with the inequality $|U|<2/(D_db\sc)$ implies that
 $$
\lim_{\g\searrow 0\atop\g\in (0,1]}a_{\infty}(\g)=0.
$$ 

It follows from \eqref{eq_2_point_correlation_with_field} that
\begin{align*}
&\lim_{\beta\to \infty,\beta\in \R_{>0}\atop \text{with
 }\frac{\beta\theta}{2}\notin \pi (2\Z+1)}\lim_{L\to \infty\atop L\in
 \N}\frac{\Tr (e^{-\beta (\sH+i\theta \sS_z+\sF)}\sA_1)}{\Tr e^{-\beta
 (\sH+i\theta \sS_z
 +\sF)}}\\
&=-a_{\infty}(\g)\frac{D_d}{2}\int_{\G_{\infty}^*}d\bk\frac{1}{\sqrt{E(\bk)^2+a_{\infty}(\g)^2}}(\hrho,\hrho).
\end{align*}
Then by recalling \eqref{eq_dispersion_measure_divided} and 
sending $\g$ to $0$ we reach the claimed equality.

\eqref{item_ODLRO_zero}: By Lemma \ref{lem_covariance_final_characterization},
 \eqref{eq_4_point_function_characterization},
 \eqref{eq_4_point_function_estimate_boundary} and
 \eqref{eq_gap_equation}
\begin{align}
&\limsup_{L\to \infty\atop L\in \N}\left|
\frac{\Tr (e^{-\beta(\sH+i\theta
 \sS_z)}\psi_{\hrho\hbx\ua}^*\psi_{\hrho\hbx\da}^*\psi_{\heta\hby\da}\psi_{\heta\hby\ua})}{\Tr e^{-\beta(\sH+i\theta \sS_z)}}\right|\label{eq_basic_inequality_for_ODLRO_zero}\\&\le 1_{(\hrho,\hbx)=(\heta,\hby)}\left|\lim_{L\to\infty\atop
 L\in\N}C(\D)(1\hrho\b0 0,1\hrho\b0 0)\right|+\frac{\D^2}{U^2}\notag\\
&\quad +\left|\lim_{L\to \infty\atop L\in\N}C(\D)(1\hrho\hbx 0,1\heta\hby 0)
C(\D)(2\heta\hby 0,2\hrho\hbx 0)\right|.\notag
\end{align}
Let us prove that 
\begin{align*}
\limsup_{\beta\to\infty,\beta\in \R_{>0}\atop\text{with
 }\frac{\beta\theta}{2}\notin \pi (2\Z+1)}\left\|
\lim_{L\to \infty\atop L\in \N}C(\D)(\cdot\hbx 0, \cdot\hby
 0)\right\|_{2b\times 2b}
\end{align*}
decays as $\|\hbx-\hby\|_{\R^d}\to\infty$. The $\beta$-dependent bound
 of the form \eqref{eq_full_covariance_total_decay} has no use
 here. Remind us the relation \eqref{eq_covariance_full_summation}. We
 have seen a $\beta$-independent decay property of
 $\sum_{l=N_{\beta}+1}^{\hat{N}_{\beta}}\cC_l$ in Lemma
 \ref{lem_covariance_construction} \eqref{item_covariance_construction_ODLRO}.
Let us establish a spatial decay property of 
$$
\limsup_{\beta\to\infty,\beta\in \R_{>0}\atop\text{with
 }\frac{\beta\theta}{2}\in \pi(2\Z+1)}\lim_{L\to \infty\atop L\in \N}
\lim_{h\to \infty\atop h\in
 \frac{2}{\beta}\N}\|\cC_{N_{\beta}}(\D)(\cdot\bx 0,\cdot \by
 0)\|_{2b\times 2b}.
$$
Take any $j\in \{1,2,\cdots,d\}$, $\bx,\by\in \G$. Assume that
 \eqref{eq_huge_h_condition} with $\phi=\D$ and
 \eqref{eq_huge_L_condition} hold so that we can use Lemma
 \ref{lem_properties_cutoff} and inequalities established in the proof
 of Lemma \ref{lem_real_covariances_properties}. By
 \eqref{eq_space_periodicity_application} for $l'=N_{\beta}$,
 \eqref{eq_integrand_bare_bound},
 \eqref{eq_integrand_momentum_derivative} and Lemma
 \ref{lem_properties_cutoff} \eqref{item_derivative_cutoff},\eqref{item_final_cutoff}
\begin{align*}
&\left\|\frac{L}{2\pi}(e^{-\frac{2\pi}{L}\<\bx-\by,\hbv_j\>}-1)\cC_{N_{\beta}}(\D)(\cdot\bx
 0, \cdot \by 0)\right\|_{2b\times 2b}\\
&\le \frac{c(d,M,\chi,\sc,\sa,(\hbv_j)_{j=1}^d)}{\beta L^d}\sup_{\bp\in
 \R^d}\sum_{\bk\in\G^*}1_{\chi_{N_{\beta}}(\frac{\pi}{\beta},\bk+\bp)\neq
 0}\\
&\quad\cdot \Bigg(M^{-\frac{N_{\beta}}{\sn_j}}\left\|h^{-1}(I_{2b}-e^{-\frac{i}{h}(\frac{\pi}{\beta}-\frac{\theta(\beta)}{2})I_{2b}+\frac{1}{h}E(\D)(\bk+\bp)})^{-1}\right\|_{2b\times
 2b}\\
&\qquad\quad +
\left\|\frac{\partial}{\partial \hat{k}_j}h^{-1}(I_{2b}-e^{-\frac{i}{h}(\frac{\pi}{\beta}-\frac{\theta(\beta)}{2})I_{2b}+\frac{1}{h}E(\D)(\bk+\bp)})^{-1}\right\|_{2b\times
 2b}\Bigg)\\
&\le c(d,M,\chi,\sc,\sa,(\hbv_j)_{j=1}^d)\Bigg(\frac{1}{L^d}
\sup_{\bp\in
 \R^d}\sum_{\bk\in\G^*}\left(\left(\frac{\pi}{\beta}-\frac{\theta(\beta)}{2}\right)^2+e(\bk+\bp)^2+\D^2\right)^{-\frac{1}{2}}\\
&\quad +\frac{1}{\beta L^d}\sup_{\bp\in
 \R^d}\sum_{\bk\in\G^*}1_{\chi_{N_{\beta}}(\frac{\pi}{\beta},\bk+\bp)\neq
 0}\left(\left(\frac{\pi}{\beta}-\frac{\theta(\beta)}{2}\right)^2+e(\bk+\bp)^2+\D^2\right)^{-\frac{1}{2}-\frac{1}{2\sn_j}}\Bigg).
\end{align*}
In the last inequality we also used \eqref{eq_beta_relation} and the
 assumption $\beta\ge 1$. By \eqref{eq_support_description} and the
 support property of $\chi(\cdot)$
\begin{align*}
1_{\chi_{N_{\beta}}(\frac{\pi}{\beta},\bk+\bp)\neq 0}\le
 \chi(2^{-1}M^{-N_{\beta}}A(\beta,M)^{-1}e(\bk+\bp)).
\end{align*}
By substituting this inequality and using periodicity and
 \eqref{eq_dispersion_measure_divided}, for any $\bx,\by\in\G_{\infty}$
\begin{align*}
&|\<\bx-\by,\hbv_j\>|\lim_{L\to \infty\atop L\in \N}
\lim_{h\to \infty\atop h\in
 \frac{2}{\beta}\N}\|\cC_{N_{\beta}}(\D)(\cdot\bx 0,\cdot \by
 0)\|_{2b\times 2b}\\
&\le c(d,M,\chi,\sc,\sa,(\hbv_j)_{j=1}^d)\Bigg(\int_{\G_{\infty}^*}d\bk
 \frac{1}{e(\bk)}\\
&\quad +\frac{1}{\beta}\int_{\G_{\infty}^*}d\bk \chi(2^{-1}M^{-N_{\beta}}A(\beta,M)^{-1}e(\bk))\left(\left(\frac{\pi}{\beta}-\frac{\theta(\beta)}{2}\right)^2+e(\bk)^2+\D^2\right)^{-\frac{1}{2}-\frac{1}{2\sn_j}}
\Bigg)\\
&\le c(d,M,\chi,\sc,\sa,(\hbv_j)_{j=1}^d)\\
&\quad\cdot \Bigg(1 + \frac{1}{\beta}\int_{\G_{\infty}^*}d\bk 1_{e(\bk)\le 4\beta^{-1}}
\left(\left(\frac{\pi}{\beta}-\frac{\theta(\beta)}{2}\right)^2+e(\bk)^2+\D^2\right)^{-\frac{1}{2}-\frac{1}{2\sn_j}}\Bigg).
\end{align*}
In the second inequality we also used the support property of
 $\chi(\cdot)$ and recalled the definition of $A(\beta,M)$. By the claim
 \eqref{item_gap_bound} of this corollary, if $e(\bk)\le 4\beta^{-1}$
 and $\beta$ is large, $(\pi/\beta-\theta(\beta)/2)^2+e(\bk)^2+\D^2\le
 (\pi^2+17)\beta^{-2}<1$. Thus by the condition $\sn_j\ge 1$ and
Lemma \ref{lem_gap_equation_solvability}, 
\begin{align*}
&|\<\bx-\by,\hbv_j\>|\limsup_{\beta\to\infty,\beta\in \R_{>0}\atop\text{with
 }\frac{\beta\theta}{2}\in \pi(2\Z+1)}\lim_{L\to \infty\atop L\in \N}
\lim_{h\to \infty\atop h\in
 \frac{2}{\beta}\N}\|\cC_{N_{\beta}}(\D)(\cdot\bx 0,\cdot \by
 0)\|_{2b\times 2b}\\
&\le c(d,M,\chi,\sc,\sa,(\hbv_j)_{j=1}^d)\\
&\quad\cdot \limsup_{\beta\to\infty,\beta\in \R_{>0}\atop\text{with
 }\frac{\beta\theta}{2}\in \pi(2\Z+1)}
\Bigg(1
+ \frac{1}{\beta}\int_{\G_{\infty}^*}d\bk 1_{e(\bk)\le 4\beta^{-1}}
\left(\left(\frac{\pi}{\beta}-\frac{\theta(\beta)}{2}\right)^2+e(\bk)^2+\D^2\right)^{-1}\Bigg)\\
&\le c(d,M,\chi,\sc,\sa,(\hbv_j)_{j=1}^d)\left(1+\limsup_{\beta\to\infty,\beta\in \R_{>0}\atop\text{with
 }\frac{\beta\theta}{2}\in \pi(2\Z+1)}\int_{\G_{\infty}^*}d\bk \Tr
 G_{\beta,\theta,\D}(\bk)\right)\\
&\le  c(d,M,\chi,\sc,\sa,(\hbv_j)_{j=1}^d)\left(1+\frac{1}{|U|}\right).
\end{align*}
To make clear, let us remark that the inequality
\begin{align*}
\int_{\G_{\infty}^*}d\bk \Tr G_{\beta,\theta,\D}(\bk)\le \frac{2}{D_d|U|}
\end{align*}
ensured by Lemma \ref{lem_gap_equation_solvability} 
and the definition of $\D$ was used. 
By combining the above inequality with Lemma
 \ref{lem_covariance_construction}
 \eqref{item_covariance_construction_ODLRO} and recalling
 \eqref{eq_covariance_full_summation} we have that for any
 $\bx,\by\in\G_{\infty}$
\begin{align*}
&\sum_{j=1}^d|\<\bx-\by,\hbv_j\>|\limsup_{\beta\to\infty,\beta\in \R_{>0}\atop\text{with }\frac{\beta\theta}{2}\in \pi(2\Z+1)}\left\|\lim_{L\to \infty\atop L\in \N}
C(\D)(\cdot\bx 0, \cdot \by 0)\right\|_{2b\times 2b}\\
&\le
 c(d,b,(\hbv_j)_{j=1}^d,\sa,(\sn_j)_{j=1}^{d},\sc,M,\chi)\left(1+\frac{1}{|U|}\right).
\end{align*}
Now coming back to \eqref{eq_basic_inequality_for_ODLRO_zero} and using
 the claim \eqref{item_gap_bound} of this corollary again, we conclude
 that
\begin{align*}
&\lim_{\|\hbx-\hby\|_{\R^d}\to\infty}\limsup_{\beta\to \infty,\beta\in
 \R_{>0}\atop \text{with }\frac{\beta\theta}{2}\notin \pi
 (2\Z+1)}\limsup_{L\to\infty\atop L\in \N}
\left|
\frac{\Tr (e^{-\beta(\sH+i\theta
 \sS_z)}\psi_{\hrho\hbx\ua}^*\psi_{\hrho\hbx\da}^*\psi_{\heta\hby\da}\psi_{\heta\hby\ua})}{\Tr
 e^{-\beta(\sH+i\theta \sS_z)}}\right|\\
&\le 
 c(d,b,(\hbv_j)_{j=1}^d,\sa,(\sn_j)_{j=1}^{d},\sc,M,\chi)\left(1+\frac{1}{|U|}\right)^2
\lim_{\|\hbx-\hby\|_{\R^d}\to\infty}\left(\sum_{j=1}^d|\<\hbx-\hby,\hbv_j\>|\right)^{-2}\\
&=0.
\end{align*}
\end{proof}

\appendix
\section{General lemmas for the infinite-volume limit}\label{app_infinite_volume_limit}

Here we state general lemmas which we use to take the infinite-volume
limit of the thermal expectations and the free energy density of our
many-electron systems. We use these lemmas in the proof of Theorem
\ref{thm_main_theorem} in Subsection \ref{subsec_proof_theorem}. The
first lemma enables us to take the infinite-volume limit of the thermal
expectation of the Cooper pair operator. 

\begin{lemma}\label{lem_infinite_volume_expectation}
Let $f_L$, $f\in C^{2}(\R^2,\R)$, $g_L$, $g$, $u_L$, $u\in
 C(\R^2,\C)$ $(L\in\N)$. Assume that the following conditions hold.
\begin{enumerate}[(i)]
\item For any non-empty compact set $Q$ of $\R^2$
\begin{align}
&\lim_{L\to\infty\atop L\in\N}\sup_{\bx\in Q}\left|
\frac{\partial^{i+j}}{\partial x_1^i\partial x_2^j}f_L(\bx)-
\frac{\partial^{i+j}}{\partial x_1^i\partial x_2^j}f(\bx)
\right|=0,\label{eq_uniform_convergence_up_to_2nd}\\
&\quad(\forall i,j\in \N\cup\{0\}\text{ satisfying }i+j\le 2),\notag\\
&\lim_{L\to \infty\atop L\in\N}\sup_{\bx\in Q}|g_L(\bx)-g(\bx)|=0,\quad
 \lim_{L\to \infty\atop L\in\N}\sup_{\bx\in Q}|u_L(\bx)-u(\bx)|=0.\notag 
\end{align}
\item 
\begin{align*}
\sup_{L\in \N}\sup_{\bx\in\R^2}|u_L(\bx)|<\infty,\quad
\sup_{L\in \N}\sup_{\bx\in\R^2}|g_L(\bx)|<\infty.
\end{align*}
\item There exist $R,\ c\in\R_{>0}$ such that 
\begin{align}
f_L(\bx)\le -c\|\bx\|_{\R^2},\quad (\forall \bx\in\R^2\text{ with
 }\|\bx\|_{\R^2}\ge R,\ L\in \N).\label{eq_L_independent_external_decay}
\end{align}
\item There exist $\ba_L,\ba\in\R^2$ $(L\in\N)$ such that 
\begin{align}
&f_L(\ba_L)>f_L(\bx),\quad (\forall \bx\in\R^{2}\backslash \{\ba_L\},\
 L\in\N),\label{eq_L_maximum}\\
&f(\ba)>f(\bx),\quad (\forall \bx\in\R^{2}\backslash
 \{\ba\}),\label{eq_limit_maximum}\\
&H(f)(\ba)<0,\label{eq_limit_hessian_negative}\\
&g(\ba)\neq 0.\notag
\end{align}
Here $H(f)(\bx)$ denotes the Hessian of $f$.
\end{enumerate}
Then
\begin{align*}
\lim_{L\to\infty\atop L\in\N}\frac{\int_{\R^2}d\bx
 e^{L^df_L(\bx)}u_L(\bx)}{\int_{\R^2}d\bx
 e^{L^df_L(\bx)}g_L(\bx)}=\frac{u(\ba)}{g(\ba)}.
\end{align*}
\end{lemma}

\begin{proof} The proof below is essentially a digest of the part
 concerning SSB of the proof of \cite[\mbox{Theorem 1.3}]{K_BCS}. By
 basic arguments based on the assumptions one can prove the following
 properties. 
\begin{itemize}
\item 
\begin{align}
\lim_{L\to\infty\atop L\in
 \N}\ba_L=\ba.\label{eq_argument_max_convergence}
\end{align}
\item There exist $\delta\in \R_{>0}$, $L_0\in \N$ such that for any
      $L\in \N$ with $L\ge L_0$,
\begin{align}
&f_L(\bx)=f_L(\ba_L)+\int_0^1dt(1-t)\<\bx-\ba_L,H(f_L)(t(\bx-\ba_L)+\ba_L)(\bx-\ba_L)\>,\label{eq_taylor_expansion_2nd}\\
&\quad
 (\forall \bx\in
 \overline{B_{\delta}(\ba_L)}),\notag\\
&H(f_L)(t(\bx-\ba_L)+\ba_L)\le \frac{1}{2}H(f)(\ba),\quad (\forall
 \bx\in \overline{B_{\delta}(\ba_L)},\ t\in
 [0,1]),\label{eq_hessian_bound}\\
&f_L(\bx)-f_L(\ba_L)\le \frac{1}{2}\sup_{\by\in \R^2\backslash
 \overline{B_{\delta/2}(\ba)}}(f(\by)-f(\ba))<0,\quad (\forall \bx\in \R^2\backslash\overline{B_{\delta}(\ba_L)}).\label{eq_maximum_application}
\end{align}
\end{itemize}
Here $B_r(\bb)$ denotes $\{\bx\in\R^2\ |\ \|\bx-\bb\|_{\R^2}<r \}$ for
      $\bb\in\R^2$, $r\in \R_{>0}$. 
In fact by using \eqref{eq_uniform_convergence_up_to_2nd} with $i=j=0$,
 \eqref{eq_L_independent_external_decay}, \eqref{eq_L_maximum},
 \eqref{eq_limit_maximum} we can prove
 \eqref{eq_argument_max_convergence}. Taylor's theorem gives
 \eqref{eq_taylor_expansion_2nd} for any $\delta \in\R_{>0}$ and
 $L\in\N$. Then by
 \eqref{eq_uniform_convergence_up_to_2nd} with $(i,j)$ satisfying
 $i+j=2$, \eqref{eq_limit_hessian_negative},
 \eqref{eq_argument_max_convergence} and the continuity of the 2nd order
 derivatives of $f$ we can prove \eqref{eq_hessian_bound} with some
 $\delta$ and any $L\in\N$ satisfying $L\ge L_0$ for some $L_0$. For the
 fixed $\delta$ the property \eqref{eq_argument_max_convergence} ensures that 
$\|\ba-\ba_L\|_{\R^2}\le \delta/2$ for any $L\in\N$ satisfying $L\ge
 L_0$, if we take $L_0$ larger if necessary. 
This implies that $\R^2\backslash \overline{B_{\delta}(\ba_L)}\subset
 \R^2\backslash \overline{B_{\delta/2}(\ba)}$. 
Then for the fixed $\delta$, by taking $L_0$ larger if necessary
we can apply \eqref{eq_uniform_convergence_up_to_2nd} with
 $i+j=0$, \eqref{eq_L_independent_external_decay},
 \eqref{eq_limit_maximum}, \eqref{eq_argument_max_convergence} and the
 continuity of $f$ to prove \eqref{eq_maximum_application}. 

For any $L\in\N$ with $L\ge L_0$,
\begin{align*}
&\int_{\R^2}d\bx e^{L^df_L(\bx)}g_L(\bx)\\
&= e^{L^df_L(\ba_L)}L^{-d}\Bigg(
\int_{\overline{B_{L^{\frac{d}{2}}\delta}(\b0)}}d\bx e^{\int_0^1dt
 (1-t)\<\bx,H(f_L)(t
 L^{-\frac{d}{2}}\bx+\ba_L)\bx\>}g_L(L^{-\frac{d}{2}}\bx+\ba_L)\\
&\qquad\qquad\qquad\qquad  +
 L^d\int_{\R^2\backslash\overline{B_{\delta}(\ba_L)}}d\bx
 e^{L^d(f_L(\bx)-f_L(\ba_L))}g_L(\bx)\Bigg).
\end{align*}
It follows from the assumptions, the properties listed above and the
 dominated convergence theorem that
\begin{align*}
&\lim_{L\to \infty\atop L\in
 \N}\Bigg(
\int_{\overline{B_{L^{\frac{d}{2}}\delta}(\b0)}}d\bx e^{\int_0^1dt
 (1-t)\<\bx,H(f_L)(t
 L^{-\frac{d}{2}}\bx+\ba_L)\bx\>}g_L(L^{-\frac{d}{2}}\bx+\ba_L)\\
&\qquad\qquad +
 L^d\int_{\R^2\backslash\overline{B_{\delta}(\ba_L)}}d\bx
 e^{L^d(f_L(\bx)-f_L(\ba_L))}g_L(\bx)\Bigg)\\
&=g(\ba)\int_{\R^2}d\bx e^{\frac{1}{2}\<\bx,H(f)(\ba)\bx\>}\neq 0.
\end{align*}
The numerator can be dealt in the same way. As the result,
\begin{align*}
\lim_{L\to \infty\atop L\in \N}\frac{\int_{\R^2}d\bx e^{L^2f_L(\bx)}u_L(\bx)}{\int_{\R^2}d\bx e^{L^2f_L(\bx)}g_L(\bx)}=\frac{u(\ba)\int_{\R^2}d\bx
 e^{\frac{1}{2}\<\bx,H(f)(\ba)\bx\>}}{g(\ba)\int_{\R^2}d\bx
 e^{\frac{1}{2}\<\bx,H(f)(\ba)\bx\>}}=\frac{u(\ba)}{g(\ba)}.
\end{align*}
\end{proof}

Next let us prove a lemma which is used to prove the existence of the
infinite-volume limit of the correlation function in the case that the
physical parameters are not on the phase boundary. 

\begin{lemma}\label{lem_infinite_volume_correlation}
Let $f_L$, $f\in C^{2}(\R,\R)$, $g_L$, $g$, $u_L$, $u\in C(\R,\C)$
 $(L\in\N)$. Assume that the following conditions hold. 
\begin{enumerate}[(i)]
\item For any $r\in \R_{>0}$
\begin{align*}
&\lim_{L\to\infty\atop L\in\N}\sup_{x\in [-r,r]}\left|
\frac{d^{i}}{dx^i}f_L(x)- \frac{d^{i}}{dx^i}f(x)\right|=0,\quad (\forall
 i\in \{0,1,2\}),\\
&\lim_{L\to \infty\atop L\in\N}\sup_{x\in [-r,r]}|u_L(x)-u(x)|=0,\quad
\lim_{L\to \infty\atop L\in\N}\sup_{x\in [-r,r]}|g_L(x)-g(x)|=0.
 \end{align*}
\item 
\begin{align*}
\sup_{L\in \N}\sup_{x\in \R}|u_L(x)|<\infty,\quad 
\sup_{L\in \N}\sup_{x\in \R}|g_L(x)|<\infty.
\end{align*}
\item There exist $R,\ c\in\R_{>0}$ such that 
\begin{align*}
f_L(x)\le -c |x|,\quad (\forall x\in\R\text{ with
 }|x|\ge R,\ L\in \N).
\end{align*}
\item There exist $a_L,a\in\R_{\ge 0}$ $(L\in\N)$ such that 
\begin{align*}
&f_L(a_L)>f_L(x),\quad (\forall x\in\R_{\ge 0}\backslash \{a_L\},\
 L\in\N),\\
&f(a)>f(x),\quad (\forall x\in\R_{\ge 0}\backslash
 \{a\}),\\
&\frac{d}{dx}f_L(a_L)=0,\quad (\forall L\in \N),\\
&\frac{d^2}{dx^2}f(a)<0,\quad g(a)\neq 0.
\end{align*}
Moreover if $a=0$, there exists $L_0\in \N$ such that $a_L=0$ ($\forall
      L\in \N$ with $L\ge L_0$).
\end{enumerate}
Then
\begin{align}
\lim_{L\to\infty\atop L\in\N}\frac{\int_0^{\infty}dx x
 e^{L^df_L(x)}u_L(x)}{\int_{0}^{\infty}dx x 
 e^{L^df_L(x)}g_L(x)}=\frac{u(a)}{g(a)}.\label{eq_1_dimensional_expectation}
\end{align}
\end{lemma}

\begin{proof}
The following argument is a generalization of the part concerning ODLRO
 of the proof of \cite[\mbox{Theorem 1.3}]{K_BCS}. The assumptions
 imply the following statements.
\begin{itemize}
\item 
$$
\lim_{L\to \infty\atop L\in\N}a_L=a.
$$
\item There exist $\delta\in\R_{>0}$, $L_1\in\N$ such that for any $L\in
      \N$ with $L\ge L_1$,
\begin{align*}
&f_L(x)=f_L(a_L)+\int_0^1dt (1-t)\frac{d^2
 f_L}{dx^2}(t(x-a_L)+a_L)(x-a_L)^2,\\
&\quad (\forall x\in [a_L-\delta,a_L+\delta]),\\
&\frac{d^2 f_L}{dx^2}(t(x-a_L)+a_L)\le
 \frac{1}{2}\frac{d^2f}{dx^2}(a)<0,\quad (\forall x\in
 [a_L-\delta,a_L+\delta],\ t\in [0,1]),\\
&f_L(x)-f_L(a_L)\le \frac{1}{2}\sup_{y\in\R_{\ge
 0}\backslash[a-\frac{\delta}{2},a+\frac{\delta}{2}]}(f(y)-f(a))<0,\\
&\quad
 (\forall x\in\R_{\ge 0}\backslash[a_L-\delta,a_L+\delta]).
\end{align*}
\end{itemize}
Then let us observe that for $L\in \N$ satisfying $L\ge
 1_{a>0}L_1+1_{a=0}\max\{L_0,L_1\}$,
\begin{align*}
&\int_0^{\infty}dx x e^{L^df_L(x)}g_L(x)\\
&=e^{L^df_L(a_L)}\\
&\quad\cdot \Bigg(L^{-\frac{d}{2}}\int_{-L^{\frac{d}{2}}\min\{a_L,\delta\}}^{L^{\frac{d}{2}}\delta}dx
 (L^{-\frac{d}{2}}x+a_L)e^{\int_0^1dt
 (1-t)f_L^{(2)}(tL^{-\frac{d}{2}}x+a_L)x^2}g_L(L^{-\frac{d}{2}}x+a_L)\\
&\qquad\quad + \int_{\R_{\ge 0}\backslash[a_L-\delta,a_L+\delta]}dx x
 e^{L^d(f_L(x)-f_L(a_L))}g_L(x)\Bigg)\\
&=e^{L^df_L(a_L)}(1_{a>0}L^{-\frac{d}{2}}+1_{a=0}L^{-d})\\
&\quad\cdot \Bigg(\int_{-L^{\frac{d}{2}}\min\{a_L,\delta\}}^{L^{\frac{d}{2}}\delta}dx(1_{a>0}(L^{-\frac{d}{2}}x+a_L)+1_{a=0}x)
e^{\int_0^1dt
 (1-t)f_L^{(2)}(tL^{-\frac{d}{2}}x+a_L)x^2}\\
&\qquad\qquad\qquad\qquad\cdot g_L(L^{-\frac{d}{2}}x+a_L)\\
&\qquad\quad +(1_{a>0}L^{\frac{d}{2}}+1_{a=0}L^{d})\int_{\R_{\ge
 0}\backslash[a_L-\delta,a_L+\delta]}dx x
 e^{L^d(f_L(x)-f_L(a_L))}g_L(x)\Bigg).
\end{align*}
By using the above properties we can take the limit $L\to\infty$ as
 follows.
\begin{align*}
&\lim_{L\to \infty\atop L\in\N}\Bigg(
\int_{-L^{\frac{d}{2}}\min\{a_L,\delta\}}^{L^{\frac{d}{2}}\delta}dx(1_{a>0}(L^{-\frac{d}{2}}x+a_L)+1_{a=0}x)
e^{\int_0^1dt
 (1-t)f_L^{(2)}(tL^{-\frac{d}{2}}x+a_L)x^2}\\
&\qquad\qquad\qquad\qquad\cdot g_L(L^{-\frac{d}{2}}x+a_L)\\
&\qquad\quad +(1_{a>0}L^{\frac{d}{2}}+1_{a=0}L^{d})\int_{\R_{\ge
 0}\backslash[a_L-\delta,a_L+\delta]}dx x
 e^{L^d(f_L(x)-f_L(a_L))}g_L(x)\Bigg)\\
&=1_{a>0}a g(a) \int_{-\infty}^{\infty}dx
 e^{\frac{1}{2}f^{(2)}(a)x^2}
+1_{a=0} g(a) \int_{0}^{\infty}dx
 e^{\frac{1}{2}f^{(2)}(a)x^2}\neq 0.
\end{align*}
The numerator can be decomposed and sent to the limit in
 the same way. Consequently,
\begin{align*}
&\lim_{L\to\infty\atop L\in\N}\frac{\int_0^{\infty}dx x
 e^{L^df_L(x)}u_L(x)}{\int_{0}^{\infty}dx x 
 e^{L^df_L(x)}g_L(x)}\\
&=\Bigg(1_{a>0}au(a)\int_{-\infty}^{\infty}dx
 e^{\frac{1}{2}f^{(2)}(a)x^2}+1_{a=0}u(a)\int_{0}^{\infty}dx
x e^{\frac{1}{2}f^{(2)}(a)x^2}\Bigg)\\
&\quad\cdot \Bigg(1_{a>0}ag(a)\int_{-\infty}^{\infty}dx
 e^{\frac{1}{2}f^{(2)}(a)x^2}+1_{a=0}g(a)\int_{0}^{\infty}dx
x e^{\frac{1}{2}f^{(2)}(a)x^2}\Bigg)^{-1}\\
&=\frac{u(a)}{g(a)}.
\end{align*}
\end{proof}

We need to estimate the correlation function in the case that the
physical parameters are on the phase boundary. We need the next lemma
for the purpose.

\begin{lemma}\label{lem_infinite_volume_correlation_estimate}
Let $n_0\in 2\N$, $f_L$, $f\in C^{n_0}(\R,\R)$, $g_L\in C(\R,\R)$, $u_L$, $u\in
 C(\R,\C)$ $(L\in\N)$. Assume that the following
 conditions hold. 
\begin{enumerate}[(i)]
\item For any $r\in \R_{>0}$
\begin{align*}
&\lim_{L\to\infty\atop L\in\N}\sup_{x\in [-r,r]}\left|
\frac{d^{n}}{dx^n}f_L(x)- \frac{d^{n}}{dx^n}f(x)\right|=0,\quad (\forall
 n\in \{0,1,\cdots,n_0\}),\\
&\lim_{L\to \infty\atop L\in\N}\sup_{x\in [-r,r]}|u_L(x)-u(x)|=0.
 \end{align*}
\item 
\begin{align*}
\sup_{L\in \N}\sup_{x\in \R}|u_L(x)|<\infty,\quad 
\sup_{L\in \N}\sup_{x\in \R}|g_L(x)|<\infty.
\end{align*}
\item There exist $R,\ c\in\R_{>0}$ such that 
\begin{align*}
f_L(x)\le -c |x|,\quad (\forall x\in\R\text{ with
 }|x|\ge R,\ L\in \N).
\end{align*}
\item There exist $a_L,a\in\R_{\ge 0}$ $(L\in\N)$ such that 
\begin{align*}
&f_L(a_L)>f_L(x),\quad (\forall x\in\R_{\ge 0}\backslash \{a_L\},\
 L\in\N),\\
&f(a)>f(x),\quad (\forall x\in\R_{\ge 0}\backslash
 \{a\}),\\
&\frac{d^n}{dx^n}f(a)=0,\quad (\forall n\in \{1,2,\cdots,n_0-1\}),\\
&\frac{d^{n_0}}{dx^{n_0}}f(a)<0.
\end{align*}
\item 
$$
\inf_{L\in\N}\inf_{x\in\R}g_L(x)>0.
$$
\end{enumerate}
Then
\begin{align*}
\limsup_{L\to\infty\atop L\in\N}\left|\frac{\int_0^{\infty}dx x
 e^{L^df_L(x)}g_L(x)u_L(x)}{\int_{0}^{\infty}dx x 
 e^{L^df_L(x)}g_L(x)}\right|\le \sup_{x\in[0,a]}|u(x)|.
\end{align*}
\end{lemma}

\begin{proof}
Since $f^{(n_0)}(a)<0$ and
 $f^{(n_0)}(\cdot)$ is continuous, there exists
 $\delta_0\in\R_{>0}$ such that 
\begin{align}
f^{(n_0)}(x)\le
 \frac{2}{3}f^{(n_0)}(a),\quad (\forall x\in
 [a-\delta_0,a+\delta_0]).\label{eq_intermediate_hessian_bound}
\end{align}
Take any $\delta\in (0,\delta_0)$. By using the assumptions and
 \eqref{eq_intermediate_hessian_bound} we can prove the following
 statements.
\begin{itemize}
\item 
$$
\lim_{L\to \infty\atop L\in\N}a_L=a.
$$
\item There exists $L_0\in\N$ such that for any $L\in
      \N$ with $L\ge L_0$,
\begin{align*}
&f_L(x)=\sum_{n=0}^{n_0-1}\frac{1}{n!}f_L^{(n)}(a_L)(x-a_L)^n\\
&\qquad\qquad+\int_0^1dt \frac{(1-t)^{n_0-1}}{(n_0-1)!}
f_L^{(n_0)}(t(x-a_L)+a_L)(x-a_L)^{n_0},\\
&\quad (\forall x\in [a_L-\delta,a_L+\delta]),\\
&f_L^{(n_0)}(t(x-a_L)+a_L)\le
 \frac{1}{2}f^{(n_0)}(a)<0,\quad (\forall x\in
 [a_L-\delta,a_L+\delta],\ t\in [0,1]),\\
&f_L(x)-f_L(a_L)\le \frac{1}{2}\sup_{y\in\R_{\ge
 0}\backslash[a-\frac{\delta}{2},a+\frac{\delta}{2}]}(f(y)-f(a))<0,\\
&\quad
 (\forall x\in\R_{\ge 0}\backslash[a_L-\delta,a_L+\delta]).
\end{align*}
\item 
\begin{align*}
\lim_{L\to\infty\atop L\in \N}f_L^{(n)}(a_L)=
 f^{(n)}(a)=0,\quad (\forall n\in \{1,2,\cdots,n_0-1\}).
\end{align*}
\end{itemize}
Let us observe that 
\begin{align}
&\int_0^{\infty}dx x e^{L^df_L(x)}u_L(x)g_L(x)\label{eq_correlation_estimate_typical}\\
&=e^{L^df_L(a_L)} L^{-\frac{d}{n_0}}\notag\\
&\quad\cdot\Bigg(\int_{-L^{\frac{d}{n_0}}\min\{a_L,\delta\}}^{L^{\frac{d}{n_0}}\delta}dx
 (L^{-\frac{d}{n_0}}x+a_L)\notag\\
&\qquad\qquad\cdot e^{\sum_{n=1}^{n_0-1}\frac{1}{n!}L^{(1-\frac{n}{n_0})d}f_L^{(n)}(a_L)x^n+\int_0^1dt
 \frac{1}{(n_0-1)!}(1-t)^{n_0-1}f_L^{(n_0)}(tL^{-\frac{d}{n_0}}x+a_L)x^{n_0}}\notag\\
&\qquad\qquad\cdot u_L(L^{-\frac{d}{n_0}}x+a_L)
g_L(L^{-\frac{d}{n_0}}x+a_L)\notag\\
&\qquad\quad + L^{\frac{d}{n_0}}\int_{\R_{\ge 0}\backslash[a_L-\delta,a_L+\delta]}dx x
 e^{L^d(f_L(x)-f_L(a_L))}u_L(x)g_L(x)\Bigg).\notag
\end{align}
By decomposing the denominator in the same way as above we can derive that
\begin{align*}
&\left|\frac{\int_0^{\infty}dx x
 e^{L^df_L(x)}g_L(x)u_L(x)}{\int_{0}^{\infty}dx x 
 e^{L^df_L(x)}g_L(x)}\right|\\
&\le \Bigg(\int_{-L^{\frac{d}{n_0}}\min\{a_L,\delta\}}^{L^{\frac{d}{n_0}}\delta}dx
 (L^{-\frac{d}{n_0}}x+a_L)\notag\\
&\qquad\qquad\cdot e^{\sum_{n=1}^{n_0-1}\frac{1}{n!}L^{(1-\frac{n}{n_0})d}f_L^{(n)}(a_L)x^n+\int_0^1dt
 \frac{1}{(n_0-1)!}(1-t)^{n_0-1}f_L^{(n_0)}(tL^{-\frac{d}{n_0}}x+a_L)x^{n_0}}\notag\\
&\qquad\qquad\cdot |u_L(L^{-\frac{d}{n_0}}x+a_L)|
g_L(L^{-\frac{d}{n_0}}x+a_L)\notag\\
&\qquad\quad + L^{\frac{d}{n_0}}\int_{\R_{\ge 0}\backslash[a_L-\delta,a_L+\delta]}dx x
 e^{L^d(f_L(x)-f_L(a_L))}|u_L(x)|g_L(x)\Bigg)\\
&\quad\cdot \Bigg(
\int_{-L^{\frac{d}{n_0}}\min\{a_L,\delta\}}^{L^{\frac{d}{n_0}}\delta}dx
 (L^{-\frac{d}{n_0}}x+a_L)\notag\\
&\qquad\qquad\cdot e^{\sum_{n=1}^{n_0-1}\frac{1}{n!}L^{(1-\frac{n}{n_0})d}f_L^{(n)}(a_L)x^n+\int_0^1dt
 \frac{1}{(n_0-1)!}(1-t)^{n_0-1}f_L^{(n_0)}(tL^{-\frac{d}{n_0}}x+a_L)x^{n_0}}\notag\\
&\qquad\qquad\cdot g_L(L^{-\frac{d}{n_0}}x+a_L)\Bigg)^{-1}\\
&\le \sup_{x\in [0,\delta+a_L]}|u_L(x)|\\
&\quad + \Bigg(L^{2\frac{d}{n_0}}e^{\sum_{n=1}^{n_0-1}\frac{L^d}{n!}|f_L^{(n)}(a_L)|\delta^n}\int_{\R_{\ge 0}\backslash[a_L-\delta,a_L+\delta]}dx x
e^{L^d(f_L(x)-f_L(a_L))}|u_L(x)|g_L(x)\Bigg)\\
&\qquad\quad\cdot \Bigg(\int_{0}^{L^{\frac{d}{n_0}}\delta}dx x 
e^{\int_0^1dt
 \frac{1}{(n_0-1)!}(1-t)^{n_0-1}f_L^{(n_0)}(tL^{-\frac{d}{n_0}}x+a_L)x^{n_0}}
g_L(L^{-\frac{d}{n_0}}x+a_L)\Bigg)^{-1}\\
&\le \sup_{x\in [0,\delta+a_L]}|u_L(x)|\\
&\quad +\sup_{L\in
 \N}\sup_{x\in\R}|u_L(x)g_L(x)|\left(\inf_{L\in\N}\inf_{x\in\R}g_L(x)\right)^{-1}\\
&\qquad\cdot \Bigg(L^{2\frac{d}{n_0}}e^{\sum_{n=1}^{n_0-1}\frac{L^d}{n!}|f_L^{(n)}(a_L)|\delta^n}\int_{\R_{\ge
 0}\backslash[a_L-\delta,a_L+\delta]}dx x
 e^{L^d(f_L(x)-f_L(a_L))}\Bigg)\\
&\qquad\cdot 
\Bigg(\int_{0}^{L^{\frac{d}{n_0}}\delta}dx x 
e^{\int_0^1dt
 \frac{1}{(n_0-1)!}(1-t)^{n_0-1}f_L^{(n_0)}(tL^{-\frac{d}{n_0}}x+a_L)x^{n_0}}\Bigg)^{-1}.
\end{align*}
Then by using the properties listed in the beginning of the proof we can
 deduce that 
\begin{align*}
&\limsup_{L\to\infty\atop L\in\N}\left|\frac{\int_0^{\infty}dx x
 e^{L^df_L(x)}g_L(x)u_L(x)}{\int_{0}^{\infty}dx x 
 e^{L^df_L(x)}g_L(x)}\right|\le \lim_{L\to\infty\atop L\in\N}\sup_{x\in
 [0,2\delta+a]}|u_L(x)|=\sup_{x\in [0,2\delta+a]}|u(x)|.
\end{align*}
The arbitrariness of $\delta$ implies the result.
\end{proof}

Finally let us show a lemma which ensures the convergence of the free
energy density in the infinite-volume limit when it is applied in
practice.

\begin{lemma}\label{lem_infinite_volume_logarithm}
Let $n_0\in 2\N$, 
$f_L$, $f\in C^{n_0}(\R,\R)$, $g_L \in C(\R,\R)$, $(L\in
 \N)$. Assume that these functions satisfy the same conditions as in
 Lemma \ref{lem_infinite_volume_correlation_estimate}. Then
\begin{align*}
\lim_{L\to\infty\atop L\in \N}\frac{1}{L^d}\log\left(\int_0^{\infty}dx x
 e^{L^df_L(x)}g_L(x)\right)=f(a).
\end{align*}
\end{lemma}

\begin{proof}
Since the assumptions are same, we can transform the integral inside the
 logarithm in the same way as in the proof of Lemma
 \ref{lem_infinite_volume_correlation_estimate}. We use the
 following equality close to
 \eqref{eq_correlation_estimate_typical}. For $\eps\in \{1,-1\}$
\begin{align*}
&\int_0^{\infty}dx x e^{L^df_L(x)}g_L(x)\\
&=L^{-\frac{d}{n_0}}e^{L^df_L(a_L)
+\eps\sum_{n=1}^{n_0-1}\frac{L^d}{n!}|f_L^{(n)}(a_L)|\delta^n}\\
&\quad\cdot \Bigg(\int_{-L^{\frac{d}{n_0}}\min\{a_L,\delta\}}^{L^{\frac{d}{n_0}}\delta}dx
 (L^{-\frac{d}{n_0}}x+a_L)
 e^{\sum_{n=1}^{n_0-1}\frac{L^d}{n!}(L^{-\frac{n}{n_0}d}f_L^{(n)}(a_L)x^n -\eps|f_L^{(n)}(a_L)|\delta^n)}\\
&\qquad\qquad\qquad\qquad\qquad\cdot
e^{\int_0^1dt
 \frac{1}{(n_0-1)!}(1-t)^{n_0-1}f_L^{(n_0)}(tL^{-\frac{d}{n_0}}x+a_L)x^{n_0}}g_L(L^{-\frac{d}{n_0}}x+a_L)\notag\\
&\qquad\quad + L^{\frac{d}{n_0}}\int_{\R_{\ge 0}\backslash[a_L-\delta,a_L+\delta]}dx x
 e^{-\eps \sum_{n=1}^{n_0-1} \frac{L^d}{n!}|f_L^{(n)}(a_L)|\delta^n+
L^d(f_L(x)-f_L(a_L))}g_L(x)\Bigg).
\end{align*}
By taking $\eps=1$ this implies that 
\begin{align}
&\frac{1}{L^d}\log\left(\int_0^{\infty}dx x e^{L^df_L(x)}g_L(x)\right)\label{eq_1_dimensional_logarithm_right}\\
&\le f_L(a_L)+\sum_{n=1}^{n_0-1}\frac{1}{n!}|f_L^{(n)}(a_L)|\delta^n\notag\\
&\quad +\frac{1}{L^d}\log\Bigg(
\int_{-L^{\frac{d}{n_0}}\min\{a_L,\delta\}}^{L^{\frac{d}{n_0}}\delta}dx
 (L^{-\frac{d}{n_0}}x+a_L)
e^{\int_0^1dt
 \frac{1}{(n_0-1)!}(1-t)^{n_0-1}f_L^{(n_0)}(tL^{-\frac{d}{n_0}}x+a_L)x^{n_0}}\notag\\
&\qquad\qquad\qquad\qquad\qquad\qquad\cdot g_L(L^{-\frac{d}{n_0}}x+a_L)\notag\\
&\qquad\qquad\qquad +\int_{\R_{\ge 0}\backslash[a_L-\delta,a_L+\delta]}dx x
 e^{L^d(f_L(x)-f_L(a_L))}g_L(x)\Bigg)\notag\\
&\le
 f_L(a_L)+\sum_{n=1}^{n_0-1}\frac{1}{n!}|f_L^{(n)}(a_L)|\delta^n\notag\\
&\quad +\frac{1}{L^d}\log\Bigg((\delta+a_L)\sup_{L\in\N}\sup_{y\in
 \R}g_L(y)\notag\\
&\qquad\qquad\qquad\cdot
\int_{-L^{\frac{d}{n_0}}\delta}^{L^{\frac{d}{n_0}}\delta}dx
 e^{\int_0^1dt
 \frac{1}{(n_0-1)!}(1-t)^{n_0-1}f_L^{(n_0)}(tL^{-\frac{d}{n_0}}x+a_L)x^{n_0}}\notag\\
&\qquad\qquad\qquad +\sup_{L\in\N}\sup_{y\in
 \R}g_L(y) \int_{\R_{\ge 0}\backslash[a_L-\delta,a_L+\delta]}dx x
 e^{L^d(f_L(x)-f_L(a_L))}\Bigg).\notag
\end{align}
On the other hand, by taking $\eps=-1$
\begin{align}
&\frac{1}{L^d}\log\left(\int_0^{\infty}dx x
 e^{L^df_L(x)}g_L(x)\right)\label{eq_1_dimensional_logarithm_left}\\
&\ge
 f_L(a_L)-\sum_{n=1}^{n_0-1}\frac{1}{n!}|f_L^{(n)}(a_L)|\delta^n+\frac{1}{L^d}\log(L^{-\frac{2d}{n_0}})\notag\\
&\quad +\frac{1}{L^d}\log\Bigg(\inf_{L\in\N}\inf_{y\in
 \R}g_L(y)
\int_{0}^{L^{\frac{d}{n_0}}\delta}dx x
 e^{\int_0^1dt \frac{1}{(n_0-1)!}(1-t)^{n_0-1}f_L^{(n_0)}(tL^{-\frac{d}{n_0}}x+a_L)x^{n_0}}\Bigg).\notag
\end{align}
By using the properties listed in the beginning of the proof of Lemma
 \ref{lem_infinite_volume_correlation_estimate} we can show that both
 the right-hand side of \eqref{eq_1_dimensional_logarithm_right} and
 that of \eqref{eq_1_dimensional_logarithm_left} converge to $f(a)$ as
 $L\to \infty$.
\end{proof}

\section*{Acknowledgments}
This work was supported by JSPS KAKENHI Grant Number 
26870110.

\section*{Supplementary List of Notations}
\subsection*{Parameters and constants}
\begin{center}
\begin{tabular}{lll}
Notation & Description & Reference \\
\hline
$b$ & number of sites in unit cell &  Subsection
 \ref{subsec_model_theorem}\\
$\sc$ & positive constant $(\ge 1)$ appearing in bounds &  Subsection
 \ref{subsec_model_theorem}\\
 & on $E(\cdot)$ and $e(\cdot)$  &  \\ 
$\sn_j$  & positive numbers $(\in\N)$ appearing in bounds 
  & \eqref{eq_dispersion_derivative},
 \eqref{eq_hopping_matrix_derivative}\\
$(j=1,\cdots,d)$ &  on derivatives of $e(\cdot)^2$ and $E(\cdot)$ & \\
$\sa$ & positive constant $(>1)$ appearing in bounds &
 \eqref{eq_dispersion_measure}, \eqref{eq_dispersion_measure_divided} \\
 & on integrals of $e(\cdot)$ & \\
$D_d$ & $|\det(\hbv_1,\hbv_2,\cdots,\hbv_d)|^{-1}(2\pi)^{-d}$ &
 Subsection \ref{subsec_model_theorem}\\
$\theta(\beta)$ & projection of $\theta$ to $[0,\frac{2\pi}{\beta})$ & beginning
 of \\
  & & Section
 \ref{sec_formulation}\\
$N$ & $4b\beta h L^d$, cardinality of $I$ & beginning of\\
 & &  Subsection \ref{subsec_general_estimation}\\
$\hat{N}_{\beta}$ & largest scale in IR integration & Subsection
 \ref{subsec_generalized_covariances} \\
  & &  and beginning of \\
 & & Subsection
 \ref{subsec_real_covariance}\\ 
 $N_{\beta}$ & smallest scale in IR integration & Subsection
 \ref{subsec_generalized_covariances} \\
 & &  and beginning of \\
 & & Subsection \ref{subsec_real_covariance}\\
 $M$ & parameter to control support size  of cut-off& Subsection
 \ref{subsec_generalized_covariances}\\
$c_0$ & positive constant $(\ge 1)$ appearing in bounds &  Subsection
 \ref{subsec_generalized_covariances}\\
 & on scale-dependent covariances & \\
$c_{end}$ & positive constant appearing in $\|\cdot\|_{1,\infty}$-norm
  & \eqref{eq_scale_covariance_decay_bound} \\
 & bound of covariance of scale $N_{\beta}$ &  
\end{tabular}
\end{center}

\newpage
\subsection*{Sets and spaces}
\begin{center}
\begin{tabular}{lll}
Notation & Description & Reference \\
\hline
$\G$ & $\left\{\sum_{j=1}^dm_j\bv_j\ \Big|\ \begin{array}{l} m_j\in \{0,1,\cdots,L-1\},\\(j=1,2,\cdots,d)\end{array}\right\}$ & Subsection \ref{subsec_model_theorem}\\
$\G^*$ & $\left\{ \sum_{j=1}^d\hat{m}_j\hbv_j\ \Big|\
 \begin{array}{l} \hat{m}_j\in \{0,\frac{2\pi}{L},\cdots,2\pi-\frac{2\pi}{L}\},\\
 (j=1,2,\cdots,d)\end{array}
\right\}$ & Subsection \ref{subsec_model_theorem}\\
$\G_{\infty}$ & $\left\{\sum_{j=1}^dm_j\bv_j\ \big|\ m_j\in \Z\
 (j=1,2,\cdots,d)\right\}$ & Subsection \ref{subsec_model_theorem}\\
$\G_{\infty}^*$ & $\left\{\sum_{j=1}^d\hat{k}_j\hbv_j\ \big|\ \hat{k}_j\in [0,2\pi]\
 (j=1,2,\cdots,d)\right\}$ & Subsection \ref{subsec_model_theorem}\\
$\cB$ & $\{1,2,\cdots,b\}$ & Subsection \ref{subsec_model_theorem}\\
$\Mat(n,\C)$ & set of $n\times n$ complex matrices & Subsection
 \ref{subsec_model_theorem}\\
$I_0$ & $\{1,2\}\times \cB\times \G\times [0,\beta)_h$ & Section
 \ref{sec_formulation}\\
$I$ & $I_0\times \{1,-1\}$ & Section \ref{sec_formulation}\\
$\cV$ & complex vector space spanned by $\{\psi_X\}_{X\in I}$ & Section
 \ref{sec_formulation}\\
$\bigwedge\cV$ & Grassmann algebra generated by $\{\psi_X\}_{X\in I}$ & Section \ref{sec_formulation}\\
$I^0$ & $\{1,2\}\times\cB\times\G\times\{0\}\times\{1,-1\}$ & Subsection
 \ref{subsec_necessary_notions}\\
$\bigwedge_{even}\cV$ & subspace of $\bigwedge \cV$ consisting of even
 polynomials & Subsection \ref{subsec_necessary_notions}\\
$C(\overline{D},\bigwedge_{even}\cV)$ & set of continuous maps from
 $\overline{D}$ to $\bigwedge_{even}\cV$ & Subsection
 \ref{subsec_integration_without}\\
$C^{\o}(D,\bigwedge_{even}\cV)$ & set of analytic maps from
 $D$ to $\bigwedge_{even}\cV$ & Subsection
 \ref{subsec_integration_without}
\end{tabular}
\end{center}

\subsection*{Functions and maps}
\begin{center}
\begin{tabular}{lll}
Notation & Description & Reference \\
\hline
$\sV$  & $\frac{U}{L^d}\sum_{(\rho,\bx),(\eta,\by)\in\cB\times\G}
\psi_{\rho\bx\ua}^*\psi_{\rho\bx\da}^*\psi_{\eta\by\da}\psi_{\eta\by\ua}$ & Subsection
 \ref{subsec_model_theorem}\\
$r_L$ & map from $\G_{\infty}$ to $\G$ &  Subsection
 \ref{subsec_model_theorem}\\
$E(\cdot)$ & map from $\R^d$ to $\Mat(b,\C)$, &  Subsection
 \ref{subsec_model_theorem}\\
 &  hopping matrix in momentum
 space & \\
$e(\cdot)$ & non-negative function on $\R^d$ & Subsection
 \ref{subsec_model_theorem}\\
$\sH_0$  & $\frac{1}{L^d}\sum_{(\rho,\bx),(\eta,\by)\in\cB\times\G}\sum_{\s\in
 \spin}\sum_{\bk\in \G^*}$ & Subsection
 \ref{subsec_model_theorem}\\
 & $\cdot e^{i\<\bx-\by,\bk\>}E(\bk)(\rho,\eta)
\psi_{\rho\bx\s}^*\psi_{\eta\by\s}$ & \\
$\sS_z$ &
 $\frac{1}{2}\sum_{(\rho,\bx)\in\cB\times\G}(\psi_{\rho\bx\ua}^*\psi_{\rho\bx\ua}-\psi_{\rho\bx\da}^*\psi_{\rho\bx\da})$ & Subsection \ref{subsec_model_theorem}\\
$\sF$ & $\g \sum_{(\rho,\bx)\in \cB\times\G}(\psi_{\rho\bx\ua}^*\psi_{\rho\bx\da}^*+\psi_{\rho\bx\da}\psi_{\rho\bx\ua})$ & Subsection \ref{subsec_model_theorem}\\
$G_{x,y,z}(\cdot)$ & map from $\R^d$ to $\Mat(b,\C)$ parameterized by $x,y,z$ &
 \eqref{eq_matrix_valued_notation}\\
$C(\phi)(\cdot)$ & function on $(\{1,2\}\times\cB\times
 \G_{\infty}\times [0,\beta))^2$ & Section \ref{sec_formulation}\\
  &  parameterized by $\phi$, 
 full covariance & \\
$E(\phi)(\cdot)$ & map from $\R^d$ to $\Mat(2b,\C)$ parameterized by
 $\phi$ & Section \ref{sec_formulation}\\
$\cR_{\beta}$  & map from $(\{1,2\}\times\cB\times \G\times
 \frac{1}{h}\Z)^n$ to $I_0^n$ & beginning of\\ 
 & or from $(\{1,2\}\times\cB\times \G\times
 \frac{1}{h}\Z\times \{1,-1\})^n$ to $I^n$ & Subsection
 \ref{subsec_general_estimation}
\end{tabular}
\end{center}

\subsection*{Other notations}
\begin{center}
\begin{tabular}{lll}
Notation & Description & Reference \\
\hline
$\bv_j$  & basis of $\R^d$ & Subsection
 \ref{subsec_model_theorem}\\
$(j=1,\cdots,d)$ & & \\ 
$\hbv_j$  & dual basis of $\{\bv_j\}_{j=1}^d$ & Subsection
 \ref{subsec_model_theorem}\\
$(j=1,\cdots,d)$ & & \\
$I_n$ & $n\times n$ unit matrix & beginning of \\
   & &  Subsection \ref{subsec_examples}\\
$V(\psi)$ & polynomial of $\bigwedge\cV$ consisting of quadratic part
& Section \ref{sec_formulation}\\
 &  and quartic part  & \\
$W(\psi)$ & quartic polynomial of $\bigwedge\cV$ & Section
 \ref{sec_formulation}\\
$A^1(\psi)$ & quadratic polynomial of $\bigwedge\cV$  & Section
 \ref{sec_formulation}\\
$A^2(\psi)$ & polynomial of $\bigwedge\cV$ consisting of quadratic part
& Section \ref{sec_formulation}\\
 & and quartic part & \\ 
$A(\psi)$ & $\la_1 A^1(\psi)+\la_2A^2(\psi)$ &
 \eqref{eq_Grassmann_artificial_term}\\
$V(u)(\psi)$ & same as $V(\psi)$, apart from having $u(\in\C)$ &
 beginning of \\
 &  in place
 of $U$ & Subsection \ref{subsec_integration_without}\\
$W(u)(\psi)$ & same as $W(\psi)$, apart from having $u(\in\C)$ &
 beginning of \\
 &  in place
 of $U$  & Subsection \ref{subsec_integration_without}\\
$[\cdot,\cdot]_{1,\infty,r}$ &
 $\sup_{u\in\overline{D(r)}}[f(u),g]_{1,\infty}$ & beginning of \\
 & & Subsection \ref{subsec_integration_without}\\
$[\cdot,\cdot]_{1,r,r'}$ &
 $\sup_{(u,\bla)\in\overline{D(r)}\times
 \overline{D(r')}^2}[f(u,\bla),g]_{1}$ & beginning of \\
 & & Subsection \ref{subsec_integration_with}
\end{tabular}
\end{center}

\end{document}